\documentclass[10pt]{article}
\usepackage{setspace}
\usepackage{amsmath,amssymb,amsfonts,bbm,verbatim,multirow,color,booktabs}
\usepackage{epic,eepic,psfrag,epsfig}
\usepackage[sf,bf,SF,footnotesize]{subfigure}
\usepackage{graphicx}
\usepackage{algorithm}
\usepackage{algpseudocode}
\usepackage{amsthm,breakcites}
\usepackage{amscd}
\usepackage{epsfig}
\usepackage{mathtools}
\usepackage{lipsum} 
\usepackage{tabularx}
\usepackage{graphicx}
\usepackage{array}
\usepackage{multirow}
\usepackage[title,toc,titletoc,page]{appendix}
\usepackage{rotating}
\usepackage{fullpage}
\usepackage[table]{xcolor}

\usepackage[round]{natbib}
\usepackage{setspace}
\usepackage{enumerate}
\usepackage{array}
\usepackage{authblk}
\usepackage[small]{caption}
\RequirePackage[colorlinks,citecolor=blue,linkcolor=blue,urlcolor=blue,breaklinks]{hyperref}

\usepackage[normalem]{ulem}

\newtheorem{thm}{Theorem}

\newtheorem{lemma}[thm]{Lemma}

\newtheorem{prop}[thm]{Proposition}
\newtheorem{cor}[thm]{Corollary}
\newtheorem{eg}{Example}
\newtheorem{rmk}{Remark}

\def\hat{\widehat}

\usepackage{xcolor}
\usepackage[draft,inline,nomargin,index]{fixme}
\fxsetup{theme=color,mode=multiuser}
\FXRegisterAuthor{tb}{atb}{\color{red} TB}
\FXRegisterAuthor{ab}{aab}{\color{green} AB}

\allowdisplaybreaks
\begin{document}
\openup .01em

\title{Nonparametric inference for ratios of densities via uniformly valid and powerful permutation tests}
\author{Alberto Bordino and Thomas B. Berrett}
\affil{Department of Statistics, University of Warwick, Coventry, CV47AL, United Kingdom \\ alberto.bordino@warwick.ac.uk \quad tom.berrett@warwick.ac.uk}
\date{}

\maketitle
\begin{abstract}
We propose the density ratio permutation test, a hypothesis test that assesses whether the ratio between two densities is proportional to a known function based on independent samples from each distribution. The test uses an efficient Markov Chain Monte Carlo scheme to draw weighted permutations of the pooled data, yielding exchangeable samples and finite sample validity. For power, if the statistic is an integral probability metric, our procedure is consistent under mild assumptions on the defining function class; specializing to a reproducing kernel Hilbert space, we introduce the shifted maximum mean discrepancy and prove minimax optimality of our test when a normalized difference between the densities lies in a Sobolev ball. We extend to the case of an unknown density ratio by estimating it on an independent training sample and derive type~I error bounds in terms of the estimation error as well as power results. This allows adapting our method to conditional two sample testing, making it a versatile tool for assessing covariate-shift and related assumptions, which frequently arise in transfer learning and causal inference. Finally, we validate our theoretical findings through experiments on both simulated and real-world datasets.
\end{abstract}

\begin{scriptsize}
Some key words: Conditional two-sample testing; Density ratio; Distributional shift; Minimax optimality; Permutation test.
\end{scriptsize}

\section{Introduction}
\label{Sec:Intro}

Density ratios are foundational objects in statistics. Commonly known as the importance in the importance sampling literature \citep[][]{tokdar2010importance}, they can be used to implement importance weighting techniques \citep[e.g.][]{kahn1951estimation, horvitz1952generalization, owen2000safe}, which assign varying weights to data points to correct for biases and emphasize relevant observations. These quantities remain central to contemporary statistics for several reasons. For instance, many modern analyses involve data drawn from different sources or environments \citep[][]{Koh2020WILDSAB}. When distributions differ across datasets, principled procedures are needed to transport information so that inferences target the population of interest. Density ratio tools provide exactly this mechanism, quantifying how one distribution departs from another and enabling methods that combine heterogeneous datasets for prediction and inference \citep[e.g.][]{storkey2009when, weiss2016survey}. In medical applications, practitioners might be interested in making predictions in a particular experimental setting or using specific equipment, while also utilizing data collected under different conditions or with alternative tools \citep[e.g.][]{guan2022domainadaptation}. Likewise, in the context of large language models, thoughtfully incorporating human-authored data alongside model-generated content during training helps maintain diversity and enhance overall performance \citep[e.g.][]{ji2025overviewlargelanguagemodels}. Beyond transportability, density ratios underpin a wider range of statistical problems that admit a classification formulation, where the Bayes optimal decision rule amounts to thresholding a class conditional ratio. A prominent example is simulation based inference, which refers to methodologies that leverage simulators to perform Bayesian analysis when the likelihood is intractable or unavailable  \citep[e.g.][]{Deistler2025SBIGuide}. In this context, many pipelines estimate density ratios, e.g.~via neural ratio estimators, to target posterior-to-prior or likelihood-to-evidence ratios \citep{Hermans2020LFMARE, Thomas2022LFIRE, Miller2022ContrastiveNRE, Delaunoy2022BNRE, leonte2025simulationbasedinferencetelescopingratio}, thereby enabling downstream tasks such as posterior approximation. Additional applications include nonparametric regression \citep{DRE_wainwright23}, efficient estimation with additional incomplete data \citep{berrett2024efficient}, testing under distributional shifts \citep{thams23jrrsb}, quantile function estimation \citep{quantile2013DRM}, reinforcement learning \citep[e.g.][]{sutton98barto}, and causal inference \citep{Robins2000MarginalSM}.

Although the statistical literature on density ratio estimation is well developed \citep[e.g.][]{Sugiyama2010review}, a persistent practical challenge is not only how to estimate the ratio but also whether a given estimate is adequate for the downstream task. Rather than plugging estimated ratios directly into importance weighting or transportability analyses, practitioners in many settings expect diagnostics before trusting the estimate. For example, common approaches to the estimation of causal treatment effects rely on the estimation of propensity scores, often assuming a parametric model, in order to balance covariate distributions between treatment and control groups. Best practice requires diagnostics to confirm the weights achieve balance. \citet[][Section~4]{Austin2015IPTWBestPractice} recommend assessing covariate balance using reweighted moments, graphical distributional comparisons via boxplots or empirical cumulative distribution functions, and summary metrics such as the weighted Kolmogorov–-Smirnov distance. Such diagnostics are also implemented in popular software; see, for example~\citet{zhou2022PSweight} and the references therein. In parallel, transportability analyses emphasize that suitability of weights must be evaluated, not assumed: as \citet{Westreich2017Transportability} write, ``In fact, the transportability of observational studies to a particular target population of interest is not guaranteed and must be evaluated carefully.''

Another setting where statistical tests for density ratios are required is simulation based inference, where modern procedures can yield miscalibrated posteriors leading to invalid downstream conclusions \citep{hermans2022a}. Although formal statistical tests exist for methods that target the posterior directly, as summarized in Section~3 and Appendix~A5 in \citet[][]{Deistler2025SBIGuide}, for neural ratio estimation the literature offers diagnostics rather than formal tests. The need for quantitative checks is evidenced both in the statistical literature, e.g.~\citet[][Figure~3]{Hermans2020LFMARE} and~\citet[][Figure~3]{Miller2022ContrastiveNRE}, and in real scientific workflows, such as in high energy physics \citep[]{ATLAS2025NSBI}. In particular, these works assess density ratio quality using the same diagnostics as listed above, alongside the receiver operating characteristic metric for a classifier that distinguishes the target distribution from a reweighted version of the reference distribution, trained on independent samples from each using a density-ratio-weighted empirical~loss. 

The foregoing discussion highlights the need to supplement density ratio estimation with inferential procedures that quantify uncertainty. In this paper, we focus on the foundational problem of testing
\begin{align}\label{eq:testing_problem}
    H_0:  g \propto r \, f,
\end{align}
based on independent data $(X_1,\ldots,X_n,Y_1,\ldots,Y_m) \sim P_f^{\otimes n} \otimes P_g^{\otimes m}$, where $P_f,P_g$ are distributions on a domain $\mathcal{X}$ 
with densities $f,g$ with respect to a common dominating measure, and where $r$ is a fixed function. This can be viewed as the most fundamental nonparametric hypothesis testing problem for an unknown density ratio $g/f$, where we test proportionality to a known function, analogous to goodness-of-fit testing for densities \citep[e.g.][]{ingster03} where we would test the equality of an unknown density and a fixed function. While simple, this problem itself is of interest in applications. For example, in simulation based inference the quality of a ratio estimator can be assessed using independent or held-out data; see \cite{rogozhnikov2016reweighting,ATLAS2025NSBI} for examples from the physics literature. The testing problem~\eqref{eq:testing_problem} can also be equivalently reformulated as a two-sample testing problem given biased data \citep[e.g.][]{KangNelsonDistr}, for which asymptotic results are available; see also the closely related problem of independence testing with biased data \citep{tenzer2022testing}. More generally, nonparametric inference under biased sampling is well studied, with known weight functions arising through size-biased sampling and many other practical mechanisms; see, for example, \cite{cox1969some,qin1993empirical,efromovich2004density} and the references therein. It can also be seen that~\eqref{eq:testing_problem} is equivalent to conditional two-sample testing~\citep[e.g.][]{kim2024conditional} when the ratio of covariate densities can be assumed to be known, as would be the case in a Model-X framework~\citep{candes2016panning} or randomized trials with perfect compliance~\citep{lei2021conformal}. Nonetheless, moving beyond this foundational problem, in Section~\ref{sec:family_shift} we extend our study to test hypotheses that $g/f$ lies in a restricted function class, for example parametric or depending only on a proper subset of variables. This, in turn, enables assessing a range of modelling assumptions, including~tests for distributional shifts such as covariate shift, a setting common in domain adaptation and transfer learning~\citep[e.g.][]{Koh2020WILDSAB, covariate_shift_ramsdas19, jin2023modelfree}.

Our assessment of \eqref{eq:testing_problem} is based on the novel density ratio permutation test defined in Algorithm~\ref{alg:star_algo}, which uses an appropriate resampling scheme to calibrate any chosen test statistic and provide a $p$-value. Permutation methods are widely used in hypothesis testing, due to their guaranteed validity and strong power properties in both classical asymptotic~\citep{hoeffding52permutation,lehmann2006testing} and non-asymptotic~\citep{berrett2020optimal,kim2022minimax} senses. These methods approximate the null distribution of the test statistic by uniformly relabelling observations, which in classical two-sample settings is sufficient to satisfy the randomization hypothesis, thus yielding exact type~I error control \citep{lehmann2006testing}. However, this conventional approach is not appropriate for~\eqref{eq:testing_problem} since the two samples are not exchangeable under the null hypothesis. We address this limitation by introducing a modified permutation sampling approach, adapting ideas from conditional independence testing~\citep{berrett2020conditional} to the problem~\eqref{eq:testing_problem}, which produces permutations of $(X_1, \ldots, X_n, Y_1, \ldots,  Y_m)$ according to a distribution dependent on the hypothesized density ratio $r$, restoring exchangeability of the true and resampled datasets under $H_0$ and thus ensuring finite-sample validity. When $r$ is constant, our method naturally reduces to the classical permutation test. 

This non-uniform permutation sampling approach represents a departure from standard methodology, but aligns with recent developments tackling various testing problems. The weight-dependent permutation method of \citet{KangNelsonDistr} tests $f=g$ using samples from the densities $w_1 f$ and $w_2 g$, where $w_1,w_2$ are known sampling bias functions, and is proved to have asymptotic validity and power. The conditional permutation test of \citet{berrett2020conditional} for testing $X \perp\!\!\!\!\perp Y \mid Z$ draws permutations based on the knowledge of the conditional law of~$X \mid Z$. \citet{Ramdas2022PermutationTU} studies a general framework for the use of weighted permutations to test data exchangeability. Recent advances in conformal prediction under distributional shifts modify procedures to maintain coverage with non-exchangeable data \citep{covariate_shift_ramsdas19,Hu2020conformalJASA,pmlr-v235-prinster24a}. Despite this great recent interest in non-uniform permutation methods, power results are rare. To the best of our knowledge, our work is the first to prove general consistency results and the minimax rate optimality of such a uniformly-valid testing procedure.

An alternative way to tackle~\eqref{eq:testing_problem} is to use knowledge of \(r\) to reweight the data and reduce the task to classical two-sample testing. For example, one can apply importance sampling to \((X_1,\ldots,X_n)\) to create a sample from a distribution proportional to \(r f\). This can then be compared with \((Y_1,\ldots,Y_m)\) using standard two-sample tests. This broad strategy has been employed in works such as \cite{thams23jrrsb} and \cite{schrab2024epistemic} for other statistical problems. We deliberately avoid this route because importance sampling typically yields a much smaller effective sample size, which in practice leads to less powerful tests. We validate this empirically in Section~\ref{sec:simulSynthetic}. Moreover, importance sampling requires density ratios to be bounded, whereas our methods do not require such an assumption to be valid and consistent.

Our contributions are as follows. We begin by introducing an appropriate permutation distribution and sampling algorithm, through which we define $p$-values that we prove in Theorem~\ref{prop:alg_validity} are uniformly valid for any choice of test statistic; the development here follows~\cite{berrett2020conditional} and our contribution is in adapting this non-uniform permutation methodology and type~I error theory to our problem. After laying these foundations we turn to our main technical novelty, which is in the construction of appropriate test statistics and the development of power results. Our new insights into the behaviour of our sampling algorithm under the alternative hypothesis lead to the introduction of a class of integral probability metrics and their sample versions as test statistics. We prove in Theorem~\ref{thm:general_consistency} that the associated hypothesis tests are consistent in wide generality, even when the sampling algorithm is run for a finite amount of time and $r$ is unbounded. In order to provide computable test statistics we specialize our class of integral probability metrics to those based on reproducing kernel Hilbert spaces, thus generalizing the maximum mean discrepancy \citep{JMLR:v13:gretton12a}. In the canonical special case of continuous data under Sobolev smoothness conditions, we establish the minimax optimal separation rates for our problem; this is based on a detailed power analysis of our new methodology and novel minimax lower bounds. We further discuss how our methodology and power results extend to testing larger null hypotheses where $r$ is an unknown member of a model.

We conclude this section with some notation that is used throughout the paper. We denote by $\mathbb{R}_+$ the set of positive real numbers and by $\mathbb{R}_{\geq 1}$ the set of real numbers greater than or equal to one. We further set $\mathbb{N}_+ = \mathbb{N}\setminus\{0\}$. The symmetric group of all permutations over the set $[n] := \{1, \ldots, n\}$ is denoted by $\mathcal{S}_{n}$. For a set $A$, we write $\#A$ to denote its cardinality. The maximum and minimum operators are sometimes denoted by $\vee$ and $\wedge$, respectively. We define the support of a function $f$ defined on a domain $\mathcal{X}$ to be the closure of the set where $f$ does not vanish, that is, $\operatorname{supp} f := \overline{\{x \in \mathcal{X} : f(x) \neq 0\}}$. The set of bounded and continuous functions on $\mathcal{X}$ is denoted by $\mathcal{C}^0_b(\mathcal{X})$, and $\delta_x$ denotes the Dirac delta measure centred at $x$. Given $1 \leq p < \infty$ and $d \geq 1$, we define the $L^p$ norm of a function $f$ as $\|f\|_p := \left( \int_{\mathbb{R}^d} |f(x)|^p dx \right)^{1/p}$, and the corresponding $L^p(\mathbb{R}^d)$ space as the set of all measurable functions for which this norm is finite. Furthermore, $\|\cdot\|_\infty$ denotes the essential supremum norm, that is, $\|f\|_\infty := \operatorname{ess\,sup}_{x \in \mathbb{R}^d} |f(x)|$, and $L^\infty(\mathbb{R}^d)$ refers to the set of functions that are bounded almost everywhere. Finally, we use $C^{\infty}(\mathcal{X})$ to denote the space of infinitely differentiable functions on a domain $\mathcal{X}$, and $\overline{z}$ for the complex conjugate of $z \in \mathbb{C}$.

\section{Permutation methodology and validity}\label{sec:methodology}
In this section we detail our general permutation procedure and its validity guarantees, which are largely adapted to our setting from the conditional independence testing problem studied in~\cite{berrett2020conditional}. Suppose we observe $Z =  (Z_1, \ldots, Z_{n+m}) = (X_1, \ldots, X_n, Y_1, \ldots, Y_m) \in \mathcal{X}^{n+m}$, consisting of $n + m$ independent random variables, where $X_1, \ldots, X_n \overset{\mathrm{i.i.d.}}{\sim} P_f$ and $Y_1, \ldots, Y_m \overset{\mathrm{i.i.d.}}{\sim} P_g$. Assume further that $P_f$ and $P_g$ are absolutely continuous with respect to a common dominating measure $\mu$ with densities $f,g$. We propose a refined permutation test designed to detect deviations from the null hypothesis~\eqref{eq:testing_problem} where $r : \mathcal{X} \rightarrow \mathbb{R}_+$ is a fixed, unnormalized function. By the definition of $\mathbb{R}_+$ given above, we have $r(x) > 0$ for all $x \in \mathcal{X}$. For simplicity, we also assume that the supports of $f$ and $g$ are the same. If this is not the case, i.e.~if $r(x) > 0$ for all $x \in \mathcal{X}$ but $\operatorname{supp} g \ne \operatorname{supp} f$, then the testing problem becomes easier, as we expect to eventually observe signals in the regions corresponding to the symmetric difference of the supports. Similarly, if $\operatorname{supp} r \ne \mathcal{X}$ but $\operatorname{supp} g \subseteq \operatorname{supp} r$, we can restrict the sample space to $\operatorname{supp} r$ and carry out the analysis as before. If this inclusion does not hold, the testing problem again becomes easier. Finally, we shall also require that $\int g/r \, d\mu < \infty$ and $\int rf \, d\mu < \infty$, which ensure that the null hypothesis is nontrivial.

In line with permutation testing practice, the methodology we propose is to create copies $Z^{(1)}, \ldots, Z^{(H)}$, with $H \geq 1$, that are exchangeable with $Z$ under $H_0$, and define a $p$-value by
\begin{equation}\label{eq:pvalue_DRPT}
    p  = \frac{1+ \sum_{h = 1}^H \mathbbm{1}\{T(Z^{(h)}) \geq T(Z) \}}{1+H},
\end{equation}
where $T:\mathcal{X}^{n+m} \to \mathbb{R}$ is an arbitrary test statistic. Classically, these copies are uniformly shuffled versions of~$Z$, but this is not appropriate here because the two samples are not exchangeable under the null. Intuitively, under the null hypothesis $g \propto r f$, points with larger $r(x)$ are more likely to originate from $g$ than from~$f$, hence we should favour permutations that assign observations with large $r(Z_i)$ to the last $m$ positions. Based on this, we set 
\begin{equation}\label{eq:sampling_1}
   Z^{(h)} = Z_{\sigma^{(h)}} \qquad\text{with}\qquad 
   \mathbb{P}\!\left\{\sigma^{(h)}=\sigma \mid Z \right\}
   = \frac{\prod_{i \in \{n+1,\ldots,n+m\}} r\!\big(Z_{\sigma(i)}\big)}
          {\sum_{\tilde{\sigma}\in\mathcal S_{n+m}} \prod_{i \in \{n+1,\ldots,n+m\}} r\!\big(Z_{\tilde{\sigma}(i)}\big)},
\end{equation}
where, for any vector $x = (x_1, \ldots, x_{n+m})$ and permutation $\sigma \in \mathcal S_{n+m}$, we define $x_\sigma := (x_{\sigma(1)}, \ldots, x_{\sigma(n+m)})$. As we see in the proof of Theorem~\ref{prop:DRPT_validity} below, if we condition on the unordered set of observed data values, then the distribution of $Z$ is the same as the distribution of $Z^{(h)}$ under $H_0$. This leads to the exchangeability of $(Z,Z^{(1)},\ldots,Z^{(H)})$ and the validity of our tests, as summarized in the following result.

\begin{thm}\label{prop:DRPT_validity}
Assume that $H_0: g \propto r \, f$ is true. Suppose that $\sigma^{(1)}, \ldots, \sigma^{(H)}$ are drawn i.i.d.~from \eqref{eq:sampling_1} and write $Z^{(h)} = Z_{\sigma^{(h)}}$ for all $h \in [H]$. Then the sequence $\left(Z, Z^{(1)}, \ldots, Z^{(H)} \right)$ is exchangeable. In particular, for any statistic $T: \mathcal{X}^{n+m} \to \mathbb{R}$, the $p$-value defined in~\eqref{eq:pvalue_DRPT} is valid, satisfying $\mathbb{P}\{p \leq \alpha\} \leq \alpha$ for all $\alpha \in [0,1]$.
\end{thm}

In practice, to conduct a statistical test we need to be able to sample permutations $\sigma^{(1)}, \ldots, \sigma^{(H)}$ from~\eqref{eq:sampling_1}. We now address the challenge of generating these samples efficiently through Algorithm \ref{alg:pairwise_sampler}. 

\begin{algorithm}[!ht]
\caption{Pairwise sampler}\label{alg:pairwise_sampler}
\begin{algorithmic}[1]
\State Initial permutation $\sigma_0$, integer $S \geq 1$. Call $K = \min\{n,m\}$.
\For{$t \in [S]$} 
\State \parbox[t]{\dimexpr\linewidth-\algorithmicindent}{Sample a vector of couples $\tau_t = \{(i_1^{t}, j_1^{t}), \ldots, (i_K^{t}, j_K^{t}) \}$ such that $(i_1^{t}, \ldots, i_K^{t})$ are sampled uniformly and without replacement from $[n]$, and $(j_1^{t}, \ldots, j_K^{t})$ are sampled uniformly and without replacement from $\{n+1, \ldots, n+m \}$, and initialize $\sigma_t$ to be a copy of $\sigma_{t-1}$.}
\For{$k \in [K]$}
\State \parbox[t]{\dimexpr\linewidth-\algorithmicindent}{Draw a Bernoulli random variable $B_{i_k,j_k}^t$ with 
\begin{equation}\label{eq:p_no_tilde}
\mathbb{P}\{B_{i_k,j_k}^t = 1\} = \frac{r\bigl(Z_{\sigma_{t-1}(i_k^t)}\bigr)}{r\bigl(Z_{\sigma_{t-1}(i_k^t)}\bigr) + r\bigl(Z_{\sigma_{t-1}(j_k^t)}\bigr)} \;:=\; p_{i_k, j_k}^t,
\end{equation}
and swap $\sigma_{t}(i_k^t)$ with $\sigma_{t}(j_k^t)$ if $B_{i_k,j_k}^t = 1$.}
\EndFor
\EndFor
\State \Return $\sigma_S$.
\end{algorithmic}
\end{algorithm}

Algorithm \ref{alg:pairwise_sampler} is easily parallelizable due to the disjoint structure of the pairs in $\tau_t$, and, most importantly, it accurately targets the distribution in \eqref{eq:sampling_1}, guaranteeing that the resulting Markov chain converges to the desired stationary distribution, as established in the following proposition.

\begin{prop}\label{prop:algo_sampler_valid}
For any initial permutation $\sigma_0$, the distribution \eqref{eq:sampling_1} of the permutation $\sigma$ conditional on~$Z$ is the unique stationary distribution of the Markov chain defined in Algorithm \ref{alg:pairwise_sampler}.
\end{prop}

\begin{rmk}\label{rmmk:hat_lambda_alg}
By examining the proof of Proposition~\ref{prop:algo_sampler_valid} in Appendix~\ref{appendix:proof2}, we see that other proposal mechanisms can also target the invariant distribution~\eqref{eq:sampling_1}. Specifically, Step 5 in Algorithm~\ref{alg:pairwise_sampler} switches the indices $i_k^t$ and~$j_k^t$ with probability $p_{i_k, j_k}^t$ as in~\eqref{eq:p_no_tilde}, but the key property used in the proof is $\mathbb{P}\{B_{i_k,j_k}^t = 1\}/\mathbb{P}\{B_{i_k,j_k}^t~=~0\} = r\bigl(Z_{\sigma_{t-1}(i_k^t)}\bigr)/r\bigl(Z_{\sigma_{t-1}(j_k^t)}\bigr)$. Thus, any alternative distribution on $B_{i_k,j_k}^t$ satisfying this ratio would work equally well. For example, for the later power analysis of the methodology, it is more convenient to consider
\begin{align}\label{eq:tilde_p}
    \tilde{p}_{i_k, j_k}^t := \frac{\hat{\lambda}mn \, r\bigl(Z_{\sigma_{t-1}(i_k^t)}\bigr)}
{\bigl\{n + \hat{\lambda}mr\bigl(Z_{\sigma_{t-1}(i_k^t)}\bigr)\bigr\}\bigl\{n + \hat{\lambda}mr\bigl(Z_{\sigma_{t-1}(j_k^t)}\bigr)\bigr\}}
\end{align}
with normalizing constant $\hat{\lambda}$ such that $\sum_{i = 1}^{n+m}\{n + \hat{\lambda}m r(Z_i)\}^{-1} = 1$.
\end{rmk}

\begin{rmk}\label{rmk:reciprocal}
Testing $H_0: g \propto rf$ is equivalent to testing $H_0': f \propto \tfrac{1}{r} \, g$. We now show that our method is similarly invariant under taking reciprocals of $r$ and relabelling the samples, using either $p_{i_k, j_k}^t$ or $\tilde{p}_{i_k, j_k}^t$. For~$p_{i_k, j_k}^t$, observe that the appropriate quantity after relabelling is
\[
(p_{i_k, j_k}^t)^\prime \;:=\; \frac{1/r\bigl(Z_{\sigma_{t-1}(j_k^t)}\bigr)}
{1/r\bigl(Z_{\sigma_{t-1}(i_k^t)}\bigr) + 1/r\bigl(Z_{\sigma_{t-1}(j_k^t)}\bigr)}
\;=\;
\frac{r\bigl(Z_{\sigma_{t-1}(i_k^t)}\bigr)}{r\bigl(Z_{\sigma_{t-1}(i_k^t)}\bigr) + r\bigl(Z_{\sigma_{t-1}(j_k^t)}\bigr)},
\]
which matches $p_{i_k, j_k}^t$ exactly. In the case of $\tilde{p}_{i_k, j_k}^t$, define
\[
(\tilde{p}_{i_k, j_k}^t)^\prime 
\;:=\; 
\frac{\hat{\theta}  nm /r\bigl(Z_{\sigma_{t-1}(j_k^t)}\bigr)}
{\bigl\{m + \hat{\theta}  n/r\bigl(Z_{\sigma_{t-1}(i_k^t)}\bigr)\bigr\}\bigl\{m + \hat{\theta}  n/r\bigl(Z_{\sigma_{t-1}(j_k^t)}\bigr)\bigr\}},
\]
with $\hat{\theta}$ such that $\sum_{i = 1}^{n+m}\{m + n \hat{\theta} / r(Z_i)\}^{-1} = 1$. Here, $n$ and $m$ are switched because the roles of the $X$'s and $Y$'s are interchanged. By algebraic manipulation, one finds that $\hat{\theta} = \hat{\lambda}^{-1}$. Substituting back into $(\tilde{p}_{i_k, j_k}^t)^\prime$ shows $(\tilde{p}_{i_k, j_k}^t)^\prime = \tilde{p}_{i_k, j_k}^t$, confirming the invariance of our algorithm under reciprocal transformations of~$r$.
\end{rmk}

By Proposition \ref{prop:algo_sampler_valid}, Algorithm \ref{alg:pairwise_sampler}, when executed for sufficiently many steps $S$, generates a copy $Z_{\sigma_S}$ which acts as an appropriate control for $Z$ in testing $H_0$. In fact, we can make a much stronger statement. Under the null hypothesis, conditionally on the unordered set of observed data values, the distribution of $Z$ coincides with that of $Z_\sigma$ when $\sigma$ is sampled from~\eqref{eq:sampling_1}. This implies that $Z_{\sigma_S}$ represents a draw from the exact target distribution for any value of $S$, thereby providing a valid control for $Z$ regardless of the number of steps taken. However, exact type I error control relies on the stronger requirement of the exchangeability of $(Z,Z^{(1)},\ldots,Z^{(H)})$. This can be achieved through a star-shaped sampling scheme, following the methodology introduced by \cite{besag89star_sampler} and subsequently applied in permutation-based approaches by \cite{berrett2020conditional} and \cite{Ramdas2022PermutationTU}. This results in the general testing procedure of Algorithm~\ref{alg:star_algo}, which is the basis of our later methodological and theoretical contributions.

\begin{algorithm}[!htbp]
\caption{Density ratio permutation test}\label{alg:star_algo}
\begin{algorithmic}[1]
\State Data $Z =  (Z_1, \ldots, Z_{n+m}) = (X_1, \ldots, X_n, Y_1, \ldots, Y_m)$, initial permutation $\sigma_0$, integers $S \geq 1$, $H \geq 1$, and test statistic $T$.
\State Let $\sigma_*$ be the output of Algorithm~\ref{alg:pairwise_sampler} after $S$ steps, initialized at $\sigma_0$.
\For{$h \in [H]$, independently} 
\State Let $\sigma^{(h)}$ be the output of Algorithm~\ref{alg:pairwise_sampler} after $S$ steps, initialized at $\sigma_*$.
\EndFor
\State Compute $Z^{(h)} := Z_{\sigma^{(h)}}$ for all $h \in [H]$.
\State \Return $p := (1+H)^{-1}(1+ \sum_{h = 1}^H \mathbbm{1}\{T(Z^{(h)}) \geq T(Z)\})$.

\end{algorithmic}
\end{algorithm}

Algorithm \ref{alg:star_algo}, when initialized with $\sigma_0 = \mathrm{id}_{n+m}$, provides an exchangeable sampling mechanism, since the permutation $\sigma_*$ lies $S$ steps away from each of the permutations $(\sigma_0, \sigma^{(1)}, \ldots, \sigma^{(H)})$ and the Markov chain associated to Algorithm~\ref{alg:pairwise_sampler} is reversible, as shown in the proof of Proposition~\ref{prop:algo_sampler_valid}. The following result verifies exchangeability and ensures that the results of Theorem~\ref{prop:DRPT_validity} remain satisfied when the permuted vectors $Z^{(1)}, \ldots, Z^{(H)}$ are obtained via Algorithm~\ref{alg:star_algo}, thereby showing that the density ratio permutation test is~valid.

\begin{thm}\label{prop:alg_validity}
Let $p$ be the output of Algorithm \ref{alg:star_algo} when initialized at $\sigma_0 = \mathrm{id}_{n+m}$. If $H_0: g \propto r \, f$ is true, then $\mathbb{P}\{p \leq \alpha\} \leq \alpha$ for all $\alpha \in [0,1]$.
\end{thm} 

Finally, although the permutation distribution is generally complex and Algorithm~\ref{alg:pairwise_sampler} is required to generate a sample from it, in Appendix~\ref{appenidix:discreteDRPT} we explicitly characterize it using Fisher's noncentral hypergeometric distribution in the case of discrete data with finite support, significantly reducing computational costs.

\section{Integral probability metrics and consistency theory}\label{sec:IPM+consistency}
\subsection{Integral probability metrics}\label{sec:integral probability metric}

We begin by establishing the  consistency of our test under mild assumptions. A sequence of tests is consistent against a given class of alternatives if, as the sample size tends to infinity, the test rejects the null with probability~1 under every fixed alternative.  In the classical $r \equiv 1$ case, this follows from \cite{hoeffding52permutation}: 
letting~$\sigma$ be uniformly distributed over the symmetric group and assuming for simplicity $n/m \to \tau>0$, we have, for all $\varphi \in C_b^0(\mathcal{X})$,
\begin{equation}\label{eq:stdPermutation_is_consistent}
    n^{-1} \sum_{i=1}^n \varphi\, \!\bigl(Z_{\sigma(i)}\bigr) 
    \xrightarrow{\mathbb{P}} \int \varphi\, h_1 \, d\mu 
    \quad \text{and} \quad 
    m^{-1} \sum_{j=n+1}^{n+m} \varphi\, \!\bigl(Z_{\sigma(j)}\bigr) 
    \xrightarrow{\mathbb{P}} \int \varphi\, h_1 \, d\mu,
\end{equation}
where $C_b^0(\mathcal{X})$ denotes the space of bounded continuous functions on $\mathcal{X}$, and $h_1 = \tau(1+\tau)^{-1}\, f + (1+\tau)^{-1}\, g$. This means that the empirical distributions of both the first $n$ and the last $m$ permuted samples converge to the same mixture of $f$ and $g$, so that $Z_\sigma$ asymptotically satisfies the null. As a result, if we base the $p$-value~\eqref{eq:pvalue_DRPT} on an appropriate non-negative test statistic that is approximately zero only under the null, we shall see that $p \approx 0$ under every alternative, thus yielding consistency.

We will now extend this for general density ratios $r$. To this end, we first show that there exists a unique density $h$ which satisfies the null hypothesis in finite samples while preserving the distribution of the combined data.

\begin{lemma}\label{lemma:unique_h}
There exists a unique density $h$ such that \[
\frac{n}{n+m} h + \frac{m}{n+m}\frac{r h}{\int r h  d\mu} = \frac{n}{n+m} f + \frac{m}{n+m}g, 
\]
and it is of the form $ h = (nf + mg)/(n + \lambda_0 m r)$ for a suitable constant $\lambda_0 > 0$ such that $\int h d\mu = 1$.
\end{lemma}

We will use $\lambda_0$ to refer to such a normalizing constant from now on. The previous result, together with the fact that $\lambda_0 = (\int r h d\mu)^{-1}$, derived in the proof of Lemma~\ref{lemma:unique_h} in Appendix~\ref{appendix:proof3.1}, implies that $H_0$ is equivalent to $f = h$ almost surely. This motivates the introduction, at the population level, of an integral probability metric of the form 
\[
T_{\mathcal{F},r}(f,g) := \sup_{\varphi \in \mathcal{F}} \left| \int \frac{n + m}{n + \lambda_0 mr} \, (\lambda_0 r f - g) \, \varphi \, d\mu \right| =  \frac{n+m}{m} \, \sup_{\varphi \in \mathcal{F}} \left| \int \left(f-h \right) \, \varphi \, d\mu \right|, 
\]
for a suitable function class $\mathcal{F}$. It follows from Lemma~\ref{lemma:unique_h} that we have $T_{\mathcal{F},r}(f,g) = 0$ under $H_0$, but in order to have the reverse implication for a full characterization of the null we need additional assumptions on $\mathcal{F}$. Following similar lines to \cite{JMLR:v13:gretton12a}, we prove the following lemma. 
\begin{lemma}\label{lemma:TFcharacterisesNull}
    Let $\mathcal{F} = \{\varphi \in \mathcal{H} : \|\varphi \|_\mathcal{H} \leq 1\}$, where $\mathcal{H}$ is a dense subset of $C^0_b(\mathcal{X})$ with respect to $\| \cdot \|_\infty$, and $\| \cdot \|_\mathcal{H}$ is a norm on $\mathcal{H}$. Then $T_{\mathcal{F},r}$ characterizes $H_0$, meaning that $T_{\mathcal{F},r} = 0$ if and only if $H_0$ is true.
\end{lemma}
Although this population measure of discrepancy depends on the sample sizes, we show in Proposition~\ref{prop:limitingTF} in Appendix~\ref{app:additionalResults} that $T_{\mathcal{F},r}(f,g) \to \sup_{\varphi \in \mathcal{F}}|\int \{f - (g/r)(\int g/r d\mu)^{-1} \} \, \varphi \, d\mu|$ when $n/m \to 0$ and  $T_{\mathcal{F},r}(f,g) \to \sup_{\varphi \in \mathcal{F}}|\int \{(rf)(\int rf d\mu)^{-1} - g \} \, \varphi \, d\mu|$ when $n/m \to \infty$. This means that $T_{\mathcal{F},r}(f,g)$ is of constant order at fixed alternatives, even with imbalanced $n,m$. While it would be possible to use either of these simple limiting values as a population measure of discrepancy, we find $T_{\mathcal{F},r}$ to be a natural choice due to its invariance as in Remark~\ref{rmk:reciprocal} to relabelling of the samples; in~Proposition~\ref{prop:invPsi} in Appendix~\ref{app:additionalResults} we prove that $\psi_r := (n+m) m^{-1} \, (f - h)$ is unchanged under this transformation. Furthermore, from a technical perspective, it allows a more elegant and straightforward analysis of Algorithm~\ref{alg:pairwise_sampler}, which in turn leads to the consistency of the test given in Algorithm~\ref{alg:star_algo} based on the empirical version~\eqref{eq:empricalMMD} of $T_{\mathcal{F},r}$, as shown in the next section.

\subsection{General consistency theory}\label{sec:consistency}

In the following, we will denote by $N(A, \delta, \| \cdot \|_*)$ the $\delta$-covering number \citep[e.g.][]{wainwright2019high}  of the set~$A$ with respect to the norm $\| \cdot \|_*$.

\begin{thm}\label{thm:general_consistency}
    Fix $\alpha \in (0,1)$ and $H > \lceil 1/\alpha - 1\rceil$. Let $\mathcal{H}$ be a dense subset of $C^0_b(\mathcal{X})$  with respect to $\| \cdot \|_\infty$, and suppose there exists fixed $\gamma > 0$ such that $\| \cdot \|_\infty \leq \gamma \| \cdot \|_\mathcal{H}$. Further, suppose $N\left(\{\|\varphi\|_\mathcal{H} \leq 1\}, \delta, \|\cdot \|_\infty \right)$ is finite for all $\delta > 0$ and define
    \begin{equation}\label{eq:empricalMMD}
        T(z_1, \ldots, z_{n+m}) := \sup_{\|\varphi\|_\mathcal{H} \leq 1} \frac{n+m}{nm} \left|\sum_{i = 1}^{n} \frac{\hat{\lambda} m r(z_i)}{n + \hat{\lambda} m r(z_i)}\varphi(z_i)  - \sum_{j = n+1}^{n+m} \frac{n}{n + \hat{\lambda} m r(z_j)}\varphi(z_j)\right|,
    \end{equation}
    where $\hat{\lambda}$ satisfies $\sum_{i = 1}^{n+m}\{n + \hat{\lambda}m r(z_i)\}^{-1} = 1$. Fix $f,g$ such that $H_0$ is not true and, letting $W \sim n(n+m)^{-1} f + m(n+m)^{-1} g$, suppose that we have $\mathbb{E}[\left\{  r(W) \right\}^2 ] \,\mathbb{E}[\left\{  r(W) \right\}^{-2}]  \leq V_0$
    for some fixed $V_0 \geq 1$. Further suppose that, writing $\tau=\max(m/n,n/m)$, we have $S/\{\tau \log(m+n)\} \rightarrow \infty$. Then  the test associated to the $p$-value returned by Algorithm~\ref{alg:star_algo} with parameters $S, H$ and test statistic~\eqref{eq:empricalMMD} is consistent, i.e.~for all $\alpha \in (0,1)$ we have $\mathbb{P}\{p \leq \alpha \} \to 1$ as $n, m \to \infty$.
\end{thm}

Here $\hat \lambda$ acts as an estimator of $\lambda_0$; theoretical results quantifying the estimation error can be found in~the proof of this theorem and in Lemma~\ref{lemma:hat_lambda} in Appendix~\ref{appendix:proof3.1}. Note that $\sum_{i = 1}^{n+m}\{n + \hat{\lambda}m r(Z_i)\}^{-1} = 1$ is equivalent~to 
\begin{equation}
\label{Eq:EqualSumsLambda}
\sum_{i = 1}^n \frac{\hat{\lambda}m r(Z_{\sigma(i)})}{n + \hat{\lambda}m r(Z_{\sigma(i)})} = \sum_{j = n+1}^{n+m} \frac{n}{n + \hat{\lambda}m r(Z_{\sigma(j)})}
\end{equation}
for all permutations $\sigma \in \mathcal{S}_{n+m}$; this fact is useful throughout our proofs and allows us to see $\hat{\lambda}$ as a normalization factor for $r$ in any division of the empirical distribution $\hat{H}_{n,m}$ into two samples. As for the second moment condition on $r$, this is satisfied with $V_0 = C^2/c^2$ if $c \leq r(x) \leq C$ for all $x \in \mathcal{X}$.  However, we require something significantly~weaker. To illustrate this point, let $\phi(\cdot)$ be the density of the standard Gaussian distribution and take $f(x) = \phi(x)$ and $g(x) = \phi(x - \mu)$ with $\mu > 0$. Then, for $r(x) = g(x)/f(x) = \exp\left\{\mu x - \mu^2/2 \right\}$ we have $\lambda_0 = 1$ and $\mathbb{E}\left\{ r^2(Z) \right\} \, \mathbb{E}\left\{ 1/r^2(Z) \right\} = (n+m)^{-2} (ne^{\mu^2} + m e^{3\mu^2})(ne^{3\mu^2} + m e^{\mu^2}) \leq e^{6 \mu^2}$, so that the assumption holds by taking $V_0 =  e^{6 \mu^2}$. Also, the assumption $\| \cdot \|_\infty \leq \gamma \| \cdot \|_\mathcal{H}$ is not overly restrictive and it is satisfied, e.g.~when $\mathcal{H}$ is a reproducing kernel Hilbert space based on a uniformly bounded kernel. 

The proof of Theorem \ref{thm:general_consistency} is based on the following~lemma, which gives an approximation to the empirical distribution of our permuted data.
\begin{lemma}
\label{lemma:convergence_in_d}
Fix $S \in \mathbb{N}$ and let $\sigma_S$ be the result of running Algorithm~\ref{alg:pairwise_sampler}, with acceptance probabilities given in~\eqref{eq:tilde_p}, for $S$ steps. Assume that $\max(m,n) \geq 3$ and that we have
\begin{align}\label{eq:boundedness_kappa}
        \biggl\{ \frac{1}{n+m} \sum_{i=1}^{n+m} r(Z_i) \biggr\}\biggl\{ \frac{1}{n+m} \sum_{i=1}^{n+m} \frac{1}{r(Z_i)} \biggr\} \leq \kappa
\end{align}
for some $\kappa \geq 1$. Further, define the empirical measure $\hat{H}_{n,m} = \sum_{i=1}^{n+m} \{n + \hat{\lambda} m r(Z_i)\}^{-1} \delta_{Z_i}$, where $\hat{\lambda} > 0$ is chosen such that $\sum_{i = 1}^{n+m}\{n + \hat{\lambda}m r(Z_i)\}^{-1} = 1$. Writing $\tau=\max(m/n,n/m)$, for all $\varphi \in \mathcal{C}^0_b(\mathcal{X})$ we have
\begin{equation}
\label{Eq:LemmaRHS}
    \mathbb{E}\left[\left(\frac{1}{n} \sum_{i = 1}^n \varphi(Z_{\sigma_S(i)}) - \int \varphi d\hat{H}_{n,m}\right)^2 \biggm| Z \right] \leq \frac{64 \kappa \|\varphi\|_\infty^2 \min(m,n)}{n^2} + \|\varphi\|_\infty^2 \exp\biggl( - \frac{S}{8\tau \kappa} \biggr).
\end{equation}
\end{lemma}

\begin{rmk}\label{lemma:runningChainLongerEnough}
We see that, provided $S \gtrsim \tau \kappa \log n$, the approximation error due to running the chain only for a finite number of steps is of the same order as the error we would have if the chain had been run for an infinite number of steps. This notion of the convergence of the Markov chain is sufficient for our testing context and allows us to prove that tests that run in finite time have high power. 
\end{rmk}

Besides being crucial in proving consistency, Lemma~\ref{lemma:convergence_in_d} shows us that the empirical distribution of the first $n$ permuted samples converges in distribution to the limiting version of $h d\mu$ defined in Lemma~\ref{lemma:unique_h}. This extends~\eqref{eq:stdPermutation_is_consistent} for general density ratios~$r$ satisfying the moment assumptions of Theorem~\ref{thm:general_consistency}. Finally, it is instructive to elaborate on the proof of this result, as the underlying strategy is novel and serves as a key component in the proof of other results, such as Theorem~\ref{thm:UB_minimax} below. To this end, let $S_n^t:=\sum_{i = 1}^n \varphi(Z_{\sigma_t(i)}) - n\int \varphi d\hat{H}_{n,m}$, where $\sigma_t$ is the permutation at time $t$ generated by Algorithm~\ref{alg:pairwise_sampler}. The proof strategy is to show that this can be seen as Lyapunov function~\citep[e.g.][]{hairer2021convergence} satisfying the geometric drift condition $\mathbb{E}[(n^{-1} S_n^{t+1})^2 \mid Z, \sigma_t] - (n^{-1} S_n^{t})^2 \leq A - B(n^{-1}S_n^t)^2$ for appropriate constants $A,B$. In proving this bound it is very convenient to use acceptance probabilities $\tilde{p}^t_{i_k, j_k}$ as defined in~$\eqref{eq:tilde_p}$. Given this drift condition the result follows using standard machinery, with the first term on the right-hand side of~\eqref{Eq:LemmaRHS} corresponding to $A/B$ and the second term capturing the convergence of the chain to stationarity.

\section{Shifted maximum mean discrepancy and minimax optimality}\label{sec:sMMD+optimality}

\subsection{Shifted maximum mean discrepancy}\label{sec:shifted maximum mean discrepancy}

While Theorem \ref{thm:general_consistency} holds under mild assumptions, its practical utility depends on whether the supremum can be explicitly computed and a closed-form expression for $T_{\mathcal{F},r}$ can be derived. To make this feasible, we focus on the specific case where $\mathcal{H}$ is a reproducing kernel Hilbert space. In this setting, we establish the following representation for the integral probability metric $T_{\mathcal{F},r}$.
\begin{prop}\label{prop:TF_explicit}
    Let $k(\cdot, \cdot)$ be a measurable kernel satisfying $\int_{\mathcal{X}} \sqrt{k(x, x)} \frac{(n+m)(\lambda_0 r(x) f(x) + g(x))}{n + m\lambda_0 r(x)} d \mu(x) < \infty$ and let~$\mathcal{H}$ be the associated reproducing kernel Hilbert space.   \begin{align*}
         T_{\mathcal{F},r}^2&(f,g) = \left(\sup_{\|\varphi\|_{\mathcal{H}} \leq 1} \left| \int \frac{n+m}{n + \lambda_0 m r} (\lambda_0 r f - g) \varphi \, d\mu \right| \right)^2 \\
        & = \mathbb{E}_{X,X' \sim f}\left[\frac{\lambda_0^2 (n+m)^2 r(X)r(X') k(X,X')}{\{n + \lambda_0 m r(X) \} \{n  + \lambda_0 m r(X')\}} \right] + \mathbb{E}_{Y,Y' \sim g}\left[\frac{(n+m)^2 k(Y,Y')}{\{n + \lambda_0m r(Y)  \} \{n + \lambda_0 m r(Y')\}} \right]  \\
        & - 2 \,  \mathbb{E}_{\substack{X \sim f \\ Y \sim g}}\left[\frac{\lambda_0 (n+m)^2 r(X) k(X,Y)}{\{n  + \lambda_0 m r(X) \} \{n + \lambda_0 m r(Y)\}} \right].
    \end{align*}
\end{prop}
First, when $r \equiv 1$ we have $\lambda_0 = 1$, leading to $T_{\mathcal{F},r \equiv 1}^2(f,g) = \operatorname{MMD}^2_k(f,g)$, where $\operatorname{MMD}_k(\cdot, \cdot)$ denotes the maximum mean discrepancy based on the kernel $k(\cdot, \cdot)$;  see Equation~(1) and Lemma~6 in \citet{JMLR:v13:gretton12a}. Consequently, our integral probability metric $T_{\mathcal{F},r}$, which we will now refer to as the shifted maximum mean discrepancy and denote by $\operatorname{MMD}_{r, k}(f,g)$, generalizes the classical one to account for the presence of a shift factor, while reducing to it in the case of a constant~$r$. A similar quantity to $\operatorname{MMD}_{r, k}(f,g)$, but without the denominators and $\lambda_0$, was proposed in \cite{kim2024conditional} for conditional two-sample testing, where its definition was motivated by importance weighting. In contrast, we naturally derive $\operatorname{MMD}_{r, k}(f,g)$ from Lemma \ref{lemma:unique_h}, and see that the presence of the denominator plays a crucial role, as it is linked to the invariance under relabelling the samples in a similar spirit to Remark~\ref{rmk:reciprocal}.

To ensure that $\operatorname{MMD}_{r, k}(f,g)$ characterizes the null hypothesis, in line with Lemma \ref{lemma:TFcharacterisesNull}, restrictions on $\mathcal{H}$ must ensure that $\operatorname{MMD}_{r, k}(f,g) = 0$ if and only if $H_0$ holds. This requirement has been extensively studied in the statistical literature, and reproducing kernel Hilbert spaces satisfying this property are referred to as characteristic, with their defining kernels termed universal. For a comprehensive discussion see \cite{Sriperumbudur2010} and references therein. Notably, reproducing kernel Hilbert spaces induced by Gaussian and Laplacian kernels on $\mathbb{R}^d$ are characteristic. 

Turning now to sample versions of the shifted maximum mean discrepancy, applying the same arguments as in Proposition \ref{prop:TF_explicit} to empirical measures shows that the test statistic in \eqref{eq:empricalMMD} becomes\begin{align}\label{eq:V-stat}
    &V(x_1, \ldots, x_n, y_1, \ldots, y_m) = \frac{(n+m)^2}{n^2m^2} \sum_{i, j = 1}^n \frac{\hat{\lambda}^2 r(x_i)r(x_j)k(x_i,x_j)}{\{n/m + \hat{\lambda} r(x_i) \} \{n/m + \hat{\lambda} r(x_j)\}}  \nonumber \\
    & \quad + \frac{(n+m)^2}{m^4} \sum_{i, j = 1}^m \frac{ k(y_i,y_j)}{\{n/m + \hat{\lambda} r(y_i) \} \{n/m + \hat{\lambda} r(y_j)\}} - \frac{2 (n+m)^2}{nm^3} \sum_{i = 1}^n \sum_{j = 1}^m \frac{\hat{\lambda} r(x_i) k(x_i,y_j)}{\{n/m + \hat{\lambda} r(x_i) \} \{n/m + \hat{\lambda} r(y_j)\}},
\end{align} 
where $\hat{\lambda}$ satisfies~\eqref{Eq:EqualSumsLambda}. Ignoring the dependence on $\hat{\lambda}$, this estimator is a V-statistic and is asymptotically unbiased for $\operatorname{MMD}^2_{r, k}(f,g)$ since $\hat{\lambda}\overset{\mathbb{P}}{\rightarrow} \lambda_0$, as shown in the proof of Theorem~\ref{thm:general_consistency} in Appendix~\ref{appendix:proof3.1}. Following Theorem~\ref{thm:general_consistency}, the test given in Algorithm~\ref{alg:star_algo} using~\eqref{eq:V-stat} with $k(\cdot, \cdot)$ from a characteristic reproducing kernel Hilbert space is consistent, provided that~$\|\cdot\|_\infty\leq\gamma\,\|\cdot\|_\mathcal{H}$ and $N\left({\|\varphi\|_\mathcal{H} \leq 1}, \delta, \|\cdot \|_\infty \right)$ is finite for all $\delta > 0$. The former condition is verified for uniformly bounded kernels, as the Cauchy--Schwarz inequality and the representer theorem imply $|\varphi(x)| = |\langle \varphi, k(\cdot, x) \rangle_\mathcal{H}| \leq \sqrt{k(x,x)} \| \varphi \|_\mathcal{H}$. This is a common assumption in the reproducing kernel Hilbert space literature~\citep[e.g.][Definition~30]{JMLR:v13:gretton12a}, and it is satisfied by many standard choices of $k(\cdot, \cdot)$. The latter condition is also mild, and it is satisfied for common kernels such as the exponential kernel~\citep[e.g.][Lemma~D.2]{yang2020functionapproximationreinforcementlearning}.

\subsection{Minimax rate optimality}\label{sec:optimality}
While Theorem \ref{thm:general_consistency} shows the consistency of our test under mild assumptions, we can further show its minimax rate optimality when $\mathcal{X} = \mathbb{R}^d$ under the extra assumption that $\psi_r := (n+m) \, m^{-1} \, (f-h)$ lies in a Sobolev ball. We shall also require that $r$ is uniformly bounded, i.e.~there exist $c,C>0$ such that $c \leq r(x) \leq C$ for all $x \in \mathcal{X}$, and balanced sample sizes; these are stronger assumptions than those needed to prove consistency and might not be necessary, but simplify an already involved analysis. It is also more convenient to consider the test statistic 
\begin{align}\label{eq:U_stat}
    &U(x_1, \ldots, x_n, y_1, \ldots, y_m) = \frac{(n+m)^2}{n^2m^2} \sum_{i \neq j = 1}^n \frac{\hat{\lambda}^2 r(x_i)r(x_j)k_\zeta(x_i,x_j)}{\{n/m + \hat{\lambda} r(x_i) \} \{n/m + \hat{\lambda} r(x_j)\}}  \nonumber \\
    &  + \frac{(n+m)^2}{m^4} \sum_{i \neq j = 1}^m \frac{ k_\zeta(y_i,y_j)}{\{n/m + \hat{\lambda} r(y_i) \} \{n/m + \hat{\lambda} r(y_j)\}} - \frac{2(n+m)^2}{nm^3} \sum_{i = 1}^n \sum_{j = 1}^m \frac{\hat{\lambda} r(x_i) k_\zeta(x_i,y_j)}{\{n/m + \hat{\lambda} r(x_i) \} \{n/m + \hat{\lambda} r(y_j)\}},
\end{align} 
where $k_\zeta(\cdot, \cdot)$ is a multivariate kernel with bandwidth $\zeta \geq 1$ which will be defined later. Note that apart from a slightly unusual normalization and the dependence on~$\hat{\lambda}$, \eqref{eq:U_stat} is essentially a U-statistic which serves as an estimator of $\operatorname{MMD}^2_{r, k_\zeta}(f,g)$. This is a slight modification of~\eqref{eq:V-stat}, with diagonal terms removed, that we expect to have similar performance but which is simpler to analyze. Coming back to the definition of~$k_\zeta$, let $K:\mathbb{R}\to\mathbb{R}$ be a bounded, continuous, even, characteristic positive semidefinite kernel with $K\in L^{1}(\mathbb{R})\cap L^{4}(\mathbb{R})$, $\int_{\mathbb{R}}K(u)\,du=1$, and $0<K(0)<\infty$. Given a bandwidth $\zeta \geq 1$, we define a characteristic kernel on $\mathbb{R}^d$ by
$k_\zeta(x, y)
\;:=\;
\zeta^d \prod_{i=1}^d  K\left(\zeta(x_i - y_i)\right)$ for $x,y \in \mathbb{R}^d$.  For example, choosing $K(u) \;=\; \tfrac{1}{\sqrt{\pi}}e^{-u^2}$ recovers the Gaussian kernel in $\mathbb{R}^d$, while choosing $K(u) \;=\;  \tfrac{1}{2} e^{-|u|}$ yields the Laplace kernel. Further useful properties of $k_\zeta$ can be found in the proof of Theorem \ref{thm:UB_minimax}. Extensions to separate kernels $K_i(\cdot)$ with distinct bandwidths $\zeta_i$ along each dimension are straightforward and can be obtained by following similar arguments as in~\cite{schrab2023mmd}.

We now characterize the optimality of the testing procedure within the minimax framework, aiming to identify the smallest separation between the null and alternative hypotheses that allows for reliable discrimination. Smoothness is assumed under the alternative but not under the null, as we already showed in Section \ref{sec:methodology} that Algorithm~\ref{alg:star_algo} test guarantees uniform, non-asymptotic control of the type~I error. For fixed $r: \mathcal{X} \rightarrow \mathbb{R}_+$ such that $ 0< c \leq r(x) \leq C$ for all $x \in \mathcal{X}$, and for $\rho > 0$, we are interested in testing 
\begin{equation}\label{eq:l2_separation}
    H_0 : g \propto rf \quad \text{ vs. } \quad H_1^r(\rho): \| \psi_r \|_2 > \rho,
\end{equation}
and aim to find the smallest value of $\rho$ such that there exists a test with uniform error control. Arguing as in Section~\ref{sec:integral probability metric}, it appears more natural in our setting to define the separation in terms of $\psi_r$.  Nonetheless, when $r$ is uniformly bounded, alternative choices such as $\| g - \bar{r}f\|_2$ with $\bar{r} = r(\int r f d\mu)^{-1}$ are also valid; we prove in Proposition~\ref{eq:prop:equivalenceSeparations} in Appendix~\ref{appendix:proof4} that these norms are equivalent when $\mathcal{X}$ is compact.

Let $\Psi$ denote the collection of randomized tests based on the combined dataset $(X_1, \ldots, X_n, Y_1, \ldots, Y_m)$. Each test is represented by a function $\varphi: \mathcal{X}^{n+m} \to [0,1]$, where, upon observing $(x_1, \ldots, x_n, y_1, \ldots, y_m)$, the null hypothesis $H_0$ is rejected with probability $\varphi(x_1, \ldots, x_n, y_1, \ldots, y_m)$. Additionally, define $\Psi(\alpha)$ as the subset of $\Psi$ consisting of tests with size $\alpha$, where $\alpha \in (0,1)$. Define the $d$-dimensional Sobolev ball with smoothness parameter $s>0$ and radius $L>0$ as
\begin{align}\label{eq:Sobolev_ball}
    \mathcal{S}_d^s(L):=\left\{p \in L^1(\mathbb{R}^d) \cap L^2(\mathbb{R}^d):\int_{\mathbb{R}^d}\|\xi\|_2^{2 s}|\widehat{p}(\xi)|^2 \mathrm{~d} \xi \leq(2 \pi)^d L^2\right\},
\end{align}
where $\hat{p}$ denotes the Fourier transform of $p$, i.e.~$\widehat{p}(\xi):=\int_{\mathbb{R}^d} p(x) e^{-i \langle x, \xi \rangle} d x$ for all $\xi \in \mathbb{R}^d$. Throughout this section we will also assume $ m \leq n \leq \tau \, m$ for some $\tau \in [1,\infty)$, which is to say that $n$ and $m$ are of the same order. This assumption is commonly employed in the literature on two-sample testing~\citep[e.g.][]{schrab2023mmd,LiYuanJMLROptimal}. For fixed $r:\mathcal{X} \rightarrow \mathbb{R}_+$ such that $0 < c \leq r(x) \leq C$ for all $x \in \mathcal{X}$, we may then define the minimax separation to be $\rho^*_r \equiv \rho^*_r\left(n, m, \theta, \alpha, \beta \right):=\inf \left\{\rho>0: \inf_{\varphi \in \Psi(\alpha)} \sup_{(f,g) \in \mathcal{S}^r_\theta(\rho)} \mathbb{E}_P(1-\varphi) \leq \beta \right\}$, where $P = P_{f}^{\otimes n}\otimes P_{g}^{\otimes m}$, $\theta = (d, \tau, M, s, L) \in \mathbb{N}_+ \times \mathbb{R}_{\geq 1} \times \mathbb{R}_+^3$ and 
\begin{align}
    \mathcal{S}^r_\theta(\rho) := & \left\{(f, g): \max (\|f\|_{\infty},\|g\|_{\infty}) \leq M, \, \| \psi_r \|_2 > \rho \text{ and } \psi_r \in \mathcal{S}_d^s(L) \right\}. 
\end{align}
We can prove an upper bound on the minimax separation by showing that the test in Algorithm~\ref{alg:star_algo}, based on the test statistic~\eqref{eq:U_stat} with bandwidth $\zeta = n^{\frac{2}{4s+d}}$, has good power when $\rho$ is sufficiently large. 

\begin{thm}\label{thm:UB_minimax}
    Let $\mathcal{X} = \mathbb{R}^d$ and fix $\alpha, \beta \in (0,1)$ such that $\alpha + \beta < 1$. There exists $C_1 = C_1(c, C, \theta, \alpha, \beta) > 0$ such that $    \rho^*_r \leq C_1 \,  n^{-2 s /(4 s+d)}$.
\end{thm}
Regarding optimality, a matching lower bound has been established for the classical two-sample testing problem over Sobolev balls,  showing that there exists a constant $c_1 = c_1(\theta, \alpha, \beta)$ such that $\rho^*_{r \equiv 1} \;\geq\; c_1\,n^{-2s/(4s + d)}$; for example, see Theorems 3 and 5 in~\citet[][]{LiYuanJMLROptimal}. In our setting, the minimax separation rate $\rho_r^*$ will typically depend on the specific choice of $r$. Nevertheless, we derive the following matching lower bound under the assumptions that $r$ is bounded above and below. 

\begin{thm}\label{thm:LB_minimax}
 Fix $\alpha, \beta \in (0,1)$ such that $\alpha + \beta < 1$ and suppose that $s \in \mathbb{N}_+$. There exists $c_1 = c_1(c, C, \theta, \alpha, \beta) > 0$ such that for all $r: \mathcal{X} \to \mathbb{R}_+$ satisfying $0 < c \leq r(\cdot) \leq C$, we have $\rho^*_r \geq c_1 \,  n^{-2 s /(4 s+d)}$.
\end{thm}
When $r$ is uniformly bounded, the combined results of Theorems \ref{thm:UB_minimax} and \ref{thm:LB_minimax} demonstrate that $\rho^*_r \asymp n^{-2 s /(4 s+d)}$, confirming that Algorithm~\ref{alg:star_algo} with test statistic \eqref{eq:U_stat} achieves optimality. Future work could study the dependence of $\rho_r^*$ on $r$, though by analogy with the goodness-of-fit testing problem for densities, this is likely to be technically demanding. Indeed, it is known~\citep[e.g.][]{balakrishnan2019hypothesis} that minimax rates of convergence for testing the null hypothesis that $f=f_0$ given $X_1,\ldots,X_n \overset{\mathrm{i.i.d.}}{\sim} f$ intricately depend on the specific choice of $f_0$, with the hardest version of the problem being when $f_0$ is a uniform density. Although the simulations in Section \ref{sec:simulSynthetic} and the results for binary data in Appendix \ref{appenidix:discreteDRPT} seem to suggest that the hardest shifts to test are those closer to $r \equiv 1$, it is still unclear if this is the case in full generality. Here, hardness is meant in the information-theoretic sense; the impact of $r$ on the computational runtime of our procedures is discussed in Remark~\ref{lemma:runningChainLongerEnough}.

Finally, we note that the assumption $s \in \mathbb{N}_+$ is common in the statistical literature \citep{LiYuanJMLROptimal}, and plays a critical role in the lower bound construction. In particular, it ensures that the bump functions introduced in Lemma \ref{lem:smooth_bumps_measurable_r} in Appendix~\ref{appendix:proof3.1} are orthogonal with respect to the Sobolev inner product $\langle \phi_1, \phi_2 \rangle_{\mathcal{S}^s_d}:=\int_{\mathbb{R}^d} \|\xi\|_2^{2s} \, \widehat{\phi_1}(\xi) \, \overline{\widehat{\phi_2}(\xi)}\,d\xi$. This orthogonality holds when $s \in \mathbb{N}_+$ because the bump functions are supported on disjoint sets. However, this argument breaks down for general $s > 0$, as disjoint support no longer guarantees orthogonality in the Sobolev inner product. Extending the construction to non-integer smoothness may still be possible using techniques from \cite{butucea2007goodness} and~\cite{meynaoui2019adaptive}.

\section{Extensions to unknown density ratios}\label{sec:family_shift}
We move beyond the case of a simple null, and consider the harder problem of testing
\begin{align}\label{eq:H_0Family}
    H_0^\mathcal{R}: \text{there exists  } r_\star \in  \mathcal{R} \text{ such that } g \propto r_\star \, f,
\end{align}
where $\mathcal{R}$ is a suitable function class. If $\mathcal{R}$ is compact in a topology with respect to which $r \mapsto \|\psi_r\|_2$ is continuous, this is equivalent to saying that $\inf_{r \in \cal R} \|\psi_r\|_2 = 0$. We will focus on cases where $\cal R$ is either a parametric family or a class whose members act only on a proper subset of the variables. In the former case,~\eqref{eq:H_0Family} can serve as an important preliminary tool for model selection. For instance, it is common to assume that the density ratio belongs to a specific parametric family of functions \citep{Qin1998InferencesFC, sugiyamaDensityRatioAnalysis}, with maximum likelihood estimation used for inference. See \cite{Qin1998InferencesFC} and the references therein for a comprehensive overview of applications that rely on this modelling assumption.

Furthermore, when $\mathcal R$ consists of functions acting on a proper subset of variables, i.e.~$r(x)=\tilde r(x_S)$ for $S\subset [d]$, \eqref{eq:H_0Family} specializes to tests for structured shifts such as marginal or conditional changes, thus providing useful diagnostics for practitioners. As a concrete example in this setting, we consider the  conditional two-sample testing problem \citep[e.g.,][]{kim2024conditional}, where we observe two independent samples $(X_i^{(1)},Y_i^{(1)})_{i=1}^{n_1}$ drawn independently from a joint density $f_{XY}^{(1)}$, and $(X_i^{(2)},Y_i^{(2)})_{i=1}^{n_2}$ drawn independently from $f_{XY}^{(2)}$, and we test the equality of conditional densities $f_{Y|X}^{(1)}(y\mid x)=f_{Y|X}^{(2)}(y\mid x)$ for $f_X^{(1)}$-almost every $x$. Because this is equivalent to $ f_{XY}^{(2)}(x,y)= (f_X^{(2)}(x) /f_X^{(1)}(x)) \,f_{XY}^{(1)}(x,y)$, conditional two-sample testing can be cast as a hypothesis test on density ratios. Since the marginal density ratio is unknown and must be estimated in practice, this yields an instance of the composite problem~\eqref{eq:H_0Family}. Even if this ratio were known, as would be the case in a trial with perfect compliance~\citep[e.g.][]{lei2021conformal} or a Model-X setting~\citep{candes2016panning}, the resulting test would reduce to the simpler null \eqref{eq:testing_problem}, which remains non-trivial and can be addressed via Algorithm~\ref{alg:star_algo}. 

To tackle \eqref{eq:H_0Family}, we first estimate the density ratio on a training sample using density-ratio estimation methods \citep[][]{sugiyamaKuLSIF-statistical,Sugiyama2012bregman}, and then feed this estimate into Algorithm~\ref{alg:star_algo}, which is run on an independent sample. At the testing stage, however, plugging an estimated ratio into Algorithm~\ref{alg:star_algo} can inflate the type~I error. In this regard, we show that if $H_0^\mathcal{R}$ is satisfied and $\hat{r}$ is a good approximation of the true $r_{\star}$, then the excess type~I error of Algorithm~\ref{alg:star_algo} is bounded by the total variation distance between the product measures of the normalized versions of $\hat{r} \, f$ and $r_{\star} \, f$. For any two distributions $P_1, P_2$ defined on the same probability space, the total variation distance is defined as $\operatorname{TV}\left(P_1, P_2\right)=\sup _A\left|P_1(A)-P_2(A)\right|$, where the supremum is taken over all measurable sets. In the result below, we assume that both the approximation $\hat{r}$ of the true $r_\star$ and the test statistic $T: \mathcal{X}^{n+m} \to \mathbb{R}$ are deterministic, meaning that they are selected independently of $Z$.

\begin{prop}\label{prop:robustness}
  Assume $H_0^\mathcal{R}$ is true, and let $r_\star$ be such that $g \propto r_\star \, f$. Write $\hat r$ for an approximation of $r_\star$, and define  $p_{\hat{r}}$ to the output of Algorithm \ref{alg:star_algo} based on $\hat r$ with fixed parameters $S \geq 1, H \geq 1$. Calling $\check{r} = \left(\int \hat{r} f  d\mu\right)^{-1} \hat{r}$, $\bar{r}_\star = \left(\int r_\star f d\mu \right)^{-1} r_\star$, for all $\alpha \in [0,1]$ we have $\mathbb{P}\{p_{\hat r} \leq \alpha \} \leq \alpha + \operatorname{TV}\left(\{\check{r} \, f\}^{\otimes{m}}, \{\bar{r}_{\star} \, f\}^{\otimes{m}} \right)$.
\end{prop}

For Proposition~\ref{prop:robustness} to have practical implications, we need $\operatorname{TV}\left(\{\check{r} \, f\}^{\otimes{(n-N)}}, \{\bar{r}_{\star} \, f\}^{\otimes{(n-N)}} \right) \overset{\mathbb P}{\to} 0$. We discuss conditions under which this holds for conditional two-sample testing in Example~\ref{ex:DRE} in Appendix~\ref{app:additionalResults}, where we construct separate training and test samples via sample splitting. As in Theorem~4 of \citet{berrett2020conditional}, however, Proposition~\ref{prop:robustness} provides a worst-case guarantee taken over an arbitrary test statistic~$T$, which may be chosen adversarially to maximize sensitivity to errors in estimating $r_\star$. In practice, for the concrete choices of $T$ considered in Subsection~\ref{sec:shifted maximum mean discrepancy}, the procedure appears substantially more robust than this worst-case bound suggests; see Figure~\ref{fig:combined!}(c) for supporting simulations. Nonetheless, without structural assumptions ensuring that the total variation term is negligible, Proposition~\ref{prop:robustness} does not yield universal finite-sample validity. In the case of conditional two-sample testing this limitation is intrinsic: Theorem 1 in \citet[][]{kim2024conditional} shows that any valid conditional two-sample test has no power against any alternative, echoing the hardness result of \cite{shah2020hardness} for conditional independence testing. Thus, in this setting the lack of universal validity is not a deficiency of our method, but a consequence of the problem’s inherent difficulty. Whether analogous hardness phenomena hold for other instances of~\eqref{eq:H_0Family} remains an open question.

For power results, we can use Theorem~\ref{thm:general_consistency} to show that whenever~\eqref{eq:H_0Family} does not hold, the test given in Algorithm~\ref{alg:star_algo} and based on an approximation $\hat r \in \cal R$ has high power, provided that the key moment condition of Theorem~\ref{thm:general_consistency} is satisfied.
\begin{prop}\label{prop:alternativeRobustness2}
Assume that $\hat r \in \mathcal R$ satisfies the moment condition of Theorem~\ref{thm:general_consistency} with $\hat r$ in place of $r$, and that the remaining conditions hold. Write $p_{\hat r}$ for the output of Algorithm~\ref{alg:star_algo} with parameters $S \geq 1, H \geq 1$, run with $\hat r$ and based on~\eqref{eq:empricalMMD}. If~\eqref{eq:H_0Family} does not hold, for all $\alpha \in (0,1)$ we have $\mathbb{P}\{p_{\hat{r}} \leq \alpha \} \to 1$~as $n,m \to \infty$.
\end{prop}
This result is an immediate consequence of Theorem~\ref{thm:general_consistency}. In particular, if~\eqref{eq:H_0Family} does not hold, no element of $\cal R$ is proportional to $g/f$. Therefore, applying Theorem~\ref{thm:general_consistency} to the test based on $\hat r \in \cal R$ shows that our procedure has high power whenever the theorem’s assumptions are met. In practice, if $\hat r$ is an estimate based on a held out sample that satisfies $\mathbb{P}( \mathbb{E}[\left\{  \hat r(W) \right\}^2 \mid \hat{r} \, ] \,\mathbb{E}[\left\{  \hat r(W) \right\}^{-2} \mid \hat{r}  \, ]  \leq V_0 ) \geq 1 - \eta$ for $\eta \in (0,1)$, $V_0 \geq 1$ and $W \sim n(n+m)^{-1} f + m(n+m)^{-1} g$, we can use the previous result to show that for all $\alpha \in (0,1)$ we have $\liminf_{n,m \to \infty}\mathbb{P}\{p_{\hat{r}} \leq \alpha \} \geq 1 - \eta$. Now, the applicability of Proposition~\ref{prop:alternativeRobustness2} simply requires that $\hat r$ lies in $\cal R$ and satisfies the two-sided second-moment condition. If we also assume that $\hat r$ is close to a suitable $r_\star \in \cal R$, we can bound the excess type II error in finite samples. This is based on the following result. 
\begin{prop}\label{prop:alternativeRobustness}
    Write $\hat r$ for a deterministic approximation of~$r_\star$, and define $\eta(Z) := \max_{i \in [n+m]} |\log \hat r(Z_i) - \log r_\star (Z_i)|$. Let $p_{\hat{r}}$ denote the p-value in~\eqref{eq:pvalue_DRPT} computed from
    $(Z_{\hat{r}}^{(1)}, \ldots, Z_{\hat{r}}^{(H)})$, where $Z_{\hat{r}}^{(h)} := Z_{\sigma^{(h)}}$ and $(\sigma^{(1)}, \ldots, \sigma^{(H)})$ are i.i.d.\ samples from~\eqref{eq:sampling_1} with $\hat{r}$ used as the density ratio. Define $p_{r_\star}$ analogously, replacing $\hat r$ by $r_\star$ throughout. For all $\alpha \in (0,1)$ we have $\mathbb{P}\{p_{\hat r} \leq \alpha\} 
        \;\geq\; 
        \mathbb{P}\{p_{r_\star} \leq \alpha\} - H\,\mathbb{E}_Z[\sqrt{1 - \exp\{-2 m \eta(Z)\}} \, ]$.
\end{prop}
Proposition~\ref{prop:alternativeRobustness} immediately leads to the consistency of our procedure provided that $r_\star$ satisfies $\|\psi_{r_\star}\|_2 > 0$ and the assumptions of Theorem~\ref{thm:general_consistency}, and the estimator $\hat r$, computed on an held-out training data of size $N_\mathrm{train}$, is such that $H\,\mathbb{E}_Z[\sqrt{1 - e^{-2 m \eta(Z)}} \mid \hat r \,] \xrightarrow{\mathbb{P}}~0$. For example, for a univariate logistic regression model with Gaussian data we expect $\mathbb{E}_Z[\sqrt{\eta(Z)} \mid \hat r \,] \lesssim \log^{1/4}(n+m) \, N_\mathrm{train}^{-1/4}$ with high probability, suggesting that the error term vanish as long as $N_\mathrm{train} \gg H^4 m^2 \log(n+m)$. The assumption that $\hat r$ remains close to~$r_\star$ even when the model~$\mathcal R$ is misspecified parallels existing theory on density-ratio estimation, which ensures that, under suitable regularity conditions, density-ratio estimators converge to the projection of the true density ratio $g/f$ onto $\mathcal R$ with respect to an appropriate divergence~\citep{sugiyamaDensityRatioAnalysis, Sugiyama2012bregman}. Analogous consistency guarantees continue to hold when the $p$-values are
computed via Algorithm~\ref{alg:star_algo}. More importantly, under additional smoothness assumptions on both the null and the alternative, Propositions~\ref{prop:robustness} and~\ref{prop:alternativeRobustness} provide finite-sample control of the excess type~I and type~II errors for the composite null~\eqref{eq:H_0Family}; combining this with our finite-sample guarantees available for the simple null, such as those in Section~\ref{sec:optimality}, can lead to stronger non-asymptotic power results. This general approach was adopted in \cite{zheng2025generativeconditionaldistributionequality} for the conditional two-sample testing problem.

\section{Simulations and real-data applications}\label{sec:simul}
\subsection{Synthetic data}\label{sec:simulSynthetic}
Code and dataset access for reproducing the simulations are available at \url{https://github.com/abordino/DRPT}. An implementation of Algorithm~\ref{alg:star_algo} is provided in the R-package \texttt{DRPT} \citep{R-DRPT}. We now define variants of Algorithm~\ref{alg:star_algo}, each with $S=50$ and $H=99$, that differ in their choice of test statistic, as well as classical permutation tests applied to resampled data. (E1) corresponds to Algorithm~\ref{alg:star_algo} with the empirical shifted maximum mean discrepancy~\eqref{eq:U_stat} and Gaussian kernel $k(x,y) = \zeta^d e^{-\zeta^2 \|x- y\|_2^2}$; the bandwidth $\zeta$ is chosen via the median heuristic. (E2) uses rejection sampling on $X_{1:n}$ to obtain $W_{1:n'} \sim \bar r f$, then applies a uniform permutation test based on the classical maximum mean discrepancy with Gaussian kernel on $(W_{1:n'}, Y_{1:m})$, using the median heuristic for the bandwidth.  (E3) corresponds to Algorithm~\ref{alg:star_algo} with V-statistic~\eqref{eq:V-stat} and collision kernel $k(x,y)=\sum_{j=0}^J \mathbbm{1}\{x=j\}\mathbbm{1}\{y=j\}$ for fixed $J \geq 1$.

In this section, we empirically validate the performance of Algorithm~\ref{alg:star_algo} on synthetic data. Fig.~\ref{fig:combined!}(a) shows the performance of Algorithm~\ref{alg:star_algo} in the case of bivariate data. Here, we let $P_q$ be an absolutely continuous distribution on $[0,1]$ with density $q(x) = 2x$, and choose $P_f = \operatorname{Unif}([0,1]^2)$ and $P_g = (1+\eta)^{-1}P_q^{\otimes 2} +  \eta (1+\eta)^{-1}\operatorname{Beta}(\frac{1}{2},\frac{1}{2})^{\otimes{2}}$ for $\eta \in \{0, \ldots, 0.80\}$. When $\eta = 0$, $g$ satisfies the null hypothesis with $r(x, y) = 4xy$, while larger values of $\eta$ correspond to greater departures from the null. The purple curve corresponds to (E1), and we compare it against (E2), shown in green. We also consider a similar setting, shown in orange and blue, where we fix $P_f = \operatorname{Unif}([0,1]^2)$ and $P_g = (1+\eta)^{-1}P_q \otimes \operatorname{Unif}([0,1]) +  \eta (1+\eta)^{-1}\operatorname{Beta}(\frac{1}{2},\frac{1}{2})^{\otimes{2}}$; this implies that the null is satisfied for $r^\prime(x,y) = 2x$ when $\eta = 0$. We fix $n=m=150$ and repeat each test $500$ times in each setting. Algorithm~\ref{alg:star_algo} demonstrates superior performance compared to the rejection sampling-based procedure, which is to be expected due to the reduction in effective sample size inherent in such resampling methods. This justifies our choice to forego rejection-sampling schemes and utilize the full sample directly. Furthermore, less extreme shifts, such as $r^\prime$, seem harder to test, as Algorithm~\ref{alg:star_algo} shows greater power with more peaked density ratios; this trend is also reflected in the empirical results for binary data in Appendix~\ref{appenidix:discreteDRPT}.

In the case of binary data, we fix $n = 100$ and $m \in \{100, 200, 500, 2000\}$, $r = (1,3)$ and choose $f =  \operatorname{Ber}(1/2)$ and $g =  \operatorname{Ber}\left(3(1-\eta)/4+ \eta/4 \right)$ with $\eta \in \{0, \ldots, 0.90\}$.  The null hypothesis holds for $\eta = 0$, with increasing values of $\eta$ representing greater deviations from the null. We run (E3) $5000$ times in each setting and report the average decision. The results, illustrated in Fig.~\ref{fig:combined!}(b), indicate that our methodology performs well even with imbalanced sample sizes, achieving greatest power in the most imbalanced case, which corresponds to the largest sample sizes in this setting.

We also empirically verify that Proposition~\ref{prop:robustness} characterizes a worst-case scenario, while in practice the type~I error increases only mildly when $\hat{r}$ is sufficiently accurate and $T$ is chosen appropriately. Let $\phi$ denote the density of a standard Gaussian and $\Phi$ its cumulative distribution function, and set $f(x) = \phi(x)$ and $g(x) = \phi(x - \mu)$ with $\mu > 0$ so that $r_\star = \exp\{\mu x  - \mu^2/2\}$. If we run Algorithm~\ref{alg:star_algo} using $\hat r(x) = \exp\{\nu x  - \nu^2/2\}$ with $\nu >0$, meaning that the mean is misspecified, theory predicts that the excess type~I error is upper bounded by $\operatorname{TV}(\mathcal{N}(\mu, 1)^{\otimes m}, \mathcal{N}(\nu, 1)^{\otimes m}) = 2\Phi(\sqrt{m} \,|\mu - \nu|/2) - 1$. An approximation of this quantity when $n = m = 250$, $\mu = 1$, $\nu \in \mu + \{0.05, 0.10, 0.15, 0.20, 0.25\}$ is shown in red in Fig.~\ref{fig:combined!}(c). In contrast, running (E1) under the same setup and using the average decision over $300$ repetitions in each setting to estimate the size yields a much smaller inflation of the type~I error in practice, as shown by the purple dots. 

Finally, we study a scenario motivated by propensity weighting in causal inference. Let covariates $Z \sim \mathcal{N}(0, \Sigma)$ with $\Sigma \in \mathbb{R}^{10 \times 10}$, and assign treatment $D \in \{0,1\}$ according to $\mathbb{P}(D = 1 \mid Z) = \{1+\exp(\eta_\gamma(Z))\}^{-1}$ with $\eta_\gamma(z) = \beta_0 + \beta^\top z + \gamma\{\sin(10 z_1) + z_2 z_3\}$ and $\gamma \in \{0, 0.25, 0.5, 1, 2\}$. Writing $f(z)$ and $g(z)$ for the conditional densities of $Z$ given $D = 0$ and $D=1$, respectively, and $\pi := \mathbb{P}(D = 1) \in (0,1)$, the density ratio is $r(z) := g(z)/f(z) = \frac{1-\pi}{\pi}\exp\{\eta_\gamma(z)\}$. In medical applications, such ratios are often used to reweight untreated patients when estimating causal effects. For example, let $n_0 := \#\{i : D_i = 0\}$ and $n_1 := \#\{i : D_i = 1\}$, and assume well-defined potential outcomes $(Y(0), Y(1))$, $Y = D Y(1) + (1-D) Y(0)$, and $(Y(0), Y(1)) \perp \!\!\!  \perp D \mid Z$. Then $n_1^{-1} \sum_{i : D_i = 1} Y_i - n_0^{-1} \sum_{i : D_i = 0} r(Z_i) Y_i$ is an unbiased estimator of the average treatment effect on the treated, $\mathbb{E}[Y(1) - Y(0) \mid D = 1]$. In the simulation study, we estimate $r$ for each $\gamma \in \{0, 0.25, 0.5, 1, 2\}$ using linear logistic or kernel logistic regression \citep[e.g.][]{Sugiyama2010review} on a training sample of size $1000$, apply (E1) to an independent test sample of size $200$, repeat this $200$ times, and report the average rejection rate as an estimate of power. The results in Fig.~\ref{fig:combined_2}(a) show that kernel logistic regression performs well for all values of $\gamma$, as expected from its flexibility, whereas linear logistic regression is misspecified for $\gamma > 0$, and our test has high power in detecting the inadequacy of this estimator.

\begin{figure}
    \centering
    \includegraphics[width=1\linewidth]{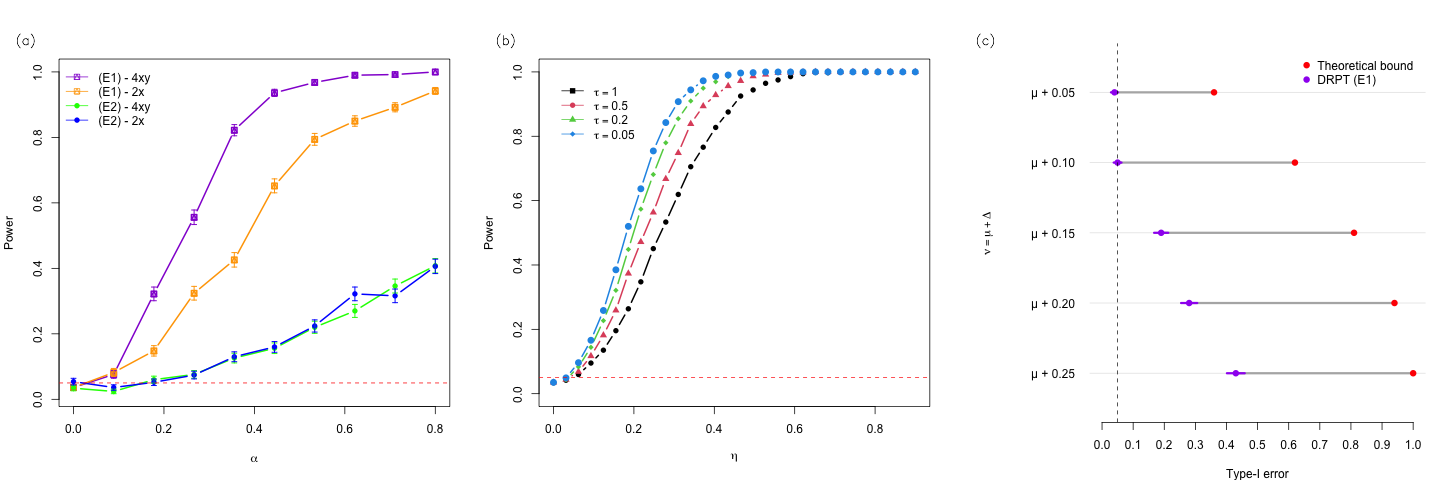}
    \caption{(a) Empirical power of (E1) for bivariate data on the unit square: purple, $r(x,y)=4xy$; orange, $r'(x,y)=2x$. Green and blue: alternative approaches based on (E2). (b) Estimated power of (E3) for binary data with varying sample sizes and $r=(1,3)$; $\tau := n/m$. (c) Empirical validation of Proposition~\ref{prop:robustness} for Gaussian data with misspecified mean: red, theoretical bound $\alpha + \operatorname{TV}(\mathcal{N}(\mu,1)^{\otimes m}, \mathcal{N}(\nu,1)^{\otimes m}) = \alpha + 2\Phi(\sqrt{m}|\mu-\nu|/2) - 1$; purple, estimate of type~I error over $300$ runs of (E1) based on $\hat r$. In all panels, error bars denote $\pm 1$ standard error.}
    \label{fig:combined!}
\end{figure}

\subsection{Real-world application}\label{sec:simulReal}
\subsubsection{New York-frisk and Stroop-effect datasets}\label{sec:simul_stroop}
When working with simple classes of distributional shifts, our hypothesis tests can be inverted in the usual way to furnish confidence sets, offering useful insights to practitioners. Formally, writing $C_\alpha$ for the set of all functions $r$ that are not rejected by our test at level $\alpha$, we have $\mathbb{P}(g/f \in C_\alpha) \geq 1-\alpha$. To illustrate this, we consider two practical scenarios, starting with the New York-frisk dataset for the years 2011 and 2012, which contains detailed records of police stop-and-frisk encounters. Restricting attention to stops labelled ``criminal possession of a weapon'' and using a binary indicator for whether any weapon was found, we partition the sample into $X_i$'s, Black individuals with $n_B=82 \, 626$, and $Y_i$'s, White individuals with $n_W=6 \, 383$. A conjecture in the literature, illustrated in Fig.~1(b) of \cite{goel2016precinct} and recalled in Section 8.1 in \cite{Koh2020WILDSAB}, suggests that Black individuals are approximately five times less likely to possess a weapon compared to White individuals. We examine this claim using (E3) and test
$w_1/w_0 = r b_1/b_0$ for $r \in \{0.1, \ldots, 7\}$, where $w_1$ and~$b_1$ denote the probabilities of weapon possession for White and Black individuals respectively. The blue curve in Fig.~\ref{fig:combined_2}(b) shows that (E3) does not reject for $r\in[5,6]$, consistent with Black individuals being about five times less likely than White individuals to be found carrying a weapon conditional on being frisked. Repeating the analysis for the years 2015–2016, where $n_W^\prime {=}138$ and $n_B^\prime{=}2\,298$, (E3) does not reject for $r \in [2,4]$, potentially suggesting a change in practices over time. To illustrate the informativeness of confidence intervals based on Algorithm~\ref{alg:star_algo}, we also report Wald-type confidence intervals computed as follows. Define $z_{1-\alpha/2}$ to be the $(1-\alpha/2)$ quantile of the standard normal distribution and $\operatorname{logit}(p):=\log\{p/(1-p)\}$ for $p\in(0,1)$. Let $\hat w_1 := n_W^{-1}\sum_i Y_i$ with $\sqrt{n_W}(\hat w_1-w_1)\xrightarrow{d} \mathcal{N}\!\big(0,(1-w_1)w_1\big)$, so by the delta method $\operatorname{Var}\{\operatorname{logit}(\hat w_1)\}\approx\{n_W w_1(1-w_1)\}^{-1}\approx\{n_W\hat w_1\}^{-1}+\{n_W(1-\hat w_1)\}^{-1}$. Similarly, for $\hat b_1:= n_B^{-1}\sum_i X_i$, we have $\operatorname{Var}\{\operatorname{logit}(\hat b_1)\}\approx\{n_B\hat b_1\}^{-1}+\{n_B(1-\hat b_1)\}^{-1}$. Since $\log r=\operatorname{logit}(w_1)-\operatorname{logit}(b_1)$, we estimate $\widehat{\log r}:=\operatorname{logit}(\hat w_1)-\operatorname{logit}(\hat b_1)$ and form the Wald confidence interval  $\widehat{\log r}\pm z_{1-\alpha/2}\sqrt{\{n_W\hat w_1\}^{-1}+\{n_W(1-\hat w_1)\}^{-1}+\{n_B\hat b_1\}^{-1}+\{n_B(1-\hat b_1)\}^{-1}}$. Exponentiating yields the intervals for $r$ shown in the figure; its asymmetry reflects the log-scale construction. We see that the Wald intervals essentially coincide with ours in this simple setting, confirming the practical utility of our test.

\begin{figure}
    \centering
    \includegraphics[width=1\linewidth]{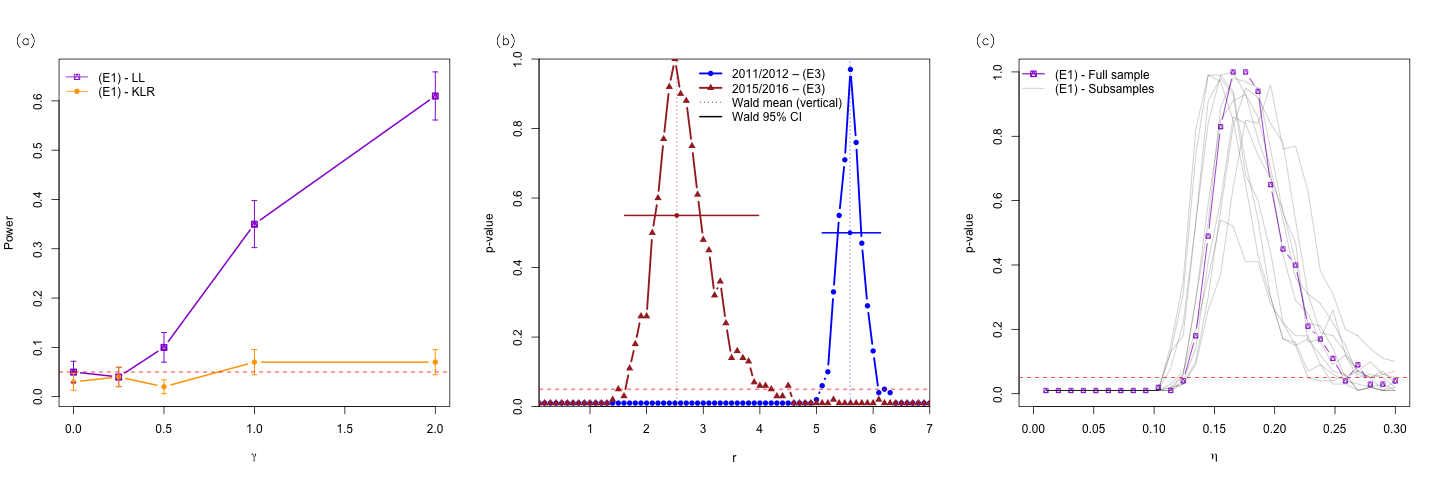}\caption{%
(a) Empirical power of (E1) in the causal inference setting, using estimates of the true density ratio
$r(z) = \frac{1-\pi}{\pi}\exp\{\beta_0 + \beta^\top z + \gamma(\sin(10 z_1) + z_2 z_3)\}$
obtained via linear logistic (LL) and kernel logistic (KLR) regression; error bars show $\pm 1$ standard error.
(b) $p$-values of (E3) for testing $w_1/w_0 = r\, b_1/b_0$ for $r \in \{0.1,\ldots,7\}$ on the New York-frisk
datasets for 2011/2012 (blue) and 2015/2016 (brown); horizontal lines give Wald-type confidence intervals.
(c) $p$-values of (E1) for testing $g(y)\propto e^{y/\eta} f(y)$ for $\eta \in \{0.01,\dots,0.3\}$ on the Stroop data;
sensitivity analysis over $10$ random subsamples of size $n=m=100$.}
    \label{fig:combined_2}
\end{figure}

In a similar spirit, we also consider an application dealing with the well-known Stroop effect \citep{stroop1935interference}. This refers to the cognitive interference observed when an individual attempts to name the colour of the ink of a word that spells a different colour, compared to naming the colour of the ink when the word and ink colour match. Typically, reaction times and error rates differ significantly between these congruent and incongruent conditions, indicating distinct underlying distributions for concordant and non-concordant stimuli. The data set consists of observations from 131 individuals, each with two recorded values: $X$, representing the time taken to name a set of concordant pairs and $Y$, the time taken to name an equal number of discordant pairs. We standardize the data, and test the hypothesis $g(y) \propto e^{y/\eta} f(y)$ for varying $\eta \in \{0.01, \ldots, 0.30\}$ using (E1). The results in Fig.~\ref{fig:combined_2}(c) show that there are values for which such modelling assumptions cannot be ruled out, especially in a neighbourhood of $0.20$. Given the small sample size compared to the New York frisk dataset, we conduct a sensitivity analysis by running (E1) on $10$ random subsamples of size 
$n=m=100$. The resulting thin black lines align closely with the purple line corresponding to the full sample.

\subsubsection{Simulation Based Inference on the Two Moons dataset}\label{sec:simulation based inference}
We illustrate how our methodology provides a practical diagnostic for simulation based inference, yielding formal decision rules for the goodness-of-fit of posterior-to-prior/likelihood-to-evidence estimates. As noted in the Introduction, when the likelihood is intractable yet a simulator allows sampling from \(p(x\mid\theta)\), many simulation based inference pipelines estimate the ratio $ r(\theta, x):= p(\theta \mid x)/p(x) = p(x \mid \theta)/p(x) = p(\theta, x)/p(\theta) p(x)$, for example by training a classifier to distinguish pairs $(\theta,x)$ drawn from the joint $p(\theta,x)$ versus the product $p(\theta)p(x)$. We now show how to use Algorithm~\ref{alg:star_algo} to formally test the quality of estimates of~$r(\theta, x)$. Concretely, simulate $\{(\theta_i,X_i)\}_{i=1}^N \overset{\mathrm{i.i.d.}}{\sim} p(\theta,x)$ by drawing $\theta_i \sim p(\theta)$ and $X_i \mid \theta_i$ from the simulator, and simulate $\{(\tilde\theta_i,\tilde X_i)\}_{i=1}^N \overset{\mathrm{i.i.d.}}{\sim} p(\theta)p(x)$ by repeating the same two-step procedure independently and retaining one component from each run. Given an approximation $\hat r(\theta,x)$ of $r(\theta,x)$, trained on an independent split of data, we test the null $p(\theta,x) \propto \hat r(\theta,x)\,p(x)\,p(\theta)$. Applying Algorithm~\ref{alg:star_algo} only requires $\hat r(\theta,x)$ being pointwise evaluable for fixed $(\theta,x)$, in addition to having access to the simulator, which is standard in this setting.

We validate the procedure on the Two Moons Dataset, a standard simulation based inference benchmark~\citep{greenberg2019apt, lueckmann2021benchmarking}. Samples are of the form $(\theta_1, \theta_2, X_1, X_2)$, where $(\theta_1,\theta_2)^\top\sim\mathrm{Unif}([-1,1]^2)$ and  $(X_1, X_2) \mid (\theta_1, \theta_2) =\big(r\cos\alpha+0.25, \, r\sin\alpha\big)^\top+\big(-2^{-1/2}|\theta_1+\theta_2|,\, 2^{-1/2}(-\theta_1+\theta_2)\big)^\top$, with $\alpha\sim\mathrm{Unif}(-\pi/2,\pi/2)$ and $r\sim \mathcal{N}(0.1,0.01^2)$. This simulator induces a highly
non-Gaussian posterior over $(\theta_1,\theta_2)$, which is bimodal with two separated crescent-shaped modes for typical observations \citep[][Figure~7]{chattejee24}, making it a
challenging benchmark for posterior approximation methods. We consider two ratio estimators, \(\hat r_{\mathrm{NRE}}(\theta,x)\) and \(\hat r_{\mathrm{BNRE}}(\theta,x)\) \citep{Delaunoy2022BNRE}, trained on \(N_{\text{train}}\in\{2^4,2^7,2^9,2^{12},2^{15}\}\) independent samples. Details on the the training phase are provided in Appendix~\ref{appendix:trainingNRE}. For each of these, we run (E1) based on such estimates of the density ratio on $n = m = 200$ independent test samples, repeat this $300$ times per setting, and report the proportion of rejections as an estimate of power. Fig.~\ref{fig:Moon}(a) shows our test rejects the null for $\hat r_\mathrm{BNRE}$ more often, in line with its conservative, prior-leaning design that shrinks estimates toward 1, trading accuracy for coverage \citep{Delaunoy2022BNRE}. For completeness, Fig.~\ref{fig:Moon}(b) reports receiver operating characteristic curves for classifiers distinguishing $5 \cdot 10^4$ joint samples from an equal number of product-of-marginals samples weighted by $\hat r_\mathrm{NRE}$ and $\hat r_\mathrm{BNRE}$. Better estimators produce curves closer to the diagonal; no option clearly dominates, while Algorithm~\ref{alg:star_algo} offers clearer guidance in this setting.

\begin{figure}
    \centering
    \includegraphics[width=0.80\linewidth]{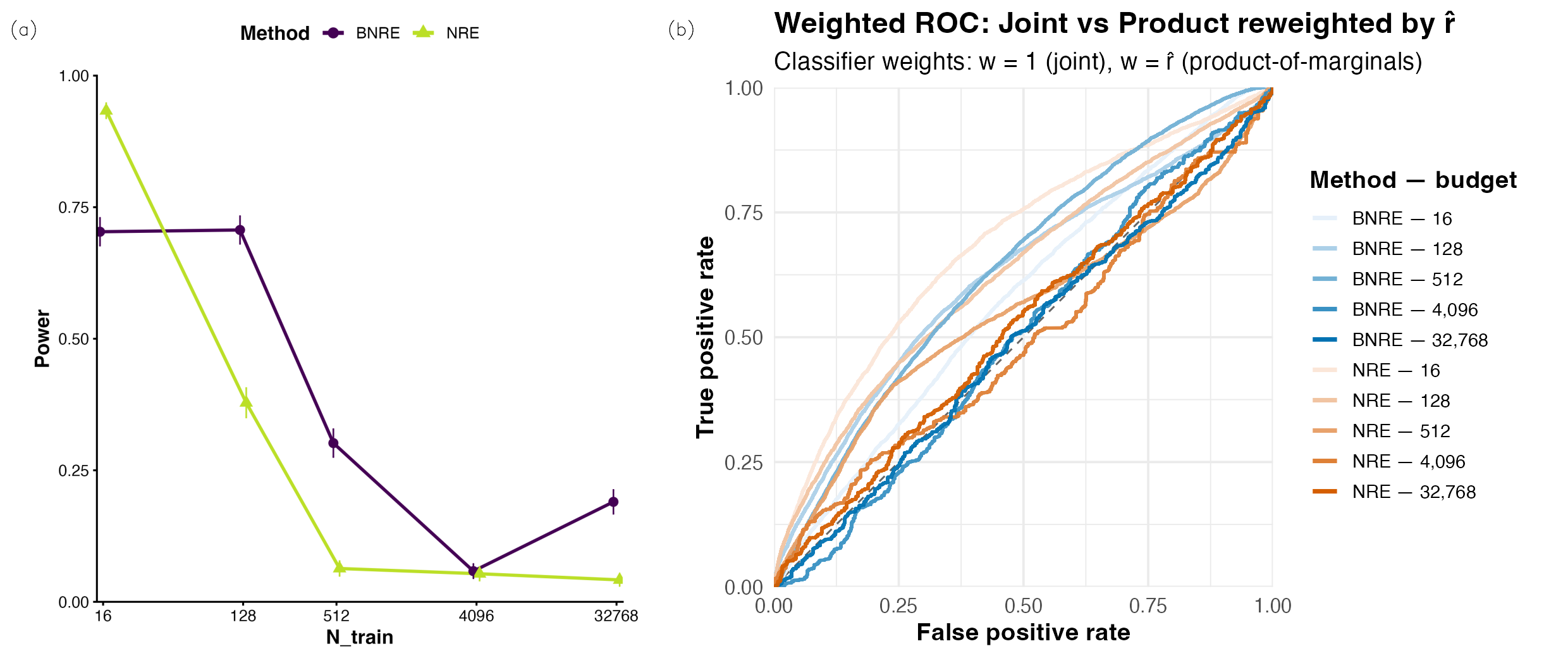}
\caption{%
(a) Empirical power of (E1) for testing
$p(\theta,x)\propto \hat r(\theta,x)\,p(x)\,p(\theta)$ with
$\hat r\in\{\hat r_\mathrm{NRE},\hat r_\mathrm{BNRE}\}$, the estimators in~\cite{Delaunoy2022BNRE}
trained on $N_\mathrm{train}\in\{2^4,2^7,2^9,2^{12},2^{15}\}$ independent samples; error bars show $\pm 1$ standard error.
(b) Receiver operating characteristic curves for classifiers distinguishing joint samples from
product-of-marginals samples weighted by $\hat r_\mathrm{NRE}$ and $\hat r_\mathrm{BNRE}$; curves closer to the diagonal indicate better estimators.%
}
    \label{fig:Moon}
\end{figure}

The same idea yields a test for an approximate posterior $\hat p(\theta\mid x)$, rather than~$\hat r(\theta,x)$, as it is sufficient to replace $\hat r(\theta,x)$ by $\hat p(\theta\mid x)/p(\theta)$. Again, Algorithm~\ref{alg:star_algo} only requires pointwise evaluation of $\hat p(\theta\mid x)$. There are many formal tests for assessing the quality of posterior approximation. Compared to these, Algorithm~\ref{alg:star_algo} avoids training a classifier on held out data \citep{LopezPazOquab2017C2ST, Linhart2023lC2ST} and does not require sampling from the approximate posterior, as in \cite{chattejee24}. While sampling from a learned posterior is often feasible, when this is impractical Algorithm~\ref{alg:star_algo} provides a viable alternative.

\subsubsection{Conditional two-sample testing on the diamonds datasets}\label{sec:simul_diamonds}
Finally, we assess the performance of Algorithm~\ref{alg:star_algo} in the conditional two-sample testing problem, as outlined in Section \ref{sec:family_shift}. Our analysis utilizes the diamonds dataset, which is available in the R-package \texttt{ggplot2}, and contains 53,490 observations with 10 features. Following \cite{kim2024conditional}, we designate the price variable as~$Y$ and use the six numerical variables (carat, depth, table, x, y, z) as predictors $X$. Prior to the analysis, we standardize both $X$ and $Y$. To introduce covariate shift, we implement biased sampling procedures: $X^{(1)}$ is sampled uniformly from the original feature space, while $X^{(2)}$ is sampled with probabilities proportional to $\exp(-x_1^2)$, with $x_1$ representing the first feature of $X$. Under the null hypothesis, the response variable $Y$ is uniformly sampled for both $Y^{(1)}$ and $Y^{(2)}$. Under the alternative hypothesis, $Y^{(1)}$ remains uniformly sampled, while $Y^{(2)}$ is sampled with probabilities proportional to $\exp(-y)$, where $y$ corresponds to the dataset's $Y$ values. We compare our method with two tests from \cite{kim2024conditional}: the single-split classifier test and the linear-time maximum mean discrepancy. More extensive simulations with additional alternative methods are provided in Appendix~\ref{appendix:extraSimul}. In our simulation, we used an 80/20 split, allocating 80\% of the total sampling budget $N \in \{200, 400, 800, 1200, 1600, 2000\}$ to marginal density ratio estimation and 20\% to testing; the other methods use an even split, as in their simulation study. Consistently with the other methods, the marginal density ratio is estimated either with the linear logistic or the kernel logistic regression. The testing phase was performed using (E1) and the average decision over $200$ repetitions is reported. All three methods control type~I error, with the possible exception of the classifier-based test, likely due to high estimation error in the classifier. The maximum mean discrepancy method matches ours in size but shows slightly lower power, as expected from its linear-time statistic, which trades power for computational~speed-up.

\begin{table}[!ht]
\centering
\caption{Results for the conditional two-sample testing problem on the diamonds dataset using (E1), compared with alternative methods. LL: linear logistic regression; KLR: kernel logistic regression; CLF: single-split classifier test; MMD-$l$: linear-time maximum mean discrepancy test. $N$ denotes the total sampling budget.}
\label{table:diamonds_main}
\scalebox{0.9}{%
\begin{tabular}{lllrrrrrr}
\toprule
\centering
Estimator & Hypothesis & Test & $N= 200$ & $N=400$ & $N=800$ & $N=1200$ & $N=1600$ & $N=2000$ \\
\midrule
LL & Null & CLF & 0.0900 & 0.0650 & 0.0750 & 0.0450 & 0.0675 & 0.0425 \\
LL & Null & MMD-$l$ & 0.0750 & 0.0650 & 0.0750 & 0.0500 & 0.0575 & 0.0575 \\
LL & Null & DRPT (E1) & 0.0600 & 0.0550 & 0.0850 & 0.0600 & 0.0500 & 0.0350 \\
\midrule
LL & Alternative & CLF & 0.1575 & 0.2050 & 0.2650 & 0.3600 & 0.3925 & 0.4800 \\
LL & Alternative & MMD-$l$ & 0.0650 & 0.0675 & 0.0850 & 0.0825 & 0.0975 & 0.0850 \\
LL & Alternative & DRPT (E1) & 0.1500 & 0.1800 & 0.4150 & 0.5450 & 0.6350 & 0.7150 \\
\midrule
KLR & Null & CLF & 0.0750 & 0.0575 & 0.0675 & 0.0350 & 0.0475 & 0.0450 \\
KLR & Null & MMD-$l$ & 0.0600 & 0.0625 & 0.0725 & 0.0575 & 0.0550 & 0.0600 \\
KLR & Null & DRPT (E1) & 0.0400 & 0.0350 & 0.0600 & 0.0500 & 0.0600 & 0.0300 \\
\midrule
KLR & Alternative & CLF & 0.0975 & 0.1525 & 0.2600 & 0.3550 & 0.3675 & 0.4450 \\
KLR & Alternative & MMD-$l$ & 0.0725 & 0.0675 & 0.0800 & 0.0900 & 0.0950 & 0.1050 \\
KLR & Alternative & DRPT (E1) & 0.1250 & 0.1750 & 0.3750 & 0.5000 & 0.6300 & 0.6650 \\
\midrule
\bottomrule
\end{tabular}
}
\end{table}

\bigskip

\section*{Acknowledgements} The research of the second author was supported by European Research Council Starting Grant 101163546.

\bibliography{bib}

@article{tenzer2022testing,
  title={{T}esting independence under biased sampling},
  author={Tenzer, Yaniv and Mandel, Micha and Zuk, Or},
  journal={Journal of the American Statistical Association},
  volume={117},
  number={540},
  pages={2194--2206},
  year={2022},
  publisher={Taylor \& Francis}
}

@article{rogozhnikov2016reweighting,
  title={{R}eweighting with boosted decision trees},
  author={Rogozhnikov, Alex},
  journal={Journal of Physics: Conference Series},
  volume={762},
  number={1},
  pages={012036},
  year={2016},
  organization={IOP Publishing}
}

@article{zhou2022PSweight,
  author = {Zhou, Tianhui and Tong, Guangyu and Li, Fan and Thomas, Laine E. and Li, Fan},
  title = {{P}{S}weight: {A}n {R} {P}ackage for {P}ropensity {S}core {W}eighting {A}nalysis},
  journal = {The R Journal},
  year = {2022},
  note = {https://doi.org/10.32614/RJ-2022-011},
  volume = {14},
  issue = {1},
  issn = {2073-4859},
  pages = {282-300}
}

@article{cox1969some,
  title={{S}ome sampling problems in technology},
  author={Cox, David R},
  journal={New Developments in Survey
Sampling },
  volume={1},
  pages={506--527},
  year={1969}
}

@article{efromovich2004density,
  title={{D}ensity estimation for biased data},
  author={Efromovich, Sam},
  journal={The Annals of Statistics},
  volume={32},
  number={3},
  pages={1137--1161},
  year={2004}
}

@article{qin1993empirical,
  title={{E}mpirical likelihood in biased sample problems},
  author={Qin, Jing},
  journal={The Annals of Statistics},
  volume={21},
  number={3},
  pages={1182--1196},
  year={1993},
  publisher={Institute of Mathematical Statistics}
}

@article{owen2000safe,
  author  = {Owen, Art B. and Zhou, Yi},
  title   = {{S}afe and {E}ffective {I}mportance {S}ampling},
  journal = {Journal of the American Statistical Association},
  volume  = {95},
  number  = {449},
  pages   = {135--143},
  year    = {2000}
}

@misc{ji2025overviewlargelanguagemodels,
      title={{A}n {O}verview of {L}arge {L}anguage {M}odels for {S}tatisticians}, 
      author={Wenlong Ji and Weizhe Yuan and Emily Getzen and Kyunghyun Cho and Michael I. Jordan and Song Mei and Jason E Weston and Weijie J. Su and Jing Xu and Linjun Zhang},
      year={2025},
      eprint={2502.17814},
      archivePrefix={arXiv},
      primaryClass={stat.ML},
      url={https://arxiv.org/abs/2502.17814}, 
}

@techreport{kahn1951estimation,
  author       = {Kahn, Herman and Harris, Theodore E.},
  title        = {{E}stimation of particle transmission by random sampling},
  institution  = {National Bureau of Standards},
  series       = {Applied Mathematics Series},
  number       = {12},
  pages        = {27--30},
  year         = {1951},
}

@incollection{storkey2009when,
  author       = {Storkey, Amos},
  title        = {{W}hen {T}raining and {T}est {S}ets {A}re {D}ifferent: {C}haracterizing {L}earning {T}ransfer},
  booktitle    = {{D}ataset {S}hift in {M}achine {L}earning},
  series       = {Neural Information Processing Series},
  pages        = {3--28},
  publisher    = {MIT Press},
  year         = {2009},
}

@article{weiss2016survey,
  author  = {Weiss, Karl and Khoshgoftaar, Taghi M. and Wang, Dingding},
  title   = {{A} {S}urvey of {T}ransfer {L}earning},
  journal = {Journal of Big Data},
  volume  = {3},
  number  = {1},
  pages   = {9},
  year    = {2016}
}

@article{tokdar2010importance,
  author  = {Tokdar, Surya T. and Kass, Robert E.},
  title   = {{I}mportance {S}ampling: {A} {R}eview},
  journal = {Wiley Interdisciplinary Reviews: Computational Statistics},
  volume  = {2},
  number  = {1},
  pages   = {54--60},
  year    = {2010}
}

@article{guan2022domainadaptation,
  author  = {Guan, Hao and Liu, Mingxia},
  title   = {{D}omain {A}daptation for {M}edical {I}mage {A}nalysis: {A} {S}urvey},
  journal = {IEEE Transactions on Biomedical Engineering},
  volume  = {69},
  number  = {3},
  pages   = {1173--1185},
  year    = {2022}
}

@article{horvitz1952generalization,
  author  = {Horvitz, D. G. and Thompson, D. J.},
  title   = {{A} generalization of sampling without replacement from a finite universe},
  journal = {Journal of the American Statistical Association},
  volume  = {47},
  number  = {260},
  pages   = {663--685},
  year    = {1952}
}

@article{Sugiyama2010review,
  author  = {Sugiyama, Masashi and Suzuki, Taiji and Kanamori, Takafumi},
  title   = {{D}ensity {R}atio {E}stimation: {A} {C}omprehensive {R}eview},
  journal = {RIMS Kokyuroku},
  year    = {2010},
  pages   = {10--31},
}

@article{ingster1987minimax,
  author    = {Yu. I. Ingster},
  title     = {{M}inimax {T}esting of {N}onparametric {H}ypotheses on a {D}istribution {D}ensity in the $L_p$ {M}etrics},
  journal   = {Theory of Probability and Its Applications},
  volume    = {31},
  number    = {2},
  pages     = {333--337},
  year      = {1987},
  publisher = {SIAM}
}

@article{butucea2007goodness,
  author       = {Cristina Butucea},
  title        = {{G}oodness-of-fit testing and quadratic functional estimation from indirect observations},
  journal      = {Annals of Statistics},
  volume       = {35},
  number       = {5},
  pages        = {1907--1930},
  year         = {2007},
  eprint       = {math/0612361},
  archivePrefix= {arXiv},
  primaryClass = {math.ST}
}

@article{goel2016precinct,
  title={{P}recinct or prejudice? {U}nderstanding racial disparities in {N}ew {Y}ork {C}ity's stop-and-frisk policy},
  author={Goel, Sharad and Rao, Justin M. and Shroff, Ravi},
  journal={The Annals of Applied Statistics},
  volume={10},
  number={1},
  pages={365--394},
  year={2016},
  publisher={Institute of Mathematical Statistics}
}

@InProceedings{pmlr-v235-prinster24a,
  title = 	 {{C}onformal {V}alidity {G}uarantees {E}xist for {A}ny {D}ata {D}istribution (and {H}ow to {F}ind {T}hem)},
  author =       {Prinster, Drew and Stanton, Samuel Don and Liu, Anqi and Saria, Suchi},
  booktitle = 	 {{P}roceedings of the 41st {I}nternational {C}onference on {M}achine {L}earning},
  pages = 	 {41086--41118},
  year = 	 {2024},
  volume = 	 {235},
  series = 	 {Proceedings of Machine Learning Research}
}

@inproceedings{Koh2020WILDSAB,
  title={{W}{I}{L}{D}{S}: {A} {B}enchmark of in-the-{W}ild {D}istribution {S}hifts},
  author={Pang Wei Koh and Shiori Sagawa and Henrik Marklund and Sang Michael Xie and Marvin Zhang and Akshay Balsubramani and Weihua Hu and Michihiro Yasunaga and Richard Lanas Phillips and Irena Gao and Tony Lee and Etiene David and Ian Stavness and Wei Guo and Berton A. Earnshaw and Imran S. Haque and Sara Meghan Beery and Jure Leskovec and Anshul Kundaje and Emma Pierson and Sergey Levine and Chelsea Finn and Percy Liang},
  booktitle={{I}nternational {C}onference on {M}achine {L}earning},
  year={2020}
}

@article{LiYuanJMLROptimal,
  author  = {Tong Li and Ming Yuan},
  title   = {{O}n the {O}ptimality of {G}aussian {K}ernel {B}ased {N}onparametric {T}ests against {S}mooth {A}lternatives},
  journal = {Journal of Machine Learning Research},
  year    = {2024},
  volume  = {25},
  number  = {334},
  pages   = {1--62}
}

@article{Nguyen2010minimax,
  title={{E}stimating {D}ivergence {F}unctionals and the {L}ikelihood {R}atio by {C}onvex {R}isk {M}inimization},
  author={XuanLong Nguyen and Martin J. Wainwright and Michael I. Jordan},
  journal={IEEE Transactions on Information Theory},
  year={2008},
  volume={56},
  pages={5847-5861}
}

@article{sugiyamaDensityRatioAnalysis,
  title={{T}heoretical {A}nalysis of {D}ensity {R}atio {E}stimation},
  author={Takafumi Kanamori and Taiji Suzuki and Masashi Sugiyama},
  journal={IEICE Transactions on Fundamentals of Electronics, Communications
and Computer Sciences},
  year={2010},
  volume={93-A},
  pages={787-798}
}

@article{sugiyamaKuLSIF-statistical,
  title={{S}tatistical analysis of kernel-based least-squares density-ratio estimation},
  author={Takafumi Kanamori and Taiji Suzuki and Masashi Sugiyama},
  journal={Machine Learning},
  year={2012},
  volume={86},
  pages={335-367}
}

@article{DRE_wainwright23,
author = {Cong Ma and Reese Pathak and Martin J. Wainwright},
title = {{Optimally tackling covariate shift in {R{K}{H}{S}}-based nonparametric regression}},
volume = {51},
journal = {The Annals of Statistics},
number = {2},
publisher = {Institute of Mathematical Statistics},
pages = {738 -- 761},
year = {2023}
}

@article{huangBiometrics,
    author = {Huang, Ming-Yueh and Qin, Jing and Huang, Chiung-Yu},
    title = {{E}fficient data integration under prior probability shift},
    journal = {Biometrics},
    volume = {80},
    number = {2},
    pages = {ujae035},
    year = {2024}
}

@article{KangNelsonDistr,
author = {Qing Kang and Paul I. Nelson},
title = {{P}ermutation tests from biased samples for the equality of two distributions},
journal = {Journal of Nonparametric Statistics},
volume = {21},
number = {3},
pages = {305--319},
year = {2009},
publisher = {Taylor \& Francis}
}

@article{chenLei24,
      title={{D}e-{B}iased {T}wo-{S}ample {U}-{S}tatistics {W}ith {A}pplication {T}o {C}onditional {D}istribution {T}esting}, 
      author={Yuchen Chen and Jing Lei},
  journal   = {Machine Learning},
  year      = {2025},
  volume    = {114},
  number    = {2},
  pages     = {33}}

@inproceedings{Delaunoy2022BNRE,
  title     = {{T}owards {R}eliable {S}imulation-{B}ased {I}nference with {B}alanced {N}eural {R}atio {E}stimation},
  author    = {Delaunoy, Arnaud and Hermans, Joeri and Rozet, Fran{\c{c}}ois and Wehenkel, Antoine and Louppe, Gilles},
  booktitle = {{A}dvances in {N}eural {I}nformation {P}rocessing {S}ystems},
  editor    = {Oh, Alice H. and Agarwal, Alekh and Belgrave, Danielle and Cho, Kyunghyun},
  year      = {2022},
  volume    = {35},
  pages     = {20025--20037},
  eprint    = {2208.13624},
  archivePrefix = {arXiv},
  primaryClass  = {stat.ML}
}

@inproceedings{LopezPazOquab2017C2ST,
  title     = {{R}evisiting {C}lassifier {T}wo-{S}ample {T}ests},
  author    = {Lopez-Paz, David and Oquab, Maxime},
  booktitle = {{I}nternational {C}onference on {L}earning {R}epresentations},
  year      = {2017},
  note      = {}
}

@misc{zheng2025generativeconditionaldistributionequality,
      title={{A} {G}enerative {C}onditional {D}istribution {E}quality {T}esting {F}ramework and {I}ts {M}inimax {A}nalysis}, 
      author={Siming Zheng and Meifang Lan and Tong Wang and Yuanyuan Lin},
      year={2025},
      eprint={2509.17729},
      archivePrefix={arXiv},
      primaryClass={cs.LG},
      url={https://arxiv.org/abs/2509.17729}, 
}

@article{Thomas2022LFIRE,
  title   = {{L}ikelihood-{F}ree {I}nference by {R}atio {E}stimation},
  author  = {Thomas, Owen and Dutta, Ritabrata and Corander, Jukka and Kaski, Samuel and Gutmann, Michael U.},
  journal = {Bayesian Analysis},
  year    = {2022},
  volume  = {17},
  number  = {1},
  pages   = {1--31},
}

@misc{chattejee24,
      title={{A} {K}ernel-{B}ased {C}onditional {T}wo-{S}ample {T}est {U}sing {N}earest {N}eighbors (with {A}pplications to {C}alibration, {R}egression {C}urves, and {S}imulation-{B}ased {I}nference)}, 
      author={Anirban Chatterjee and Ziang Niu and Bhaswar B. Bhattacharya},
      year={2024},
      eprint={2407.16550},
      archivePrefix={arXiv},
      primaryClass={stat.ME},
      url={https://arxiv.org/abs/2407.16550}, 
}

@article{Hu2020conformalJASA,
  title={{A} {T}wo-{S}ample {C}onditional {D}istribution {T}est {U}sing {C}onformal {P}rediction and {W}eighted {R}ank {S}um},
  author={Xiaoyu Hu and Jing Lei},
  journal={Journal of the American Statistical Association},
  year={2020},
  volume={119},
  pages={1136 - 1154}
}

@article{sugiyiamaKLIEPrkhs,
  title={{D}irect importance estimation for covariate shift adaptation},
  author={Masashi Sugiyama and Taiji Suzuki and Shinichi Nakajima and Hisashi Kashima and Paul von B{\"u}nau and Motoaki Kawanabe},
  journal={Annals of the Institute of Statistical Mathematics},
  year={2008},
  volume={60},
  pages={699-746}
}

@article{Sugiyama2012bregman,
  title={{D}ensity-ratio matching under the {B}regman divergence: a unified framework of density-ratio estimation},
  author={Masashi Sugiyama and Teruyuki Suzuki and Takafumi Kanamori},
  journal={Annals of the Institute of Statistical Mathematics},
  year={2012},
  volume={64},
  pages={1009-1044}
}

@article{ATLAS2025NSBI,
  author  = {{ATLAS Collaboration}},
  title   = {{A}n implementation of neural simulation-based inference for parameter estimation in {A}{T}{L}{A}{S}},
  journal = {Reports on Progress in Physics},
  year    = {2025},
  volume  = {88},
  pages   = {067801},
}

@inproceedings{Hermans2020LFMARE,
  author    = {Hermans, Joeri and Begy, Volodimir and Louppe, Gilles},
  title     = {{L}ikelihood-free {M}{C}{M}{C} with {A}mortized {A}pproximate {R}atio {E}stimators},
  booktitle = {{P}roceedings of the 37th {I}nternational {C}onference on {M}achine {L}earning},
  series    = {Proceedings of Machine Learning Research},
  volume    = {119},
  year      = {2020}
}

@inproceedings{Miller2022ContrastiveNRE,
  author    = {Miller, Benjamin Kurt and Weniger, Christoph and Forr{\'e}, Patrick},
  title     = {{C}ontrastive {N}eural {R}atio {E}stimation},
  booktitle = {{A}dvances in {N}eural {I}nformation {P}rocessing {S}ystems},
  year      = {2022}
}

@Manual{R-DRPT,
  title  = {{D}{R}{P}{T}: {D}ensity {R}atio {P}ermutation {T}est},
  author = {Alberto Bordino and Thomas B. Berrett},
  year   = {2025},
  note   = {R package version 1.1},
  url    = {https://CRAN.R-project.org/package=DRPT}
}

@article{
hermans2022a,
title={{A} {C}risis {I}n {S}imulation-{B}ased {I}nference? {B}eware, {Y}our {P}osterior {A}pproximations {C}an {B}e {U}nfaithful},
author={Joeri Hermans and Arnaud Delaunoy and Fran{\c{c}}ois Rozet and Antoine Wehenkel and Volodimir Begy and Gilles Louppe},
journal={Transactions on Machine Learning Research},
issn={2835-8856},
year={2022},
note={}
}

@misc{leonte2025simulationbasedinferencetelescopingratio,
      title={{S}imulation-based inference via telescoping ratio estimation for trawl processes}, 
      author={Dan Leonte and Raphaël Huser and Almut E. D. Veraart},
      year={2025},
      eprint={2510.04042},
      archivePrefix={arXiv},
      primaryClass={stat.ML},
      url={https://arxiv.org/abs/2510.04042}
}

@inproceedings{Linhart2023lC2ST,
  author    = {Linhart, Julia and Gramfort, Alexandre and Rodrigues, Pedro L. C.},
  title     = {$\ell$-{C}2{S}{T}: {L}ocal {D}iagnostics for {P}osterior {A}pproximations in {S}imulation-{B}ased {I}nference},
  booktitle = {{A}dvances in {N}eural {I}nformation {P}rocessing {S}ystems 36},
  year      = {2023}
}

@article{Austin2015IPTWBestPractice,
  author  = {Austin, Peter C. and Stuart, Elizabeth A.},
  title   = {{M}oving {T}owards {B}est {P}ractice {W}hen {U}sing {I}nverse {P}robability of {T}reatment {W}eighting ({I}{P}{T}{W}) {U}sing the {P}ropensity {S}core to {E}stimate {C}ausal {T}reatment {E}ffects in {O}bservational {S}tudies},
  journal = {Statistics in Medicine},
  year    = {2015},
  volume  = {34},
  pages   = {3661--3679},
}

@misc{Deistler2025SBIGuide,
  author       = {Deistler, Michael and Boelts, Jan and Steinbach, Peter and Moss, Guy and Moreau, Thomas and Gloeckler, Manuel and Rodrigues, Pedro L. C. and Linhart, Julia and Lappalainen, Janne K. and Miller, Benjamin Kurt and Gon{\c c}alves, Pedro J. and Lueckmann, Jan-Matthis and Schr{\"o}der, Cornelius and Macke, Jakob H.},
  title        = {{S}imulation-{B}ased {I}nference: {A} {P}ractical {G}uide},
  year         = {2025},
  eprint       = {2508.12939},
  archivePrefix= {arXiv},
  primaryClass = {stat.ML}
}

@article{Westreich2017Transportability,
  author  = {Westreich, Daniel and Edwards, Jessie K. and Lesko, Catherine R. and Stuart, Elizabeth and Cole, Stephen R.},
  title   = {{T}ransportability of {T}rial {R}esults {U}sing {I}nverse {O}dds of {S}ampling {W}eights},
  journal = {American Journal of Epidemiology},
  year    = {2017},
  volume  = {186},
  number  = {8},
  pages   = {1010--1014},
}

@article{quantile2013DRM,
author = {Jiahua Chen and Yukun Liu},
title = {{Quantile and quantile-function estimations under density ratio model}},
volume = {41},
journal = {The Annals of Statistics},
number = {3},
publisher = {Institute of Mathematical Statistics},
pages = {1669 -- 1692},
year = {2013}
}

@misc{kim2024conditional,
      title={{G}eneral {F}rameworks for {C}onditional {T}wo-{S}ample {T}esting}, 
      author={Seongchan Lee and Suman Cha and Ilmun Kim},
      year={2024},
      eprint={2410.16636},
      archivePrefix={arXiv},
      primaryClass={stat.ML},
      url={https://arxiv.org/abs/2410.16636}, 
}

@article{Robins2000MarginalSM,
  title={{M}arginal {S}tructural {M}odels and {C}ausal {I}nference in {E}pidemiology},
  author={James M. Robins and Miguel A. Hern{\'a}n and Babette A. Brumback},
  journal={Epidemiology},
  year={2000},
  volume={11},
  pages={550-560}
}

@book{sutton98barto,
author = {Sutton, Richard S. and Barto, Andrew G.},
title = {{R}einforcement {L}earning: {A}n {I}ntroduction},
year = {2018},
publisher = {A Bradford Book}
}

@misc{berrett2024efficient,
      title={{E}fficient estimation with incomplete data via generalised {A{N}{O}{V}{A}} decomposition}, 
      author={Thomas B. Berrett},
      year={2024},
      eprint={2409.05729},
      archivePrefix={arXiv},
      primaryClass={math.ST},
      url={https://arxiv.org/abs/2409.05729}, 
}

@article{kim2022minimax,
  title={{M}inimax optimality of permutation tests},
  author={Kim, Ilmun and Balakrishnan, Sivaraman and Wasserman, Larry},
  journal={Ann. Statist.},
  volume={50},
  number={1},
  pages={225--251},
  year={2022},
  publisher={Institute of Mathematical Statistics}
}

@book{wainwright2019high,
  title={{H}igh-dimensional {S}tatistics: {A} {N}on-asymptotic {V}iewpoint},
  author={Wainwright, Martin J},
  year={2019},
  publisher={Cambridge University Press}
}

@article{candes2016panning,
  title={{P}anning for gold: {M}odel-{X} knockoffs for high-dimensional controlled variable selection},
  author={Cand\`es, Emmanuel and Fan, Yingying and Janson, Lucas and Lv, Jinchi},
  journal={J. Roy. Statist. Soc. Ser. B},
  volume={80},
  number={3},
  pages={551--577},
  year={2018}
}

@article{balakrishnan2019hypothesis,
  title={{H}ypothesis testing for densities and high-dimensional multinomials: {S}harp local minimax rates},
  author={Balakrishnan, Sivaraman and Wasserman, Larry},
  journal={Ann. Statist.},
  volume={47},
  pages={1893--1927},
  year={2019},
  publisher={Institute of Mathematical Statistics}
}

@article{meynaoui2019adaptive,
  title={{A}daptive test of independence based on {H{S}{I}{C}} measures},
  author={Albert, M{\'e}lisande  and Laurent, B{\'e}atrice  and Marrel, Amandine and Meynaoui, Anouar},
  journal   = {The Annals of Statistics},
  year      = {2022},
  volume    = {50},
  number    = {2},
  pages     = {858--879},
publisher = {Institute of Mathematical Statistics}
}

@article{shah2020hardness,
  title={{T}he hardness of conditional independence testing and the generalised covariance measure},
  author={Shah, Rajen D and Peters, Jonas},
  journal={Ann. Statist.},
  volume={48},
  pages={1514--1538},
  year={2020},
  publisher={Institute of Mathematical Statistics}
}

@article{berrett2020conditional,
  title={{T}he conditional permutation test for independence while controlling for confounders},
  author={Berrett, Thomas B and Wang, Yi and Barber, Rina Foygel and Samworth, Richard J},
  journal={J. Roy. Statist. Soc. Ser. B},
  volume={82},
  pages={175--197},
  year={2020},
  publisher={Wiley Online Library}
}

@book{lehmann2006testing,
  title={{T}esting {S}tatistical {H}ypotheses},
  author={Lehmann, Erich L and Romano, Joseph P},
  year={2005},
  publisher={Springer Science \& Business Media}
}

@book{lee1990ustatistics,
  author    = {A. J. Lee},
  title     = {{U}-{S}tatistics: {T}heory and {P}ractice},
  year      = {1990},
  publisher = {Routledge},
  address   = {New York}
}

@article{berrett2020optimal,
  title={{O}ptimal rates for independence testing via {${U}$}-statistic permutation tests},
  author={Berrett, Thomas B and Kontoyiannis, Ioannis and Samworth, Richard J},
  journal={Ann. Statist.},
  volume={49},
  pages={2457--2490},
  year={2021}
}

@inproceedings{yang2020functionapproximationreinforcementlearning,
      title={{O}n {F}unction {A}pproximation in {R}einforcement {L}earning: {O}ptimism in the {F}ace of {L}arge {S}tate {S}paces}, 
      author={Zhuoran Yang and Chi Jin and Zhaoran Wang and Mengdi Wang and Michael I. Jordan},
 booktitle = {{A}dvances in {N}eural {I}nformation {P}rocessing {S}ystems},
 pages = {13903--13916},
 volume = {33},
 year = {2020}
}

@article{lei2021conformal,
  title   = {{C}onformal inference of counterfactuals and individual treatment effects},
  author  = {Lei, Lihua and Cand{\`e}s, Emmanuel J.},
  journal = {Journal of the Royal Statistical Society: Series B (Statistical Methodology)},
  volume  = {83},
  number  = {5},
  pages   = {911--938},
  year    = {2021},
}

@InProceedings{greenberg2019apt,
  title     = {{A}utomatic {P}osterior {T}ransformation for {L}ikelihood-{F}ree {I}nference},
  author    = {Greenberg, David and Nonnenmacher, Marcel and Macke, Jakob},
  booktitle = {{P}roceedings of the 36th {I}nternational {C}onference on {M}achine {L}earning},
  series    = {Proceedings of Machine Learning Research},
  volume    = {97},
  pages     = {2404--2414},
  year      = {2019},
  publisher = {PMLR}
}

@inproceedings{lueckmann2021benchmarking,
  title     = {{B}enchmarking {S}imulation-{B}ased {I}nference},
  author    = {Lueckmann, Jan-Matthis and Boelts, Jan and Greenberg, David and Gon{\c{c}}alves, Pedro J. and Macke, Jakob H.},
  booktitle = {{P}roceedings of the 24th {I}nternational {C}onference on {A}rtificial {I}ntelligence and {S}tatistics},
  series    = {Proceedings of Machine Learning Research},
  volume    = {130},
  pages     = {343--351},
  year      = {2021},
  editor    = {Banerjee, Arindam and Fukumizu, Kenji},
  publisher = {PMLR},
}

@article{Ramdas2022PermutationTU,
  title={{P}ermutation {T}ests {U}sing {A}rbitrary {P}ermutation {D}istributions},
  author={Aaditya Ramdas and Rina Foygel Barber and Emmanuel J. Cand{\`e}s and Ryan J. Tibshirani},
  journal={Sankhya A},
  year={2022},
  volume={85},
  pages={1156 - 1177}
}

@article{hoeffding52permutation,
author = {Wassily Hoeffding},
title = {{The {L}arge-{S}ample {P}ower of {T}ests {B}ased on {P}ermutations of {O}bservations}},
volume = {23},
journal = {The Annals of Mathematical Statistics},
number = {2},
publisher = {Institute of Mathematical Statistics},
pages = {169 -- 192},
year = {1952}
}

@article{besag89star_sampler,
 author = {Julian Besag and Peter Clifford},
 journal = {Biometrika},
 number = {4},
 pages = {633--642},
 publisher = {[Oxford University Press, Biometrika Trust]},
 title = {{G}eneralized {M}onte {C}arlo {S}ignificance {T}ests},
 urldate = {2024-06-08},
 volume = {76},
 year = {1989}
}

@article{jin2023modelfree,
  title   = {{M}odel-free selective inference under covariate shift via weighted conformal p-values},
  author  = {Jin, Ying and Cand{\`e}s, Emmanuel J.},
  journal = {Biometrika},
  year    = {2025},
  note    = {Advance article}
}

@inproceedings{covariate_shift_ramsdas19,
 author = {Tibshirani, Ryan J and Foygel Barber, Rina and Candes, Emmanuel and Ramdas, Aaditya},
 booktitle = {{A}dvances in {N}eural {I}nformation {P}rocessing {S}ystems},
 pages = {},
 title = {{C}onformal {P}rediction {U}nder {C}ovariate {S}hift},
 volume = {32},
 year = {2019}
}

@article{Qin1998InferencesFC,
  title={{I}nferences for case-control and semiparametric two-sample density ratio models},
  author={Yongsong Qin},
  journal={Biometrika},
  year={1998},
  volume={85},
  pages={619-630}
}

@book{mccullagh1989generalized,
  author = {McCullagh, P. and Nelder, J.A.},
  lccn = {99013896},
  publisher = {Chapman \& Hall},
  series = {Chapman and Hall/CRC Monographs on Statistics and Applied Probability Series},
  title = {{G}eneralized {L}inear {M}odels, {S}econd {E}dition},
  year = 1989
}

@book{ingster03,
author = {Ingster, Yuri and Suslina, Irina},
year = {2003},
pages = {},
title = {{N}onparametric {G}oodness-of-{F}it {T}esting {U}nder {G}aussian {M}odels},
publisher={Springer Science \& Business Media},
volume = {169}
}

@article{schrab2023mmd,
  author  = {Antonin Schrab and Ilmun Kim and Mélisande Albert and Béatrice Laurent and Benjamin Guedj and Arthur Gretton},
  title   = {{M}{M}{D} {A}ggregated {T}wo-{S}ample {T}est},
  journal = {Journal of Machine Learning Research},
  year    = {2023},
  volume  = {24},
  number  = {194},
  pages   = {1--81}
}

@article{JMLR:v13:gretton12a,
  author  = {Arthur Gretton and Karsten M. Borgwardt and Malte J. Rasch and Bernhard Sch{{\"o}}lkopf and Alexander Smola},
  title   = {{A} {K}ernel {T}wo-{S}ample {T}est},
  journal = {Journal of Machine Learning Research},
  year    = {2012},
  volume  = {13},
  number  = {25},
  pages   = {723-773}
}

@inproceedings{schrab2024epistemic,
      title={{C}redal {T}wo-{S}ample {T}ests of {E}pistemic {I}gnorance}, 
  author    = {Chau, Siu Lun and Schrab, Antonin and Gretton, Arthur and Sejdinovic, Dino and Muandet, Krikamol},
  booktitle = {{P}roceedings of {T}he 28th {I}nternational {C}onference on {A}rtificial {I}ntelligence and {S}tatistics ({A}{I}{S}{T}{A}{T}{S} 2025)},
  series    = {Proceedings of Machine Learning Research},
  volume    = {258},
  pages     = {127--135},
  year      = {2025},
  editor    = {Li, Yingzhen and Mandt, Stephan and Agrawal, Shipra and Khan, Emtiyaz}
}

@misc{yan2024distancekernel,
      title={{D}istance and {K}ernel-{B}ased {M}easures for {G}lobal and {L}ocal {T}wo-{S}ample {C}onditional {D}istribution {T}esting}, 
      author={Jian Yan and Zhuoxi Li and Xianyang Zhang},
      year={2024},
      eprint={2210.08149},
      archivePrefix={arXiv},
      primaryClass={stat.ME},
      url={https://arxiv.org/abs/2210.08149}, 
}

@article{thams23jrrsb,
    author = {Thams, Nikolaj and Saengkyongam, Sorawit and Pfister, Niklas and Peters, Jonas},
    title = "{Statistical testing under distributional shifts}",
    journal = {Journal of the Royal Statistical Society Series B: Statistical Methodology},
    volume = {85},
    number = {3},
    pages = {597-663},
    year = {2023}
}

@article{stroop1935interference,
  author    = {Stroop, J. Ridley},
  title     = {{S}tudies of {I}nterference in {S}erial {V}erbal {R}eactions},
  journal   = {Journal of Experimental Psychology},
  year      = {1935},
  volume    = {18},
  number    = {6},
  pages     = {643--662}
}

@article{Sriperumbudur2010,
author = {Sriperumbudur, Bharath K. and Gretton, Arthur and Fukumizu, Kenji and Sch\"{o}lkopf, Bernhard and Lanckriet, Gert R.G.},
title = {{H}ilbert {S}pace {E}mbeddings and {M}etrics on {P}robability {M}easures},
year = {2010},
publisher = {JMLR.org},
volume = {11},
  journal = {Journal of Machine Learning Research},
pages = {1517–1561},
numpages = {45}
}

@misc{hairer2021convergence,
    author = {Hairer, Martin},
    year = {2021},
    title = {{C}onvergence of {M}arkov processes},
    publisher = {Lecture notes},
}

@book{Dudley_2002, place={Cambridge}, edition={2}, series={Cambridge Studies in Advanced Mathematics}, title={{R}eal {A}nalysis and {P}robability}, publisher={Cambridge University Press}, author={Dudley, R. M.}, year={2002}, collection={Cambridge Studies in Advanced Mathematics}}

@incollection{BretagnolleHuber1978,
  author    = {Bretagnolle, J. and Huber, C.},
  title     = {Estimation des densit{\'e}s: risque minimax},
  booktitle = {S{\'e}minaire de Probabilit{\'e}s XII},
  series    = {Lecture Notes in Mathematics},
  volume    = {649},
  pages     = {342--363},
  publisher = {Springer},
  address   = {Berlin},
  year      = {1978},
  note      = {Univ. Strasbourg, Strasbourg, 1976/1977}
}

\begin{appendices}
Appendix~\ref{appendix:proof} contains the proofs of all results stated in the main text, while Appendix~\ref{app:additionalResults} presents and proves some results that were only briefly mentioned there. Appendix~\ref{appendix:extraSimulSec} expands on the simulation studies, and Appendix~\ref{appenidix:discreteDRPT} analyzes Algorithm~\ref{alg:star_algo} in the discrete, finite-support setting.

\section{Proofs}\label{appendix:proof}
\subsection{Proofs for Section \ref{sec:methodology}}\label{appendix:proof2}
\begin{proof}[Proof of Theorem \ref{prop:DRPT_validity}]
To establish exchangeability under the null, we adapt ideas from~\cite{berrett2020conditional} and consider an equivalent formulation of the permutation scheme. Let $Z_{()}=\left(Z_{(1)}, \ldots, Z_{(n+m)}\right)$ be the order statistics of $Z$.  When $\mathcal{X} \subseteq \mathbb{R}$, we naturally use the standard ordering on $\mathbb{R}$. In the general case, we can select any total ordering on $\mathcal{X}$; the specific choice does not matter, as its sole purpose is to let us examine the set of 
$Z$-values without needing to know which value is associated with which data point. Define also $Z_{(p)}=\left(Z_{(p(1))}, \ldots, Z_{(p(n+m))}\right)$ for each $p \in \mathcal{S}_{n+m}$, and let $P \in \mathcal{S}_{n+m}$ be the permutation given by the ranks of the true observed vector $Z$, so that $Z=Z_{(P)}$. Under $g \propto r \, f$, we can use Bayes theorem to show that the distribution of the true ranks $P$, conditional on the order statistics $Z_{()}$, is given by
\begin{equation}\label{eq:sampling_2}
    \mathbb{P}\left\{P = p \mid Z_{()} \right\}=\frac{\prod_{i \in \{n+1, \ldots, n+m \}}  r(Z_{(p(i))})}{\sum_{\tilde{p} \in \mathcal{S}_{n+m}}\prod_{i \in \{n+1, \ldots, n+m \}}  r(Z_{(\tilde{p}(i))})}.
\end{equation}
As a result, if $P^{(1)}, \ldots, P^{(H)}$ are drawn independently from~\eqref{eq:sampling_2}, under $H_0$ the vectors $Z,\, Z_{(P^{(1)})}, \ldots, Z_{(P^{(H)})}$ are conditionally independent and identically distributed given $Z_{()}$. Consequently, after marginalising over $Z_{()}$, the sequence $\left(Z, Z_{(P^{(1)})}, \ldots, Z_{(P^{(H)})}\right)$ is exchangeable.

We now conclude the proof by showing that \eqref{eq:sampling_1} and \eqref{eq:sampling_2} are equivalent. For fixed $Z$, and thus $P$, the map $\sigma \mapsto p:=\sigma \circ P$ is a bijection over $\mathcal{S}_{n+m}$. This, together with $Z_{\sigma(i)} = Z_{(\sigma \circ P(i))} = Z_{(p(i))}$, allows showing that the weight $\prod_{i = n+1}^{n+m} r(Z_{\sigma(i)})$ in~\eqref{eq:sampling_1} can be rewritten as $\prod_{i = n+1}^{n+m} r(Z_{(p(i))})$, which corresponds to the weight in~\eqref{eq:sampling_2}. Therefore, under the change of variables $p = \sigma \circ P$, the distribution of $\sigma^{(h)}$ given $Z$ in~\eqref{eq:sampling_1} induces the distribution of $P^{(h)}  = \sigma^{(h)} \circ P$ given $Z_{()}$ in~\eqref{eq:sampling_2}, and conversely. In particular, 
$(\sigma^{(1)},\ldots,\sigma^{(H)})$ and $(P^{(1)},\ldots,P^{(H)})$ have corresponding laws via this bijection, and since $Z_{\sigma^{(h)}} = Z_{( \sigma^{(h)} \circ P)} = Z_{(P^{(h)})}$ for each $h \in [H]$, the joint law of $(Z, Z_{\sigma^{(1)}}, \ldots, Z_{\sigma^{(H)}})$ is the same as that of 
$(Z, Z_{(P^{(1)})}, \ldots, Z_{(P^{(H)})})$. As we have already shown that the latter is exchangeable, this
implies the finite-sample validity of~\eqref{eq:pvalue_DRPT} for every test statistic $T$.
\end{proof}

\begin{proof}[Proof of Proposition \ref{prop:algo_sampler_valid}]
The proof consists of simply checking the detailed balance equations for the Markov chain defined by the algorithm.
Denote with $\mathcal{P}$ the collection of all $K$ couples of indices $ \{(i_1, j_1), \ldots, (i_K, j_K) \}$ such that $(i_1, \ldots, i_K)$ contains distinct elements from $[n]$, and $(j_1, \ldots, j_K)$ contains distinct elements from $\{n+1, \ldots, n+m \}$. For any $\tau \in \mathcal{P}$ and any permutations $p, p'$, we write $p \sim_\tau p'$ if $p$ can be transformed to $p'$ by swapping any subset of the pairs in $\tau$. We now compute the transition probability matrix of the Markov chain defined by Algorithm \ref{alg:pairwise_sampler}. Every probability sign has to be intended conditionally on $Z$. For all $t \in \mathbb{N}_+$ and any permutations $p, p^{\prime}$, we have
$$
\mathbb{P}\left\{\sigma_t = p^{\prime} \mid \sigma_{t-1} = p \right\}=\frac{1}{|\mathcal{P}|} \sum_{\tau \in \mathcal{P}} \mathbb{P}\left\{\sigma_t = p^{\prime} \mid \sigma_{t-1} = p, \tau_t = \tau \right\},
$$
since at each time $t$, Step 3 of the algorithm corresponds to drawing $\tau_t \in \mathcal{P}$ uniformly at random. Next, given $\tau_t = \tau:= \{(i_1, j_1), \ldots, (i_K, j_K) \}$ and $\sigma_{t-1} = p$, it must be the case that $\sigma_t$ satisfies $\sigma_t \sim_\tau p$ by definition of Steps 4-5 of the algorithm. In light of the definition of the odds ratio for each $B_{i_k, j_k}^t$ in~\eqref{eq:p_no_tilde}, we see that for any $p^{\prime}, p^{\prime \prime} \sim_\tau p$, we have
\begin{align}\label{eq:oddsRatio}
    &\frac{\mathbb{P}\left\{\sigma_t = p^{\prime} \mid \sigma_{t-1} = p, \tau_t = \tau \right\}}{\mathbb{P}\left\{\sigma_t = p^{\prime \prime} \mid \sigma_{t-1} = p, \tau_t = \tau \right\}}= \prod_{j \in \{j_1, \ldots, j_K \}} \frac{  r(Z_{p^\prime(j)})}{  r(Z_{p^{\prime \prime}(j)})} = \prod_{j \in \{j_1, \ldots, j_K \}} \frac{  r(Z_{p^\prime(j)})}{  r(Z_{p^{\prime \prime}(j)})} \prod_{j \notin \{j_1, \ldots, j_K \}} \frac{  r(Z_{p^\prime(j)})}{  r(Z_{p^{\prime \prime}(j)})} \nonumber \\
    & = \prod_{j \in \{n+1, \ldots, n+m \}} \frac{  r(Z_{p^\prime(j)})}{  r(Z_{p^{\prime \prime}(j)})} = \frac{\mathbb{P}\left\{\sigma=p^{\prime}\right\}}{\mathbb{P}\left\{\sigma=p^{\prime \prime}\right\}},
\end{align}
where in the second equality we used the fact that $r(Z_{p^\prime(j)})/r(Z_{p^{\prime \prime}(j)}) = 1$ for all $j \notin \{j_1, \ldots, j_K \}$, while in the last step we used the definition of the distribution \eqref{eq:sampling_1}. Therefore,
\[
\mathbb{P}\left\{\sigma_t = p^{\prime} \mid \sigma_{t-1} = p \right\} = \frac{1}{|\mathcal{P}|} \sum_{\tau \in \mathcal{P}} \frac{\mathbbm{1}\left\{p^{\prime} \sim_\tau p \right\} \cdot \mathbb{P}\left\{\sigma=p^{\prime}\right\}}{\sum_{p^{\prime \prime}} \mathbbm{1}\left\{p^{\prime \prime} \sim_\tau p\right\} \cdot \mathbb{P}\left\{\sigma=p^{\prime \prime}\right\}} .
\]
This, together with the fact that $\sim_\tau$ defines an equivalence relation on $\mathcal{P}$, shows that
\begin{align}\label{eq:detailedBalance}
& \mathbb{P}\{\sigma = p\} \cdot \mathbb{P}\left\{\sigma_t = p' \mid \sigma_{t-1} = p \right\} =\frac{1}{|\mathcal{P}|} \sum_{\tau \in \mathcal{P}} \mathbb{P}\{\sigma = p\} \cdot \frac{\mathbbm{1}\left\{p' \sim_\tau p \right\} \cdot \mathbb{P}\left\{\sigma = p'\right\}}{\sum_{p^{\prime \prime}}  \mathbbm{1}\left\{p^{\prime \prime} \sim_\tau p \right\} \cdot \mathbb{P}\left\{\sigma = p^{\prime \prime}\right\}} \nonumber\\
& =\frac{1}{|\mathcal{P}|} \sum_{\tau \in \mathcal{P}} \mathbb{P}\left\{\sigma = p^{\prime}\right\} \cdot \frac{\mathbbm{1}\left\{p \sim_\tau p^{\prime}\right\} \cdot \mathbb{P}\{\sigma =p\}}{\sum_{p^{\prime \prime}} \mathbbm{1}\left\{p^{\prime \prime} \sim_\tau p^{\prime}\right\} \cdot \mathbb{P}\left\{\sigma=p^{\prime \prime}\right\}} =\mathbb{P}\left\{\sigma=p^{\prime}\right\} \cdot \mathbb{P}\left\{\sigma_t=p \mid \sigma_{t-1}=p^{\prime}\right\} .
\end{align}
This verifies the detailed balance equations, and so the Markov chain is reversible and has \eqref{eq:sampling_1} as stationary distribution. Moreover, since $r(x)$ is assumed to be positive for all $x \in \mathcal{X}$, it follows that the chain is aperiodic and irreducible, ensuring the uniqueness of the stationary distribution.
\end{proof}

\begin{proof}[Proof of Theorem~\ref{prop:alg_validity}]
The proof relies fundamentally on the reversibility of the Markov chain defined in Algorithm \ref{alg:pairwise_sampler}, a property established in the proof of Proposition \ref{prop:algo_sampler_valid}. This reversibility permits an alternative but distributionally equivalent sampling procedure whenever the input $\sigma$ is drawn from the stationary distribution: we may first sample $\sigma_*$ from distribution \eqref{eq:sampling_1} conditional on $Z$, and subsequently generate $(\sigma, \sigma^{(1)}, \ldots, \sigma^{(H)})$ through $H + 1$ independent applications of Algorithm \ref{alg:pairwise_sampler}, each running for $S$ steps and initialized at $\sigma_*$. The independence of these runs, combined with their shared initialization point $\sigma_*$, ensures that $(\sigma, \sigma^{(1)}, \ldots, \sigma^{(H)})$ are independently and identically distributed when conditioned on $\sigma_*$ and $Z$. This conditional independence directly implies the exchangeability of $\left(Z, Z^{(1)}, \ldots, Z^{(H)} \right)$ since, under $H_0$, the distribution of $Z_{\sigma_0} \mid Z_{()}$, with $\sigma_0 = \mathrm{id}_{n+m}$, coincides with the distribution of $Z_{\sigma} \mid Z_{()}$, where $\sigma$ is sampled from~\eqref{eq:sampling_1}.  This concludes the proof.
\end{proof}
 
 \subsection{Proofs for Section \ref{sec:IPM+consistency}}\label{appendix:proof3.1}
 \begin{proof}[Proof of Lemma \ref{lemma:unique_h}]
 We begin by proving that there exists a unique $\lambda_0$ such that $h = \frac{nf + mg}{n + \lambda_0 m r}$ is a density function. Define the function
\begin{align}\label{eq:F(lambda)}
     F :
     \begin{cases}
        \mathbb{R}_+ \to \mathbb{R}_+ \\
        \lambda \mapsto  \int \frac{nf + mg}{n + \lambda m r} d\mu
     \end{cases}
\end{align}
which can easily be seen to be continuous. It also straightforward to see that $F$ is strictly decreasing, since if $\lambda_2 > \lambda_1 > 0$ we have 
\[
     \frac{nf(x) + mg(x)}{n + \lambda_2 m r(x)} \leq
     \frac{nf(x) + mg(x)}{n + \lambda_1 m r(x)} \quad \text{ for all } x \in \mathcal{X},
\] 
and this inequality is strict in the support of $f$ and $g$ as we are assuming that $r(x)>0$ for all $x \in \mathcal{X}$. It is clear that $\lim_{\lambda \rightarrow 0^+} F(\lambda) = 1 + m/n > 1$, while $\lim_{\lambda \rightarrow \infty} F(\lambda) = 0<1$. Thus, by the intermediate value theorem and the strict monotonicity of $F$ we have established our first claim that $\lambda_0$ exists and is unique.

It can now be seen that $h$ satisfies the requirements of the result. Indeed, we have 
     \begin{align*}
         1 & = \int h d\mu = \int \frac{nf + mg}{n + \lambda_0 m r} d\mu = \frac{1}{n} \int \frac{n}{n + \lambda_0 m r}(nf + mg) d\mu \\
         & = \frac{1}{n} \int \left( 1 - \frac{\lambda_0 m r}{n + \lambda_0 m r}\right)(nf + mg)d\mu = \frac{1}{n} \left(\int (nf + mg)d\mu - \int (\lambda_0 m r) \frac{nf + mg}{n + \lambda_0 m r} d\mu \right) \\
         & = \frac{1}{n} \left(n+m - \lambda_0 m  \int r h d\mu \right) = 1 + \frac{m}{n} - \lambda_0 \frac{m}{n} \int r h d\mu,
     \end{align*}
     which implies that $\lambda_0 = (\int r h d\mu)^{-1}$. This shows us that $n h + m (\int r h d\mu)^{-1} r h  = n h + \lambda_0 m r h  = nf + m g$, so that $h$ gives rise to null distributions for which the distribution of the combined sample matches that of our data, concluding the existence part of the result.

     As for its uniqueness, observe that any density $h$ preserving the combined distribution under the null hypothesis must be of the form $h$ for some $\lambda_0$ positive, since  
     \begin{align*}
    n f + m g & = n h + m\frac{r h}{\int r f d\mu} = \left\{n + m r\left(\int r h d\mu\right)^{-1}\right\}h = \left(n + \lambda_0 m r \right)h.
     \end{align*}
     But the $\lambda_0$ such that $h$ integrates to $1$ is unique. This completes the proof.
 \end{proof}

\begin{proof}[Proof of Lemma \ref{lemma:TFcharacterisesNull}]
Suppose that the null hypothesis holds, so that $g=rf/\int rf d\mu$. It follows from Lemma~\ref{lemma:unique_h} and its proof that $f=h$ and that $\lambda_0=(\int rf d\mu)^{-1}$. We therefore see that $\lambda_0 r f - g =0$, so indeed $T_{\mathcal{F},r}(f,g)=0$.

We now turn to the reverse implication. Assuming that $f$ and $g$ are such that $T_{\mathcal{F},r}(f,g)=0$, we will show that $H_0$ must be true. By assumption, $\mathcal{H}$ is dense in $C^0_b(\mathcal{X})$ with respect to $\| \cdot \|_\infty$, so that for all $\phi \in C^0_b(\mathcal{X})$ and all $\epsilon > 0$ there exists $\varphi \in \mathcal{H}$ such that $\|\varphi - \phi \|_{\infty} \leq \epsilon$. 
Thus \begin{align*}
   \frac{n+m}{m}&\bigg| \int \frac{\lambda_0 r f - g}{n/m + \lambda_0 r} \phi\,  d\mu \bigg| \\
    &\leq \frac{n+m}{m} \left| \int \frac{\lambda_0 r f}{n/m + \lambda_0 r} (\phi - \varphi) d\mu \right| + \frac{n+m}{m}\left| \int \frac{\lambda_0 r f - g}{n/m + \lambda_0 r} \varphi d\mu \right| + \frac{n+m}{m} \left| \int \frac{g}{n/m + \lambda_0 r} (\phi - \varphi) d\mu \right| \\
    & \leq \int \frac{n+m}{n+\lambda_0 m r} \, \lambda_0 r f \, |\phi - \varphi| d\mu  + \|\varphi\|_\mathcal{H} \,T_{\mathcal{F},r}(f,g) + \int \frac{n+m}{n+ \lambda_0 m r} \,g \,|\phi - \varphi| d\mu  \\
    & = \int \frac{n+m}{n+\lambda_0 m r} \, \lambda_0 r f \, |\phi - \varphi| d\mu   + \int \frac{n+m}{n+ \lambda_0 m r} \,g \, |\phi - \varphi| d\mu \\
    & \leq \epsilon \, \int \frac{\lambda_0 r}{\min(1, \lambda_0 r)} \, f d\mu  + \epsilon \, \int \frac{1}{\min(1, \lambda_0 r)} \,g d\mu = \epsilon \left(\mathbb{E}[\max\{1, \lambda_0 r(X)\}] + \mathbb{E}[\max\{1, 1/\lambda_0 r(Y)\}] \right) \\
    & \leq \epsilon \, \left(2 + \mathbb{E}[\lambda_0 r(X)] + \mathbb{E}[1/\lambda_0 r(Y)]  \right) = \epsilon \, \left(2 + \lambda_0 \int rf \, d\mu + \lambda_0^{-1} \int g/r \, d\mu  \right),
\end{align*}
where in the first equality we used our assumption that $T_{\mathcal{F},r}(f,g) = 0$.  Because $\epsilon>0$ was arbitrary, and since $\int r f\,d\mu<\infty$, $\int g/r\,d\mu<\infty$ by assumption, and $\lambda_0<\infty$ by Proposition~\ref{prop:limitingTF} in Appendix~\ref{app:additionalResults}, we now see that $ \frac{n+m}{m}\int \frac{\lambda_0 r f - g}{n/m + \lambda_0 r} \phi \, d\mu = 0$ for all $\phi \in C^0_b(\mathcal{X})$, which further implies 
\[
 \frac{mg}{n + \lambda_0 m r} = \frac{\lambda_0 m r f}{n + \lambda_0 m r} \text{ a.s.}
\]
by Lemma 9.3.2 in \cite{Dudley_2002}. Simplifying this inequality we see that $g=\lambda_0 rf \propto rf$, as required.
 \end{proof}

\begin{proof}[Proof of Theorem \ref{thm:general_consistency}]
We assume throughout the proof that $H_0$ does not hold and show that $ \mathbb{P}\{p > \alpha \} \rightarrow 0$ as $n,m \rightarrow \infty$. Recall our notation $\bar{R}$ and $\check{R}$ for the arithmetic and harmonic means of the $r$ values, respectively, and write $A$ for the event that $\bar{R} \leq 2 \mathbb{E}(\bar{R})$ and $1/\check{R} \leq 2\mathbb{E}(1/\check{R})$. On this event, we have
\[
    \bar{R}/\check{R} \leq 4\mathbb{E}(\bar{R}) \mathbb{E}(\check{R}^{-1}) \leq 4 \sqrt{ \mathbb{E}\{ r(Z)^2\} \mathbb{E}\{r(Z)^{-2}\}} \leq 4\sqrt{V_0},
\]
which will allow us to apply Lemma~\ref{lemma:convergence_in_d}. Since we assume that $H > \lceil 1/\alpha - 1\rceil$ we have $\alpha(1+H)-1 > 0$, and we can therefore use Markov's inequality to see that
 \begin{align*}
     \mathbb{P}&\{p > \alpha \}  =  \mathbb{P}\left\{1+\sum_{h=1}^H \mathbbm{1}\{T(Z_{\sigma^{(h)}}) \geq T(Z)\} > \alpha(1+H) \right\} \\
     & \leq \frac{\mathbb{E}\left[ \sum_{h=1}^H \mathbbm{1}\{T(Z_{\sigma^{(h)}}) \geq T(Z)\} \right]}{\alpha(1+H)-1} = \frac{H}{\alpha(1+H)-1}\mathbb{P}\{T(Z_{\sigma^{(1)}})\geq T(Z) \} \\
     & \leq \frac{H}{\alpha(1+H)-1} \bigl[ \mathbb{P}\{T(Z_{\sigma^{(1)}})\geq T(Z), A \} + \mathbb{P}(A^\complement) \bigr],
 \end{align*}
 where the penultimate step follows from exchangeability. We begin by controlling $\mathbb{P}(A^\complement)$ by combining Chebyshev's inequality with the fact that $\max(\mathbb{E}[\left\{  \lambda_0 r(Z) \right\}^2 ], \,\mathbb{E}[\left\{ \lambda_0 r(Z) \right\}^{-2}] ) \leq V_0$. In particular, this follows from $\mathbb{E}[\left\{  r(Z) \right\}^2 ] \,\mathbb{E}[\left\{  r(Z) \right\}^{-2}]  \leq V_0$ and \[ \frac{1}{\mathbb{E}(1/\check{R})} = \frac{n+m}{n \int f/r \, d\mu + m\int g/r \, d\mu} \leq \frac{1}{\lambda_0} \leq \frac{1}{n+m} \biggl( n \int rf \, d\mu + m \int rg d \mu\biggr) = \mathbb{E}(\bar{R}),\]
 where the latter holds by Jensen's inequality. Together, these allow showing
\begin{align}\label{eq:ProbA_Complement}
     \mathbb{P}(A^\complement) 
     & \leq \mathbb{P} \biggl( \frac{\bar{R}}{\mathbb{E}(\bar{R})} > 2 \biggr) + \mathbb{P} \biggl( \frac{1}{\check{R} \, \mathbb{E}(1/\check{R})} > 2 \biggr) \leq  \frac{\mathrm{Var}(\bar{R})}{\mathbb{E}^2(\bar{R})} + \frac{\mathrm{Var}(1/\check{R})}{\mathbb{E}^2(1/\check{R})} \leq  \mathrm{Var} \bigl( \lambda_0 \bar{R} \bigr) + \mathrm{Var} \biggl( \frac{1}{\lambda_0 \check{R}} \biggr) \leq \frac{2V_0}{m+n}.
 \end{align}
Using the shorthand $\sigma=\sigma^{(1)}$ it now suffices to consider $T(Z_{\sigma})$ and $T(Z)$ on the event $A$.

We will first prove that $\mathbb{E}[|T(Z_{\sigma})| \mathbbm{1}_A] = \mathbb{E}[T(Z_{\sigma})\mathbbm{1}_A] \rightarrow 0$. Write 
\[
W_\varphi^\sigma := \frac{1}{n} \left(\sum_{i  = 1}^{n} \frac{\hat{\lambda}m r(Z_{\sigma(i)})}{n + \hat{\lambda}m r(Z_{\sigma(i)})} \varphi(Z_{\sigma(i)}) - \sum_{j  = n+1}^{n+m} \frac{n}{n + \hat{\lambda}m r(Z_{\sigma(j)})} \varphi(Z_{\sigma(j)})\right),
\]
so that we have $T(Z_\sigma) = \frac{n+m}{m}\sup_{\| \varphi \|_\mathcal{H} \leq 1} |W_\varphi^\sigma|$. In particular, $W_\varphi^\sigma$ coincides with the term inside the conditional expectation in the statement of Lemma~\ref{lemma:convergence_in_d}. Write $N \equiv N(\delta) \equiv N\left(\{\|\varphi\|_\mathcal{H} \leq 1\}, \delta, \|\cdot \|_\infty \right)$ for the covering number of $\{\|\varphi\|_\mathcal{H} \leq 1\}$ with respect to $\| \cdot \|_\infty$ and let $\{\psi_1, \ldots, \psi_N\}$ be an associated $\delta$-cover. Now, for a generic function $\varphi_0 \in \mathcal{H}$ such that $\|\varphi_0\|_\mathcal{H} \leq 1$ we have
\begin{align}
\label{Eq:ConsistencyDecompositionNew}
    \mathbb{E}&\left[\sup_{\| \varphi \|_\mathcal{H} \leq 1} |W_\varphi^\sigma| \, \mathbbm{1}_A\right] = \mathbb{E}\left[\sup_{\| \varphi \|_\mathcal{H} \leq 1} (|W_\varphi^\sigma| - |W_{\varphi_0}^\sigma|)\mathbbm{1}_A \right] + \mathbb{E}[|W_{\varphi_0}^\sigma|\mathbbm{1}_A ] \nonumber \\
    & \leq \mathbb{E}\left[\sup_{\substack{\| \varphi_1 \|_\mathcal{H} \leq 1 \\ \| \varphi_2 \|_\mathcal{H} \leq 1}} (|W_{\varphi_1}^\sigma| - |W_{\varphi_2}^\sigma|) \mathbbm{1}_A \right] + \mathbb{E}[|W_{\varphi_0}^\sigma| \mathbbm{1}_A] \leq \mathbb{E}\left[\sup_{\substack{\| \varphi_1 \|_\mathcal{H} \leq 1 \\ \| \varphi_2 \|_\mathcal{H} \leq 1}} \left|W_{\varphi_1}^\sigma - W_{\varphi_2}^\sigma\right| \mathbbm{1}_A \right] + \mathbb{E}[|W_{\varphi_0}^\sigma|\mathbbm{1}_A] \nonumber \\
    & \leq 2 \mathbb{E}\left[\sup_{\| \varphi_1 - \varphi_2\|_\infty \leq \delta} \left|W_{\varphi_1}^\sigma - W_{\varphi_2}^\sigma\right|\right] + 2 \mathbb{E}\left[\max_{i \in [N]} \left|W_{\psi_1}^\sigma - W_{\psi_i}^\sigma\right|\mathbbm{1}_A\right] + \mathbb{E}[|W_{\varphi_0}^\sigma|\mathbbm{1}_A] \nonumber \\
    &\leq 2 \mathbb{E} \left[ \sup_{\|\varphi\|_\infty \leq \delta} |W_\varphi^\sigma| \right] + (4N+1) \sup_{\|\varphi\|_\mathcal{H} \leq 1} \mathbb{E}[|W_{\varphi}^\sigma|\mathbbm{1}_A ],
\end{align}
where the penultimate inequality follows from a one-step discretization argument as in \citet[][Equation 5.34]{wainwright2019high} and the final inequality follows on bounding the maximum by a sum, using the fact that $W_{\varphi_1}^\sigma - W_{\varphi_2}^\sigma = W_{\varphi_1 - \varphi_2}^\sigma$ for any $\varphi_1,\varphi_2$, and using the triangle inequality to say that $\max_{i \in [N]}\|\psi_1-\psi_i\|_\mathcal{H} \leq 2$. We will now bound each term on the right-hand side of~\eqref{Eq:ConsistencyDecompositionNew} separately. As for the first, observe that by~\eqref{Eq:EqualSumsLambda} we have
\begin{align*}
    |W_{\varphi}^\sigma| & = \frac{1}{n} \left|\sum_{i  = 1}^{n} \frac{\hat{\lambda}m r(Z_{\sigma(i)})}{n + \hat{\lambda}m r(Z_{\sigma(i)})} \varphi(Z_{\sigma(i)}) - \sum_{j  =n+1}^{n+m} \frac{n}{n + \hat{\lambda}m r(Z_{\sigma(j)})} \varphi(Z_{\sigma(j)})\right| \\
    & \leq \|\varphi\|_\infty \biggl\{ \frac{1}{n} \sum_{i=1}^n \frac{\hat{\lambda}m r(Z_{\sigma(i)})}{n + \hat{\lambda}m r(Z_{\sigma(i)})} + \sum_{j=n+1}^{n+m} \frac{1}{n + \hat{\lambda}m r(Z_{\sigma(j)})} \biggr\} \\
    & = \frac{2}{n} \|\varphi\|_\infty \min \biggl\{ \sum_{i=1}^n \frac{\hat{\lambda}m r(Z_{\sigma(i)})}{n + \hat{\lambda}m r(Z_{\sigma(i)})}, \, \sum_{j=n+1}^{n+m} \frac{n}{n + \hat{\lambda}m r(Z_{\sigma(j)})} \biggr\} \leq \frac{2\min(m,n)}{n} \|\varphi\|_\infty
\end{align*}
for any $\varphi \in \mathcal{H}$, which implies that  
\[
    \mathbb{E} \left[ \sup_{\|\varphi\|_\infty \leq \delta} |W_\varphi^\sigma| \right] \leq \frac{2\min(m,n)}{n} \delta.
\]
We now turn to the second term in~\eqref{Eq:ConsistencyDecompositionNew}. With our assumption that $\|\cdot\|_\infty \leq \gamma \|\cdot\|_\mathcal{H}$ for a constant $\gamma$, it follows from the Cauchy--Schwarz inequality,~\eqref{eq:bound(I)^2} and Lemma~\ref{lemma:convergence_in_d} that for any $\varphi$ such that $\|\varphi\|_\mathcal{H} \leq 1$ we have
\[
    \mathbb{E}[|W_{\varphi_0}^\sigma|\mathbbm{1}_A ] \leq \sqrt{ \mathbb{E}\left[ (W_\varphi^\sigma)^2 \mathbbm{1}_A \right] } \leq \gamma \biggl\{ \frac{16 \,V_0^{1/4} \min(m^{1/2},n^{1/2})}{n} + \exp\biggl( - \frac{S}{64\tau V_0^{1/2}} \biggr) \biggr\}.
\]
Putting these bounds together and using  $T(Z_\sigma) = \frac{n+m}{m}\sup_{\| \varphi \|_\mathcal{H} \leq 1} |W_\varphi^\sigma|$, we see that for any $\delta>0$ we have
\begin{align}\label{eq:PermutedGoesToZero}
    \mathbb{E} & [ T(Z_\sigma) \mathbbm{1}_A] \nonumber \\
    & \leq \frac{4\min(m,n) (n+m)}{nm} \delta + 
    (4N(\delta)+1) \, \gamma \, \biggl\{ \frac{16 \, V_0^{1/4} \sqrt{\min(n,m)} (n+m)}{nm} + \frac{n+m}{m} \exp\biggl( - \frac{S}{64\tau V_0^{1/2}} \biggr) \biggr\} \nonumber \\
    & \leq 8\delta +  
    (4N(\delta)+1) \, \gamma \, \biggl\{ \frac{32 \, V_0^{1/4}}{\sqrt{\min(n,m)}} + \frac{n+m}{m} \exp\biggl( - \frac{S}{64\tau V_0^{1/2}} \biggr) \biggr\}.
\end{align}
Since we assume that $N\left(\{\|\varphi\|_\mathcal{H} \leq 1\}, \delta, \|\cdot \|_\infty \right)$ is finite for all $\delta > 0$, and that $S \gg \tau \log(m+n)$, there exists a sequence $(\delta_n)_{n \in \mathbb{N}_+}$ such that this right-hand side converges to zero. \\

We now turn to the unpermuted statistic $T(Z)$. Recall the population quantity 
\[
T_{\mathcal{F},r} \equiv T_{\mathcal{F},r}(f,g) =  \, \sup_{\|\varphi\|_\mathcal{H} \leq 1} |I_\varphi|, \quad \text{ where }  \quad I_\varphi := \frac{n+m}{m} \int \frac{\lambda_0 r f - g}{n/m + \lambda_0 r} \varphi \, d\mu 
\]
with $\lambda_0$ such that $\int \frac{n f + m g}{n + \lambda_0 m r} d\mu = 1$. As we are working under the alternative hypothesis, we know by Lemma~\ref{lemma:TFcharacterisesNull} and Proposition~\ref{prop:limitingTF} in Appendix~\ref{app:additionalResults} that $T_{\mathcal{F},r}>0$ and we may fix $\varphi$ such that $\|\varphi\|_\mathcal{H} \leq 1$ and $I_\varphi \geq T_{\mathcal{F},r}/2$. Then we have
\begin{align*}
    T(Z) &\geq \frac{n+m}{nm} \sum_{i=1}^n \frac{\hat{\lambda}mr(X_i)}{n + \hat{\lambda} m r(X_i)} \varphi(X_i) - \frac{n+m}{nm} \sum_{j=n+1}^{n+m} \frac{n}{n + \hat{\lambda} m r(Y_j)} \varphi(Y_j) \\
    & \geq \frac{1}{2} T_{\mathcal{F},r} + \frac{n+m}{m} \biggl( \frac{1}{n} \sum_{i=1}^n \frac{\hat{\lambda} m r(X_i) \varphi(X_i)}{n + \hat{\lambda} m r(X_i)}  - \int \frac{\lambda_0 m r}{n + \lambda_0 mr} \varphi f - \frac{1}{m} \sum_{j=n+1}^{n+m} \frac{m \varphi(Y_j)}{n + \hat{\lambda} m r(Y_j)} + \int \frac{m}{n + \lambda_0 mr} \varphi g \biggr) \\
    & \geq \frac{1}{2} T_{\mathcal{F},r} - \gamma \, \biggl| \frac{\hat{\lambda}}{\lambda_0} -1 \biggr| \, \frac{n+m}{m}  \sum_{i = 1}^{n+m} \frac{\lambda_0 m r(Z_i)}{(n+ \lambda_0 m r(Z_i))(n+ \hat \lambda m r(Z_i))}  \\
    & \quad \quad - \frac{n+m}{m} \biggl| \frac{1}{n} \sum_{i=1}^n \frac{\lambda_0 m r(X_i) \varphi(X_i)}{n + \lambda_0 m r(X_i)}  - \int \frac{\lambda_0 m r}{n + \lambda_0 mr} \varphi f - \frac{1}{m} \sum_{j=n+1}^{n+m} \frac{m \varphi(Y_j)}{n + \lambda_0 m r(Y_j)} + \int \frac{m}{n + \lambda_0 mr} \varphi g \biggr| \\
    & \geq \frac{1}{2} T_{\mathcal{F},r} - \gamma \, \biggl| \frac{\hat{\lambda}}{\lambda_0} -1 \biggr| \, \frac{1}{n+m}  \sum_{i = 1}^{n+m} \frac{\max\{1, \lambda_0 r(Z_i)\}}{\min\{1, \hat \lambda r(Z_i)\}}  \\
    & \quad \quad - \frac{n+m}{m} \biggl| \frac{1}{n} \sum_{i=1}^n \frac{\lambda_0 m r(X_i) \varphi(X_i)}{n + \lambda_0 m r(X_i)}  - \int \frac{\lambda_0 m r}{n + \lambda_0 mr} \varphi f - \frac{1}{m} \sum_{j=n+1}^{n+m} \frac{m \varphi(Y_j)}{n + \lambda_0 m r(Y_j)} + \int \frac{m}{n + \lambda_0 mr} \varphi g \biggr| \\
    & =: \frac{1}{2} T_{\mathcal{F},r} - \gamma \, \biggl| \frac{\hat{\lambda}}{\lambda_0} -1 \biggr| \, \frac{1}{n+m}  \sum_{i = 1}^{n+m} \frac{\max\{1, \lambda_0 r(Z_i)\}}{\min\{1, \hat \lambda r(Z_i)\}} - R_1.
\end{align*}
The final term is a sum of independent mean-zero random variables, so that we may use Chebyshev's inequality to control the probability of this event. Indeed, using again the boundedness assumption on the second moments of $r$ we can see that 
\begin{align}\label{eq:HoeffdingUnpermuted}
        \mathbb{P}&(R_1 \geq \epsilon) \nonumber \\
        & =\mathbb{P} \biggl(\frac{n+m}{m} \biggl| \frac{1}{n} \sum_{i=1}^n \frac{\lambda_0 m r(X_i) \varphi(X_i)}{n + \lambda_0 m r(X_i)}  - \int \frac{\lambda_0 m r}{n + \lambda_0 mr} \varphi f - \frac{1}{m} \sum_{j=n+1}^{n+m} \frac{m \varphi(Y_j)}{n + \lambda_0 m r(Y_j)} + \int \frac{m}{n + \lambda_0 mr} \varphi g \biggr| \geq \epsilon \biggr) \nonumber \\
        & \leq \frac{(n+m)^2}{\epsilon^2 m^2} \left\{n^{-1} \operatorname{Var}\left(\frac{\lambda_0 m r(X) \varphi(X)}{n + \lambda_0 m r(X)} \right) +m^{-1} \operatorname{Var}\left(\frac{ m \varphi(Y)}{n + \lambda_0 m r(Y)} \right)\right\} \nonumber \\
         & \leq \frac{\gamma^2(n+m)^2}{\epsilon^2 m^2} \biggl\{n^{-1} \min \biggl( \frac{m^2}{n^2} \mathbb{E}[\lambda_0^2 r^2(X)] , \, 1 \biggr) + \frac{m}{n^2} \min \biggl(1, \frac{n^2}{m^2} \mathbb{E}[\lambda_0^{-2} r^{-2}(Y)] \biggr) \biggr\} \nonumber \\
    & \leq \frac{\gamma^2 (n+m)^2}{\epsilon^2 m^2} \biggl\{\frac{1}{n} \min \biggl( \frac{m^2(m+n)}{n^3} V_0, \, 1 \biggr) + \frac{m}{n^2} \min \biggl(1, \frac{n^2(m+n)}{m^3} V_0 \biggr) \biggr\} \nonumber \\
    & \leq \frac{\gamma^2 (n+m)^2}{\epsilon^2 m^2} V_0 \biggl\{ \mathbbm{1}_{\{n \leq m\}} \biggl( \frac{1}{n} + \frac{m+n}{m^2} \biggr) + \mathbbm{1}_{\{n > m\}} \biggl( \frac{m^2(m+n)}{n^4} + \frac{m}{n^2} \biggr) \biggr\} \leq \frac{12 \gamma^2 V_0}{\epsilon^2 \, \min(n,m)}.
\end{align}
It remains for us to control the term depending on $\hat{\lambda}/\lambda_0$ under the event $A$. We first note that our assumption on $r$ gives us bounds on $\hat{\lambda}$ and $\lambda_0$. Exactly as in the proof of Lemma~\ref{lemma:convergence_in_d}, on the event $A$ we have that $1/\{2\mathbb{E}(\bar{R})\} \leq \hat{\lambda} \leq 2 \mathbb{E}(\check{R}^{-1})$. As also claimed above, by a very similar argument we also have that $1/\mathbb{E}(\bar{R}) \leq \lambda_0 \leq \mathbb{E}(\check{R}^{-1})$. We may view $\hat{\lambda}$ as an $M$-estimator of $\lambda_0$. Indeed, writing $\phi_\lambda(z) = (n+m)/\{n+m\lambda r(z)\}-1$, we defined $\hat{\lambda}$ to be the solution to
\[
    \hat{S}_{m,n}(\lambda) := \frac{1}{n+m} \sum_{i=1}^{n+m} \phi_\lambda(Z_i) = 0.
\]
Similarly, $\lambda_0$ can be thought of as the solution to $S_{m,n}(\lambda):= \mathbb{E}\{\hat{S}_{m,n}(\lambda)\} = 0$.  By Jensen's inequality we have that
\begin{align*}
    - \hat{S}_{m,n}'(\lambda) = \sum_{i=1}^{n+m} \frac{mr(Z_i)}{\{n+\lambda mr(Z_i)\}^2} \geq \frac{(n+m)^2}{\sum_{i=1}^{n+m} \frac{\{n+\lambda mr(Z_i)\}^2}{mr(Z_i)} } = \frac{m(n+m)}{n^2\check{R}^{-1} + 2\lambda m n + \lambda^2 m^2 \bar{R}},
\end{align*}
which is decreasing in $\lambda$. Thus, for any $\epsilon \in (0,1)$ we have
\begin{align*}
    \mathbb{P}&(|\hat{\lambda}/\lambda_0 - 1| \geq \epsilon, A) = \mathbb{P}(|\hat{\lambda}-\lambda_0| \geq \epsilon \lambda_0, A) \leq \mathbb{P}\biggl(|\hat{S}_{m,n}(\hat{\lambda}) - \hat{S}_{m,n}(\lambda_0) | \geq \epsilon \lambda_0 \bigl|\hat{S}_{m,n}'\bigl((1+\epsilon)\lambda_0 \bigr) \bigr|, A \biggr) \\
    & \leq \mathbb{P}\biggl(|\hat{S}_{m,n}(\lambda_0) | \geq \epsilon \lambda_0 \bigl|\hat{S}_{m,n}'\bigl(2\lambda_0\bigr) \bigr|, A \biggr) \leq \mathbb{P}\biggl(|\hat{S}_{m,n}(\lambda_0) | \geq \frac{(m+n)m \lambda_0 \epsilon}{2\{n^2\mathbb{E}(\check{R}^{-1}) + 2\lambda_0 m n + \lambda_0^2 m^2 \mathbb{E}(\bar{R})\}} \biggr) \\
    & \leq \mathbb{P}\biggl(|\hat{S}_{m,n}(\lambda_0) | \geq \frac{(m+n)m\epsilon}{2} \min \biggl[ \frac{ 1 }{n^2\mathbb{E}(\check{R}^{-1})\mathbb{E}(\bar{R}) + 2 m n  + m^2 }, \,  \frac{ 1}{n^2+ 2 m n +  m^2 \mathbb{E}(\bar{R})\mathbb{E}(\check{R}^{-1})} \biggr] \biggr) \\
    & \leq \mathbb{P}\biggl(|\hat{S}_{m,n}(\lambda_0) | \geq \frac{m\epsilon}{2 V_0^{1/2} (m+n) } \biggr),
\end{align*}
where the last step follows from the boundedness assumption. Now $\hat{S}_{m,n}(\lambda_0)$ is a sum of independent mean-zero random variables, so that we may again use Chebyshev's inequality to control the probability of this event. Indeed, it is straightforward to see that for $\delta=m\epsilon/\{2V_0^{1/2}(m+n)\}$ we have
\begin{align*}
    \mathbb{P}\bigl(|\hat{S}_{m,n}(\lambda_0) | &\geq \delta  \bigr) \leq \delta^{-2} \mathrm{Var} \bigl( \hat{S}_{m,n}(\lambda_0) \bigr) = \frac{1}{\delta^2(m+n)^2} \bigl\{ n \mathrm{Var}\, \phi_{\lambda_0}(X) + m\mathrm{Var}\, \phi_{\lambda_0}(Y) \bigr\} \\
    & = \frac{4V_0}{\epsilon^2} \biggl\{ n \mathrm{Var} \biggl( \frac{1-\lambda_0 r(X)}{n + m \lambda_0 r(X)} \biggr)+m\mathrm{Var} \biggl( \frac{1-\lambda_0 r(Y)}{n + m \lambda_0 r(Y)} \biggr) \biggr\} \\
    & \leq \frac{4V_0}{\epsilon^2(m+n)^2} \biggl[ n \mathbb{E} \biggl\{ \frac{\{1-\lambda_0r(X)\}^2}{\min^2(1,\lambda_0 r(X))} \biggr\} + m \mathbb{E} \biggl\{ \frac{\{1-\lambda_0r(Y)\}^2}{\min^2(1,\lambda_0 r(Y))} \biggr\}  \biggr] \\
    & \leq \frac{4V_0}{\epsilon^2(m+n)^2} \biggl[ n \mathbb{E} \biggl\{ \bigl(\lambda_0 r(X)\bigr)^2 + \biggl( \frac{1}{\lambda_0 r(X) } \biggr)^2 \biggr\} + m \mathbb{E} \biggl\{ \bigl(\lambda_0 r(Y)\bigr)^2 + \biggl( \frac{1}{\lambda_0 r(Y) } \biggr)^2 \biggr\} \biggr] \\
    & \leq \frac{8 V_0^2}{\epsilon^2(m+n)}.
\end{align*}
Finally, since on the event $A$ we have $\bar{R} \leq 2 \mathbb{E}(\bar{R})$, $\check{R}^{-1} \leq 2\mathbb{E}(\check{R}^{-1})$, $1/\{2\mathbb{E}(\bar{R})\} \leq \hat{\lambda} \leq 2 \mathbb{E}(\check{R}^{-1})$, $1/\mathbb{E}(\bar{R}) \leq \lambda_0 \leq \mathbb{E}(\check{R}^{-1})$ and $\mathbb{E}(\bar{R}) \mathbb{E}(\check{R}^{-1}) \leq V_0$, it follows that
\begin{align*}
    &\frac{1}{n+m} \sum_{i = 1}^{n+m} \frac{\max\{1, \lambda_0 r(Z_i)\}}{\min\{1, \hat \lambda r(Z_i)\}} \leq  \frac{1}{n+m} \sum_{i = 1}^{n+m}  \left( 1 + \lambda_0 r(Z_i)\right) \left( 1 + \frac{1}{\hat \lambda r(Z_i)}\right) \\
    & \leq  \frac{1}{n+m} \sum_{i = 1}^{n+m} \left(1 + \lambda_0 r(Z_i) + \frac{1}{\hat \lambda r(Z_i)} + \frac{\lambda_0}{\hat \lambda} \right)\leq \frac{1}{n+m} \sum_{i = 1}^{n+m} \left(1 + \mathbb{E}(\check{R}^{-1}) r(Z_i) + \frac{2\mathbb{E}(\bar{R}) }{r(Z_i)} + 2\mathbb{E}(\check{R}^{-1}) \mathbb{E}(\bar{R}) \right)\\
    & = 1 + \mathbb{E}(\check{R}^{-1}) \, \bar R + 2\mathbb{E}(\bar{R}) \, \check R^{-1} + 2\mathbb{E}(\check{R}^{-1}) \mathbb{E}(\bar{R}) \leq  1 + 2 \mathbb{E}(\check{R}^{-1}) \, \mathbb{E}(\bar{R}) + 4\mathbb{E}(\bar{R}) \, \mathbb{E}(\check R^{-1}) + 2\mathbb{E}(\check{R}^{-1}) \mathbb{E}(\bar{R}) \\
    & \leq 1 + 8\mathbb{E}(\check{R}^{-1}) \mathbb{E}(\bar{R}) \leq 1 + 8V_0^{1/2}.
\end{align*}
This, in conjunction with~\eqref{eq:HoeffdingUnpermuted} and the previous display, gives
\begin{align*}
    \mathbb{P}&(T(Z) < \tfrac{1}{4} T_{\mathcal{F},r}, A) \\
    & \leq  \mathbb{P}\left(\frac{1}{2} T_{\mathcal{F},r} - \gamma (1+8V_0^{1/2}) \, \biggl| \frac{\hat{\lambda}}{\lambda_0} -1 \biggr| \, - R_1 < \tfrac{1}{4}  T_{\mathcal{F},r}, A  \right) \leq \mathbb{P}\left(\gamma (1+8V_0^{1/2}) \, \biggl| \frac{\hat{\lambda}}{\lambda_0} -1 \biggr| \, + R_1 > \tfrac{1}{4} T_{\mathcal{F},r}, A  \right) \\
    & \leq \mathbb{P}\left(\gamma (1+8V_0^{1/2}) \, \biggl| \frac{\hat{\lambda}}{\lambda_0} -1 \biggr|  > \tfrac{1}{8} T_{\mathcal{F},r}, A  \right) + \mathbb{P}\left(R_1 > \tfrac{1}{8} T_{\mathcal{F},r}, A  \right) \leq \frac{2^{9} V_0^2 \gamma^2 (1+8V_0^{1/2})^2}{T_{\mathcal{F},r}^2 (n+m)} + \frac{2^{10} \gamma^2 V_0}{T_{\mathcal{F},r}^2 \min(n,m)}.
\end{align*}
We can thus bound 
\begin{align*}
    \mathbb{P}(T(Z_\sigma) \geq T(Z), A) & \leq \mathbb{P}(T(Z_\sigma) \geq \tfrac{1}{4} T_{\mathcal{F},r}, A) + \mathbb{P}(T(Z) < \tfrac{1}{4} T_{\mathcal{F},r}, A) \\
    & \leq \frac{4 \mathbb{E}[T(Z_\sigma) \mathbbm{1}_A]}{T_{\mathcal{F},r}} + \frac{2^{9} V_0^2 \gamma^2 (1+8V_0^{1/2})^2}{T_{\mathcal{F},r}^2 (n+m)} + \frac{2^{10} \gamma^2 V_0}{T_{\mathcal{F},r}^2 \min(n,m)},
\end{align*}
and use~\eqref{eq:PermutedGoesToZero} to show that the last display goes to zero. This, combined with~\eqref{eq:ProbA_Complement}, concludes the proof.
\end{proof}

\begin{proof}[Proof of Lemma \ref{lemma:convergence_in_d}]
Let $(\sigma_0,\sigma_1,\ldots,\sigma_S)$ be the permutations generated through Algorithm~\ref{alg:pairwise_sampler}. For $t \in \{0,1,\ldots,S\}$ define
\[
S_n^t := \sum_{i = 1}^n \varphi(Z_{\sigma_t(i)}) - n\int \varphi d\hat{H}_{n,m},
\]
Introduce the shorthand $r^t_i := r(Z_{\sigma_{t}(i)})$ and $\varphi_i^t := \varphi(Z_{\sigma_t(i)})$ for $i \in [n+m]$ and define $q_i^t : = \hat{\lambda}m r_i^t/(n + \hat{\lambda}m r_i^t)$. Then we see that
\begin{align}\label{eq:bound(I)^2}
    \frac{S_n^t}{n} &= \frac{1}{n} \sum_{i = 1}^n \varphi_i^t - \int \varphi d\hat{H}_{n,m}  =  \frac{1}{n} \sum_{i = 1}^n \varphi_i^t - \sum_{i  =1}^{n+m} \frac{\varphi(Z_i)}{n + \hat{\lambda}m r(Z_i)} = \frac{1}{n} \sum_{i = 1}^n \varphi_i^t - \sum_{i  =1}^{n+m} \frac{\varphi_i^t}{n + \hat{\lambda}m r_i^t} \nonumber \\
    & = \frac{1}{n} \biggl( \sum_{i  = 1}^{n} \frac{\hat{\lambda}m r_i^t}{n + \hat{\lambda}m r_i^t} \varphi_i^t - \sum_{j=n+1}^{n+m} \frac{n}{n + \hat{\lambda}m r_j^t} \varphi_j^t \biggr) = \frac{1}{n} \biggl\{ \sum_{i=1}^n q_i^t \varphi_i^t - \sum_{j=n+1}^{n+m} (1-q_j^t) \varphi_j^t \biggr\}.
\end{align}
Recall that we write $K=\min(m,n)$ and, for each $t \geq 0$, we generate $\sigma_{t+1}$ by sampling $I_1,\ldots,I_K \in [n]$ and $J_1,\ldots,J_K \in \{n+1, \ldots, n+m\}$ uniformly at random without replacement, and switching $Z_{\sigma_t(I_\ell)}$ with $Z_{\sigma_t(J_\ell)}$ with probability 
\[
\tilde{p}_{I_\ell,J_\ell}^t =  \frac{\hat{\lambda}n m r_{I_\ell}^t}{(n + \hat{\lambda}m r_{I_\ell}^t)(n + \hat{\lambda}m r_{J_\ell}^t)} = q_{I_\ell}^t(1-q_{J_\ell}^t).
\]
Given $(I_\ell),(J_\ell)$, for $\ell \in [K]$ let $B_{I_\ell,J_\ell} \sim \mathrm{Ber}(\tilde{p}_{i_\ell,j_\ell}^t)$ determine whether the $\ell$th switch occurs. 

Having introduced the necessary notation, we prove the result by analysing the change in $S_n^t$ over a single step of the algorithm, showing that $S_n^t$ can be regarded as a Lyapunov function satisfying a geometric drift condition. We have that
\begin{align}
\label{Eq:TwoTermDecomp}
    \mathbb{E}\{ (S_n^{t+1})^2 | \sigma_t, Z \} &- (S_n^t)^2 = \mathbb{E} \biggl[ 
    \biggl\{ S_n^t + \sum_{\ell=1}^K B_{I_\ell,J_\ell}(\varphi_{J_\ell}^t - \varphi_{I_\ell}^t) \biggr\}^2 \biggm| \sigma_t, Z \biggr] - (S_n^t)^2  \nonumber \\
    &= 2 S_n^t \mathbb{E} \biggl\{ \sum_{\ell=1}^K B_{I_\ell,J_\ell}(\varphi_{J_\ell}^t - \varphi_{I_\ell}^t) \biggm| \sigma_t, Z \biggr\} + \mathbb{E} \biggl[ \biggl\{ \sum_{\ell=1}^K B_{I_\ell,J_\ell}(\varphi_{J_\ell}^t - \varphi_{I_\ell}^t) \biggr\}^2 \biggm| \sigma_t, Z \biggr].
\end{align}
We treat each of the two terms in this sum separately. For the first, we may calculate that
\begin{align}
\label{Eq:Term1}
    \mathbb{E}& \biggl\{ \sum_{\ell=1}^K B_{I_\ell,J_\ell}(\varphi_{J_\ell}^t - \varphi_{I_\ell}^t) \biggm| \sigma_t, Z \biggr\} = \sum_{\ell=1}^K \mathbb{E} \bigl\{ \tilde{p}_{I_\ell,J_\ell}^t (\varphi_{J_\ell}^t - \varphi_{I_\ell}^t) \bigm| \sigma_t, Z \bigr\} = K \mathbb{E} \bigl\{ \tilde{p}_{I_1,J_1}^t (\varphi_{J_1}^t - \varphi_{I_1}^t) \bigm| \sigma_t, Z \bigr\} \nonumber \\
    &= \frac{K}{mn} \sum_{i=1}^n \sum_{j=n+1}^{n+m} q_i^t(1-q_j^t)(\varphi_j^t - \varphi_i^t) = \frac{K}{mn} \biggl( \sum_{i=1}^n q_i^t \biggr) \biggl\{ \sum_{j=n+1}^{n+m} (1-q_j^t) \varphi_j^t - \sum_{i=1}^n q_i^t \varphi_i^t \biggr\} = - \frac{K}{mn} \biggl( \sum_{i=1}^n q_i^t \biggr) S_n^t,
\end{align}
where the penultimate equality follows from~\eqref{Eq:EqualSumsLambda}. The second term in~\eqref{Eq:TwoTermDecomp} is controlled by writing
\begin{align}
\label{Eq:Term2}
    \mathbb{E} &\biggl[ \biggl\{ \sum_{\ell=1}^K B_{I_\ell,J_\ell}(\varphi_{J_\ell}^t - \varphi_{I_\ell}^t) \biggr\}^2 \biggm| \sigma_t, Z \biggr] \nonumber \\
    &= \sum_{\ell=1}^K \mathbb{E} \bigl\{ B_{I_\ell,J_\ell}(\varphi_{I_\ell}^t - \varphi_{J_\ell}^t)^2 \bigm| \sigma_t, Z \bigr\} + \sum_{\ell_1 \neq \ell_2} \mathbb{E} \bigl\{ B_{I_{\ell_1},J_{\ell_1}} B_{I_{\ell_2},J_{\ell_2}}(\varphi_{I_{\ell_1}}^t - \varphi_{J_{\ell_1}}^t)(\varphi_{I_{\ell_2}}^t - \varphi_{J_{\ell_2}}^t) \bigm| \sigma_t, Z \bigr\} \nonumber \\
    & = \frac{K}{mn} \sum_{i=1}^n \sum_{j=n+1}^{n+m} \tilde{p}_{ij}^t (\varphi_i^t - \varphi_j^t)^2 + \frac{K(K-1)}{n(n-1)m(m-1)} \sum_{i_1 \neq i_2} \sum_{j_1 \neq j_2} \tilde{p}_{i_1j_1}^t \tilde{p}_{i_2j_2}^t (\varphi_{i_1}^t - \varphi_{j_1}^t) (\varphi_{i_2}^t - \varphi_{j_2}^t) \nonumber \\
    & = \frac{K}{mn} \sum_{i=1}^n \sum_{j=n+1}^{n+m} \tilde{p}_{ij}^t (\varphi_i^t - \varphi_j^t)^2 + \frac{K(K-1)}{n(n-1)m(m-1)} \biggl[ \sum_{i_1, i_2=1}^n \sum_{j_1,j_2=n+1}^{n+m} \tilde{p}_{i_1j_1}^t \tilde{p}_{i_2j_2}^t (\varphi_{i_1}^t - \varphi_{j_1}^t) (\varphi_{i_2}^t - \varphi_{j_2}^t) \nonumber \\
    & \hspace{50pt} - \sum_{i_1,i_2=1}^n \sum_{j=n+1}^{n+m} \tilde{p}_{i_1j}^t \tilde{p}_{i_2j}^t (\varphi_{i_1}^t - \varphi_{j}^t) (\varphi_{i_2}^t - \varphi_{j}^t) - \sum_{i=1}^n \sum_{j_1,j_2=n+1}^{n+m} \tilde{p}_{ij_1}^t \tilde{p}_{ij_2}^t (\varphi_{i}^t - \varphi_{j_1}^t) (\varphi_{i}^t - \varphi_{j_2}^t) \nonumber \\
    & \hspace{50pt} + \sum_{i=1}^n \sum_{j=n+1}^{n+m} (\tilde{p}_{ij}^t)^2 (\varphi_{i}^t - \varphi_{j}^t)^2\biggr] \nonumber \\
    & = \frac{K}{mn} \sum_{i=1}^n \sum_{j=n+1}^{n+m} \tilde{p}_{ij}^t \biggl\{ 1 + \frac{K-1}{(m-1)(n-1)} \tilde{p}_{ij}^t \biggr\} (\varphi_i^t - \varphi_j^t)^2 + \frac{K(K-1)}{n(n-1)m(m-1)} \biggl[ \biggl( \sum_{i=1}^n q_i^t \biggr)^2 (S_n^t)^2 \nonumber \\
    & \hspace{50pt} - \sum_{j=n+1}^{n+m} (1-q_j^t)^2 \biggl\{ \sum_{i=1}^n q_i^t(\varphi_i^t - \varphi_j^t) \biggr\}^2 - \sum_{i=1}^{n} (q_i^t)^2 \biggl\{ \sum_{j=n+1}^{n+m} (1-q_j^t)(\varphi_i^t - \varphi_j^t) \biggr\}^2 \biggr] \nonumber \\
    & \leq \biggl(\sum_{i=1}^n q_i^t \biggr)^2 \biggl\{ \frac{8K \|\varphi\|_\infty^2}{mn}  + \frac{K(K-1)}{n(n-1)m(m-1)} (S_n^t)^2 \biggr\},
\end{align}
where for the final inequality we again use~\eqref{Eq:EqualSumsLambda}. Combining~\eqref{Eq:TwoTermDecomp},~\eqref{Eq:Term1} and~\eqref{Eq:Term2}, using our assumption that $\max(m,n) \geq 3$ and using the fact that by~\eqref{Eq:EqualSumsLambda} we have $\sum_{i=1}^n q_i^t \leq K$, we see that
\begin{align*}
    \mathbb{E}\{ &(S_n^{t+1})^2 | \sigma_t, Z \} - (S_n^t)^2 \leq - \frac{2K}{mn} \biggl( \sum_{i=1}^n q_i^t \biggr) (S_n^t)^2 + \biggl(\sum_{i=1}^n q_i^t \biggr)^2 \biggl\{ \frac{8K \|\varphi\|_\infty^2}{mn}  + \frac{K(K-1)}{n(n-1)m(m-1)} (S_n^t)^2 \biggr\} \\
    & = - \frac{2K}{mn}\biggl( \sum_{i=1}^n q_i^t \biggr) \biggl\{1 - \frac{K-1}{2(m-1)(n-1)} \sum_{i=1}^n q_i^t \biggr\} (S_n^t)^2 + \frac{8K \|\varphi\|_\infty^2}{mn} \biggl(\sum_{i=1}^n q_i^t \biggr)^2 \\
    & \leq - \frac{K}{2mn}\biggl( \sum_{i=1}^n q_i^t \biggr) (S_n^t)^2 + \frac{8K \|\varphi\|_\infty^2}{mn} \biggl(\sum_{i=1}^n q_i^t \biggr)^2.
\end{align*}

We now control $\sum_i q_i^t$ using our boundedness assumption on $r$. Before doing so we establish some basic consequences of this assumption. First, we see that it implies bounds on $\hat{\lambda}$. Indeed, writing $\bar{R}=(n+m)^{-1} \sum_{i=1}^{n+m} r(Z_i)$ for the arithmetic mean of the observed $r$ values, by Jensen's inequality we have that
\[
    1= \sum_{i=1}^{n+m} \frac{1}{n+\hat{\lambda}mr(Z_i)} \geq \biggl[ \frac{1}{(n+m)^2} \sum_{i=1}^{n+m} \{n+\hat{\lambda}mr(Z_i)\} \biggr]^{-1} = \biggl( \frac{n}{n+m} + \frac{m}{n+m} \hat{\lambda} \bar{R} \biggr)^{-1}.
\]
Similarly, writing $\check{R}=\{(n+m)^{-1}\sum_{i=1}^{n+m} r(Z_i)^{-1}\}^{-1}$ for the harmonic mean of the observed $r$ values, we have that
\[
    1=\sum_{i=1}^{n+m} \frac{1}{m+ \hat{\lambda}^{-1} n r(Z_i)^{-1}} \geq \biggl( \frac{m}{n+m} + \frac{n}{m+n} \hat{\lambda}^{-1} \check{R}^{-1} \biggr)^{-1},
\]
where the first identity follows by the symmetry of our problem, as discussed in Remark~\ref{rmk:reciprocal}. These two bounds together imply that $\bar{R}^{-1} \leq \hat{\lambda} \leq \check{R}^{-1}$. Together with our assumption that $\bar{R}/\check{R} \leq \kappa$, this implies that $\max(\hat{\lambda}^{-1} \check{R}^{-1},\hat{\lambda} \bar{R}) \leq (\hat{\lambda}^{-1} \check{R}^{-1}) \cdot (\hat{\lambda} \bar{R}) = \bar{R}/\check{R} \leq \kappa$. Having bounds on $\hat{\lambda}$, we now turn to the control of $\sum_i q_i^t$ itself. Let $(\cdot)$ be a reordering of our data such that $r(Z_{(1)}) \leq r(Z_{(2)}) \leq \ldots \leq r(Z_{(m+n)})$. Then, again using Jensen's inequality, we have that
\begin{align*}
    \sum_{i=1}^n q_i^t \geq \sum_{i=1}^n \frac{m \hat{\lambda} r(Z_{(i)})}{n+m \hat{\lambda} r(Z_{(i)})} &\geq n \biggl\{ \frac{1}{n} \sum_{i=1}^n \frac{n+m \hat{\lambda} r(Z_{(i)})}{m \hat{\lambda} r(Z_{(i)})} \biggr\}^{-1} = \frac{mn}{m + \hat{\lambda}^{-1} \sum_{i=1}^n r(Z_{(i)})^{-1}} \\
    &\geq \frac{mn}{m+(n+m) \hat{\lambda}^{-1} \check{R}^{-1}} \geq \frac{mn}{m+(n+m) \kappa} \geq \frac{1}{2} \frac{mn}{(m+n)\kappa}.
\end{align*}

Having bounded $\sum_i q_i^t$ we may now conclude the proof. Writing $\tau=\max(m,n)/\min(m,n)$ and arguing inductively for the third inequality, we now see that
\begin{align*}
    n^{-2}\mathbb{E}\{ (S_n^{t+1})^2| Z \} &\leq n^{-2}\mathbb{E}\{ (S_n^{t})^2| Z \} + n^{-2}\mathbb{E} \biggl\{ - \frac{K}{2mn}\biggl( \sum_{i=1}^n q_i^t \biggr) (S_n^t)^2 + \frac{8K \|\varphi\|_\infty^2}{mn} \biggl(\sum_{i=1}^n q_i^t \biggr)^2 \biggm| Z \biggr\} \\
    & \leq n^{-2}\biggl\{1-\frac{K}{4(m+n)\kappa}\biggr\} \mathbb{E} \{(S_n^t)^2 | Z\} + \frac{8 \|\varphi\|_\infty^2 K^3}{mn^3}  \\
    & \leq \frac{8 \|\varphi\|_\infty^2 K^3}{mn^3} \sum_{s=0}^{t} \biggl(1-\frac{1}{8 \tau \kappa}\biggr)^s + \biggl(1-\frac{1}{8 \tau \kappa}\biggr)^{t+1} n^{-2} \mathbb{E}\{(S_n^0)^2 | Z\} \\
    & \leq \frac{64 \|\varphi\|_\infty^2 K^3}{mn^3}\tau\kappa + \|\varphi\|_\infty^2 \exp\biggl( - \frac{t+1}{8\tau\kappa} \biggr) \leq \frac{64 \|\varphi\|_\infty^2\kappa K}{n^2} + \|\varphi\|_\infty^2 \exp\biggl( - \frac{t+1}{8\tau\kappa} \biggr),
\end{align*}
and the result follows upon taking $t=S-1$.

\end{proof}

\subsection{Proofs for Section~\ref{sec:sMMD+optimality}}
\begin{proof}[Proof of Proposition \ref{prop:TF_explicit}]
    This proof borrows ideas from the proof of Lemma 4 in \cite{JMLR:v13:gretton12a}. Define the linear operator 
    \[
    T_{rf} :
    \begin{cases}
         \mathcal{H} \to \mathbb{R} \\
         \varphi \mapsto \int \frac{n+m }{n + \lambda_0 m  r} \, \lambda_0 r f \,  \varphi \, d\mu. 
    \end{cases}
    \]
    Write $d\mu_x = d\mu(x)$. Using the reproducing property of the reproducing kernel Hilbert space, i.e.~$\varphi(x) = \langle \varphi, k(\cdot, x) \rangle_\mathcal{H}$, we can show that this operator is bounded, since for all $\varphi \in \mathcal{H}$ we have
    \begin{align*}
        |T_{rf}\varphi| & \leq \int_{\mathcal{X}} \frac{(n+m)\lambda_0 r(x) f(x) }{n + \lambda_0 m  r(x)} |\varphi(x)| d\mu_x =  \int_{\mathcal{X}} \frac{(n+m)\lambda_0 r(x) f(x) }{n + \lambda_0 m r(x)} \left| \langle \varphi, k(\cdot, x) \rangle_\mathcal{H} \right| d\mu_x \\
        & \leq \|\varphi\|_\mathcal{H} \int_{\mathcal{X}}  \sqrt{k(x,x)} \frac{(n+m)\lambda_0  r(x) f(x) }{n + \lambda_0 m r(x)} d\mu_x,
    \end{align*}
    which shows that $|T_{rf}\varphi|/\|\varphi\|_\mathcal{H}$ is bounded uniformly in $\varphi$. The same is true for the linear operator
    \[
    T_{g} :
    \begin{cases}
         \mathcal{H} \to \mathbb{R} \\
         \varphi \mapsto \int \frac{(n+m) \, g}{n + \lambda_0 m r} \varphi \, d\mu, 
    \end{cases}
    \]
    hence the Riesz representation theorem implies that there exist $m_{rf}, m_g \in \mathcal{H}$ such that $T_{rf}\varphi = \langle m_{rf}, \varphi \rangle_\mathcal{H}$ and $T_{g}\varphi = \langle m_{g}, \varphi \rangle_\mathcal{H}$. Furthermore, using again the reproducing property of $\mathcal{H}$, we have that 
    \begin{align*}
        m_{rf}(t) &=  \langle m_{rf}, k(t, \cdot) \rangle_\mathcal{H} = T_{rf} k(t,\cdot)  = \int_{\mathcal{X}} \frac{(n+m)\lambda_0 r(x) f(x) }{n + \lambda_0 m r(x)} k(x,t) d\mu_x
    \end{align*}
    and, similarly, $m_g(t) = \int_{\mathcal{X}} \frac{(n+m) g(x)}{n + \lambda_0 m r(x)} k(x,t) d\mu_x$. This implies that
    \begin{align*}
         T_{\mathcal{F},r}^2(f,g)& = \left(\sup_{\|\varphi\|_\mathcal{H} \leq 1} \left| \int \frac{(n+m)(\lambda_0 r f - g)}{n + \lambda_0 m r} \varphi d\mu \right| \right)^2 = \left(\sup_{\|\varphi\|_\mathcal{H} \leq 1} \left| T_{rf} \varphi - T_g \varphi \right| \right)^2 \\
        & = \left(\sup_{\|\varphi\|_\mathcal{H} \leq 1} |\langle m_{rf} - m_g, \varphi  \rangle_{\mathcal{H}}| \right)^2 = \| m_{rf} - m_g \|_\mathcal{H}^2 = \langle m_{rf}, m_{rf} \rangle_\mathcal{H} + \langle m_{g}, m_{g} \rangle_\mathcal{H} - 2\langle m_{rf}, m_{g} \rangle_\mathcal{H} \\
        & = \int_{\mathcal{X}} \frac{(n+m)\lambda_0 r(t) f(t) }{n + \lambda_0 m r(t)} \left( \int_{\mathcal{X}} \frac{(n+m)\lambda_0 r(x) f(x) }{n + \lambda_0 m r(x)} k(x,t) d\mu_x \right) d\mu_t \\
        &\hspace{50pt} + \int_{\mathcal{X}} \frac{(n+m)g(t) }{n + \lambda_0 m r(t)} \left( \int_{\mathcal{X}} \frac{(n+m)g(x)}{n + \lambda_0 m r(x)} k(x,t) d\mu_x \right) d\mu_t \\
        & \hspace{50pt} -2 \int_{\mathcal{X}} \frac{(n+m)\lambda_0 r(t) f(t) }{n + \lambda_0 m r(t)} \left( \int_{\mathcal{X}} \frac{(n+m) g(x) }{n + \lambda_0 m r(x)} k(x,t) d\mu_x \right) d\mu_t \\
        & = \mathbb{E}_{X,X'}\left[\frac{\lambda_0^2 (n+m)^2 r(X)r(X')}{\{n + \lambda_0 m r(X) \} \{n + \lambda_0 m r(X')\}} k(X,X')\right]  \\
        &\hspace{50pt} + \mathbb{E}_{Y,Y'}\left[\frac{(n+m)^2}{\{n + \lambda_0 m r(Y) \} \{n + \lambda_0 m r(Y')\}} k(Y,Y')\right] \\
        & \hspace{50pt} - 2 \mathbb{E}_{X,Y}\left[\frac{\lambda_0 (n+m)^2 r(X)}{\{n + \lambda_0 m r(X) \} \{n + \lambda_0 m r(Y)\}} k(X,Y)\right],
    \end{align*}
    where $X, X^\prime \overset{\mathrm{i.i.d.}}{\sim} f$ independently of $Y, Y^\prime \overset{\mathrm{i.i.d.}}{\sim} g$, and concludes the proof.
\end{proof}

\begin{proof}[Proof of Theorem \ref{thm:UB_minimax}]
Consider a kernel function $K : \mathbb{R} \to \mathbb{R}$ that satisfies the assumptions in Subsection~\ref{sec:optimality}.  In particular, $K \in L^1(\mathbb{R}) \cap L^4(\mathbb{R})$ implies that $K \in L^2(\mathbb{R})$ using the basic inequality $t^2 \leq t + t^4$ for $t \geq 0$.  Based on this, we define the constants
\[
\kappa_j(d) := \left( \int_{\mathbb{R}} |K(u)|^j \, du \right)^d < \infty
\quad \text{ for } j \in \{1, 2, 4\},
\]
which are finite by assumption. Since the kernel is fixed in advance, we omit its dependence from the notation. We frequently use the analytical properties of the product kernel $k_\zeta$ with bandwidth $\zeta \geq 1$, which satisfies
\[
\int_{\mathbb{R}^d} |k_\zeta(x, y)|^j\, dy
=  \zeta^{(j-1)d} \left( \int_{\mathbb{R}} |K(u)|^j\, du \right)^d
= \kappa_j(d)\, \zeta^{(j-1)d} \quad \text{ for all } x \in \mathbb{R}^d \text{ and } j \in \{1,2,4 \}.
\]
As a result, using the fact that $\max(\|f\|_\infty, \|g\|_\infty) \leq M$, it follows that for any $x \in \mathbb{R}^d$ we have
\begin{align*}
    \mathbb{E}[|k_\zeta(x,X)|^j] &=  \int_{\mathbb{R}^d} |k_\zeta(x, y)|^j \, f(y) dy \leq M \int_{\mathbb{R}^d} |k_\zeta(x, y)|^j\, dy = M \kappa_j(d) \zeta^{(j-1)d} \quad \text{ for } j \in \{1,2,4\},
\end{align*}
and similarly for the expectation of $|k_\zeta(x,Y)|^j$. Furthermore, these bounds imply the same bounds on the quantities $\mathbb{E}[|k_\zeta(X_1,X_2)|^j],\, \mathbb{E}[|k_\zeta(Y_1,Y_2)|^j]$ and $\mathbb{E}[|k_\zeta(X,Y)|^j]$ for all $j \in \{1,2,4\}$. Now, we already know that Algorithm~\ref{alg:star_algo} controls the type~I error at a nominal level $\alpha$, so we can bound the minimax separation $\rho_r^*$ by controlling its type~II error. In order to do this, fix $\beta \in (0, 1-\alpha)$, choose $H \geq 2 \lceil \frac{1}{\alpha \beta}-1 \rceil$, and suppose $(f,g) \in \mathcal{S}_\theta^r(\rho)$ satisfies 
\begin{align}\label{eq:2moments}
         \mathrm{MMD}^2_{r, k_\zeta}(f,g) \geq \max \left\{2|\mathbb{E}[U(Z) - U(Z_\sigma)] - \mathrm{MMD}^2_{r, k_\zeta}(f,g)|, \left(\frac{8}{\alpha \beta} \operatorname{Var}[U(Z_\sigma) - U(Z) ] \right)^{1/2} \right\}.
\end{align}
Then, a double application of Markov's inequality shows that 
 \begin{align*}
\mathbb{P}\left\{p>\alpha \right\} & =\mathbb{P}\left(1+\sum_{h=1}^H \mathbbm{1}\left\{U(Z_{\sigma^{(h)}}) \geq U(Z) \right\}>(1+H) \alpha  \right) \leq \frac{1+H \, \mathbb{P}\left\{U(Z_\sigma) \geq U(Z) \right\} }{(1+H) \alpha} \\
& \leq \frac{1}{(1+H) \alpha}\left(1+\frac{H \, \operatorname{Var}[U(Z_\sigma) -U(Z)]}{\left\{\mathbb{E}[U(Z_\sigma) -U(Z)]\right\}^2}\right) \leq \frac{1}{(1+H) \alpha}\left(1+\frac{H \alpha \beta}{2}\right) \leq \beta.
\end{align*}
Bounding the terms on the right-hand side of \eqref{eq:2moments} is therefore sufficient to establish an upper bound on the minimax separation with respect to the squared shifted maximum mean discrepancy metric. However, since our goal is to characterize the separation in terms of the $L^2$ distance defined in \eqref{eq:l2_separation}, we must also relate $\mathrm{MMD}^2_{r, k_\zeta}$ to this $L^2$ norm. The proof proceeds in two main steps: first, we analyze the expectation and variance terms appearing on the right-hand side of \eqref{eq:2moments} and derive appropriate upper bounds. In the second step, we express the squared shifted maximum mean discrepancy as the sum of the square of the separation metric in \eqref{eq:l2_separation} and the bias term $\|\psi_r - \varphi_\zeta \ast \psi_r\|_2^2$, where we define $\varphi_{\zeta}(u)\;:=\;  \zeta^d \, \prod_{i = 1}^d K\!\bigl(\zeta\, u_i\bigr)$ for $u \in \mathbb{R}^d$, so that $k_\zeta(x, y) = \varphi_{\zeta}(x - y)$ for all $x,y \in \mathbb{R}^d$. We then control the bias term using the smoothness assumptions associated with the Sobolev class. Combining these two steps yields an upper bound on $\rho^*_r$. Throughout the following we set $\zeta = n^{\frac{2}{4s+d}}$ so that in particular $n^{-2} \zeta^d = n^{-\frac{8s}{4s+d}} \leq 1$. Also, for parameters $a_1, \ldots, a_k$, we denote by $Q(a_1, \ldots, a_k)$ a constant that depends only on these parameters. Its value may change from line to line, but it may only depend on $a_1, \ldots, a_k$. \\

\noindent \textbf{$\bullet$ Mean and variance of $U(Z_\sigma)$}\\
We begin by analyzing the second moment of the permuted statistic $U(Z_\sigma)$. Throughout this bullet point, we rescale the kernel so that $K(0) = 1$; this normalization is inessential and serves only to simplify notation, and the bound in \eqref{eq:moments_permuted} remains valid after absorbing the resulting scaling factor into $Q_0$. Furthermore, sums over~$i$'s are to be intended for $i$ varying in $[n]$, sums over $j$'s are to be intended for $j$ varying in $\{n+1, \ldots, n+m\}$, and sums over $k$'s are to be intended for $k$ varying in $[n+m]$. The strategy is to relate the distribution of $U_\sigma := U(Z_\sigma)$,
where $\sigma$ is sampled from $\eqref{eq:sampling_1}$, to the distribution of $U_{\sigma_t} \equiv U_t$, where $\sigma_t$ is the permutation at time $t \in \mathbb{N}$ of an equivalent version of  Algorithm~\ref{alg:pairwise_sampler} that is initialized at stationarity. More precisely, we consider a procedure that at each time step $t$ samples $i \in [n]$ and $j \in \{n+1, \ldots, n+m\}$ uniformly at random, and switches $Z_{\sigma_t(i)}$ with $Z_{\sigma_t(j)}$ with probability 
\[
\tilde{p}_{i,j}^t := \mathbb{P}\{\text{switch } i \text{ and } j \text{ at time } t \mid i,j \text{ are selected}\} =  \frac{\hat{\lambda}n m r_i^t}{(n + \hat{\lambda}m r_i^t)(n + \hat{\lambda}m r_j^t)},
\]
where $r^t_i := r(Z_{\sigma_{t}(i)})$ for all $i \in [n+m]$. As already outlined in Remark~\ref{rmmk:hat_lambda_alg}, this algorithm still targets the distribution \eqref{eq:sampling_1} since $\tilde{p}_{i,j}^t/\tilde{p}^t_{j,i} = r_i^t/r_j^t$. Write 
\begin{align*}
    K_{ij}^t & := k_\zeta(Z_{\sigma_t(i)}, Z_{\sigma_t(j)}) \text{ and }K_{i}^t := k_\zeta(Z_{\sigma_t(i)}, \cdot), \\
    q_i^t &:= \frac{\hat{\lambda} m r_i^t}{n + \hat{\lambda} m r_i^t} \text{ and } S_n^t := \sum_{i} q_i^t = \sum_j (1-q_j^t),
\end{align*}
where the last equality holds by~\eqref{Eq:EqualSumsLambda}. In particular we have $\tilde{p}_{i,j}^t = q_i^t (1-q_j^t)$. Observe that under the assumption that $0 < c \leq r(x) \leq C < \infty$ for all $x \in \mathcal{X}$, we have that $C^{-1} \leq \hat{\lambda}\leq c^{-1}$, which implies $mc \, (mc + nC)^{-1} \leq q_i^t \leq mC \, (mC + nc)^{-1}$ and $S_n^t \geq nmc \, (nC + mc)^{-1}$. We define the V-statistic
\[
V_t = \frac{(n+m)^2}{n^2 m^2} \left\|\sum_i q_i^t K_i^t - \sum_j (1-q_j^t)K_j^t \right\|_\mathcal{H}^2 =: \|G_t \|_\mathcal{H}^2 = \langle G_t, G_t\rangle_\mathcal{H},
\]
which coincides with~\eqref{eq:V-stat} evaluated on our data at time $t$ with kernel $k_\zeta$, and will be convenient for our analysis. We further define the U-statistic 
\begin{align*}
    U_t & = V_t - \frac{(n+m)^2 \zeta^d}{n^2 m^2} \sum_i (q_i^t)^2 - \frac{(n+m)^2 \zeta^d}{n^2 m^2} \sum_j (1- q_j^t)^2   \\
    & =V_t - \frac{(n+m)^2 \zeta^d}{n^2 m^2} \sum_i (q_i^t)^2 - \frac{(n+m)^2 \zeta^d}{n^2 m^2} \sum_j \{1 + (q_j^t)^2  -2q_j^t\} \\
    & =  V_t - \frac{(n+m)^2 \zeta^d}{n^2 m^2} \sum_k (q_k^t)^2 + \frac{(n+m)^2 \zeta^d}{n^2 m} - \frac{2(n+m)^2\zeta^d}{n^2 m^2} \sum_j (1-q_j^t) \\
    & = V_t  - \frac{2(n+m)^2\zeta^d}{n^2 m^2} S_n^t \underbrace{ - \frac{(n+m)^2 \zeta^d}{n^2 m^2} \sum_k (q_k^t)^2 + \frac{(n+m)^2 \zeta^d}{n^2 m}}_{\text{permutation independent}},
\end{align*}
which has the same distribution as $U_\sigma$ for each $t \in \mathbb{N}$. If we swap $i$ and $j$ at time $t$, then  
\begin{align*}
U_{t+1} - U_t & = 
\left\|G_t + \frac{n+m}{nm}(K_j^t - K_i^t)\right\|_\mathcal{H}^2 - \|G_t\|_\mathcal{H}^2 - \frac{2(n+m)^2\zeta^d}{n^2 m^2}(q_j^t - q_i^t) \\
& = \frac{2(n+m)}{nm} \langle G_t, K_j^t - K_i^t\rangle + \frac{(n+m)^2}{n^2 m^2}\|K_j^t - K_i^t\|_\mathcal{H}^2 - \frac{2(n+m)^2\zeta^d}{n^2 m^2}(q_j^t - q_i^t).
\end{align*}
By stationarity, we therefore have
\begin{align}
\label{Eq:UStatDiffDecomp}
     0 & = nm \, \mathbb{E}[U^2_{t+1} - U^2_t] \nonumber \\
    & = \sum_{i,j} \mathbb{E}\left[q_i^t(1-q_j^t) \left(\left\{U_t + \frac{2(n+m)}{nm} \langle G_t, K_j^t - K_i^t\rangle + \frac{(n+m)^2}{n^2 m^2}\|K_j^t - K_i^t\|_\mathcal{H}^2 \right. \right. \right. \nonumber\\
    & \left. \left. \left. \hspace{250pt}- \frac{2(n+m)^2\zeta^d}{n^2m^2}(q_j^t - q_i^t) \right\}^2 -U^2_t \right) \right] \nonumber \\
    & = \sum_{i,j} \mathbb{E}\left[q_i^t(1-q_j^t) \cdot 2 U_t \left\{\frac{2(n+m)}{n m} \langle G_t, K_j^t - K_i^t\rangle + \frac{(n+m)^2}{n^2 m^2}\|K_j^t - K_i^t\|_\mathcal{H}^2 - \frac{2(n+m)^2\zeta^d}{n^2 m^2}(q_j^t - q_i^t) \right\}\right] \nonumber \\
    & \quad \quad + \sum_{i,j} \mathbb{E}\left[q_i^t(1-q_j^t) \left\{\frac{2(n+m)}{nm} \langle G_t, K_j^t - K_i^t\rangle + \frac{(n+m)^2}{n^2m^2}\|K_j^t - K_i^t\|_\mathcal{H}^2 - \frac{2(n+m)^2\zeta^d}{n^2 m^2}(q_j^t - q_i^t) \right\}^2 \right],
\end{align}
which is the sum of a linear and a quadratic term. In order to simplify this right-hand side we now provide some useful identities. We notice that 
\[
\sum_{i,j} q_i^t(1-q_j^t) (K_j^t - K_i^t) = S_n^t\left\{\sum_j (1-q_j^t)K_j^t - \sum_i q_i^t K_i^t \right\} = -\frac{nm}{n+m}S_n^t G_t,
\]
\[
\sum_{i,j} q_i^t(1-q_j^t) \|K_j^t - K_i^t\|_\mathcal{H}^2 = \sum_{i,j} q_i^t(1-q_j^t) (2\zeta^d - 2K_{ij}^t) = 2\zeta^d \, (S_n^t)^2 -  2 \sum_{i,j} q_i^t(1-q_j^t) K_{ij}^t,
\]
and
\[
 \sum_{i,j} q_i^t(1-q_j^t)(q_j^t - q_i^t) = S_n^t \left\{\sum_j q_j^t(1-q_j^t) - \sum_i (q_i^t)^2 \right\} = S_n^t \left\{m - S_n^t - \sum_k (q_k^t)^2 \right\}.
\]
 Using these identities, the linear term in~\eqref{Eq:UStatDiffDecomp} gives 
\begin{align}\label{eq:LinearTerm}
    & \sum_{i,j} \mathbb{E}\left[q_i^t(1-q_j^t) \cdot 2 U_t \left\{\frac{2(n+m)}{n m} \langle G_t, K_j^t - K_i^t\rangle + \frac{(n+m)^2}{n^2 m^2}\|K_j^t - K_i^t\|_\mathcal{H}^2 - \frac{2(n+m)^2\zeta^d}{n^2 m^2}(q_j^t - q_i^t) \right\}\right] \nonumber \\
    &= \mathbb{E}\left[ 2 U_t \left\{\frac{2(n+m)}{n m} \langle G_t, \sum_{i,j} (K_j^t - K_i^t)q_i^t(1-q_j^t)\rangle + \frac{(n+m)^2}{n^2 m^2} \sum_{i,j}  \|K_j^t - K_i^t\|_\mathcal{H}^2 \, q_i^t(1-q_j^t) \right. \right. \nonumber\\
    & \left. \left. \hspace{250 pt} - \frac{2 (n+m)^2 \zeta^d}{n^2 m^2} \sum_{i,j} (q_j^t - q_i^t) q_i^t(1-q_j^t) \right\}\right] \nonumber \\
    & =  \mathbb{E}\left[2U_t \left\{-2S_n^t \left(U_t + \frac{2 (n+m)^2\zeta^d}{n^2 m^2} S_n^t + \frac{(n+m)^2 \zeta^d}{n^2 m^2} \sum_k (q_k^t)^2 - \frac{(n+m)^2 \zeta^d}{n^2 m} \right) \right. \right. \nonumber \\
    & \left.  \left.  \quad \quad \quad + \frac{2 (n+m)^2 \zeta^d }{n^2 m^2}(S_n^t)^2 -  \frac{2(n+m)^2}{n^2 m^2} \sum_{i,j} q_i^t(1-q_j^t) K_{ij}^t - \frac{2(n+m)^2\zeta^d}{n^2 m^2}S_n^t \left(m - S_n^t - \sum_k (q_k^t)^2 \right)\right\}\right] \nonumber \\
    & = -4 \mathbb{E}\left[S_n^t U^2_t\right] - 4 \mathbb{E}\left[U_t \frac{(n+m)^2}{n^2 m^2} \sum_{i,j}q_i^t(1-q_j^t)K_{ij}^t \right].
\end{align}
As for the quadratic term, it is useful to define $G_t^{-(i,j)}:= G_t - (n+m)n^{-1}m ^{-1}\{q_i^t K_i^t - (1-q_j^t)K_j^t\}$ and write
\begin{align}\label{eq:QuadraticTerm}
    & \frac{2(n+m)}{nm} \langle G_t, K_j^t - K_i^t\rangle + \frac{(n+m)^2}{n^2 m^2}\|K_j^t - K_i^t\|_\mathcal{H}^2 - \frac{2(n+m)^2\zeta^d}{n^2 m^2}(q_j^t - q_i^t) \nonumber \\
    & = \frac{2(n+m)}{nm} \langle G_t - \frac{n+m}{n m}\{q_i^t K_i^t - (1-q_j^t)K_j^t\}, K_j^t - K_i^t\rangle + \frac{2(n+m)^2}{n^2 m^2} \langle q_i^t K_i^t - (1-q_j^t)K_j^t, K_j^t - K_i^t\rangle \nonumber \\
    & \quad \quad  + \frac{(n+m)^2}{n^2 m^2}\|K_j^t - K_i^t\|_\mathcal{H}^2 - \frac{2(n+m)^2\zeta^d}{n^2 m^2}(q_j^t - q_i^t)\nonumber \\
    & = \frac{2 (n+m)}{nm} \langle G_t^{-(i,j)}, K_j^t - K_i^t\rangle + \frac{2(n+m)^2}{n^2 m^2} \langle q_i^t K_i^t - (1-q_j^t)K_j^t, K_j^t - K_i^t\rangle \nonumber \\
    & \quad \quad + \frac{2(n+m)^2}{n^2 m^2}(\zeta^d - K^t_{ij}) - \frac{2(n+m)^2\zeta^d}{n^2 m^2}(q_j^t - q_i^t) \nonumber \\
    & = \frac{2(n+m)}{n m} \langle G_t^{-(i,j)}, K_j^t - K_i^t\rangle + \frac{2(n+m)^2}{n^2 m^2} (q_i^t K_{ij}^t - \zeta^d q_i^t - \zeta^d (1-q_j^t) + (1-q_j^t)K^t_{ij})  \nonumber \\
    & \quad \quad + \frac{2(n+m)^2}{n^2 m^2}(\zeta^d - K^t_{ij}) - \frac{2(n+m)^2\zeta^d}{n^2 m^2}(q_j^t - q_i^t) \nonumber \\
    & = \frac{2 (n+m)}{nm} \langle G_t^{-(i,j)}, K_j^t - K_i^t\rangle + \frac{2(n+m)^2}{n^2 m^2}(q_i^t - q^t_j )K^t_{ij}.
\end{align}
Combining~\eqref{Eq:UStatDiffDecomp},~\eqref{eq:LinearTerm},~\eqref{eq:QuadraticTerm} gives
\begin{align}\label{eq:variance_U_perm}
    4 &\mathbb{E}\left[S_n^t U^2_t\right] = - 4 \mathbb{E}\left[U_t \frac{(n+m)^2}{n^2 m^2} \sum_{i,j}q_i^t(1-q_j^t)K_{ij}^t \right] \nonumber \\
    &\hspace{50pt} + \sum_{i,j} \mathbb{E}\left[q_i^t(1-q_j^t) \left\{\frac{2(n+m)}{n m} \langle G_t^{-(i,j)}, K_j^t - K_i^t\rangle + \frac{2(n+m)^2}{n^2m^2}(q_i^t - q^t_j)K^t_{ij} \right\}^2 \right] \nonumber \\
    & \leq 4 \mathbb{E}\left[|U_t| \frac{(n+m)^2}{n^2 m^2} \sum_{i,j}q_i^t(1-q_j^t)|K_{ij}^t| \right] \nonumber \\
    & \hspace{50pt}+ \sum_{i,j} \mathbb{E}\left[q_i^t(1-q_j^t) \left\{\frac{2(n+m)}{nm} \langle G_t^{-(i,j)}, K_j^t - K_i^t\rangle + \frac{2(n+m)^2}{n^2 m^2}(q_i^t - q^t_j )K^t_{ij} \right\}^2 \right] \nonumber \\
    & \leq 4 \sqrt{\mathbb{E}[U_t^2] \mathbb{E}\left[\frac{(n+m)^4}{n^4 m^4} \Big\{\sum_{i,j}q_i^t(1-q_j^t)|K_{ij}^t| \Big\}^2 \right] } + \frac{8(n+m)^2}{n^2 m^2} \sum_{i,j} \mathbb{E}\left[q_i^t(1-q_j^t) \langle G_t^{-(i,j)}, K_j^t - K_i^t\rangle^2 \right] \nonumber \\
    & \quad \quad \quad \quad \quad \quad \quad  \quad \quad  \quad \quad \quad \quad \quad \quad  \quad \quad \quad \quad \quad \quad \quad \quad + \frac{8(n+m)^4}{n^4 m^4} \sum_{i,j} \mathbb{E}\left[q_i^t(1-q_j^t) (q_i^t - q^t_j )^2(K^t_{ij})^2 \right] \nonumber \\
    & \leq 4 \sqrt{\mathbb{E}[U_t^2] \mathbb{E}\left[\frac{(n+m)^4}{n^{3}m^3} \sum_{i,j}\{q_i^t(1-q_j^t)K_{ij}^t \}^2\right] } + \frac{8(n+m)^2}{n^2 m^2} \sum_{i,j} \mathbb{E}\left[q_i^t(1-q_j^t) \langle G_t^{-(i,j)}, K_j^t - K_i^t\rangle^2 \right] \nonumber \\
    & \quad \quad \quad \quad \quad \quad  \quad \quad \quad \quad \quad \quad \quad \quad \quad \quad  \quad \quad \quad \quad \quad \quad \quad \quad  + \frac{8(n+m)^4}{n^4 m^4} \sum_{i,j} \mathbb{E}\left[q_i^t(1-q_j^t) (q_i^t - q^t_j )^2(K^t_{ij})^2 \right] \nonumber \\
    & \leq 12 \max\left\{ \sqrt{\mathbb{E}[U_t^2] \mathbb{E}\left[\frac{(n+m)^4}{n^3m^3} \sum_{i,j}\{q_i^t(1-q_j^t)K_{ij}^t \}^2 \right] }\, ,  \frac{2(n+m)^2}{n^2 m^2} \sum_{i,j} \mathbb{E}\left[q_i^t(1-q_j^t) \langle G_t^{-(i,j)}, K_j^t - K_i^t\rangle^2 \right], \right. \nonumber \\
    & \quad \quad \quad \quad  \quad \quad \quad \quad \quad \quad \quad \quad  \quad \quad \quad \quad \quad \quad  \quad \quad \quad \quad \quad \quad \quad \quad \left.  \frac{2(n+m)^4}{n^4 m^4} \sum_{i,j} \mathbb{E}\left[q_i^t(1-q_j^t) (q_i^t - q^t_j )^2(K^t_{ij})^2 \right] \right\}.
\end{align}
We will now bound each of the three quantities inside the maximum separately. If the first quantity reaches the maximum, we can use $m \leq n \leq \tau m$, $q_k^t \in [0,1]$ and $S_n^t \geq \frac{nmc}{nC + mc}$ to show that 
\begin{align}\label{eq:K1}
    \frac{nmc}{nC + mc} \mathbb{E}&\left[U^2_t\right] \leq \mathbb{E}\left[S_n^t U^2_t\right] \leq 3 \sqrt{\mathbb{E}[U_t^2] \mathbb{E}\left[\frac{(n+m)^4}{n^3 m^3} \sum_{i,j}\{q_i^t(1-q_j^t)K_{ij}^t \}^2\right] }  \nonumber \\
    & \leq  3 \sqrt{\mathbb{E}[U_t^2] \mathbb{E}\left[\frac{(n+m)^4}{n^3 m^3} \sum_{i,j}(K_{ij}^t)^2\right] } \leq  3 \sqrt{\mathbb{E}[U_t^2] \mathbb{E}\left[\frac{(n+m)^4}{n^3 m^3} \sum_{k_1 \neq k_2}(K_{k_1, k_2}^t)^2\right] }  \nonumber \\
    & = 3 \sqrt{\mathbb{E}[U_t^2] \mathbb{E}\left[\frac{(n+m)^4}{n^3 m^3} \left\{\sum_{i_1 \neq i_2} k_\zeta^2(X_{i_1}, X_{i_2}) +  \sum_{j_1 \neq j_2} k_\zeta^2(Y_{j_1}, Y_{j_2}) + 2 \sum_{i, j} k_\zeta^2(X_{i}, Y_{j}) \right\}\right]}  \nonumber \\
    &  = 3 \sqrt{\mathbb{E}[U_t^2] \mathbb{E}\left[\frac{(n+m)^4}{n^3 m^3} \left\{n(n-1) k_\zeta^2(X_{1}, X_{2}) +  m(m-1) k_\zeta^2(Y_{1}, Y_{2}) + 2 nm k_\zeta^2(X_1, Y_1) \right\}\right]}  \nonumber \\
    & \leq 3 \sqrt{\mathbb{E}[U_t^2] \, \frac{(n+m)^6}{n^3 m^3} \, M \kappa_2(d) \zeta^d } \leq  3 \sqrt{(1+\tau)^6\, \tau^{-3} \, \mathbb{E}[U_t^2] \, M \kappa_2(d) \zeta^d }. 
\end{align}
Note that in the fourth inequality we bound $\sum_{i,j} (K^t_{i, j})^2 $ by $ \sum_{k_1 \neq k_2} (K^t_{k_1, k_2})^2 $, the latter being more advantageous to analyze due to its permutation invariance. Making~\eqref{eq:K1} explicit for the expectation of $U_t^2$ shows that there exists a constant $Q_0 \equiv Q_0(c, C, d, \tau, M)$ such that $\mathbb{E}[U_t^2] \leq Q_0 \frac{\zeta^d}{n^2}$, in the case where the first quantity in the maximum in~\eqref{eq:variance_U_perm} dominates. As for the case when the third quantity dominates, we can argue exactly as in~\eqref{eq:K1} to show that 
\begin{align}\label{eq:K2}
    \frac{(n+m)^4}{n^4 m^4} & \sum_{i,j} \mathbb{E}\left[q_i^t(1-q_j^t) (q_i^t - q^t_j )^2(K^t_{ij})^2 \right] \leq \frac{4(n+m)^6}{n^4 m^4} \, M \kappa_2(d) \zeta^d,
\end{align}
which implies that $\mathbb{E}[U_t^2] \leq Q_0 \frac{\zeta^d}{n^3}$. Finally, as for the term involving $G_t^{-(i,j)}$, we have 
\[
\sum_{i,j} \mathbb{E}\left[q_i^t(1-q_j^t) \langle G_t^{-(i,j)}, K_j^t - K_i^t\rangle^2 \right] \leq 2 \sum_{i,j} \mathbb{E}\left[\langle G_t^{-(i,j)}, K_i^t\rangle^2 \right] + 2 \sum_{i,j} \mathbb{E}\left[\langle G_t^{-(i,j)}, K_j^t\rangle^2 \right],
\]
and we can analyze the evolution of the right hand side through a version of Algorithm \ref{alg:pairwise_sampler} where indices $i$ and $j$ remain fixed. In other words, the Algorithm works the same, but we are just allowed to swap indices $\tilde{i} \neq i$ with $\tilde{j} \neq j$ with probability $q_{\tilde{i}}^t (1- q^t_{\tilde{j}})$. We show in Lemma \ref{lemma:sampler_noIJ} that this procedure still preserves the stationary distribution. The two terms above can be bounded by almost identical arguments so we will restrict attention to the first. In analysing the evolution of $\langle G_t^{-(i,j)}, K_i^t\rangle^2$, it is useful to write 
\begin{align*}
    \sum_{\tilde{i} \neq i}\sum_{\tilde{j} \neq j} & q_{\tilde{i}}^t (1-q_{\tilde{j}}^t)(K_{\tilde{j}}^t - K_{\tilde{i}}^t) = \left(\sum_{\tilde{i} \neq i} q_{\tilde{i}}^t\right)\left( \sum_{\tilde{j} \neq j} (1-q_{\tilde{j}}^t) K_{\tilde{j}}^t \right) - \left(\sum_{\tilde{j} \neq j} (1-q_{\tilde{j}}^t)\right)\left( \sum_{\tilde{i} \neq i} q_{\tilde{i}}^t K_{\tilde{i}}^t \right) \\
    & = - \frac{nm}{n+m}S_n^t G_t^{-(i,j)} + (1-q_{j}^t) \sum_{\tilde{i} \neq i} q_{\tilde{i}}^t K_{\tilde{i}}^t - q_{i}^t \sum_{\tilde{j} \neq j} (1-q_{\tilde{j}}^t) K_{\tilde{j}}^t.
\end{align*}
Under stationarity, we thus have 
\begin{align*}
    & 0 = (n-1)(m-1) \, \mathbb{E}[\langle G_{t+1}^{-(i,j)}, K_i^t\rangle^2 - \langle G_t^{-(i,j)}, K_i^t\rangle^2] \\
    & = \sum_{\tilde{i} \neq i}\sum_{\tilde{j} \neq j} \mathbb{E}\left[q_{\tilde{i}}^t (1-q_{\tilde{j}}^t) \left\{ \langle  G_t^{-(i,j)} + \tfrac{n+m}{nm}(K_{\tilde{j}}^t - K_{\tilde{i}}^t ), K_i^t \rangle^2 - \langle  G_t^{-(i,j)}, K_i^t \rangle^2 \right\}\right] \\
    & = \sum_{\tilde{i} \neq i}\sum_{\tilde{j} \neq j} \mathbb{E}\left[\frac{(n+m)^2}{n^2 m^2} \, q_{\tilde{i}}^t (1-q_{\tilde{j}}^t)  \langle  K_{\tilde{j}}^t - K_{\tilde{i}}^t, K_i^t \rangle^2  \right] + \tfrac{2(n+m)}{n m} \, \mathbb{E}\left[ \langle  G_t^{-(i,j)}, K_i^t\rangle \langle \sum_{\tilde{i} \neq i}\sum_{\tilde{j} \neq j} q_{\tilde{i}}^t (1-q_{\tilde{j}}^t)(K_{\tilde{j}}^t - K_{\tilde{i}}^t), K_i^t \rangle \right] \\
    & = \sum_{\tilde{i} \neq i}\sum_{\tilde{j} \neq j} \mathbb{E}\left[\frac{(n+m)^2}{n^2 m^2} \, q_{\tilde{i}}^t (1-q_{\tilde{j}}^t)  \langle  K_{\tilde{j}}^t - K_{\tilde{i}}^t, K_i^t \rangle^2  \right] - 2 \mathbb{E}[S_n^t \langle G_t^{-(i,j)}, K_i^t \rangle^2]  \\
    & \qquad\qquad +\tfrac{2(n+m)}{n m} \, \mathbb{E}\left[ \langle  G_t^{-(i,j)}, K_i^t\rangle \langle  (1-q_{j}^t) \sum_{\tilde{i} \neq i} q_{\tilde{i}}^t K_{\tilde{i}}^t - q_{i}^t \sum_{\tilde{j} \neq j} (1-q_{\tilde{j}}^t) K_{\tilde{j}}^t, K_i^t \rangle \right] \\
    & \leq \tfrac{2(n+m)^2}{n m^2} \sum_{\tilde{j} \neq j} \mathbb{E} \left[ \langle K_{\tilde{j}}^t, K_i^t \rangle^2 \right] + \tfrac{2(n+m)^2}{n^2 m} \sum_{\tilde{i} \neq i} \mathbb{E} \left[ \langle K_{\tilde{i}}^t, K_i^t \rangle^2 \right] - 2 \mathbb{E}[S_n^t \langle G_t^{-(i,j)}, K_i^t \rangle^2]  \\
    & \qquad\qquad +\tfrac{2(n+m)}{n m} \, \mathbb{E}\left[ \langle  G_t^{-(i,j)}, K_i^t\rangle \langle  (1-q_{j}^t) \sum_{\tilde{i} \neq i} q_{\tilde{i}}^t K_{\tilde{i}}^t - q_{i}^t \sum_{\tilde{j} \neq j} (1-q_{\tilde{j}}^t) K_{\tilde{j}}^t, K_i^t \rangle \right] \\
    & \leq \tfrac{(n+m)^4}{n^2 m^2} \max_{k_1 \neq k_2} \mathbb{E}[k^2_\zeta(Z_{k_1}, Z_{k_2})] - 2 \mathbb{E}[S_n^t \langle G_t^{-(i,j)}, K_i^t \rangle^2]  \\
    & \qquad\qquad + \tfrac{2(n+m)}{n m} \, \mathbb{E}\left[ \langle  G_t^{-(i,j)}, K_i^t\rangle \langle  (1-q_{j}^t) \sum_{\tilde{i} \neq i} q_{\tilde{i}}^t K_{\tilde{i}}^t - q_{i}^t \sum_{\tilde{j} \neq j} (1-q_{\tilde{j}}^t) K_{\tilde{j}}^t, K_i^t \rangle \right] \\
    & \leq  \tfrac{(n+m)^4}{n^2 m^2} M \kappa_2(d) \zeta^d - 2 \mathbb{E}\left[ S_n^t \langle  G_t^{-(i,j)}, K_i^t\rangle^2 \right] \\
    & \quad \quad + \tfrac{2(n+m)}{n m} \sqrt{\mathbb{E}\left[ \langle  G_t^{-(i,j)}, K_i^t\rangle^2 \right] \mathbb{E}\left[ \langle  (1-q_{j}^t) \sum_{\tilde{i} \neq i} q_{\tilde{i}}^t K_{\tilde{i}}^t - q_{i}^t \sum_{\tilde{j} \neq j} (1-q_{\tilde{j}}^t) K_{\tilde{j}}^t, K_i^t \rangle^2 \right]} \\
    & \leq  \tfrac{(n+m)^4}{n^2 m^2} M \kappa_2(d) \zeta^d - 2 \mathbb{E}\left[ S_n^t \langle  G_t^{-(i,j)}, K_i^t\rangle^2 \right] \\
    & \quad \quad + 2\sqrt{2} \tfrac{n+m}{nm} \sqrt{\mathbb{E}\left[ \langle  G_t^{-(i,j)}, K_i^t\rangle^2 \right] \mathbb{E}\left[ \langle  \sum_{\tilde{i} \neq i} q_{\tilde{i}}^t K_{\tilde{i}}^t, K_i^t \rangle^2 + \langle \sum_{\tilde{j} \neq j} (1-q_{\tilde{j}}^t) K_{\tilde{j}}^t, K_i^t \rangle^2 \right]} \\
    & \leq  \tfrac{(n+m)^4}{n^2 m^2} M \kappa_2(d) \zeta^d - 2 \mathbb{E}\left[ S_n^t \langle  G_t^{-(i,j)}, K_i^t\rangle^2 \right] + 2\sqrt{2} \tfrac{(n+m)^{2}}{nm}  \sqrt{\mathbb{E}\left[ \langle  G_t^{-(i,j)}, K_i^t\rangle^2 \right] \max_{k_1 \neq k_2} \mathbb{E}[k^2_\zeta(Z_{k_1}, Z_{k_2})]} \\
    & \leq \tfrac{(n+m)^4}{n^2 m^2} M \kappa_2(d) \zeta^d - 2 \mathbb{E}\left[ S_n^t \langle  G_t^{-(i,j)}, K_i^t\rangle^2 \right] + 2\sqrt{2} \tfrac{(n+m)^{2}}{nm}  \sqrt{M \kappa_2(d) \zeta^d \, \mathbb{E}\left[ \langle  G_t^{-(i,j)}, K_i^t\rangle^2 \right]}.
\end{align*}
Using again the fact that $S_n^t \geq \frac{nmc}{nC + mc}$, we can employ the previous calculations to show that 
\begin{align*}
    \frac{2nmc}{nC + mc} \mathbb{E}\left[ \langle  G_t^{-(i,j)}, K_i^t\rangle^2 \right] & \leq 2 \mathbb{E}\left[ S_n^t \langle  G_t^{-(i,j)}, K_i^t\rangle^2 \right] \\
    & \leq \tfrac{(n+m)^4}{n^2 m^2} M \kappa_2(d) \zeta^d + 2\sqrt{2} \tfrac{(n+m)^{2}}{nm} \sqrt{M \kappa_2(d) \zeta^d \mathbb{E}\left[ \langle  G_t^{-(i,j)}, K_i^t\rangle^2 \right]} \\
    & \leq 2 (\tau + 1)^4 \tau^{-2} \, \max \left\{M \kappa_2(d) \zeta^d, \,\sqrt{M \kappa_2(d) \zeta^d \mathbb{E}\left[ \langle  G_t^{-(i,j)}, K_i^t \rangle^2 \right]} \right\}.
\end{align*}
In particular, this yields  $\mathbb{E}[ \langle  G_t^{-(i,j)}, K_i^t\rangle^2] \leq Q_0 \frac{\zeta^d}{n}$, which further implies that $\mathbb{E}[U_t^2] \leq Q_0 \frac{\zeta^d}{n^2}$ also when the second quantity in~\eqref{eq:variance_U_perm} attains the maximum. Combining this with \eqref{eq:K1} and \eqref{eq:K2} gives
\begin{align}\label{eq:moments_permuted}
    \max\left\{\mathbb{E}^2[|U_\sigma|], \operatorname{Var}[U_\sigma] \right\} \leq \mathbb{E}[U^2_\sigma] \leq Q_0 \frac{\zeta^d}{n^2},
\end{align}
and concludes the analysis for $U_\sigma$. \\

\noindent \textbf{$\bullet$ Mean and variance of $U_\mathrm{id}$}\\
In what follows, we define \(U_{\mathrm{id}} := \, U(Z)\), which
targets \(\mathrm{MMD}^2_{r, k_\zeta}(f,g)\), and we study its
mean and variance. This would be straightforward if we used the normalization factors $\{n(n-1)\}^{-1}$ and $\{m(m-1)\}^{-1}$ instead of $n^{-2}$ and $m^{-2}$, and if $\hat{\lambda}$ were not random. The former problem is easy to address. Define 
\begin{align}\label{eq:tilde_U_stat}
    & \tilde{U}_\mathrm{id} = \frac{(n+m)^2}{n(n-1)m^2} \sum_{i_1 \neq i_2 = 1}^n \frac{\hat{\lambda}^2 r(X_{i_1})r(X_{i_2})k_\zeta(X_{i_1},X_{i_2})}{\{\frac{n}{m} + \hat{\lambda} r(X_{i_1}) \} \{\frac{n}{m} + \hat{\lambda} r(X_{i_2})\}}  \nonumber \\
    & + \frac{(n+m)^2}{m^3(m-1)} \sum_{j_1 \neq j_2 = 1}^m \frac{ k_\zeta(Y_{j_1},Y_{j_2})}{\{\frac{n}{m} + \hat{\lambda} r(Y_{j_1}) \} \{\frac{n}{m} + \hat{\lambda} r(Y_{j_2})\}} - \frac{2(n+m)^2}{nm^3} \sum_{i = 1}^n \sum_{j = 1}^m \frac{\hat{\lambda} r(X_{i}) k_\zeta(X_{i},Y_j)}{\{\frac{n}{m} + \hat{\lambda} r(X_{i}) \} \{\frac{n}{m} + \hat{\lambda} r(Y_j)\}},
\end{align} 
and observe that $\tilde{U}_\mathrm{id} - U_\mathrm{id} $ equals
\begin{align*}
     \frac{(n+m)^2}{n^2(n-1) m^2} &\sum_{i_1 \neq i_2 = 1}^n \frac{\hat{\lambda}^2 r(X_{i_1})r(X_{i_2})k_\zeta(X_{i_1},X_{i_2})}{\{\frac{n}{m} + \hat{\lambda} r(X_{i_1}) \} \{\frac{n}{m} + \hat{\lambda} r(X_{i_2})\}} + \frac{(n+m)^2}{m^4(m-1)} \sum_{j_1 \neq j_2 = 1}^m \frac{ k_\zeta(Y_{j_1},Y_{j_2})}{\{\frac{n}{m} + \hat{\lambda} r(Y_{j_1}) \} \{\frac{n}{m} + \hat{\lambda} r(Y_{j_2})\}}.
\end{align*}
Due to the boundedness assumption on $r$, this implies that there exists a constant $Q_1 \equiv Q_1(c, C, \tau) > 0$ such that
\begin{align}\label{eq:mean_Uid}
    |\mathbb{E}[U_\mathrm{id}] - \mathrm{MMD}^2_{r, k_\zeta}(f,g)| &\leq  |\mathbb{E}[\tilde{U}_\mathrm{id}] - \mathrm{MMD}^2_{r, k_\zeta}(f,g)| + \mathbb{E}[|U_\mathrm{id} - \tilde{U}_\mathrm{id}|] \nonumber \\
    & \leq |\mathbb{E}[\tilde{U}_\mathrm{id}]- \mathrm{MMD}^2_{r, k_\zeta}(f,g)|  + \frac{Q_1}{n^3} \, \mathbb{E}\left[ \sum_{i_1 \neq i_2 = 1}^n |k_\zeta(X_{i_1},X_{i_2})| + \sum_{j_1 \neq j_2 = 1}^m |k_\zeta(Y_{j_1},Y_{j_2})| \right] \nonumber \\
    & \leq |\mathbb{E}[\tilde{U}_\mathrm{id}]-\mathrm{MMD}^2_{r, k_\zeta}(f,g)| + \frac{Q_1 M \kappa_1(d)}{n},
\end{align}
and
\begin{align}\label{eq:var_Uid}
    \operatorname{Var}[U_\mathrm{id}] & \leq  2 \operatorname{Var}[\tilde{U}_\mathrm{id}] + 2\operatorname{Var}[U_\mathrm{id} - \tilde{U}_\mathrm{id}] \nonumber \\
    & \leq 2\operatorname{Var}[\tilde{U}_\mathrm{id}] + \frac{Q_1}{n^2} \left(\max_{i_1 \neq i_2}\mathbb{E}[k_\zeta^2(X_{i_1},X_{i_2})] + \max_{j_1 \neq j_2}\mathbb{E}[k_\zeta^2(Y_{j_1},Y_{j_2})]\right) \leq 2\operatorname{Var}[\tilde{U}_\mathrm{id}] + \frac{Q_1 M \kappa_2(d) \zeta^d}{n^2}.
\end{align}
Since $\zeta \ge 1$, we have $1/n \leq \zeta^{d/2}/n = \sqrt{\zeta^{d}/n^{2}}$, and therefore the additional mean term in~\eqref{eq:mean_Uid} is of no larger order than the second-moment scale given in~\eqref{eq:moments_permuted}, namely $\sqrt{\mathbb{E}[U_\sigma^{2}]} \lesssim \sqrt{\zeta^{d}/n^{2}}$. Similarly, the additional variance term in~\eqref{eq:var_Uid} matches the order in~\eqref{eq:moments_permuted}. Thus, it suffices to control the moments of $\tilde{U}_{\mathrm{id}}$. \\

\noindent \textbf{$\bullet$ Mean of $\tilde{U}_\mathrm{id}$}\\
We now address the harder problem of having $\hat{\lambda}$ in $\tilde{U}_\mathrm{id}$, instead of the non-random quantity $\lambda_0$ appearing in the definition of $\mathrm{MMD}^2_{r, k_\zeta}(f,g)$. To overcome this issue, for $\lambda > 0$ define
\begin{align}\label{eq:G_lambda}
    G(\lambda) := &\frac{(n+m)^2}{n(n-1) m^2} \sum_{i_1 \neq i_2 = 1}^n \frac{\lambda^2 r(X_{i_1})r(X_{i_2})k_\zeta(X_{i_1},X_{i_2})}{\{n/m + \lambda r(X_{i_1}) \} \{n/m + \lambda r(X_{i_2})\}}  \nonumber \\
    &  + \frac{(n+m)^2}{m^3(m-1)} \sum_{j_1 \neq j_2 = 1}^m \frac{ k_\zeta(Y_{j_1},Y_{j_2})}{\{n/m + \lambda r(Y_{j_1}) \} \{n/m + \lambda r(Y_{j_2})\}} \nonumber \\
    &- \frac{2(n+m)^2}{nm^3} \sum_{i = 1}^n \sum_{j = 1}^m \frac{\lambda r(X_{i}) k_\zeta(X_{i},Y_j)}{\{n/m + \lambda r(X_{i}) \} \{n/m + \lambda r(Y_j)\}},
\end{align}
so that $\tilde{U}_\mathrm{id} = G(\hat{\lambda})$. By expanding $G$ around $\lambda_0$ using a Taylor sum up to the second order we get $G(\hat{\lambda}) = G(\lambda_0) + (\hat{\lambda} - \lambda_0) \, G^\prime(\lambda_0) +  \frac{1}{2}(\hat{\lambda} - \lambda_0)^2 \, G^{\prime\prime} (\check{\lambda})$, where the random $\check{\lambda}$ is in between $\hat{\lambda}$ and $\lambda_0$. We can use this identity to bound the mean of $\tilde{U}_\mathrm{id}$ as follows: 
\begin{align}\label{eq:meanG_hat}
    &|\mathbb{E}[G(\hat{\lambda})] -  \mathrm{MMD}^2_{r, k_\zeta}(f,g)| \nonumber \\
    &= |\mathbb{E}[G(\lambda_0)]  - \mathrm{MMD}^2_{r, k_\zeta}(f,g) + \mathbb{E}[(\hat{\lambda} - \lambda_0)G^\prime(\lambda_0)] + \frac{1}{2}\mathbb{E}[(\hat{\lambda} - \lambda_0)^2  G^{\prime\prime} (\check{\lambda})]  | \nonumber \\
    & \leq |\mathbb{E}[(\hat{\lambda} - \lambda_0)  G^\prime(\lambda_0)] | +  \left|\tfrac{1}{2}\mathbb{E}[(\hat{\lambda} - \lambda_0)^2  G^{\prime\prime}(\check{\lambda})]\right|\leq \sqrt{\mathbb{E}[(\hat{\lambda} - \lambda_0)^2]  \mathbb{E}[\{G^\prime(\lambda_0)\}^2]} + \tfrac{1}{2}\sqrt{\mathbb{E}[(\hat{\lambda} - \lambda_0)^4]  \mathbb{E}[\{G^{\prime \prime}(\check{\lambda})\}^2]} \nonumber \\
    & \leq Q_1 \sqrt{\frac{1}{n}  \mathbb{E}[\{G^\prime(\lambda_0)\}^2]} + Q_1 \sqrt{\frac{1}{n^2}  \mathbb{E}[\{G^{\prime \prime}(\check{\lambda})\}^2]} \nonumber \\
    & = Q_1 \sqrt{\frac{1}{n}  \{\mathbb{E}[G^\prime(\lambda_0)]\}^2 + \frac{1}{n} \operatorname{Var}[G^\prime(\lambda_0)]} + Q_1 \sqrt{\frac{1}{n^2}  \mathbb{E}[\{G^{\prime \prime}(\check{\lambda})\}^2]}.
\end{align}
Note that in the first bound we used the triangle inequality together with the fact that the U-statistic $G(\lambda_0)$ is unbiased for $\mathrm{MMD}^2_{r, k_\zeta}$, while in the third one we used Lemma \ref{lemma:hat_lambda}. Now, as for the last term, we can use the fact that $0 < c \leq r(\cdot) \leq C$ to show that the first and second derivatives of $\lambda \mapsto \frac{\lambda^j}{\{n/m + \lambda r(x)\} \{n/m + \lambda r(y)\}}$ are uniformly bounded for all $\lambda > 0, \, x,y \in \mathbb{R}^d$ and $j \in \{0, 1, 2\}$. We can thus argue as we did for the second term in~\eqref{eq:var_Uid} to get 
\begin{align}\label{eq:expG''}
    \mathbb{E}[\{G^{\prime \prime}(\check{\lambda})\}^2] & \leq Q_1 \max_{k_1 \neq k_2} \mathbb{E}\left[ k_\zeta^2(Z_{k_1}, Z_{k_2}) \right] \leq Q_1 M \kappa_2(d) \zeta^d.
\end{align}
For the first two terms, observe that $G^\prime(\lambda_0)$ is a two-sample second-order U-statistic with asymmetric kernel $\check{h}(x_1, x_2, y_1, y_2) = b_{XX}(x_1, x_2)k_\zeta(x_1, x_2) +  b_{YY}(y_1, y_2)k_\zeta(y_1, y_2) + b_{XY}(x_1, y_2)k_\zeta(x_1, y_2) +  b_{XY}(x_2, y_1)k_\zeta(x_2, y_1)$, where
\begin{align*}
    \begin{dcases}
    b_{XX}(x_1, x_2) := \frac{2\lambda_0(n+m)^2 r(x_1) r(x_2)}{\left(n + \lambda_0 m r(x_1)\right) \left(n + \lambda_0 m r(x_2)\right)} - \frac{\lambda_0^2 m(n+m)^2 r^2(x_1) r(x_2)}{\left(n + \lambda_0 m r(x_1)\right)^2 \left(n + \lambda_0 m r(x_2)\right)} \\
    \hspace{75pt}- \frac{\lambda_0^2 m (n+m)^2 r(x_1) r^2(x_2)}{\left(n + \lambda_0 m r(x_1)\right) \left(n + \lambda_0 m r(x_2)\right)^2} \\
    b_{XY}(x_1, y_2) := - \frac{(n+m)^2 \, r(x_1)}{\left(n + \lambda_0 m r(x_1)\right) \left(n + \lambda_0 m r(y_2)\right)} + \frac{\lambda_0 m (n+m)^2 r^2(x_1)}{\left(n + \lambda_0 m r(x_1))^2(n + \lambda_0 m r(y_2)\right)}  \\
    \hspace{75pt} + \frac{\lambda_0 m (n+m)^2 r(x_1) r(y_2)}{\left(n + \lambda_0 m r(x_1)\right) \left(n + \lambda_0 m r(y_2)\right)^2} \\
    b_{YY}(y_1, y_2) :=  - \frac{m (n+m)^2 r(y_1)}{\left(n + \lambda_0 m  r(y_1)\right)^2 \left(n + \lambda_0 m r(y_2)\right)} - \frac{m (n+m)^2 r(y_2)}{\left(n + \lambda_0 m  r(y_1)\right) \left(n + \lambda_0 m r(y_2)\right)^2}.
\end{dcases}
\end{align*}
This is to say that \[
G^\prime(\lambda_0) = \frac{1}{n(n-1)m(m-1)}
\sum_{i_1 \neq i_2}\ \sum_{j_1 \neq j_2}
\check h(X_{i_1},X_{i_2},Y_{j_1},Y_{j_2}).
\]
We now bound the mean and the variance of $G^\prime(\lambda_0)$, starting from the latter. Since $\check h$ is not symmetric in $(x_1,x_2)$ and $(y_1,y_2)$, we introduce its within-sample symmetrization
\[
\bar h(x_1,x_2,y_1,y_2)
=\frac{1}{2!\,2!}\sum_{\pi\in {\cal S}_2}\sum_{\tau\in {\cal S}_2}
\check h\bigl(x_{\pi(1)},x_{\pi(2)},y_{\tau(1)},y_{\tau(2)}\bigr),
\]
where ${\cal S}_2$ denotes the symmetric group on two elements. Because $G'(\lambda_0)$ averages $\check h(X_{i_1},X_{i_2},Y_{j_1},Y_{j_2})$ over all ordered pairs
$i_1\neq i_2$ and $j_1\neq j_2$, replacing $\check h$ by $\bar h$ leaves the statistic unchanged.
We can thus express $G'(\lambda_0)$ as a two-sample U-statistic with a kernel symmetric within each sample, so that we may apply the classical 
variance formula for two-sample U-statistics in \citet[Equation~2, pag.~38]{lee1990ustatistics}.
Following the simplification in \citet[Equation~(69)]{kim2022minimax}, we obtain
\[
\operatorname{Var}\!\bigl[G'(\lambda_0)\bigr]
\leq
Q_2\left\{
\frac{\bar\sigma^2_{10}}{n}
+
\frac{\bar\sigma^2_{01}}{m}
+
\left(\frac{1}{n}+\frac{1}{m}\right)^2\bar\sigma^2_{22}
\right\}\leq Q_2 \left\{ \frac{\check{\sigma}^2_{10}}{n} + \frac{\check{\sigma}^2_{01}}{m} + \left(\frac{1}{n} + \frac{1}{m} \right)^2 \check{\sigma}^2_{22} \right\}
\]
for a sufficiently large universal constant $Q_2>0$, where $\check{\sigma}^2_{10} = \operatorname{Var}_{X_1}[\mathbb{E}_{X_2, Y_1, Y_2} \{\check{h}(X_1, X_2, Y_1, Y_2) \}]$, $\check{\sigma}^2_{01} = \operatorname{Var}_{Y_1}[\mathbb{E}_{X_1, X_2, Y_2} \{\check{h}(X_1, X_2, Y_1, Y_2) \}]$, $\check{\sigma}^2_{22} = \operatorname{Var}_{X_1, X_2, Y_1, Y_2}[\check{h}(X_1, X_2, Y_1, Y_2)]$, and $\bar{\sigma}^2_{10}$, $\bar{\sigma}^2_{01}$, $\bar{\sigma}^2_{22}$ are defined analogously with $\bar h$ in place of $\check h$. In particular, the final inequality uses that $\bar h$ is the within-sample permutation average of $\check h$; therefore, by Jensen's inequality, the corresponding projection variances for $\bar h$ are bounded up to universal constants by those for $\check h$, i.e.~$\bar\sigma^2_{10}\lesssim \check\sigma^2_{10}$, $\bar\sigma^2_{01}\lesssim \check\sigma^2_{01}$, and $\bar\sigma^2_{22}\le \check\sigma^2_{22}$. Now, using again the boundedness assumption on $r(\cdot)$ we have that $\max\{|b_{XX}(\cdot, \cdot)|, |b_{YY}(\cdot, \cdot)|, |b_{XY}(\cdot, \cdot)|\} \leq Q_1$, which can be used to show that $\check{\sigma}^2_{22} \leq Q_1 M \kappa_2(d) \zeta^d$ arguing as in \eqref{eq:expG''}. Furtermore,  
\begin{align*}
   \check{\sigma}^2_{10} &= \operatorname{Var}_{X_1}[\mathbb{E}_{X_2}\{b_{XX}(X_1, X_2)k_\zeta(X_1, X_2)\} + \mathbb{E}_{Y_2}\{b_{XY}(X_1, Y_2)k_\zeta(X_1, Y_2)\}] \\
   & \leq 2 \mathbb{E}_{X_1}[\mathbb{E}^2_{X_2}\{b_{XX}(X_1, X_2)k_\zeta(X_1, X_2)\} + \mathbb{E}^2_{Y_2}\{b_{XY}(X_1, Y_2)k_\zeta(X_1, Y_2)\}] \\
   & \leq 2\mathbb{E}_{X_1,X_2}[b^2_{XX}(X_1, X_2)k_\zeta^2(X_1, X_2)] + 2\mathbb{E}_{X_1,Y_2}[b^2_{XY}(X_1, Y_2)k_\zeta^2(X_1, Y_2)] \leq Q_1 M \kappa_2(d) \zeta^d.
\end{align*}
The same holds true for $\check{\sigma}^2_{01}$, and shows that  $n^{-1} \, \operatorname{Var}[G^\prime(\lambda_0)] \leq Q_1 M \kappa_2(d) n^{-2} \zeta^d$. Finally, we bound the expectation of $G^\prime(\lambda_0)$ by $\|\varphi_\zeta \ast  \psi_r \|_2$ up to constants, where $\ast$ stands for the convolution operator. Recalling $\varphi_\zeta(x-y) = k_\zeta(x,y)$ and $\psi_r  = \frac{n+m}{n+\lambda_0 m r}(\lambda_0 r f - g)$,  we have
\begin{align}\label{eq:mean_Gprime}
    \mathbb{E}[G^\prime(\lambda_0)] = &\, \mathbb{E}[b_{XX}(X_1, X_2)k_\zeta(X_1, X_2)] + \mathbb{E}[b_{YY}(Y_1, Y_2)k_\zeta(Y_1, Y_2)] + 2\mathbb{E}[b_{XY}(X, Y)k_\zeta(X, Y)] \nonumber\\
    = & - \mathbb{E}\left[\frac{2 m (n+m)^2 \lambda_0^2 r^2(X_1) r(X_2) \varphi_\zeta(X_1 - X_2)}{(n + \lambda_0 m r(X_1))^2 (n + \lambda_0 m r(X_2))}\right] + \mathbb{E}\left[ \frac{2(n+m)^2\lambda_0 r(X_1) r(X_2) \varphi_\zeta(X_1 - X_2)}{(n + \lambda_0 m r(X_1)) (n + \lambda_0 m r(X_2))} \right] \nonumber \\
    & - \mathbb{E}\left[ \frac{2m(n+m)^2 r(Y_1) \varphi_\zeta(Y_1 - Y_2)}{(n + \lambda_0  m r(Y_1))^2(n + \lambda_0 m r(Y_2))} \right]  - \mathbb{E}\left[\frac{2 (n+m)^2 r(X) \varphi_\zeta(X-Y)}{(n + \lambda_0 m r(X))(n + \lambda_0 m r(Y))}\right] \nonumber \\
    & + \mathbb{E}\left[\frac{2 m (n+m)^2 \lambda_0 r^2(X) \varphi_\zeta(X-Y)}{(n + \lambda_0 m r(X))^2(n + \lambda_0 m r(Y))} \right] + \mathbb{E}\left[ \frac{2m(n+m)^2 \lambda_0 r(X) r(Y) \varphi_\zeta(X-Y)}{(n + \lambda_0 m r(X))(n + \lambda_0 m r(Y))^2}\right] \nonumber \\
     = & \int_{\mathbb{R}^d} \frac{2 (n+m) r(x) f(x)}{(n + \lambda_0 m r(x))} \left\{ \int_{\mathbb{R}^d} \varphi_\zeta(x-y) \left(\frac{\lambda_0 (n+m) r(y) f(y)}{(n + \lambda_0 m r(y))} - \frac{(n+m) g(y)}{(n + \lambda_0 m  r(y))}  \right) dy \right\} dx \nonumber \\
     & - \int_{\mathbb{R}^d} \frac{2 m (n+m) \lambda_0 r^2(x) f(x)}{(n + \lambda_0 m r(x))^2} \left\{ \int_{\mathbb{R}^d} \varphi_\zeta(x-y) \left(\frac{\lambda_0 (n+m) r(y) f(y)}{(n + \lambda_0 m r(y))} - \frac{(n+m) g(y)}{(n + \lambda_0 m r(y))}  \right) dy \right\} dx \nonumber \\
     & + \int_{\mathbb{R}^d} \frac{2 m (n+m) r(x) g(x)}{(n + \lambda_0 m r(x))^2} \left\{ \int_{\mathbb{R}^d} \varphi_\zeta(x-y) \left(\frac{\lambda_0 (n+m) r(y) f(y)}{(n + \lambda_0 m r(y))} - \frac{(n+m) g(y)}{(n + \lambda_0 m r(y))}  \right) dy \right\} dx \nonumber \\
     = & \int_{\mathbb{R}^d} \left\{\frac{2 (n+m) r(x) f(x)}{(n + \lambda_0 m r(x))} - \frac{2 m (n+m) \lambda_0 r^2(x) f(x)}{(n + \lambda_0 m r(x))^2} + \frac{2m(n+m) r(x) g(x)}{(n + \lambda_0 m r(x))^2} \right\} (\varphi_\zeta \ast  \psi_r)(x) dx \nonumber \\
     \leq & \int_{\mathbb{R}^d} \frac{2(n+m) r(x) \{n f(x) + m g(x)\}}{(n+\lambda_0 m r(x))^2}  \left|(\varphi_\zeta \ast \psi_r)(x)\right| \, dx \leq   Q_0 \|\varphi_\zeta \ast \psi_r\|_{2}.
    \end{align}
In particular the last step follows by combining the Cauchy--Schwartz ineqaulity with the fact that the densities $f,g$ satisfy $\max(\|f\|_\infty, \|g\|_\infty) \leq M$, and $n/m + \lambda_0 r$ is bounded away from zero. Combining this with \eqref{eq:mean_Uid}, \eqref{eq:meanG_hat} and \eqref{eq:expG''} enables to conclude the analysis of the first moment of ${U}_\mathrm{id}$, showing that $|\mathbb{E}[{U}_\mathrm{id}] - \mathrm{MMD}^2_{r, k_\zeta}(f,g)| \leq Q_0\sqrt{n^{-2} \zeta^d+n^{-1}\|\varphi_\zeta \ast \psi_r \|_2^2}$, and thus
\begin{equation}\label{eq:finalBoundMeanU_id}
    |\mathbb{E}[{U}(Z)] - \, \mathrm{MMD}^2_{r, k_\zeta}(f,g)| \leq Q_0\sqrt{n^{-2} \zeta^d+n^{-1}\|\varphi_\zeta \ast \psi_r \|_2^2}.
\end{equation}

\noindent \textbf{$\bullet$Variance of $\tilde{U}_\mathrm{id}$}\\
Using the second-order Taylor approximation of $G(\hat{\lambda})$ around $\lambda_0$ gives
\begin{align}\label{eq:varUtilde}
    \operatorname{Var}[\tilde{U}_\mathrm{id}] &\leq 3\operatorname{Var}[G(\lambda_0)] + 3\operatorname{Var}[(\hat{\lambda} - \lambda_0) \, G^\prime(\lambda_0)] +  \tfrac{3}{4}\operatorname{Var}[(\hat{\lambda} - \lambda_0)^2 \, G^{\prime\prime} (\check{\lambda})] \nonumber \\
    & \leq Q_0( n^{-2} \zeta^d + n^{-1} \|\varphi_\zeta \ast \psi_r \|_2^2) + 3\mathbb{E}[(\hat{\lambda} - \lambda_0)^2 \, \{G^\prime(\lambda_0)\}^2] + \tfrac{3}{4} \mathbb{E}[(\hat{\lambda} - \lambda_0)^4\,\{G^{\prime\prime} (\check{\lambda})\}^2] \nonumber \\
    & \leq Q_0( n^{-2} \zeta^d + n^{-1} \|\varphi_\zeta \ast \psi_r \|_2^2) + 3 \sqrt{\mathbb{E}[(\hat{\lambda} - \lambda_0)^4] \, \mathbb{E}[\{G^\prime(\lambda_0)\}^4]} + \tfrac{3}{4} \sqrt{\mathbb{E}[(\hat{\lambda} - \lambda_0)^8] \, \mathbb{E}[\{G^{\prime \prime}(\check{\lambda})\}^4]} \nonumber \\
    & \leq Q_0(n^{-2} \zeta^d + n^{-1} \|\varphi_\zeta \ast \psi_r \|_2^2) + Q_1 \sqrt{n^{-2} \, \mathbb{E}[\{G^\prime(\lambda_0)\}^4]} + Q_1 \sqrt{n^{-4} \, \mathbb{E}[\{G^{\prime \prime}(\check{\lambda})\}^4]}.
\end{align} 
Note that the second inequality can be proved following similar lines as in \citet[Proposition~3]{schrab2023mmd}, while the last one follows from Lemma \ref{lemma:hat_lambda}. Similarly to \eqref{eq:expG''} and \eqref{eq:mean_Gprime}, we need to control fourth-order moments of some derivatives of $G$. Starting from the term involving $G^{\prime \prime}$, we can argue similarly to \eqref{eq:expG''} to show $\mathbb{E}[\{G^{\prime \prime}(\check{\lambda})\}^4] \leq Q_1 n^{-8} \, \mathbb{E}[\{\sum_{k_1 \neq k_2} |k_\zeta(Z_{k_1}, Z_{k_2})|\}^4] \leq Q_1 n^{-4} \, \mathbb{E}[\sum_{k_1 \neq k_2} \sum_{k_3 \neq k_4} k^2_\zeta(Z_{k_1}, Z_{k_2}) k^2_\zeta(Z_{k_3}, Z_{k_4})]$. It is now just a matter of counting what is the contribution of each term in the sum, depending on how many indices are shared. In this regard, observe that we have $\mathcal{O}(n^4)$ terms like $\mathbb{E}[k^2_\zeta(Z_1, Z_2) k^2_\zeta(Z_3, Z_4)] = \mathbb{E}[k^2_\zeta(Z_1, Z_2)]\mathbb{E}[k^2_\zeta(Z_3, Z_4)] \leq M^2 \kappa_2^2(d) \zeta^{2d}$, $\mathcal{O}(n^2)$ terms like $\mathbb{E}[k^4_\zeta(Z_1, Z_2)] \leq M \kappa_4(d) \zeta^{3d}$ and $\mathcal{O}(n^3)$ terms like $\mathbb{E}[k^2_\zeta(Z_1, Z_2) k^2_\zeta(Z_1, Z_3)] = \mathbb{E}[\mathbb{E}[k^2_\zeta(Z_1, Z_2) \,| \, Z_1] \, \mathbb{E}[k^2_\zeta(Z_1, Z_3) \, | \, Z_1]] \leq M^2 \kappa^2_2(d) \zeta^{2d}$. Since $\zeta^d/n^2 = n^{-\frac{8s}{4s+d}} \leq 1$ for our specific choice of $\zeta$, this suffices to show that $\sqrt{n^{-4} \, \mathbb{E}[\{G^{\prime \prime}(\check{\lambda})\}^4] }\leq Q_0 n^{-2}\zeta^{d}$. 

As for the term involving the first derivative of $G$, we have $\mathbb{E}[\{G^{\prime}(\lambda_0)\}^4] \leq 8\mathbb{E}[\{G^{\prime}(\lambda_0) - \mathbb{E}G^{\prime}(\lambda_0) \}^4] + 8\{\mathbb{E}G^{\prime}(\lambda_0) \}^4 \leq 8\mathbb{E}[\{G^{\prime}(\lambda_0) - \mathbb{E}G^{\prime}(\lambda_0) \}^4] + Q_1 \|\varphi_\zeta \ast \psi_r \|_2^4$. Hence, it just remains to bound the first term on the right-hand side. Define $h_{XX}(X_{1}, X_{2}) := b_{XX}(X_{1}, X_{2}) k_\zeta (X_{1}, X_{2}) - \mathbb{E}[ b_{XX}(X_{1}, X_{2}) k_\zeta (X_{1}, X_{2})]$ and similarly for $h_{YY}$ and $h_{XY}$. We thus have
\begin{align}\label{eq:Gprime_4}
    &\mathbb{E}[\{G^{\prime}(\lambda_0) - \mathbb{E}G^{\prime}(\lambda_0) \}^4] \nonumber \\
    & = \mathbb{E} \left[ \left\{ \frac{1}{n(n-1)} \sum_{i_1 \neq i_2} h_{XX}(X_{i_1}, X_{i_2}) + \frac{1}{m(m-1)} \sum_{j_1 \neq j_2} h_{YY}(Y_{j_1}, Y_{j_2}) + \frac{2}{nm} \sum_{i, j} h_{XY}(X_{i}, Y_{j}) \right\}^4 \right] \nonumber \\
   & \lesssim \frac{1}{n^8} \left(\mathbb{E}\left[\left\{\sum_{i_1\neq i_2} h_{XX}(X_{i_1},X_{i_2})\right\}^4\right]
   + \mathbb{E}\left[\left\{\sum_{j_1\neq j_2} h_{YY}(Y_{j_1},Y_{j_2})\right\}^4\right] + \mathbb{E}\left[\left\{\sum_{i, j} h_{XY}(X_i,Y_j)\right\}^4\right] \right).
\end{align}
These three terms can be bounded by almost identical arguments so we focus on the first. We expand it as $n^{-8} \sum_{i_1 \neq i_2} \sum_{i_3 \neq i_4} \sum_{i_5 \neq i_6} \sum_{i_7 \neq i_8} \mathbb{E}[ h_{XX}(X_{i_1}, X_{i_2})h_{XX}(X_{i_3}, X_{i_4})h_{XX}(X_{i_5}, X_{i_6})h_{XX}(X_{i_7}, X_{i_8})]$ 
and use a combinatorial argument to derive an upper bound. In this regard, we have $\mathcal{O}(n^8)$ terms with all distinct indices and $\mathcal{O}(n^7)$ terms with seven distinct indices, but they do not contribute to the sum since their expectations are zero. This is due to the independence among the $X$'s and the fact that $\mathbb{E}[h_{XX}(X_1, X_2)] = 0$. Moreover, we have
\begin{align*}
    \mathbb{E}&[ h_{XX}(X_{i_1}, X_{i_2})h_{XX}(X_{i_3}, X_{i_4})h_{XX}(X_{i_5}, X_{i_6})h_{XX}(X_{i_7}, X_{i_8})] \leq \mathbb{E}[ h^4_{XX}(X_{1}, X_{2})] \\
    & = \mathbb{E}[ \{ \, b_{XX}(X_{1}, X_{2}) k_\zeta(X_{1}, X_{2}) - \mathbb{E}[ b_{XX}(X_{1}, X_{2}) k_\zeta(X_{1}, X_{2})] \, \}^4] \leq 16 \mathbb{E}[ b^4_{XX}(X_{1}, X_{2}) k^4_\zeta(X_{1}, X_{2})] \\
    &\leq Q_1 \mathbb{E}[k^4_\zeta(X_{1}, X_{2})] \leq Q_1 M \kappa_4(d) \zeta^{3d},
\end{align*}
where in first inequality we repeatedly used the Cauchy--Schwarz inequality, while in the second inequality we used the fact that for $X, X^\prime$ independent and identically distributed we have $\mathbb{E}[(X - \mathbb{E}X)^4]  \leq 16 \mathbb{E}[X^4]$. Based on this, the contribution of the $\mathcal{O}(n^4)$ terms with four or less distinct indices is bounded above by $n^{-4}\zeta^{3d}$. It remains to bound those terms in which there are five or six different indices. As for the latter case, the only non-zero terms are those where the product decomposes into two independent blocks, each block containing two terms sharing one index; up to relabelling, a representative is of the form
\begin{align*}
    \mathbb{E}&[h_{XX}(X_1, X_2)h_{XX}(X_1, X_3)h_{XX}(X_4, X_5)h_{XX}(X_4, X_6)] \\
    & = \mathbb{E}[h_{XX}(X_1, X_2)h_{XX}(X_1, X_3)] \, \mathbb{E}[h_{XX}(X_4, X_5)h_{XX}(X_4, X_6)] = \mathbb{E}^2[h_{XX}(X_1, X_2)h_{XX}(X_1, X_3)] \\
    & = \mathbb{E}^2[ \mathbb{E}[h_{XX}(X_1, X_2) \mid X_1] \mathbb{E}[h_{XX}(X_1, X_3) \mid X_1]] \leq \mathbb{E}^2[ \mathbb{E}^2[h_{XX}(X_1, X_2) \mid X_1]] \\
    & = \mathbb{E}^2[ \mathbb{E}^2[\, \{ \,b_{XX}(X_1, X_2) k_\zeta(X_1, X_2) - \mathbb{E}[b_{XX}(X_1, X_2) k_\zeta(X_1, X_2) ]\, \}\, \mid X_1]] \\
    & \leq 8\mathbb{E}^2[ \mathbb{E}^2[b_{XX}(X_1, X_2) k_\zeta(X_1, X_2) \mid X_1]]  + 8\mathbb{E}^4[b_{XX}(X_1, X_2) k_\zeta(X_1, X_2)] \\
    & \leq Q_1 \mathbb{E}^2[ \mathbb{E}^2[|k_\zeta(X_1, X_2)| \mid X_1]]  + Q_1 \mathbb{E}^4[|k_\zeta(X_1, X_2)|] \leq Q_1 M^4 \kappa_1^4(d),
    \end{align*}
and since there are $\mathcal{O}(n^6)$ of these terms, their contribution is of the order $  n^{-2}$. Finally, when there are exactly five distinct indices we can have three different typical terms:
\begin{enumerate}
    \item
    \hfill$\begin{aligned}
    \mathbb{E}&[h_{XX}(X_1, X_2)h_{XX}(X_1, X_3)h_{XX}(X_1, X_4)h_{XX}(X_1, X_5)] = \mathbb{E}[\mathbb{E}^4[h_{XX}(X_1, X_2) \mid X_1]] \leq Q_1 M^4 \kappa_1^4(d);
    \end{aligned}$\hfill\mbox{}
    \item 
    \hfill$\begin{aligned}
    \mathbb{E}&[h_{XX}(X_1, X_2)h_{XX}(X_1, X_3)h_{XX}(X_1, X_4)h_{XX}(X_2, X_5)] \\
    & = \mathbb{E}[ \mathbb{E}[ h_{XX}(X_1, X_2)h_{XX}(X_1, X_3)h_{XX}(X_1, X_4)h_{XX}(X_2, X_5)\mid X_1, X_2]] \\
    & \leq \mathbb{E}[|h_{XX}(X_1, X_2) | \cdot \mathbb{E}^2[|h_{XX}(X_1, X_3)| \mid X_1] \cdot \mathbb{E}[|h_{XX}(X_2, X_5)| \mid X_2]] \leq Q_1 M^4 \kappa_1^4(d);
    \end{aligned}$\hfill\mbox{}
    \item 
    \hfill$\begin{aligned}
    \mathbb{E}&[h^2_{XX}(X_1, X_2)h_{XX}(X_3, X_4)h_{XX}(X_3, X_5)]  = \mathbb{E}[h^2_{XX}(X_1, X_2)]\,  \mathbb{E}[h_{XX}(X_3, X_4)h_{XX}(X_3, X_5)]\\
    & = \mathbb{E}[h^2_{XX}(X_1, X_2)]\ \mathbb{E}[\mathbb{E}^2[h_{XX}(X_3, X_4) \mid X_3]] \leq  Q_1 M^3 \kappa^2_1(d) \kappa_2(d) \zeta^d.
    \end{aligned}$\hfill\mbox{}
\end{enumerate}
There are $\mathcal{O}(n^5)$ of these terms, hence their contribution is of the order $\zeta^d n^{-3} \leq \zeta^{2d} n^{-3} \leq \zeta^{2d} n^{-2}$. This is enough to show that $\sqrt{n^{-2} \, (n^{-8} \, \mathbb{E}[\sum_{i_1\neq i_2} h_{XX}(X_{i_1},X_{i_2})\}^4])} \leq Q_0 \left(n^{-2} \zeta^d+n^{-1}\|\varphi_\zeta \ast \psi_r \|_2^2 \right)$ when combined with $\zeta^d/n^2 \leq 1$. This concludes the analysis for $h_{XX}$; analogous bounds can be obtained for $h_{YY}$ and $h_{XY}$. Overall, we thus have that $\sqrt{n^{-2} \mathbb{E}[\{G^{\prime}(\lambda_0)\}^4]} \leq Q_0 \left(n^{-2} \zeta^d+n^{-1}\|\varphi_\zeta \ast \psi_r \|_2^2 \right)$, which further shows that 
\[
\operatorname{Var}[U_\mathrm{id}] \leq Q_0 \left\{\frac{\zeta^d}{n^2} + \frac{\|\varphi_\zeta \ast \psi_r \|_2^2}{n} \right\}
\]
when combined with \eqref{eq:var_Uid}, \eqref{eq:varUtilde}, and \eqref{eq:Gprime_4}. \\

\noindent \textbf{$\bullet$ Relating the squared shifted maximum mean discrepancy to the $L^2$ distance}\\
Applying the previous argument with~\eqref{eq:2moments},~\eqref{eq:moments_permuted} and~\eqref{eq:finalBoundMeanU_id} and using the fact that $\sqrt{x+y} \leq \sqrt{x} + \sqrt{y}$ for $x,y \geq 0$ shows that a uniform control of the type~I and type~II errors is possible whenever 
\[
\mathrm{MMD}^2_{r, k_\zeta}(f,g) \geq Q_0 \left\{\frac{\zeta^{d/2}}{n} + \frac{\|\varphi_\zeta \ast \psi_r \|_2}{\sqrt{n}}\right\}.
\]
We now conclude the proof by relating the squared shifted maximum mean discrepancy metric to the $L^2$ distance defined in~\eqref{eq:l2_separation}, thus providing an upper bound on $\rho^*_r$. We have 
\begin{align*}
    \mathrm{MMD}^2_{r,k_\zeta}(f,g) & = \int_{\mathbb{R}^d} \int_{\mathbb{R}^d} \varphi_\zeta(x-y) \psi_r(x) \psi_r(y) dx dy = \int_{\mathbb{R}^d} \psi_r(x) (\varphi_\zeta \ast \psi_r)(x) dx \\
    & = \langle \psi_r,  \varphi_\zeta \ast \psi_r \rangle_2  = \frac{1}{2}\left(\|\psi_r\|_2^2 + \|\varphi_\zeta \ast \psi_r\|_2^2 -  \|\psi_r - \varphi_\zeta \ast \psi_r\|_2^2 \right), 
\end{align*}
hence an equivalent sufficient condition to bound the total error is given by
\begin{align*}
    \|\psi_r\|_2^2 \geq  \|\psi_r - \varphi_\zeta \ast \psi_r\|_2^2 + Q_0 \frac{\zeta^{d/2}}{n} + Q_0 \frac{\|\varphi_\zeta \ast \psi_r \|_2}{\sqrt{n}} -  \|\varphi_\zeta \ast \psi_r\|_2^2.
\end{align*}
This can be further simplified to just $\|\psi_r\|_2^2 \geq  \|\psi_r - \varphi_\zeta \ast \psi_r\|_2^2 + Q_0 \, n^{-1}\zeta^{d/2}$ in light of the fact that 
\begin{align*}
    \sqrt{\frac{Q_0^2 \|\varphi_\zeta \ast \psi_r \|_2^2}{n}} - \|\varphi_\zeta \ast \psi_r \|_2^2 \leq \frac{Q_0^2}{n} + \|\varphi_\zeta \ast \psi_r \|_2^2 - \|\varphi_\zeta \ast \psi_r \|_2^2 \leq  \frac{Q_0^2 \zeta^{d/2}}{n},
\end{align*}
as $\zeta \geq 1$ and $\sqrt{xy} \leq x + y$ for $x,y \geq 0$. Furthermore, we can argue as in \citet[Theorem 6]{schrab2023mmd} and show that if $\psi_r \in \mathcal{S}_d^s(L)$ we have $\|\psi_r - \varphi_\zeta \ast \psi_r\|_2^2 \leq Q_3 \|\psi_r\|_2^2 + Q_4 \, \zeta^{-2s}$, for some constants $Q_3 \in (0,1)$ and $Q_4 \equiv Q_4(d, s, L) > 0$. This shows that for $\zeta = n^{\frac{2}{4s+d}}$ there exists a constant $C_1 \equiv C_1(c, C, d, \tau, M,s, L, \alpha, \beta)$ such that
\[
\rho^*_r \leq C_1 \sqrt{\frac{\zeta^{d/2}}{n} + \zeta^{-2s}} = C_1 \, n^{-\frac{2s}{4s+d}},
\]
and completes the proof.
\end{proof}

\begin{lemma}\label{lemma:hat_lambda}
Assume $ m \leq n \leq \tau \, m$ for $\tau \geq 1$, and $0 < c \leq r(x) \leq C$ for all $x \in \mathcal{X}$. Let $\hat{\lambda}$ be such that 
\[
\sum_{i = 1}^n \frac{\hat{\lambda}m r(X_i)}{n + \hat{\lambda}m r(X_i)} = \sum_{j = n+1}^{n+m} \frac{n}{n + \hat{\lambda}m r(Y_j)}
\]
and $\lambda_0$ such that 
\[\int \frac{n f + mg}{n + \lambda_0 mr} d\mu = 1.
\]
There exists a constant $Q_0 \equiv Q_0(p, c, C, \tau) > 0$ such that 
\[
\mathbb{E}[|\hat{\lambda} - \lambda_0|^p] \leq Q_0 \, n^{-p/2} \text{ for all } p \in \mathbb{N}. 
\]
\end{lemma}

 \begin{proof} 
We know that $\hat{\lambda}$ is a Z-estimator, being the solution with respect to $\lambda$ of \[
\Psi_{n,m}(\lambda) := \frac{1}{n+m} \sum_{k = 1}^{n+m} \psi_\lambda(Z_k) = 0, \quad \text{ with } \quad \psi_\lambda(x) := \frac{n+m}{n+\lambda m r(x)} - 1,
\]
whereas $\lambda_0$ is a population quantity and satisfies 
\begin{align}\label{eq:tildeLambda}
    \int \frac{nf}{n + \lambda_0 m r} d\mu + \int \frac{mg}{n + \lambda_0 m r} d\mu = 1.
\end{align}
We can use the assumptions $ m \leq n \leq \tau  m$ and $0 < c \leq r(\cdot) \leq C$ to show that
\[
\frac{\partial \psi_\lambda(x)}{\partial \lambda} = - \frac{(n+m) m r(x)}{\{n+\lambda m r(x)\}^2} \leq - \frac{c\tau^{-1}(1+\tau^{-1}) n^2}{(1+C/c)^2n^2} = -\frac{c^3(1+\tau)}{\tau^2(c+C)^2} =: - a(c, C, \tau) \equiv - a < 0,
\]
which implies that for all $\epsilon > 0$ we have $ \{|\hat{\lambda} - \lambda_0| < \epsilon\} \supseteq \{|\Psi_{n,m}(\hat{\lambda}) - \Psi_{n,m}(\lambda_0)| < a \epsilon\}  = \{|\Psi_{n,m}(\lambda_0)| < a \epsilon\}$, using the fact that $\Psi_{n,m}(\hat{\lambda}) = 0$ by definition of $\hat{\lambda}$. As a result, the previous inclusion shows that the boundedness assumption on $r(\cdot)$ allows to relate how close $\hat{\lambda}$ is to $\lambda_0$ with $ |\Psi_{n,m}(\lambda_0)|$, which is easier to analyze since it is the sum of zero mean, bounded independent and identically distributed~random variables. In this regard, we have
\begin{align}\label{eq:Hoeffding_h}
    \mathbb{P}&\{|\hat{\lambda} - \lambda_0| \geq  \epsilon\} \leq \mathbb{P}\{|\Psi_{n,m}(\lambda_0)| \geq a \epsilon\} = \mathbb{P} \left\{\left| \frac{1}{n+m}\sum_{k = 1}^{n+m} \left(\frac{n+m}{n+\lambda_0mr(Z_k)} - 1\right) \right| \geq a \epsilon \right\} \nonumber \\
    & = \mathbb{P} \left\{\left| \frac{1}{n}\sum_{i = 1}^{n} \frac{n}{n+\lambda_0mr(X_i)} + \frac{1}{m}\sum_{j = 1}^{m} \frac{m}{n+\lambda_0mr(Y_j)} - 1 \right| \geq a \epsilon \right\} \nonumber \\
    & \overset{\eqref{eq:tildeLambda}}{=} \mathbb{P} \left\{\left| \frac{1}{n}\sum_{i = 1}^{n} \frac{n}{n+\lambda_0mr(X_i)} + \frac{1}{m}\sum_{j = 1}^{m} \frac{m}{n+\lambda_0mr(Y_j)} - \int \frac{nf}{n + \lambda_0mr} d\mu - \int \frac{mg}{n + \lambda_0 m r} d\mu \right| \geq a \epsilon \right\} \nonumber \\
    & \leq \mathbb{P} \left\{\left| \frac{1}{n}\sum_{i  =1}^{n} \frac{n}{n+\lambda_0mr(X_i)} - \int \frac{nf}{n + \lambda_0mr} d\mu \right| \geq \frac{a\epsilon}{2} \right\} \nonumber \\
    & \qquad\qquad\qquad+ \mathbb{P} \left\{\left|  \frac{1}{m}\sum_{j = 1}^{m} \frac{m}{n+\lambda_0mr(Y_j)} - \int \frac{mg}{n + \lambda_0 m r} d\mu \right| \geq \frac{a\epsilon}{2} \right\} \nonumber \\
    & \leq 2 \exp \left\{-\frac{n a^2 \epsilon^2}{2} \right\} + 2 \exp \left\{-\frac{m a^2 \epsilon^2}{2} \right\} \leq 2 \exp \left\{-\frac{n a^2 \epsilon^2}{2} \right\} + 2 \exp \left\{-\frac{n a^2 \epsilon^2}{2 \tau} \right\} \leq 4 \exp \left\{-\frac{n a^2 \epsilon^2}{2 \tau} \right\},
\end{align}
where in the last line we applied Hoeffding's inequality for bounded random variables to $0 \leq \frac{m}{n+\lambda_0mr(\cdot)} \leq \frac{n}{n+\lambda_0mr(\cdot)} \leq 1$, and the assumption that $m \leq n \leq \tau m$. We can thus bound the moment of order $p$ of $|\hat{\lambda} - \lambda_0|$ as 
\begin{align*}
    \mathbb{E}[|\hat{\lambda} - \lambda_0|^p] &= \int _0^\infty \mathbb{P}\{|\hat{\lambda} - \lambda_0|^p \geq \epsilon \} d\epsilon  = \int _0^\infty \mathbb{P}\{|\hat{\lambda} - \lambda_0| \geq \epsilon^{\frac{1}{p}} \} d\epsilon \\
    & = \int _0^\infty \mathbb{P}\{|\hat{\lambda} - \lambda_0| \geq \epsilon \} \,  p \, \epsilon^{p-1} d\epsilon  = \left(\frac{2\tau}{na^2} \right)^{\frac{p}{2}} \int _0^\infty \mathbb{P}\left\{|\hat{\lambda} - \lambda_0| \geq \sqrt{\frac{2\tau}{na^2}} t\right\} \,  p \, t^{p-1} dt \\
    & \leq 2 p  \left(\frac{2\tau}{na^2} \right)^{\frac{p}{2}} \int_0^\infty 2 t^{p-1} e^{-t^2} dt =  2 p  \left(\frac{2\tau}{na^2} \right)^{\frac{p}{2}} \int_0^\infty t^{\frac{p}{2}-1} e^{-t} dt = 2 p  \left(\frac{2\tau}{na^2} \right)^{\frac{p}{2}} \Gamma\Big(\frac{p}{2}\Big),
\end{align*}
where in the last equality we used the definition of the Gamma function. This completes the proof.
\end{proof}

\begin{lemma}\label{lemma:sampler_noIJ}
    Fix subsets $\bar{I} \subseteq [n]$ and $\bar{J} \subseteq \{n+1, \ldots, n+m\}$, and modify Step 3 of Algorithm~\ref{alg:pairwise_sampler} to be: Sample a vector of couples $\tau_t = \{(i_1^{t}, j_1^{t}), \ldots, (i_K^{t}, j_K^{t}) \}$ such that $(i_1^{t}, \ldots, i_K^{t})$ are sampled uniformly and without replacement from $[n]\setminus \bar{I}$, and $(j_1^{t}, \ldots, j_K^{t})$ are sampled uniformly and without replacement from $\{n+1, \ldots, n+m \}\setminus \bar{J}$. Keep all the other steps the same. Then this algorithm still has~\eqref{eq:sampling_1} as its stationary distribution.
    \end{lemma}
    \begin{proof}
    The key observation to make here is that all the steps that led to~\eqref{eq:detailedBalance} in the proof of Proposition~\ref{prop:algo_sampler_valid} remain valid in this other setting. The only difference lies in the fact that the Markov chain associated with this new algorithm is not irreducible, and hence it will not converge to~\eqref{eq:sampling_1} if we let it run long enough. Nonetheless,~\eqref{eq:sampling_1} is still a stationary distribution, and we will now show this by proving a detailed balance condition. Let $K = \min\{n - \#\bar{I},m - \#\bar{J}\}$ and let $\tilde{\mathcal{P}}$ be the set of all $K$ couples of the form $\{(i_1, j_1), \ldots, (i_K, j_K) \}$ such that $(i_1, \ldots, i_K)$ contains distinct elements from $[n]\setminus \bar{I}$, and $(j_1, \ldots, j_K)$ contains distinct elements from $\{n+1, \ldots, n+m \}\setminus \bar{J}$. For all $t \in \mathbb{N}_+$ and all permutations $p, p^{\prime}$, we have
$$
\mathbb{P}\left\{\sigma_t = p^{\prime} \mid \sigma_{t-1} = p \right\}=\frac{1}{|\tilde{\mathcal{P}}|} \sum_{\tau \in \tilde{\mathcal{P}}} \mathbb{P}\left\{\sigma_t = p^{\prime} \mid \sigma_{t-1} = p, \tau_t = \tau \right\},
$$
since at each time $t$ this new algorithm draws $\tau_t \in \tilde{\mathcal{P}}$ uniformly at random. Next, given $\tau_t = \tau$ and $\sigma_{t-1} = p$, it must be the case that $\sigma_t$ satisfies $\sigma_t \sim_\tau p$, since this new algorithm still uses Steps 4-5 of Algorithm \ref{alg:pairwise_sampler}. Arguing as in~\eqref{eq:oddsRatio} gives \[
\frac{\mathbb{P}\left\{\sigma_t = p^{\prime} \mid \sigma_{t-1} = p, \tau_t = \tau \right\}}{\mathbb{P}\left\{\sigma_t = p^{\prime \prime} \mid \sigma_{t-1} = p, \tau_t = \tau \right\}} = \frac{\mathbb{P}\left\{\sigma=p^{\prime}\right\}}{\mathbb{P}\left\{\sigma=p^{\prime \prime}\right\}},
\]
which implies that
\[
\mathbb{P}\left\{\sigma_t = p^{\prime} \mid \sigma_{t-1} = p \right\} = \frac{1}{|\tilde{\mathcal{P}}|} \sum_{\tau \in \tilde{\mathcal{P}}} \frac{\mathbbm{1}\left\{p^{\prime} \sim_\tau p \right\} \cdot \mathbb{P}\left\{\sigma=p^{\prime}\right\}}{\sum_{p^{\prime \prime}} \mathbbm{1}\left\{p^{\prime \prime} \sim_\tau p\right\} \cdot \mathbb{P}\left\{\sigma=p^{\prime \prime}\right\}}.
\]
This concludes the proof by analogous calculations to those in~\eqref{eq:detailedBalance}.
\end{proof}

\begin{proof}[Proof of Theorem \ref{thm:LB_minimax}]
For simplicity, we assume $n = m$ throughout the proof. The more general case $m \leq n \leq \tau m$ corresponds to a simpler problem, as it involves a larger sample size; thus, the lower bound derived here remains valid in that setting. For varying densities $p, q$ on $\mathbb{R}^d$ and a sufficiently small $\rho > 0$, define the set 
\begin{align}\label{eq:tildeS_theta}
    \tilde{\mathcal{S}}_\theta^r(\rho) := \mathcal{S}_\theta^r(\rho) \, \cap \, \left\{(f \equiv f_p, g \equiv g_q) \, : \, f_p  = \gamma_p \frac{p}{r}  \left(1 + r \, \frac{\gamma_q}{\gamma_p}\right), g_{q} = \gamma_q q  \left(1 + r \, \frac{\gamma_q}{\gamma_p}\right) \right\} \subseteq \mathcal{S}_\theta^r(\rho),
\end{align}
where $\gamma_p = \frac{\sqrt{B}}{A\sqrt{B} + \sqrt{A}}$ and  $\gamma_q = \frac{\sqrt{A}}{A\sqrt{B} + \sqrt{A}}$, with $A = \int_{\mathbb{R}^d} p(x)/r(x) dx$ and $B = \int_{\mathbb{R}^d} q(x) r(x) dx$. One can easily check that $\int_{\mathbb{R}^d} f_p(x) dx = \int_{\mathbb{R}^d} g_q(x) dx = 1$ for this specific choice of $\gamma_p$ and $\gamma_q$. As a result, for any test $\varphi \in \Psi(\alpha)$, and for all prior distributions $\mu_0, \mu_1$ supported on $H_0$ and $\tilde{\mathcal{S}}_\theta^r(\rho)$, respectively, Equation~\eqref{eq:tildeS_theta} shows that we can bound the total error probability as 
\begin{align}\label{eq:minimax_TV_reduction}
    \alpha &+ \sup_{(f,g) \in \mathcal{S}_\theta^r(\rho)} \mathbb{E}_P (1-\varphi) \geq \sup_{g \, \propto \, r f} \mathbb{E}_P \,\varphi + \sup_{(f,g) \in \mathcal{S}_\theta^r(\rho)} \mathbb{E}_P (1-\varphi) \nonumber \\
    & \geq \sup_{g \, \propto \, r f} \mathbb{E}_P \,\varphi + \sup_{(f,g) \in \tilde{\mathcal{S}}_\theta^r(\rho)} \mathbb{E}_P (1-\varphi) \geq \mathbb{E}_{\mu_0} \{\mathbb{E}_{P} \,\varphi\} + \mathbb{E}_{\mu_1} \{\mathbb{E}_{P} (1-\varphi)\} \geq 1 - \operatorname{TV}(\mathbb{E}_{\mu_0} P, \mathbb{E}_{\mu_1} P),
\end{align}
where we recall that $P = P_{f}^{\otimes n} \otimes P_{g}^{\otimes n}$. This demonstrates that controlling the total variation distance above is sufficient to obtain a lower bound on $\rho^*_r$. We proceed to do so for specific choices of $\mu_0$ and $\mu_1$, using a classical perturbation-based method originating from \cite{ingster1987minimax} and recently employed in two-sample and independence testing problems in \cite{meynaoui2019adaptive} and \cite{LiYuanJMLROptimal}. Start by considering $q_0(x) = q_1(x) = p_0(x) = \mathbbm{1}\{x \in [0,1]^d\}$ and define $A_i = \int_{\mathbb{R}^d} p_i(x)/r(x) dx$ and $B_i = \int_{\mathbb{R}^d} q_i(x) r(x) dx$ for $i \in \{0,1\}$. We can assume without loss of generality that $\int_{[0,1]^d} r(x) dx = 1$ so that $B_0 = B_1 = 1$; otherwise, since the problem is scale invariant, we might replace $r$ with $r/\int_{[0,1]^d} r(x) dx$ without affecting the minimax separation. Finally, define $p_1$ to be a perturbation of $p_0$ of the following form. Let $\tilde{M}$ be as in~\eqref{eq:maxRatioNorm}, fix
\[
B_n^{1/d}  = \left\lfloor \left(\frac{\sqrt{c}(2\pi)^{d/2}L}{2\sqrt{2(1+c)} \, \tilde{M} \rho} \right)^{1/s} \right\rfloor \qquad \text{and } \qquad \delta_n = \sqrt{\frac{2(1+c)}{c}} \rho \, B_n^{-1/2},
\]
and consider $\{\phi_1, \ldots, \phi_{B_n}\}$ as in~\eqref{eq:phi_k} in the proof of Lemma \ref{lem:smooth_bumps_measurable_r}, i.e.~an orthonormal set of functions in $L^2(\mathbb{R}^d)$ whose supports are disjoint and contained in $[0,1]^d$, and satisfy 
\begin{align}\label{eq:conditionPhi}
    \int_{\mathbb{R}^d} \phi_k(x) dx = \int_{\mathbb{R}^d} \frac{\phi_k(x)}{r(x)} dx = 0 \quad \text{ for all } k \in [B_n].
\end{align}
It is convenient to recall that the $\phi_k$'s are of the form $\phi_k(x)\;=\;
\frac{B_{n}^{1/2}}{\|\phi_{0,k}\|_{2}}\,
\phi_{0,k}\!\bigl(B_{n}^{1/d}\{x-x_k^{0}\}\bigr)$, where $x_k^0$ is the lower-left corner of their support, and the $\phi_{0,k}$'s satisfy $\max_{k \in [B_n]} \left\{\frac{\| \phi_{0,k}\|_{\mathcal{S}^s_d}}{\| \phi_{0,k}\|_2} \vee \frac{\| \phi_{0,k}\|_\infty}{\| \phi_{0,k}\|_2} \right\} \leq \tilde{M}$, with $\|\cdot\|_{\mathcal{S}^s_d}$ defined in~\eqref{eq:SobolevNorm}.  Based on this, we define 
\begin{align}\label{eq:p0_perturbed}
       p_1(x) \equiv p_{1,a}(x) = p_0(x) + \delta_n \, \sum_{k = 1}^{B_n} a_k \, \phi_k(x), 
\end{align}
where $a = (a_1, \ldots, a_{B_n})$ is a collection of independent and identically distributed~Rademacher random variables, meaning that $\mathbb{P}\{a_k = 1\} = \mathbb{P}\{a_k = -1\} = 1/2$ for all $k \in [B_n]$. 

With these definitions in mind, let $f_i := f_{p_i}$ and $g_i := g_{q_i}$ for $i \in \{0,1\}$, with $f_p, g_p$ as in \eqref{eq:tildeS_theta}.
Now, it is clear that $(f_0, g_0)$ satisfy the null, thus we just need to check that $(f_1, g_1) \in \tilde{\mathcal{S}}_\theta^r(\rho)$ for all $a \in \{\pm 1\}^{B_n}$, which is required to ensure that the distribution $\mu_1$ that assigns equal probability to each of them is indeed supported on $\tilde{\mathcal{S}}_\theta^r(\rho)$. We may assume that $\|f_1\|_\infty \vee \|g_1\|_\infty \leq M$; otherwise, it suffices to choose $p_0(x) = \mathbbm{1}\{x \in [0,u]^d\}$ for $u \geq 1$ sufficiently large and construct a perturbed version of it as in~$\eqref{eq:p0_perturbed}$. This is clearly sufficient for $g_1$, since the prefactor of $q$ in the definition of $g_q$ in~\eqref{eq:tildeS_theta} depends on $r$, which is uniformly bounded. A similar argument applies to $f_1$, taking into account that the magnitudes of the bumps are such that 
\[
\|\delta_n \sum_{k=1}^{B_n} a_k \phi_k \|_\infty \leq \delta_n \max_{k=1}^{B_n} \|\phi_k \|_\infty  \leq  \delta_n B_n^{1/2} \max_{k=1}^{B_n} \frac{\|\phi_{0,k} \|_\infty}{\|\phi_{0,k} \|_2} \leq\tilde{M} \delta_n B_n^{1/2} \leq \sqrt{\frac{2(1+c)}{c}} \tilde M \, \rho, 
\]
where $\rho$ is sufficiently small. In particular, we will set $\rho = c_1 n^{-2 s /(4 s+d)}$, where the constant $c_1 = c_1(c, C, \theta, \alpha, \beta)$ is chosen so that $\sqrt{2c^{-1}(1+c)} \tilde M \, \rho \leq 1$;
this also implies that $f_1$ is non-negative. As for the conditions involving $\psi_r$, observe that 
\[
\psi_r = 2 \frac{\lambda_0 r f_1 - g_1}{1 + \lambda_0 r} = 2\frac{1 + r \, \gamma_{q_1}/\gamma_{p_1}}{1 + \lambda_0 r}(\lambda_0 \gamma_{p_1} p_1 - \gamma_{q_1} q_1) = 2\gamma_{q_1}(p_1 - q_1),
\]
since $\lambda_0 = \gamma_{q_1}/\gamma_{p_1}$ satisfies $\int_{\mathbb{R}^d} \psi_r(x) dx = 0$. Hence we need to verify that $2\gamma_{q_1}\|p_1 - q_1 \|_2 > \rho$ and $2\gamma_{q_1}(p_1 - q_1) \in \mathcal{S}_d^s(L)$. Start by noticing that for all $a \in \{\pm 1\}^{B_n}$ we have 
\[
1 \geq \gamma_{q_1}^2 = (\sqrt{A_1}+1)^{-2} \geq (2A_1 + 2)^{-1} \geq \frac{c}{2(1+c)},
\]
as \begin{align*}
    A_1 &= \int_{\mathbb{R}^d} \frac{p_1(x)}{r(x)} dx = \int_{\mathbb{R}^d} \frac{p_0(x) + \delta_n \, \sum_{k = 1}^{B_n} a_k \, \phi_k(x)}{r(x)} dx \overset{\eqref{eq:conditionPhi}}{=} \int_{\mathbb{R}^d} \frac{p_0(x)}{r(x)} dx \leq \frac{1}{c}.
\end{align*} 
Then, as for the condition involving the $L^2$-norm of $\psi_r$, we have
\begin{align*}
    4\gamma_{q_1}^2 \|p_1 - q_1 \|_2^2 &= 4\gamma_{q_1}^2 \left\| \delta_n \, \sum_{k = 1}^{B_n} a_k \, \phi_k \right\|_2^2 = 4\gamma_{q_1}^2 \delta_n^2 \left\|  \, \sum_{k = 1}^{B_n} a_k \, \phi_k \right\|_2^2 = 4\gamma_{q_1}^2 \delta_n^2 \int_{\mathbb{R}^d} \left(  \, \sum_{k = 1}^{B_n} a_k \, \phi_k(x) \right)^2 dx \\
    & = 4\gamma_{q_1}^2 \delta_n^2 \, \sum_{k = 1}^{B_n} a_{k}^2  \, \int_{\mathbb{R}^d}\phi_{k}^2(x) dx = 4\gamma_{q_1}^2 \delta_n^2 \, B_n \geq \frac{2c}{(1+c)} \delta_n^2 \, B_n =  4\rho^2 > \rho^2 
\end{align*}
for our particular choice of $\delta_n$ and $B_n$. As for the smoothness condition, define for notational convenience the norm 
\begin{align}\label{eq:SobolevNorm}
\| p \|_{\mathcal{S}^s_d}^2 := \int_{\mathbb{R}^d}\|\xi\|_2^{2 s}|\widehat{p}(\xi)|^2 \mathrm{~d} \xi,
\end{align}
so that $\mathcal{S}^s_d(L) = \{p \in L^1(\mathbb{R}^d) \cap L^2(\mathbb{R}^d) : \| p \|_{\mathcal{S}^s_d}^2 \leq (2\pi)^d L^2 \}$. Furthermore,  since the iterated Laplacian $(-\Delta)^{s}$ of order $s \in \mathbb{N}_+$ is the Fourier multiplier with symbol $\|\xi\|^{2s}_2$, Plancherel's theorem gives
\[
\int_{\mathbb{R}^d} \|\xi\|_2^{2s} \, \widehat{\phi_1}(\xi) \,
                         \overline{\widehat{\phi_2}(\xi)}\,d\xi = \int_{\mathbb{R}^d}(-\Delta)^{s}\phi_1(x)\,
                         \overline{\phi_2(x)}\,dx = 0,
\]
since  $(-\Delta)^{s}\phi_{1}$ is a combination of derivatives of $\phi_{1}$ of order $2s$, and $\phi_1, \phi_2$ have disjoint supports. This implies that 
\begin{align*}
4 & \gamma_{q_1}^2  \,  \|  p_1 - q_1\|_{\mathcal{S}^s_d}^2 \leq \\
& 4 \,  \|  p_1 - q_1\|_{\mathcal{S}^s_d}^2 =    4 \left\| \delta_n \, \sum_{k = 1}^{B_n} a_k \, \phi_k \right\|_{\mathcal{S}^s_d}^2 = 4 \delta_n^2 \left\| \, \sum_{k = 1}^{B_n} a_k \, \phi_k \right\|_{\mathcal{S}^s_d}^2 =  4\delta_n^2\int_{\mathbb{R}^d}\|\xi\|_2^{2 s}\left|\sum_{k = 1}^{B_n} a_k \, \widehat{\phi_k}(\xi)\right|^2 \mathrm{~d} \xi \\
& \leq 4 \delta_n^2 B_n \max_{k=1}^{B_n} \int_{\mathbb{R}^d}\|\xi\|_2^{2 s}\left| \, \widehat{\phi_k}(\xi)\right|^2 \mathrm{~d} \xi = 4  \delta_n^2 B_n \max_{k=1}^{B_n} \int_{\mathbb{R}^d}\|\xi\|_2^{2 s}\left| \, \widehat{\frac{B_n^{1/2}}{\|\phi_{0,k}\|_2}\phi_{0,k}(B_n^{1/d} \{\cdot - x_k^0\})}(\xi)\right|^2 \mathrm{~d} \xi \\
& = 4  \delta_n^2 B_n \max_{k=1}^{B_n} \int_{\mathbb{R}^d}\|\xi\|_2^{2 s} \, \left|\frac{B_n^{-1/2}} {\|\phi_{0,k}\|_2} \, e^{- i \langle \xi, x_k^0  \rangle}\widehat{\phi_{0,k}}(\xi/B_n^{1/d})\right|^2 \mathrm{~d} \xi = 4\max_{k=1}^{B_n} \frac{ \delta_n^2}{\|\phi_{0,k}\|_2^2} \int_{\mathbb{R}^d}\|\xi\|_2^{2 s} \, \left|\widehat{\phi_{0,k}}(\xi/B_n^{1/d})\right|^2 \mathrm{~d} \xi \\
& = 4 \max_{k=1}^{B_n} \frac{ \delta_n^2}{\|\phi_{0,k}\|_2^2} \int_{\mathbb{R}^d}\|B_n^{1/d} \, \xi\|_2^{2 s} \, \left|\widehat{\phi_{0,k}}(\xi)\right|^2 B_n \mathrm{~d} \xi = 4 \delta_n^2 B_n^{\frac{2s + d}{d}} \max_{k=1}^{B_n} \frac{\|\phi_{0,k}\|_{\mathcal{S}^s_d}^2}{\|\phi_{0,k}\|_2^2} \leq 4\tilde{M}^2 \delta_n^2 B_n^{\frac{2s + d}{d}} \leq (2\pi)^d L^2
\end{align*}
for our particular choice of $\delta_n$ and $B_n$. 

It remains to control the total variation in~\eqref{eq:minimax_TV_reduction} for this specific choice of $\mu_1, \delta_n$ and $B_n$, and assess for which values of $\rho$ we can bound it above by $1-\alpha -\beta$. Writing $f_1 \equiv f_{1,a}$ and $g_1 \equiv g_{1,a} $ to highlight their dependence on $a \in \{\pm1\}^d$ through $p_1$, and using $\chi^2$ for the chi-square divergence, we look at 
\begin{align}\label{eq:LB_fromTVtoChisq}
    4\operatorname{TV}^2&(P_{f_0}^{\otimes n} \otimes P_{g_0}^{\otimes n}, \mathbb{E}_a\{P_{f_1}^{\otimes n} \otimes P_{g_1}^{\otimes n}\}) =  4\operatorname{TV}^2\left(P_{f_0}^{\otimes n} \otimes P_{g_0}^{\otimes n}, 2^{-B_n} \sum_{a \in \{\pm 1\}^d} P_{f_{1,a}}^{\otimes n} \otimes P_{g_{1,a}}^{\otimes n}\right)  \nonumber \\
    & \leq \chi^2\left(P_{f_0}^{\otimes n} \otimes P_{g_0}^{\otimes n}, 2^{-B_n} \sum_{a \in \{\pm 1\}^d} P_{f_{1,a}}^{\otimes n} \otimes P_{g_{1,a}}^{\otimes n} \right) \nonumber \\
    & = 2^{-2B_n} \sum_{a_1, a_2 \in \{\pm 1\}^d} \int_{\mathbb{R}^{2nd}} \left(\prod_{i=1}^n \frac{f_{1,a_1}(x_i)}{f_{0}(x_i)} \prod_{j=1}^n\frac{g_{1,a_1}(y_j)}{g_{0}(y_j)} \right) \left(\prod_{i=1}^n \frac{f_{1,a_2}(x_i)}{f_{0}(x_i)} \prod_{j=1}^n\frac{g_{1,a_2}(y_j)}{g_{0}(y_j)} \right) dP_0 \nonumber \\
    & = 2^{-2B_n} \sum_{a_1, a_2 \in \{\pm 1\}^d} \left(\prod_{i=1}^n \int_{\mathbb{R}^d} \frac{f_{1,a_1}(x_i)}{f_{0}(x_i)} \frac{f_{1,a_2}(x_i)}{f_{0}(x_i)} f_0(x_i) dx_i \right) \left(\prod_{j=1}^n \int_{\mathbb{R}^d} \frac{g_{1,a_1}(y_j)}{g_{0}(y_j)} \frac{g_{1,a_2}(y_j)}{g_{0}(y_j)} g_0(y_j) dy_j \right) \nonumber \\
    & = 2^{-2B_n} \sum_{a_1, a_2 \in \{\pm 1\}^d} \left(\int_{\mathbb{R}^d} \frac{f_{1,a_1}(x_1)}{f_{0}(x_1)} \frac{f_{1,a_2}(x_1)}{f_{0}(x_1)} f_0(x_1) dx_1 \right)^n \left( \int_{\mathbb{R}^d} \frac{g_{1,a_1}(y_1)}{g_{0}(y_1)} \frac{g_{1,a_2}(y_1)}{g_{0}(y_1)} g_0(y_1) dy_1 \right)^n
\end{align}
where we set $dP_0 = \prod_{i = 1}^n f_0(x_i) \prod_{j=1}^n g_0(y_j) dx dy $. We now focus on controlling the integrals in the last display. As for the latter, Equation~\eqref{eq:conditionPhi} implies that $A_0 = A_1$, hence $\gamma_{p_1} = \gamma_{p_0}$ and $\gamma_{q_1} = \gamma_{q_0}$. This gives $\frac{g_{1,a}(y_1)}{g_{0}(y_1)} = \mathbbm{1}\{x \in [0,1]^d\}$ for all $a \in \{\pm1\}^d$, and further shows that the second integral in~\eqref{eq:LB_fromTVtoChisq} is equal to one. As for the other one, similar calculations show that 
\begin{align}\label{eq:secondIntChisq}
    \int_{\mathbb{R}^d} & \frac{f_{1,a_1}(x)}{f_{0}(x)} \frac{f_{1,a_2}(x)}{f_{0}(x)} f_0(x) dx = \int_{\mathbb{R}^d} \left(1 + \delta_n \sum_{k=1}^{B_n} a_{1,k} \frac{\phi_k(x)}{p_0(x)} \right) \left(1 + \delta_n \sum_{k=1}^{B_n} a_{2,k} \frac{\phi_k(x)}{p_0(x)} \right)   f_0(x) dx \nonumber \\
    & = 1 + \delta_n^2 \sum_{k=1}^{B_n} a_{1,k} a_{2,k} \int_{\mathbb{R}^d}\frac{\phi_k^2(x)}{p_0^2(x)} f_0(x) dx + \delta_n  \sum_{k=1}^{B_n} (a_{1,k}+a_{2,k})\int_{\mathbb{R}^d}\frac{\phi_k(x)}{p_0(x)} f_0(x) dx\nonumber \\
    & = 1 + \delta_n^2 \sum_{k=1}^{B_n} a_{1,k} a_{2,k} \int_{\mathbb{R}^d}\frac{\phi_k^2(x)}{p_0^2(x)} f_0(x) dx + \delta_n  \sum_{k=1}^{B_n} (a_{1,k}+a_{2,k})\int_{\mathbb{R}^d}\frac{\phi_k(x)}{p_0(x)} \frac{p_0(x)}{r(x)} (\gamma_{p_0} +r \, \gamma_{q_0}) dx \nonumber \\
    & = 1 + \delta_n^2 \sum_{k=1}^{B_n} a_{1,k} a_{2,k} \int_{\mathbb{R}^d}\frac{\phi_k^2(x)}{p_0^2(x)} f_0(x) dx + \delta_n  \sum_{k=1}^{B_n} (a_{1,k}+a_{2,k})\left\{\gamma_{p_0} \int_{\mathbb{R}^d}\frac{\phi_k(x)}{r(x)} dx + \gamma_{q_0} \int_{\mathbb{R}^d} \phi_k(x) dx \right\} \nonumber \\
    & = 1 + \delta_n^2 \sum_{k=1}^{B_n} a_{1,k} a_{2,k} \int_{\mathbb{R}^d}\frac{\phi_k^2(x)}{p_0^2(x)} f_0(x) dx,
\end{align}
where the second equality follows from the fact that $\phi_{k_1}$ and $\phi_{k_2}$ have disjoint support when $k_1 \neq k_2$, and the final equality uses~\eqref{eq:conditionPhi}. Combining~\eqref{eq:LB_fromTVtoChisq} with~\eqref{eq:secondIntChisq} then shows 
\begin{align*}
    4&\operatorname{TV}^2(P_{f_0}^{\otimes n} \otimes P_{g_0}^{\otimes n}, \mathbb{E}_a\{P_{f_1}^{\otimes n} \otimes P_{g_1}^{\otimes n}\}) \leq 2^{-2B_n} \sum_{a_1, a_2 \in \{\pm 1\}^d} \left(1 + \delta_n^2 \sum_{k=1}^{B_n} a_{1,k} a_{2,k} \int_{\mathbb{R}^d}\frac{\phi_k^2(x)}{p_0^2(x)} f_0(x) dx \right)^n \\
    & = \mathbb{E}_{a_1, a_2} \left[\left(1 + \delta_n^2 \sum_{k=1}^{B_n} a_{1,k} a_{2,k} \int_{\mathbb{R}^d}\frac{\phi_k^2(x)}{p_0^2(x)} f_0(x) dx \right)^n \right] = \mathbb{E}_{a} \left[\left(1 + \delta_n^2 \sum_{k=1}^{B_n} a_k \int_{\mathbb{R}^d}\frac{\phi_k^2(x)}{p_0^2(x)} f_0(x) dx \right)^n \right] \\
    & \leq \mathbb{E}_{a} \left[\exp\left\{n \, \delta_n^2 \sum_{k=1}^{B_n} a_k \int_{\mathbb{R}^d}\frac{\phi_k^2(x)}{p_0^2(x)} f_0(x) dx \right\} \right] = \prod_{k=1}^{B_n} \cosh \left(n \, \delta_n^2 \int_{\mathbb{R}^d}\frac{\phi_k^2(x)}{p_0^2(x)} f_0(x) dx \right)  \\
    & \leq \exp \left\{\frac{1}{2} B_n n^2 \delta_n^4  \left(\max_{k  =1}^{B_n}\int_{\mathbb{R}^d}\frac{\phi_k^2(x)}{p_0^2(x)} f_0(x) dx\right)^2 \right\} = \exp \left\{\frac{1}{2} B_n n^2 \delta_n^4  \left(\max_{k  =1}^{B_n}\int_{\mathbb{R}^d}\frac{\gamma_{p_0} + r(x) \gamma_{q_0}}{r(x)} \phi_k^2(x) dx\right)^2 \right\} \\
    & \leq \exp \left\{\frac{C+c^2}{c^2} B_n n^2 \delta_n^4 \right\},
\end{align*}
where in the last step we used the fact that $c \leq r(\cdot) \leq C$ together with $\gamma_{q_0}  = (\sqrt{A_0} + 1)^{-1} \leq 1$ and $\gamma_{p_0}  = A_0^{-1/2} \gamma_{q_0} \leq A_0^{-1/2} = \{\int_{\mathbb{R}^d} p_0(x)/r(x)\}^{-1/2} \leq \sqrt{C}$. Now, being $B_n n^2 \delta_n^4$ of the order $n^2 \rho^{\frac{4s+d}{s}}$, the previous shows that there exists a constant $c_1 = c_1(c, C, \theta, \alpha, \beta)$ for which the previous display is upper bounded by $1-\alpha - \beta$ whenever $\rho \leq c_1 n^{-2 s /(4 s+d)}$. This concludes the proof.
\end{proof}

\begin{lemma}\label{lem:smooth_bumps_measurable_r}
Let $r:\mathbb{R}^{d}\to \mathbb{R}_+$ be such that $0<c\le r(x)\le C$ for all $x \in \mathbb{R}^d$. Fix an integer $B_{n}\ge 1$ such that $B_n^{1/d} \in \mathbb{N}_+$ and write
\(b_{n}:=B_{n}^{-1/d}\).
Partition $[0,1]^d$ into the $B_n$ disjoint cubes  
\[
Q_k=\prod_{j=1}^{d}[i_j b_{n},(i_j+1)b_{n}),\qquad \text{ with }
k=(i_1,\dots,i_d)\in\{0,\dots,b_{n}^{-1}-1\}^{d},
\]
and denote the lower–left corner of $Q_k$ by $x_k^{0}$. There exist functions $\{\phi_{0,1}, \ldots, \phi_{0, B_n}\}$ supported on $[0,1]^d$ satisfying the following conditions:
\begin{itemize}
    \item[(i)] For all $k \in [B_n]$
\[
\phi_k(x)\;=\;
\frac{B_{n}^{1/2}}{\|\phi_{0,k}\|_{2}}\,
\phi_{0,k}\!\bigl(B_{n}^{1/d}\{x-x_k^{0}\}\bigr)
\]
is $C^\infty(\mathbb{R}^{d})$ and satisfies
\[
\int_{\mathbb{R}^{d}}\phi_k(x)\,dx
    \;=\;
\int_{\mathbb{R}^{d}}\frac{\phi_k(x)}{r(x)}\,dx
    \;=\;0,
\qquad
\|\phi_k\bigr\|_{2}=1,
\qquad
\operatorname{supp}\phi_k\subseteq Q_k;
\]
\item[(ii)] There exist a constant $\tilde{M} > 0$ such that 
\[
\max_{k  =1}^{B_n} \left\{\frac{\| \phi_{0,k}\|_{\mathcal{S}^s_d }}{\| \phi_{0,k}\|_2} \vee \frac{\| \phi_{0,k}\|_\infty}{\| \phi_{0,k}\|_2} \right\} \leq \tilde{M}.
\]
\end{itemize}
\end{lemma}
\begin{proof}
We will start by proving part $(i)$. Consider the standard $C^{\infty}(\mathbb{R})$ bump  
\[
\eta(t):=
\begin{cases}
\exp \left(-\frac{1}{1-t^{2}} \right), & |t|<1,\\
0, & |t|\ge 1,
\end{cases}
\]
and define $f_{+}(y_1)=\eta(6y_1-1)$, $f_{-}(y_1)=\eta(6y_1-3)$ and $f_{0}(y_1)=\eta(6y_1-5)$. It is immediate to see that $f_{+}$, $f_-$ and $f_0$ are supported on $[0,\tfrac13]$, $[\tfrac13,\tfrac23]$
and $[\tfrac23,1]$, respectively. Based on this, define three bumps in $\mathbb{R}^d$ by 
\[
\phi_{+}(y)=f_{+}(y_1) \prod_{l=2}^d \eta(y_l),\quad
\phi_{-}(y)=f_{-}(y_1) \prod_{l=2}^d \eta(y_l),\quad
\phi_{0}(y)=f_{0}(y_1)\prod_{l=2}^d \eta(y_l)
\]
for all $y = (y_1, \ldots, y_d) \in \mathbb{R}^d$, and observe that all three functions are in $C^\infty(\mathbb{R}^d)$ and supported in $[0,\tfrac13]\times[0,1]^{d-1}$, $[\tfrac13,\tfrac23]\times[0,1]^{d-1}$
and $[\tfrac23,1]\times[0,1]^{d-1}$, respectively. Fix a cube $Q_k$ and write points in it as
\(x=x_k^{0}+b_{n}y\) with $y\in[0,1]^d$.
Because $r^{-1}$ is bounded on $Q_k$,
each of the numbers  
\[
a_k=\int_{\mathbb{R}^{d}}\frac{\phi_{+}(y)}{r(x_k^{0}+b_{n}y)}\,dy,\;
b_k=\int_{\mathbb{R}^{d}}\frac{\phi_{-}(y)}{r(x_k^{0}+b_{n}y)}\,dy,\;
c_k=\int_{\mathbb{R}^{d}}\frac{\phi_{0}(y)}{r(x_k^{0}+b_{n}y)}\,dy
\]
is finite and strictly positive. If $a_k \neq b_k$, set
\[
u_k = \frac{b_k - c_k}{b_k - a_k},
v_k = \frac{a_k - c_k}{b_k - a_k}, w_k = 1\qquad \text{ and } \qquad
\phi_{0,k}(y) = u_k\,\phi_{+}(y) - v_k\,\phi_{-}(y) - w_k \phi_{0}(y).
\]
If instead $a_k = b_k$ but $a_k \neq c_k$, we switch the roles of $\phi_{-}$ and $\phi_{0}$ and apply the same formula. Finally, if $a_k = b_k = c_k$, we simply set $u_k = w_k = 1$ and $v_k = 0$, so that $\phi_{0,k}(y) = \phi_{+}(y) - \phi_{0}(y)$. A straightforward calculation shows that $\phi_{0,k}$ has zero average both with respect to Lebesgue measure and with respect to the weight $r^{-1}$ evaluated at $x_k^{0}+b_n y$:
\[
\int_{\mathbb{R}^d}\phi_{0,k}(y)\,dy
      =\left(\int_{\mathbb{R}^d}\phi_{+}(y)\,dy\right)(u_k-v_k-w_k)=0,
\qquad
\int_{\mathbb{R}^d}\frac{\phi_{0,k}(y)}
                      {r(x_k^{0}+b_n y)}\,dy
      =u_k a_k - v_k b_k - w_k c_k = 0 .
\]
We now verify that the function 
\begin{align}\label{eq:phi_k}
    \phi_k(x):=\frac{B_n^{1/2}}{\|\phi_{0,k}\|_{2}}\,
           \phi_{0,k}\!\bigl(B_n^{1/d}\{x-x_k^{0}\}\bigr)
\end{align}
satisfies the desired properties. Because $dx=b_n^{d}dy=B_n^{-1}dy$, the two zero–average identities above
translate to the $x$-scale, yielding  
\[
\int_{\mathbb{R}^d}\phi_k(x)\,dx = 
\int_{\mathbb{R}^d} \frac{\phi_k(x)}{r(x)}\,dx = 0.
\]
Moreover $\|\phi_k\|_{2}^{2} =B_{n}\,\|\phi_{0,k}\|_{2}^{-2}\,B_{n}^{-1}\,\|\phi_{0,k}\|_{2}^{2} =1$. Finally, the support of $\phi_k$ is contained in $Q_k$, and different cubes do not intersect, hence for $k_1 \neq k_2$ we also have $\langle\phi_{k_1},\phi_{k_2}\rangle_{2}=0$. This completes the proof of the of first part of the statement.

As for part $(ii)$, set $J_\ell := \|\phi_+\|_\ell = \|\phi_-\|_\ell = \|\phi_0\|_\ell$ for $\ell \in \{\infty, 2, \mathcal{S}^s_d\}$, and observe that these constants depend only on the shape of the function $\eta$. The disjoint structure of the supports of $\phi_+, \phi_-, \phi_0$ gives
\[
\begin{cases}
    \|\phi_{0,k}\|_\infty = J_\infty \max\{|u_k|, |v_k|, |w_k|\}\\
    \|\phi_{0,k}\|_2 = J_2 \sqrt{u_k^2 + v_k^2 + w_k^2}\\
    \|\phi_{0,k}\|_{\mathcal{S}^s_d} = J_{\mathcal{S}^s_d}\sqrt{u_k^2 + v_k^2 + w_k^2}
\end{cases}
\]
for all $k \in [B_k]$, and implies that 
\[
\frac{\| \phi_{0,k}\|_{\mathcal{S}^s_d}}{\| \phi_{0,k}\|_2} = \frac{J_{\mathcal{S}^s_d}}{J_2} \qquad \text{and} \qquad\frac{\| \phi_{0,k}\|_{\infty}}{\| \phi_{0,k}\|_2} = \frac{J_\infty \max\{|u_k|, |v_k|, |w_k|\}}{J_2 \sqrt{u_k^2 + v_k^2 + w_k^2}} \leq \frac{J_\infty (|u_k|+|v_k|+|w_k|)}{J_2 \sqrt{u_k^2 + v_k^2 + w_k^2}} \leq \frac{\sqrt{3} J_\infty}{J_2}.
\]
Note the denominator is always well-defined for our choices of $(u_k, v_k, w_k)$ as $u_k^2 +v_k^2 + w_k^2 \geq 1$. This shows that the claim in part $(ii)$ holds with 
\begin{align}\label{eq:maxRatioNorm}
    \tilde{M}:= \frac{J_{\mathcal{S}^s_d}}{J_2} \vee \frac{\sqrt{3} J_\infty}{J_2}
\end{align}
and concludes the proof of the lemma.
\end{proof}

\subsection{Proofs for Section \ref{sec:family_shift}}\label{appendix:proof4}
\begin{proof}[Proof of Proposition \ref{prop:robustness}]
    The reason we cannot directly apply Theorem \ref{prop:alg_validity} is that, although $g \propto r_\star f$ for a certain $r \in {\cal R}$, Algorithm~\ref{alg:star_algo} now generates $Z^{(1)}, \ldots, Z^{(H)}$ using an approximation $\hat{r}$ to $r_\star$. To address this mismatch, let $\tilde{X} = (\tilde{X}_1, \ldots, \tilde{X}_n)$, $\tilde{Y} = (\tilde{Y}_1, \ldots, \tilde{Y}_m)$ be such that $\tilde{X} \perp \!\!\! \perp \tilde{Y}$, $\tilde{X}_i \overset{\mathrm{i.i.d.}}{\sim} f$ and $\tilde{Y}_i \overset{\mathrm{i.i.d.}}{\sim} \check{r} \, f$. Define $\tilde{Z} = (\tilde{X}, \tilde{Y})$ and let $\tilde{Z}^{(1)}, \ldots, \tilde{Z}^{(H)}$ be draws of Algorithm~\ref{alg:star_algo} based on $\hat{r}$ when we sample from the values of $\tilde{Z}$ instead of $Z$. That is, for every $h \in [H]$ independently we have 
    \[
    \tilde{Z}^{(h)} = \tilde{Z}_{\tilde{\sigma}^{(h)}} \quad \text { where } \quad \mathbb{P}\left\{\tilde{\sigma}^{(h)} = \sigma \mid \tilde{Z} \right\} \propto \prod_{i \in \{n+1, \ldots, n+m \}}  \hat{r}(\tilde{Z}_{\sigma(i)}),
    \]
    where $\tilde{Z}_{\sigma}$ is defined analogously to $Z_{\sigma}$. Next, by comparing to the sampling mechanism outlined in Algorithm~\ref{alg:star_algo}, we observe that the $\tilde{Z}^{(h)}$'s, conditional on $\tilde{Z}$, are generated with the same mechanism as the $Z^{(h)}$'s conditional on $Z$. That is, for every $z \in \mathcal{X}^{n+m}$ we have 
    \[
    \left((\tilde{Z}^{(1)}, \ldots, \tilde{Z}^{(H)}) \mid \tilde{Z} = z \right) \overset{d}{=} \left((Z^{(1)}, \ldots, Z^{(H)}) \mid Z = z \right).
    \]
    We can now use the fact that, if $(V \mid U = u) \overset{d}{=} (V^\prime \mid U^\prime = u)$ for all $u$, then $\operatorname{TV}((U,V), (U^\prime, V^\prime)) = \operatorname{TV}(U, U^\prime)$. It follows that 
    \begin{align}\label{eq:TV_simplified}
            \operatorname{TV}\left((\tilde{Z}, \tilde{Z}^{(1)}, \ldots, \tilde{Z}^{(H)}) \right. & \left., (Z, Z^{(1)}, \ldots, Z^{(H)}) \right) = \operatorname{TV}(\tilde{Z}, Z) = \operatorname{TV}\left((\tilde{X}, \tilde{Y}), (X,Y) \right) \nonumber \\
            & = \operatorname{TV}(\tilde{Y}, Y) = \operatorname{TV}\left(\{\check{r}\cdot f\}^{\otimes{m}}, \{\bar{r}_{\star} \cdot f\}^{\otimes{m}} \right).
    \end{align}
 Hence, if we define 
    \[
    A_{\alpha} := \left\{(z, z^{(1)}, \ldots, z^{(H)}) : \frac{1+\sum_{h = 1}^H \mathbbm{1}\{T(z^{(h)}) \geq T(z) \}}{1+H} \leq \alpha \right\},
    \]
we can bound the type~I error as 
    \begin{align*}
        \mathbb{P}\{p_{\hat{r}} \leq \alpha\} &= \mathbb{P}\{(Z, Z^{(1)}, \ldots, Z^{(H)}) \in A_\alpha\} \\
        & \leq \mathbb{P}\{(\tilde{Z}, \tilde{Z}^{(1)}, \ldots, \tilde{Z}^{(H)}) \in A_\alpha\} +  \operatorname{TV}\left((\tilde{Z}, \tilde{Z}^{(1)}, \ldots, \tilde{Z}^{(H)}), (Z, Z^{(1)}, \ldots, Z^{(H)}) \right) \\
        &  = \mathbb{P}\{(\tilde{Z}, \tilde{Z}^{(1)}, \ldots, \tilde{Z}^{(H)}) \in A_\alpha\} + \operatorname{TV}\left(\{\check{r}\cdot f\}^{\otimes{m}}, \{\bar{r}_\star \cdot f\}^{\otimes{m}} \right) \leq \alpha + \operatorname{TV}\left(\{\check{r} \cdot f\}^{\otimes{m}}, \{\bar{r}_{\star} \cdot f\}^{\otimes{m}} \right),
    \end{align*}
    where the first inequality follows from the definition of total variation distance, the second equality from~\eqref{eq:TV_simplified}, and the final inequality from the exchangeability of $(\tilde{Z}, \tilde{Z}^{(1)}, \ldots, \tilde{Z}^{(H)})$. This arises from the fact that the~$\tilde{Y}_i$'s are independent and identically distributed~from a distribution proportional to $\hat{r} \, f$, and Algorithm~\ref{alg:star_algo} copies $\tilde{Z}^{(1)}, \ldots, \tilde{Z}^{(H)}$ are generated using the same approximation $\hat{r}$.
\end{proof}

\begin{proof}[Proof of Proposition~\ref{prop:alternativeRobustness}]
    Arguing as in the proof of Proposition~\ref{prop:robustness}, if we define 
    \[
    A_{\alpha} := \left\{(z, z^{(1)}, \ldots, z^{(H)}) : \frac{1+\sum_{h = 1}^H \mathbbm{1}\{T(z^{(h)}) \geq T(z) \}}{1+H} \leq \alpha \right\},
    \]
    we can bound
    \begin{align*}
        \mathbb{P}\{p_{\hat r} > \alpha \} & = \mathbb{P}\{(Z, Z_{\hat r}^{(1)}, \ldots, Z_{\hat r}^{(H)} ) \in A_\alpha^\complement \} \\
        & \leq \mathbb{P}\{(Z, Z_{r_\star}^{(1)}, \ldots, Z_{r_\star}^{(H)} ) \in A_\alpha^\complement \} + \operatorname{TV}\{(Z, Z_{\hat r}^{(1)}, \ldots, Z_{\hat r}^{(H)} ), (Z, Z_{r_\star}^{(1)}, \ldots, Z_{r_\star}^{(H)} ) \} \\
        & = \mathbb{P}\{p_{r_\star} > \alpha \} + \operatorname{TV}\{(Z, Z_{\hat r}^{(1)}, \ldots, Z_{\hat r}^{(H)} ), (Z, Z_{r_\star}^{(1)}, \ldots, Z_{r_\star}^{(H)} ) \}, 
    \end{align*}
    which shows that it is sufficient to control the total variation in the last display to conclude the proof. In this regard, let $Q_r(\cdot \mid z)$ be the law of $Z_r^{(1)}$ conditionally on $Z = z$, and $\pi_r(\cdot \mid z)$ be the law of~$\sigma^{(1)}$ conditionally on $Z = z$ in the case where~\eqref{eq:sampling_1} is based on $r$, meaning that for all $C \subseteq \mathcal S_{n+m}$ we have
    \[
    \pi_r(C \mid z) = \frac{\sum_{\sigma \in C} \prod_{i \notin [n]} r(z_{\sigma(i)})}{\sum_{\sigma \in \mathcal{S}_{n+m}} \prod_{i \notin [n]} r(z_{\sigma(i)})}.
    \]
 Based on this, we can write
     \begin{align*}
       \operatorname{TV}\{(Z, Z_{\hat r}^{(1)}, & \ldots, Z_{\hat r}^{(H)} ), (Z, Z_{r_\star}^{(1)}, \ldots, Z_{r_\star}^{(H)} ) \} \leq \mathbb{E}\left[\operatorname{TV}\{(Z_{\hat r}^{(1)}, \ldots, Z_{\hat r}^{(H)} \mid Z), ( Z_{r_\star}^{(1)}, \ldots, Z_{r_\star}^{(H)} \mid Z) \}\right] \\
       & = \mathbb{E}\left[\operatorname{TV}\{Q_{\hat{r}}(\cdot \mid Z)^{\otimes H}, Q_{r_\star}(\cdot \mid Z)^{\otimes H} \}\right]  \leq H \, \mathbb{E}\left[\operatorname{TV}\{ Q_{\hat{r}}(\cdot \mid Z), Q_{r_\star}(\cdot \mid Z) \}\right],
    \end{align*}
    where in the penultimate step we used the fact that $(Z_{r}^{(1)}, \ldots, Z_{r}^{(H)})$ are i.i.d.~conditionally on $Z$. Now, let  \[
\operatorname{KL}(P,Q) := \begin{cases}
\int \log\bigl(\frac{dP}{dQ}\bigr)\, dP & \text{if } P \ll Q, \\[1mm]
+\infty & \text{otherwise},
\end{cases}
\] 
denote the Kullback-Leibler divergence between two distributions defined on 
the same measurable space, and for all measurable $B \subseteq \mathcal X^{n+m}$ define $C_B := \{\sigma \in {\cal S}_{n+m} : Z_\sigma \in B\}$. By the sampling rule~\eqref{eq:sampling_1}, for all measurable $B \subseteq \mathcal X^{n+m}$ we have 
    \begin{align*}
        Q_r(B \mid Z) &= \mathbb{P}\{Z_r^{(1)} \in B \mid Z\} = \sum_{\sigma : z_\sigma \in B} \mathbb{P}\{\sigma^{(1)} = \sigma \mid Z\,; r\}  = \frac{\sum_{\sigma : z_\sigma \in B} \prod_{i \notin [n]} r(Z_{\sigma(i)})}{\sum_{\sigma \in \mathcal{S}_{n+m}} \prod_{i \notin [n]} r(Z_{\sigma(i)})} = \pi_r(C_B \mid Z),
    \end{align*}
hence
\begin{align*}
    \operatorname{TV}\{ & Q_{\hat{r}}(\cdot \mid Z), Q_{r_\star}(\cdot \mid Z) \} = \sup_{B \subseteq {\mathcal X}^{n+m}} \left|Q_{\hat r}(B \mid Z) - Q_{r_\star}(B \mid Z) \right| \\
    & = \sup_{B \subseteq \mathcal X^{n+m}} \left|\pi_{\hat r}(C_B \mid Z) -\pi_{r_\star}(C_B \mid Z) \right| \leq \sup_{C \subseteq \mathcal S_{n+m}} \left|\pi_{\hat r}(C \mid Z) -\pi_{r_\star}(C \mid Z )\right| = \operatorname{TV}\{\pi_{\hat r}(\cdot \mid Z), \pi_{r_\star}(\cdot \mid Z)\} \\
    & \leq \sqrt{ 1 - \exp \left\{-\operatorname{KL}(\pi_{\hat r}(\cdot \mid Z), \pi_{r_\star}(\cdot \mid Z))\right\}} \leq \sqrt{ 1 - \exp \left\{-2m\eta \right\}}.
\end{align*}
In particular, the first inequality holds since $\{C_B : \text{measurable } B \subseteq \mathcal X^{n+m}\} \subseteq \{C : C \subseteq \mathcal S_{n+m}\}$, the second one follows from the Bretagnolle--Huber inequality \citep[][Lemma 2.1]{BretagnolleHuber1978}, while the third one uses \[
\frac{\prod_{i \notin [n]} \hat r(Z_{\sigma(i)})}{\prod_{i \notin [n]} r_\star(Z_{\sigma(i)})} \quad , \quad \frac{\sum_{\sigma \in \mathcal{S}_{n+m}}\prod_{i \notin [n]} r_\star (Z_{\sigma(i)})}{\sum_{\sigma \in \mathcal{S}_{n+m}} \prod_{i \notin [n]} \hat r(Z_{\sigma(i)})} \quad \in [e^{-m \eta}, e^{m \eta}],
\]
which follows from the definition of $\eta := \max_{i \in [n+m]} |\log \hat r(Z_i) -\log r_\star (Z_i)|$.
This concludes the proof.
\end{proof}

\section{Additional theoretical results}\label{app:additionalResults}

We first derive an explicit representation of the population discrepancy measure
\(T_{\mathcal{F},r}(f,g)\), introduced in Section~\ref{sec:integral probability metric},
in the regime of highly unbalanced sample sizes. Since $\lambda_0$ is a continuous function of $n/m$, the fact that $\lim_{t \to 0^+} \lambda_0 = \int g/r \, d\mu$ and $ \lim_{t \to \infty}\lambda_0^{-1} = \int rf \, d\mu$ also ensures that $\lambda_0$ is bounded away from zero and infinity in every regime of $n$ and $m$, provided that $\int g/r \, d\mu < \infty$ and $\int rf \, d\mu < \infty$. This applies also to $T_{\mathcal{F},r}(f,g)$ when $f,g$ are fixed alternatives.

\begin{prop}\label{prop:limitingTF}
    Recall $T_{\mathcal{F},r}(f,g) := \sup_{\varphi \in \mathcal{F}} \left| \int \frac{n/m + 1}{n/m + \lambda_0 r} \, (\lambda_0 r f - g) \, \varphi \, d\mu \right|$, and assume $\int g/r \, d\mu < \infty$, $\int rf \, d\mu < \infty$ and  $\|\varphi\|_\infty \leq \gamma$ for all $\varphi \in \mathcal{F}$. Then, we have 
    \[
    \lambda_0 \to \int \frac{g}{r} \, d\mu, \quad  T_{\mathcal{F},r}(f,g) \to \sup_{\varphi \in \mathcal{F}} \left|\int \left\{f - \left(\frac{g}{r}\right)\left(\int \frac{g}{r} d\mu\right)^{-1} \right\} \, \varphi \, d\mu\right|
    \]
    when $n/m \to 0^+$. Similarly,  
    \[
    \lambda_0^{-1} \to \int r f d\mu, \quad T_{\mathcal{F},r}(f,g) \to \sup_{\varphi \in \mathcal{F}}\left|\int \left\{(rf)\left(\int rf d\mu\right)^{-1} - g \right\} \, \varphi \, d\mu\right|
    \]
    when $n/m \to \infty$.
\end{prop}
\begin{proof}
We focus on the case \(t := n/m \to 0^+\). For each fixed \(t>0\), define
\[
F_t :
     \begin{cases}
        \mathbb{R}_+ \to [-1, \infty) \\
        \lambda \mapsto  \int \frac{tf + g}{t + \lambda r} d\mu - 1,
     \end{cases} \quad \text{ and } \quad F_0 :
     \begin{cases}
        \mathbb{R}_+ \to [-1, \infty) \\
        \lambda \mapsto  \int \frac{g}{ \lambda r} d\mu - 1. 
     \end{cases}
\]
Arguing as in Lemma~\ref{lemma:unique_h}, there exists a unique \(\lambda_0(t)>0\) such that \(F_t(\lambda_0(t)) = 0\). Similarly, $F_0$ is continuous, strictly decreasing, with
$\lim_{\lambda \to 0^+}F_0(\lambda)=\infty$, $\lim_{\lambda\to\infty}F_0(\lambda)= - 1$, so there is
a unique $\lambda_\star\in(0,\infty)$ such that $F_0(\lambda_\star)=0$, namely $\lambda_\star= \int g/r \, d\mu$. We will now show that $\lambda_0(t) \to \lambda_\star$ as $t \to 0^+$. In this regard, for $0 < a< \lambda_\star < b < \infty$ and for fixed $x \in \mathcal X$, we have 
\begin{align*}
 \sup_{\lambda\in[a,b]} \left|
\frac{t f(x) + g(x)}{t + \lambda r(x)} - \frac{g(x)}{\lambda r(x)}
\right| & = \sup_{\lambda\in[a,b]} \left| \frac{t(\lambda r(x) f(x) - g(x))}{(t+\lambda r(x))\lambda r(x)} \right|
\\
& \leq \min\left(1, \frac{t}{a r(x)} \right) \sup_{\lambda\in[a,b]}\frac{\lambda r(x) f(x) + g(x)}{\lambda r(x)} \leq  f(x) + \frac{g(x)}{a r(x)}. 
\end{align*}
This both shows that the left-hand side converges to zero as $t \to 0^+$, and that 
\begin{align*}
    \sup_{\lambda\in[a,b]} |F_t(\lambda) - F_0(\lambda)| \leq \int \sup_{\lambda\in[a,b]} \left|
\frac{t f(x) + g(x)}{t + \lambda r(x)} - \frac{g(x)}{\lambda r(x)}
\right| d\mu \leq 1 + \frac{1}{a} \int \frac{g}{r} d\mu.
\end{align*}
Taken together, the previous claims show that $F_t \to F_0$ uniformly on $[a,b]$ by the dominated convergence theorem. This is enough to show that 
$\lambda_0(t) \to \lambda_\star$ as $t \to 0^+$. Indeed, for an arbitrary $\epsilon > 0$, define $\alpha_\epsilon := F_0(\lambda_\star - \epsilon)$ and $\beta_\epsilon := F_0(\lambda_\star + \epsilon)$ with $[\lambda_\star - \epsilon, \lambda_\star + \epsilon] \subset (a,b)$. Further set $\eta_\epsilon = \frac{1}{2}\min(\alpha_\epsilon, \beta_\epsilon)$, so that $\alpha_\epsilon \geq 2\eta_\epsilon$ and $\beta_\epsilon \leq -2\eta_\epsilon$. By the uniform convergence of $F_t$ to $F_0$ on $[a,b]$, there exists $t_\epsilon >0$ such that for all $0 < t < t_\epsilon$ we have $\sup_{\lambda\in[a,b]} |F_t(\lambda) - F_0(\lambda)| \leq \eta_\epsilon$. In particular, we have
\[
F_t(\lambda_\star - \epsilon) \geq F_0(\lambda_\star - \epsilon) - |F_t(\lambda_\star - \epsilon) - F_0(\lambda_\star - \epsilon)| \geq 2\eta_\epsilon - \eta_\epsilon = \eta_\epsilon > 0, 
\]
and, similarly, $F_t(\lambda_\star + \epsilon) \leq -\eta_\epsilon < 0$.  Combined with the fact that $F_t$ is continuous and monotonically decreasing, 
this implies by the intermediate value theorem  that $\lambda_0(t) \in (\lambda_\star - \epsilon, \lambda_\star+\epsilon)$, which yields $\lambda_0(t) \to \lambda_\star$ as $t \to 0^+$ by the arbitrariness of $\epsilon$.

As for the second part of the claim, for $t>0$ define
\[
k_t(x)
:= \frac{t+1}{t+\lambda_0(t) r(x)}\big(\lambda_0(t) r(x) f(x) - g(x)\big),
\qquad \text{ and } \qquad
k(x) := f(x) - \frac{g(x)}{\lambda_\star r(x)}.
\]
In light of the previous result, $k_t(x) \to k(x)$ for all $x \in \mathcal X$ as $t \to 0^+$. Furthermore, we proved there exist $\epsilon>0$ and $t_\epsilon>0$ such that
$\lambda_0(t)\in(\lambda_\star-\varepsilon,\lambda_\star+\varepsilon)$ for all $0<t<t_\epsilon$; set
$\lambda_{\min} := \lambda_\star-\varepsilon>0$ and $\lambda_{\max} := \lambda_\star+\varepsilon<\infty$. For fixed $x \in \mathcal X$ and $0<t<t_\epsilon$, we thus have
\[
\left|\frac{t+1}{t+\lambda_0(t) r(x)}\right| \le \frac{t_\epsilon + 1}{\lambda_{\min} r(x)},
\quad \text{ and } \qquad
|\lambda_0(t) r(x) f(x) - g(x)|
\le \lambda_{\max} r(x) f(x) + g(x),
\]
so
\[
|k_t(x) - k(x)| \le |k_t(x)| +|k(x)| \leq \left\{(t_\epsilon + 1)\frac{\lambda_\mathrm{max}}{\lambda_\mathrm{min}} + 1\right\} f(x) + \frac{t_\epsilon + 2}{\lambda_{\min} r(x)} g(x).
\]
Since the right-hand side is integrable, the dominated convergence implies that 
$\int |k_t-k|\,d\mu\to 0$. 
This, together with the uniform boundedness assumption on the elements of $\cal F$, immediately yields
\begin{align*}
    \left|\sup_{\varphi \in \mathcal{F}} \left| \int k_t \varphi d\mu \right| - \sup_{\varphi \in \mathcal{F}} \left| \int k \varphi d\mu \right| \right| & \leq \sup_{\varphi \in \mathcal{F}}  \int | k_t - k| \, |\varphi| d\mu \\
    & \leq \left(\sup_{\varphi \in \mathcal{F}} \|\varphi\|_\infty \right) \int |k_t - k|  d\mu \leq \gamma \, \int |k_t - k|  d\mu,
\end{align*}
and concludes the proof for the regime $n/m \to 0^+$. 

We omit the proof for the regime $t \to \infty$, as it follows by analogous steps to the previous case.
\end{proof}

We also formally justify the claim made in Section~\ref{sec:optimality} that, on compact domains, the minimax testing rate with separation measured by $\| \psi_r\|_2$ is equivalent, up to multiplicative constants, to the minimax testing rate when the separation is measured with $\| g - rf \, (\int rf d\mu)^{-1}\|_2$. By Proposition~\ref{prop:invPsi} below, the same conclusion applies to $\| f - g/r \, (\int g/r \, d\mu)^{-1}\|_2$.

\begin{prop}\label{eq:prop:equivalenceSeparations}
    Assume that $m \leq n \leq \tau m$ for a fixed $\tau \geq 1$, $\operatorname{supp}(f) = \operatorname{supp}(g) = \operatorname{supp}(r) = [0, 1]^d$, $\|f\|_\infty \vee \|g\|_\infty \leq M$, and $c \leq r(x) \leq C$ for all $x\in [0,1]^d$. There exist positive constants $C_0 \equiv C_0(\tau, c, C, M)$ and $c_0 \equiv c_0(\tau, c, C, M)$ such that $c_0 \, \|\psi_r\|_2 \leq \|g - \bar{r} f \|_2 \leq C_0 \, \|\psi_r\|_2$, where $\bar r  = r/\int_{[0,1]^d} r f d\mu$.
\end{prop}

\begin{proof}
We can assume without loss of generality that $\int_{[0,1]^d} r f d\mu  = 1$. Otherwise, it is sufficient to replace $r$ by $r / \int_{[0,1]^d} r f d\mu  = 1$; this normalization only rescales the relevant quantities by multiplicative constants depending on $c$ and $C$ via bounds on $\int_{[0,1]^d} r f\,d\mu$, which we absorb into $c_0$ and $C_0$. Recall that $h = \frac{nf + mg}{n + \lambda_0 m  r}$, and $\lambda_0 = (\int r h d\mu)^{-1}$, which was shown in the proof of Lemma \ref{lemma:unique_h}. Define also $\tilde \psi_r := f- h = \frac{\lambda_0 m r f - mg}{n + \lambda_0 m r}$, so that $(n+m) \,\tilde \psi_r = m \, \psi_r$. We have 
\begin{align*}
    1 - \lambda_0^{-1} & = \int r f d \mu - \int r h d\mu  = \int r \tilde \psi_r d\mu,
\end{align*}
by which $\lambda_0 - 1  = \int r \tilde \psi_r d\mu/\int r h d\mu$. Based on this, we have 
\begin{align*}
    g - r f & = g - \lambda_0 r f + (\lambda_0 - 1) r f = -\frac{n + \lambda_0 m r}{m} \tilde \psi_r + \frac{\int r \tilde \psi_r d\mu}{\int r h d\mu} r f, 
\end{align*}
implying that 
\begin{align*}
     \|g - rf \|_2 &\leq \|(n/m + \lambda_0 r) \, \tilde \psi_r\|_2 + \left|\frac{\int r \tilde \psi_r d\mu}{\int r h d\mu} \right| \, \|r f \|_2 \\
     & \leq \left\{\tau + C/c  + \frac{\|r\|_2  \|r f \|_2}{\int r h d\mu} \right\} \, \|\tilde \psi_r\|_2 \leq \{\tau + C/c + C^2 M^{1/2}/c\} \, \|\tilde \psi_r\|_2, 
\end{align*}
where in the last step we used $\int_{[0, 1]^d} r^2 d\mu \leq C^2, \int_{[0, 1]^d} r^2 f^2 d\mu \leq C^2 M$, $\int_{[0, 1]^d} r h  d\mu \geq c$; in particular this follows from the fact that $\lambda_0$ is such that $\int h d \mu = 1$. 

Similarly, we can use the fact that $\lambda_0 - 1  = \int r \tilde \psi_r d\mu/\int r h d\mu$ to write
\[
1 - \lambda_0^{-1} = (1 - \lambda_0^{-1}) \int \frac{\lambda_0 m r}{n+\lambda_0 m r} \, rf \, d\mu + \int \frac{mr \, (rf - g)}{n + \lambda_0 mr} d\mu,
\]
which further implies $1 - \lambda_0^{-1} = \int \frac{mr \, (rf - g )}{n+\lambda_0 m r} \, d\mu / \int \frac{n r f}{n+\lambda_0 m r} d\mu$ upon reordering. This, together with  
\[
\tilde \psi_r = \frac{(\lambda_0 -1) mrf}{n + \lambda_0 mr} + \frac{m(rf-g)}{n + \lambda_0 m r}, 
\]
allows bounding 
\begin{align*}
    \|\tilde \psi_r \|_2 &\leq \left\|\frac{m}{n+\lambda_0 m r} (g - rf) \right\|_2 +  \left\|\frac{(\lambda_0-1)mr}{n+\lambda_0 mr} f\right\|_2 \leq \|g - rf\|_2 + \left|1-\lambda_0^{-1} \right| \|f\|_2 \\
    &\leq \|g - rf\|_2 + \left| \frac{ \int \frac{mr \, (rf - g )}{n+\lambda_0 m r} \, d\mu}{ \int \frac{n r f}{n+\lambda_0 m r} d\mu}\right| \|f\|_2  \leq \{1 + C M^{1/2} (c+C)/ c \} \|g - rf\|_2.
\end{align*}
Notice that in the third inequality we used the fact that $ \int \frac{mr}{n+\lambda_0 m r} d\mu \leq \lambda_0^{-1} \leq C$ and $ \int \frac{n r f}{n+\lambda_0 m r} d\mu \geq  \int \frac{r f}{1+\lambda_0 r} d\mu  \geq  \int \frac{r f}{1 + C/c} d\mu \geq c/(c+C)$.  This yields the reverse inequality and shows that $\|\tilde \psi_r\|_2 \asymp_{\tau, C, c, M} \|g - \bar r f\|_2$. Finally, since $m \leq n \leq \tau m$, we obtain
$\|\tilde \psi_r\|_2  = \|\frac{m}{n+m} \psi_r \|_2\asymp_{\tau} \left\| \psi_r\right\|_2$, which completes the proof.
\end{proof}

A key quantity that appears throughout the paper is the function $\psi_r : = (1+n/m) (f-h)$, which intuitively captures the discrepancy between $g$ and $rf$, as formalized in Proposition~\ref{eq:prop:equivalenceSeparations} above. We now justify the claim that $\psi_r$ is a more natural notion of separation by showing that it is invariant, up to a sign, to taking reciprocal transformations of $r$ and relabelling the samples. 
\begin{prop}\label{prop:invPsi}
Define $\psi(n,f,m,g,r) : = (n+m) \, m^{-1} \, (f-h)$, where $h \equiv h(n,f,m,g,r) = (nf + mg)/(n + \lambda_0 m r)$ for a suitable constant $\lambda_0 > 0$ such that $\int h d\mu = 1$. We have $\psi(n,f,m,g,r) = - \psi(m,g,n,f,r^{-1})$.
\end{prop}
\begin{proof}
    Arguing as in Remark~\ref{rmk:reciprocal}, we can show that the constant $\theta_0 > 0$ that makes $h^\prime \equiv h(m,g,n,f,r^{-1})  = (mg + nf)/(m + \theta_0 n r^{-1})$ a density satisfies $\theta_0  = \lambda_0^{-1}$, where $\lambda_0$ is such that  $\int h(n,f,m,g,r) d\mu = 1$. This implies that
\begin{align*}
        \psi(m,g,n,f,r^{-1}) &= (1+m/n)(g - h^\prime) = (1+m/n) \left(g  - \frac{mg + nf}{m + n\lambda_0^{-1} r^{-1}} \right) = \frac{n+m}{n} \frac{ng \lambda_0^{-1} r^{-1} - nf}{m+n\lambda_0^{-1} r^{-1}} \\
        & = \frac{n+m}{n} \frac{ng - \lambda_0 n r f}{n+\lambda_0 mr} = \frac{n+m}{m} \frac{mg - \lambda_0 m r f}{n+\lambda_0 mr} = - \psi(n,f,m,g,r),
\end{align*}
thus concluding the proof.
\end{proof}

Finally, we further discuss  Proposition~\ref{prop:robustness} in Section~\ref{sec:family_shift} and show, in the conditional two-sample setting, how to implement sample splitting so that $\operatorname{TV}(\{\check r f\}^{\otimes (n-N)}, \{\bar r_\star f\}^{\otimes (n-N)}) \overset{\mathbb P}{\to} 0$. Assume for simplicity that $n = m$ and write $Z = (X_1,\ldots,X_n,Y_1,\ldots,Y_n)$. Split $Z$ into an estimation and a testing part, $Z^{\mathrm{estim}} = (X_1,\ldots,X_N,Y_1,\ldots,Y_N)$ and $Z^{\mathrm{test}} = (X_{N+1},\ldots,X_n,Y_{N+1},\ldots,Y_n)$ with $1 \le N < n$; use $Z^{\mathrm{estim}}$ to fit $\hat r$ as an estimator of $r_\star$, then run Algorithm~\ref{alg:star_algo} based on $\hat r$ on $Z^{\mathrm{test}}$. For Proposition~\ref{prop:robustness} to be practically informative, we must identify regimes in which $N$ can be chosen so that $r_\star$ is estimated with sufficient accuracy and the resulting excess type~I error is small. As an illustrative example, we use Theorem~2 of \citet{Nguyen2010minimax} to control estimation of the marginal density ratio. Alternative density-ratio estimators, such as those analyzed in Theorem~1 and Example~1 of \citet{sugiyiamaKLIEPrkhs} or Theorem~2 of \citet{sugiyamaKuLSIF-statistical}, can yield analogous guarantees under comparable regularity conditions.

\begin{eg}\label{ex:DRE}
Consider the setting of the conditional two-sample testing problem with $ {\cal X \times Y} \subseteq \mathbb{R}^{d_X} \times \mathbb{R}^{d_Y}$, i.e.~we identify the $X$-values with $S:=[d_X] \subset [d_X + d_Y]$ and the $Y$-values with $\{d_X+1,\dots,d_X+d_Y\}$. Here, the density ratio corresponds to $f_X^{(2)}(x)/f_X^{(1)}(x)$, and we can write $r(x,y) = r_S(x)$ for a function $r_S$ acting on $\mathbb{R}^{d_X}$. We will now show that the estimation error is small provided that~$r_S$ is sufficiently smooth.  In particular, consider  $\mathcal{R} = \mathcal{S}^s_{d_X}(L)$, the Sobolev ball introduced in~\eqref{eq:Sobolev_ball}, with integer~$s$ satisfying $s > d_X/2$, and assume that $r_S$ lies in $\cal R$. Let $\hat{r}_S$ represent the first estimator introduced in \cite{Nguyen2010minimax} for $r_S$. Define also $\bar{r}_S := r_S/\int  r_S f \, d\mu_X$ and $\check{r}_S := \hat r_S/\int \hat r_S f \, d\mu_X$. If $0 < c \leq r(\cdot) \leq C$ for all $r \in \mathcal{R}$, Theorem 2 in \cite{Nguyen2010minimax} and the remarks thereafter ensure that $\int(\sqrt{\bar{r}_S f} - \sqrt{\hat{r}_S f})^2 d\mu =\mathcal{O}_\mathbb{P}(N^{-\frac{2s}{2s + d}})$. As a result, writing $\operatorname{H}^2(p,q) = 2(1 - \int \sqrt{p q} d\mu)$ for the squared Hellinger distance between densities $p$ and $q$, the Cauchy--Schwarz inequality gives
    \begin{align*}
        \operatorname{H^2}&\left(\{\check{r}_S \, f\}, \{\bar{r}_{S} \, f\} \right) = \int(\sqrt{\check{r}_S f} - \sqrt{\bar{r}_S f})^2 d\mu_X \leq 4 \int(\sqrt{\hat{r}_S f} - \sqrt{\bar{r}_S f})^2 d\mu_X = \mathcal{O}_\mathbb{P}(N^{-\frac{2s}{2s + d}}). 
    \end{align*}
 This, combined with standard inequalities, yields $\operatorname{TV}\left(\{\check{r}_S \, f\}^{\otimes{(n-N)}}, \{\bar{r}_S \, f\}^{\otimes{(n-N)}} \right) \lesssim \sqrt{(n-N) N^{-\frac{2s}{2s+d}}}$ with high probability, which implies that the  total variation vanishes for sufficiently large $N$, i.e.~$n~\ll~N^{\frac{2s}{2s+d}}$.
\end{eg}
Related results in the literature can be leveraged to derive analogous guarantees for alternative estimators of $r_\star$ across different settings. Broadly speaking, theoretical insights emphasize the importance of choosing the estimation sample size $N$ to be significantly larger than the testing sample size $n-N$ in order to keep the excess type~I error low. More broadly, Proposition~\ref{prop:robustness} together with Example~\ref{ex:DRE} shows that our method is a viable approach to conditional two-sample testing. Its key feature is permutation-based calibration, so any excess type~I error arises solely from density ratio estimation. Although our non-uniform permutation calibration is novel, the use of marginal density ratios parallels \citet{kim2024conditional}; their classifier-based and linear-time maximum mean discrepancy procedures, reweighed by an estimate of marginal importance, are asymptotically valid, and align with the spirit of Section~\ref{sec:family_shift}. The problem is well studied, with numerous alternatives \citep[][]{Hu2020conformalJASA, chattejee24,yan2024distancekernel,huangBiometrics, zheng2025generativeconditionaldistributionequality, chenLei24}; we compare Algorithm~\ref{alg:star_algo} with \citet{kim2024conditional} and several of these methods in Section~\ref{sec:simul} and Appendix~\ref{appendix:extraSimul}.

\section{Supplementary experiments and setup}\label{appendix:extraSimulSec}
\subsection{Additional simulation studies}\label{appendix:extraSimul}
We present additional simulation studies beyond those in Section~\ref{sec:simul}. First, we consider a univariate setting and compare it to the bivariate one shown in Fig.~\ref{fig:combined!}(a). We set $P_f = \mathcal{N}(0,1)$ and define $P_g = (1+\eta)^{-1} \mathcal{N}(0, 1/9) + \eta (1+\eta)^{-1} \operatorname{Exp(1)}$ for $\eta \in \{0, \ldots, 0.40\}$. For $\eta = 0$, $g$ satisfies the null hypothesis with $r(x) = e^{-4x^2}$, while larger values of $\eta$ correspond to greater departures from the null. The results are presented in Figure \ref{fig:appendix1D}. The purple curve corresponds to (E1), and we compare it against (E2), shown in green. Furthermore, we also consider a similar setting (shown by the orange and blue lines) in which $P_f = \mathcal{N}(0,1)$ and $P_g = (1+\eta)^{-1} \mathcal{N}(0, 1/3) + \eta (1+\eta)^{-1} \operatorname{Exp(1)}$ for $\eta \in \{0, \ldots, 0.40\}$, which means that the null hypothesis is satisfied for $r^\prime(x) = e^{-x^2}$. Across all four simulations we fix $n=m=150$. For each setting, each test is run $500$ times, and the mean average decision is reported as the estimated power function. Compared with Fig.~\ref{fig:combined!}(a), the four curves have approximately the same shape. Note, however, that here we consider $\eta \in \{0,\ldots,0.40\}$ rather than $\{0,\ldots,0.80\}$. Although this is not a general rule, it is interesting that with twice the dimensionality we recover the same qualitative behaviours by considering alternatives that are twice as far apart.

\begin{figure}[!ht]
    \centering
    \includegraphics[width=0.35\linewidth]{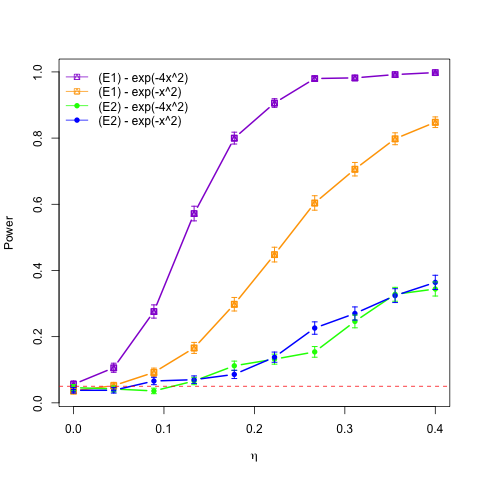}
    \caption{Purple: simulation of (E1); here $r(x) = e^{-4x^2}$. Orange: same setting but with $r^\prime(x) = e^{-x^2}$. Green and blue: alternative approaches based on (E2). Error bars show $\pm 1$ standard errors.}
    \label{fig:appendix1D}
\end{figure}

Furthermore, in Table~\ref{table:diamonds} we include a more comprehensive version of Table~\ref{table:diamonds_main}. Using the same experimental design as in Section~\ref{sec:simul_diamonds} for the diamonds dataset, we compare our method against a broader set of alternatives. As before, we report the performance of the single-split classifier-based test and the linear-time maximum mean discrepancy test from \cite{kim2024conditional}, but also include their cross-fit counterparts. Additionally, we include the conformal prediction test based on conformity scores \citep{Hu2020conformalJASA}, and the debiased conformal prediction test, which enhances the former through Neyman orthogonality and cross-fitting \citep{chenLei24}. The results are shown in Table~\ref{table:diamonds} and agree with Fig.~4 in \citet[][]{kim2024conditional}. The cross-fitted variants exhibit similar size and higher power if compared to their non–cross-fitted counterparts presented in the main body. The conformal prediction test attains the highest power in the table but shows slightly inflated type~I error. This likely reflects the added sensitivity of conformal-prediction-based ranks together with error from estimating density ratios and other nuisance functions when no orthogonalization is used.  By contrast, the debiased conformal prediction test achieves very high power, often exceeding our method, while maintaining conservative type~I error. This pattern aligns with its Neyman-orthogonal, cross-fitted construction, which mitigates first-order estimation bias.

\begin{table}[!ht]
\centering
\caption{Simulation results for the conditional two-sample testing problem on the diamonds dataset. LL stands for linear logistic regression, KLR for kernel logistic regression. (E1) is used, and compared against several alternative methods. CLF stands for the single-split classifier test, MMD-$l$ for the linear-time maximum mean discrepancy test. $\dagger$CLF and $\dagger$MMD-$l$ indicate their cross-fitted counterpart. CP indicates the conformal prediction test, and DCP its debiased version. $N$ denotes the total sampling budget.}
\label{table:diamonds}
\scalebox{0.9}{%
\begin{tabular}{lllrrrrrr}
\toprule
\centering
Estimator & Hypothesis & Test & $N = 200$ & $N = 400$ & N$= 800$ & $N = 1200$ & $N= 1600$ & $N =2000$ \\
\midrule
LL & Null & CLF & 0.0900 & 0.0650 & 0.0750 & 0.0450 & 0.0675 & 0.0425 \\
LL & Null & CP & 0.0650 & 0.0875 & 0.0925 & 0.0575 & 0.1100 & 0.0925 \\
LL & Null & $\dagger$CLF & 0.0950 & 0.0675 & 0.0850 & 0.0750 & 0.0450 & 0.0675 \\
LL & Null & $\dagger$MMD-$l$ & 0.0700 & 0.0700 & 0.0675 & 0.0550 & 0.0750 & 0.0625 \\
LL & Null & MMD-$l$ & 0.0750 & 0.0650 & 0.0750 & 0.0500 & 0.0575 & 0.0575 \\
LL & Null & DCP & 0.0375 & 0.0400 & 0.0350 & 0.0425 & 0.0325 & 0.0400 \\
LL & Null & DRPT (E1) & 0.0600 & 0.0550 & 0.0850 & 0.0600 & 0.0500 & 0.0350 \\
\midrule
LL & Alternative & CLF & 0.1575 & 0.2050 & 0.2650 & 0.3600 & 0.3925 & 0.4800 \\
LL & Alternative & CP & 0.2950 & 0.5275 & 0.6900 & 0.8700 & 0.9050 & 0.9300 \\
LL & Alternative & $\dagger$CLF & 0.2425 & 0.3675 & 0.4700 & 0.6225 & 0.6575 & 0.7575 \\
LL & Alternative & $\dagger$MMD-$l$ & 0.0975 & 0.1100 & 0.0900 & 0.1075 & 0.1125 & 0.1275 \\
LL & Alternative & MMD-$l$ & 0.0650 & 0.0675 & 0.0850 & 0.0825 & 0.0975 & 0.0850 \\
LL & Alternative & DCP & 0.1750 & 0.4150 & 0.6750 & 0.7925 & 0.8250 & 0.7950 \\
LL & Alternative & DRPT (E1) & 0.1500 & 0.1800 & 0.4150 & 0.5450 & 0.6350 & 0.7150 \\
\midrule
KLR & Null & CLF & 0.0750 & 0.0575 & 0.0675 & 0.0350 & 0.0475 & 0.0450 \\
KLR & Null & CP & 0.0450 & 0.0675 & 0.0575 & 0.0400 & 0.0675 & 0.0500 \\
KLR & Null & $\dagger$CLF & 0.0825 & 0.0375 & 0.0475 & 0.0400 & 0.0425 & 0.0500 \\
KLR & Null & $\dagger$MMD-$l$ & 0.0200 & 0.0450 & 0.0675 & 0.0650 & 0.0550 & 0.0550 \\
KLR & Null & MMD-$l$ & 0.0600 & 0.0625 & 0.0725 & 0.0575 & 0.0550 & 0.0600 \\
KLR & Null & DCP & 0.0275 & 0.0150 & 0.0275 & 0.0175 & 0.0275 & 0.0200 \\
KLR & Null & DRPT (E1) & 0.0400 & 0.0350 & 0.0600 & 0.0500 & 0.0600 & 0.0300 \\
\midrule
KLR & Alternative & CLF & 0.0975 & 0.1525 & 0.2600 & 0.3550 & 0.3675 & 0.4450 \\
KLR & Alternative & CP & 0.3100 & 0.5450 & 0.7450 & 0.9025 & 0.9600 & 0.9725 \\
KLR & Alternative & $\dagger$CLF & 0.1675 & 0.2650 & 0.3900 & 0.5750 & 0.6275 & 0.7150 \\
KLR & Alternative & $\dagger$MMD-$l$ & 0.0700 & 0.0625 & 0.0950 & 0.1000 & 0.1225 & 0.1425 \\
KLR & Alternative & MMD-$l$ & 0.0725 & 0.0675 & 0.0800 & 0.0900 & 0.0950 & 0.1050 \\
KLR & Alternative & DCP & 0.2425 & 0.4325 & 0.6850 & 0.8175 & 0.9200 & 0.9700 \\
KLR & Alternative & DRPT (E1) & 0.1250 & 0.1750 & 0.3750 & 0.5000 & 0.6300 & 0.6650 \\
\midrule
\bottomrule
\end{tabular}
}
\end{table}

\subsection{Training setup for neural ratio estimators}\label{appendix:trainingNRE}
We briefly describe the training of the neural density–ratio estimator $\hat r_{\mathrm{NRE}}$ and its balanced variant $\hat r_{\mathrm{BNRE}}$ used in Section~\ref{sec:simulation based inference}. We use the respective losses as implemented in the Python package \texttt{LAMPE}. The estimators take the concatenated input $(\theta_1, \theta_2 ,x_1, x_2)\in\mathbb{R}^4$ and output an estimate of $\log r(\theta_1, \theta_2 ,x_1, x_2)$ via a multilayer perceptron with five hidden layers, each of width $128$, and exponential linear unit activations. For each training budget $N_{\text{train}}\in\{2^4,2^7,2^9,2^{12},2^{15}\}$, we draw a fixed set of $N_{\text{train}}$ joint pairs once, shuffle each epoch, and optimize with \texttt{AdamW}, with learning rate $10^{-3}$ and gradient clipping set to $1$, for $400$ epochs using batch size $\min (256, N_{\text{train}})$.

\section{Algorithm~\ref{alg:star_algo} for discrete data}\label{appenidix:discreteDRPT}
\subsection{Methodology}\label{appenidix:discreteMethodology}
One significant drawback of permutation tests is their computational cost, which in our setting stems from choosing large values for $H$ in \eqref{eq:pvalue_DRPT} and $S$ in Algorithm~\ref{alg:star_algo}. In this subsection, we introduce an alternative method for discrete data with finite support that is computationally more efficient, since it enables direct sampling from \eqref{eq:sampling_1} without relying on the Markov Chain Monte Carlo sampler. Let $\mathcal{X} = \{0, \ldots, J\} =: \mathcal{J}$ with $J \geq 1$, and let $X_i \overset{\mathrm{i.i.d.}}{\sim} f$ and $Y_i \overset{\mathrm{i.i.d.}}{\sim} g$, where $\mathbb{P}_f\{X = j\} = f_j  $ and $\mathbb{P}_g\{Y = j\} = g_j,$
for all $j \in \mathcal{J}$, with $\sum_{j=0}^J f_j = \sum_{j=0}^J g_j= 1$. Given a sequence of positive numbers $(r_0, \ldots, r_J)$, we can write $r(x) = \sum_{j \in \mathcal{J}} r_j \mathbbm{1}\{x = j\}$, so that the testing problem becomes 
\[
H_0: g_j \propto r_j f_j \quad \text{for all } j \in \mathcal{J}. 
\]
At the sample level, the methodology relies on generating a permutation $\sigma$ given $Z$ according to the distribution~\eqref{eq:sampling_1}. Now, conditioned on the data, each permutation preserves the total number of units in each of the $J+1$ categories, where the count for category $j \in \mathcal{J}$ is denoted by $\mathrm{tot}_j$. What changes is how these values are split between the first $n$ and the last $m$ data points. We can represent this using the following table: 
\begin{center}
    \begin{tabular}{c|cc|c}
& $Z_{\sigma(1:n)}$ & $Z_{\sigma(n+1:n+m)}$ & $+$ \\
 \hline
0 & $\mathrm{tot}_0 - N_{Y,0}^\sigma$ & $N_{Y,0}^\sigma$ & $\mathrm{tot}_0$ \\
\vdots & \vdots & \vdots & \vdots \\
$J$ & $\mathrm{tot}_J - N_{Y,J}^\sigma$ & $N_{Y,J}^\sigma$ & $\mathrm{tot}_J$ \\
 \hline
$+$ & $n$ & $m$ & $m+n$ \\
    \end{tabular}
    \end{center}
where $N_{Y,j}^\sigma = \# \{i \notin [n] : Z_{\sigma(i)} = j \} = \mathbbm{1}\{Z_{\sigma(n+1)}  = j\} + \ldots + \mathbbm{1}\{Z_{\sigma(n+m)}  = j\}$ for all $j \in \mathcal{J}$. If we restrict attention to test statistics that are functions of this frequency table, which is natural given that it is a sufficient statistic in our model, we can characterize~\eqref{eq:sampling_1} through the distribution of $(N_{Y,0}^\sigma, \ldots, N_{Y,J}^\sigma) \mid Z$. In this regard, given the data and for $w=(w_0, \ldots, w_J)$ such that $\sum_{j = 0}^J w_j = m$ and $\max(0, \mathrm{tot}_j - n) \leq w_j \leq \min(m , \mathrm{tot}_j)$ for all $j \in \mathcal{J}$, we have
\begin{align}\label{eq:cond_distr_NY}
    \mathbb{P}&\{N_{Y,0}^\sigma = w_0, \ldots, N_{Y,J}^\sigma = w_J \mid Z \} = \sum_{\{\sigma \in \mathcal{S}_{n+m} : (N_{Y,0}^\sigma, \ldots, N_{Y,J}^\sigma)=w\}} \mathbb{P}\{\sigma^{(1)} = \sigma \mid Z \} \nonumber \\
    & \propto \sum_{\{\sigma \in \mathcal{S}_{n+m} : (N_{Y,0}^\sigma, \ldots, N_{Y,J}^\sigma)=w\}} \prod_{i \notin [n]}  r(Z_{\sigma(i)}) = \sum_{\{\sigma \in \mathcal{S}_{n+m} : (N_{Y,0}^\sigma, \ldots, N_{Y,J}^\sigma)=w\}} \prod_{j \in \mathcal{J}} r_j^{N_{Y,j}^\sigma} \nonumber \\
    & = \bigg(\prod_{j \in \mathcal{J}} r_j^{w_j} \bigg)\#\{\sigma \in \mathcal{S}_{n+m} : (N_{Y,0}^\sigma, \ldots, N_{Y,J}^\sigma)=w \} =  n! m!  \prod_{j \in \mathcal{J}} r_j^{w_j} \binom{\mathrm{tot}_j}{w_j} \propto  \prod_{j \in \mathcal{J}} r_j^{w_j} \binom{\mathrm{tot}_j}{w_j}.
\end{align} 
Notice that in the penultimate step we used the fact that  $\#\{\sigma \in \mathcal{S}_{n+m} : (N_{Y,0}^\sigma, \ldots, N_{Y,J}^\sigma)=w \} =  n! m! \prod_{j \in \mathcal{J}} \binom{\mathrm{tot}_j}{w_j}$, as there are $\binom{\mathrm{tot}_j}{w_j}$ ways of choosing $w_j$ many $j$'s for the last $m$ data points for all $j \in \mathcal{J}$. Considering all the possible ways in which we can further permute the first $n$ and last $m$ values gives the extra factor of $n!m!$. This shows that $(N_{Y,0}^\sigma, \ldots, N_{Y,J}^\sigma) \mid Z$ is distributed according to Fisher's multivariate noncentral hypergeometric distribution, which is a generalization of the hypergeometric distribution where sampling probabilities are adjusted by weight factors \citep[see, e.g.][Section~7]{mccullagh1989generalized}. Coming back to the testing problem \eqref{eq:testing_problem}, the previous argument shows that in the case of discrete data with finite support we can avoid sampling permutations from \eqref{eq:sampling_1}, as it is sufficient to sample $(N_{Y,0}^\sigma, \ldots, N_{Y,J}^\sigma) \mid Z$ from \eqref{eq:cond_distr_NY}. This can be done efficiently using for example the R-function \texttt{rMFNCHypergeo} from R-package \texttt{BiasedUrn}.

Compared to the Markov Chain Monte Carlo-based approach, aside from reducing computational runtime, Algorithm~\ref{alg:star_algo} copies of the table above generated through independent and identically distributed~draws from~\eqref{eq:cond_distr_NY} are conditionally independent, given the row totals, across different $h \in [H]$, while the $Z^{(h)}$'s generated by Algorithm~\ref{alg:star_algo} exhibit non-zero correlation due to their shared initialization. Nonetheless, this dependence decreases for larger and larger values of $S$. Regarding its theoretical guarantees, finite-sample validity is ensured by Theorem~\ref{prop:DRPT_validity}. 

\subsection{Power analysis}\label{appendix:powerDiscreteDRPT}
We now provide some further insights for the this version of Algorithm~\ref{alg:star_algo}, henceforth referred to as the discrete density ratio permutation test. For the case $\mathcal{X} = \{0,1\}$, i.e.~$J=1$, we recall that the testing problem \eqref{eq:testing_problem} is equivalent to  
\[
     H_0: \frac{g_1}{g_0} = \frac{r_1 f_1}{r_0 f_0},
\]
for $r_0, r_1 > 0$. As a result, we may assume without loss of generality that $r_0 = 1$, since our interest lies solely in the ratio $r_1/r_0$, and $r_1 \equiv r \geq 1$; if this condition does not hold, we can simply switch the roles of $f$ and $g$ and consider $1/r$ instead. Now, it is instructive to analyze the behaviour of the sample mean of the permuted data. In this regard, Lemma \ref{lemma:convergence_in_d} implies the following unconditional result. In what follows, with a slight abuse of notation, whenever we write an expression such as $(\tau f_1 + g_1)/(\tau + \lambda r)$ with $\tau = \infty$, it is to be interpreted as its limit as $\tau \to \infty$; for example, $\left.(\tau f_1 + g_1)/(\tau + \lambda r)\right|_{\tau = \infty} := \lim_{\tau \to \infty} (\tau f_1 + g_1)/(\tau + \lambda r) = f_1$.
\begin{cor}\label{lemma:mean_Nsigma}
Let $N_{Y,1}^\sigma = \sum_{j = n+1}^{n+m}Z_{\sigma(j)}$, where $\sigma$ is sampled according to \eqref{eq:sampling_1}. Calling $\tau := \lim_{n,m \to \infty} n/m  \in \mathbb{R}_+ \cup \{0,\infty\}$, for all $k \in \mathbb{N}$ we have $\mathbb{E}[(m^{-1}N_{Y,1}^\sigma)^k] \rightarrow \gamma_1^k$ as $n,
        m \rightarrow \infty$, where
        \begin{equation}\label{eq:gamma_1}
            \gamma_1(f_1, g_1, r, \tau) \equiv \gamma_1 = 
        \begin{cases}
            \frac{\tau f_1 + g_1}{\tau + 1} + \frac{(r-1)\frac{\tau-1}{\tau+1}(\tau f_1 + g_1)+ \tau+r-\sqrt{(\tau+r +(r-1)(\tau f_1+g_1))^2-4 (r-1) r (\tau f_1+g_1)}}{2(r-1)} \quad \text{ if } r > 1, \\
            \frac{\tau f_1 + g_1}{\tau + 1} \text{ if } r = 1.
        \end{cases}
        \end{equation}
\end{cor}

\begin{proof}
Let $\mu$ be the counting measure and equip the set $\{0,1\}$ with the discrete topology. We can show as a corollary of  Lemma~\ref{lemma:convergence_in_d} in Subsection~\ref{sec:consistency} that 
\[
\frac{1}{m} \sum_{j = n+1}^{n+m} \varphi(Z_{\sigma(j)})  \overset{\mathbb{P}}{\rightarrow} \int \varphi \, \lambda_\infty r_\mathrm{func} \frac{\tau f + g}{\tau + \lambda_\infty r_\mathrm{func}} d\mu
\]
for every bounded and continuous function $\varphi$, where $r_\mathrm{func}(x) = r\,\mathbbm{1}\{x = 1\} + \mathbbm{1}\{x = 0\}$ and $\lambda_\infty > 0 $ is the positive solution of 
\[
1 = \int \frac{\tau f + g}{\tau + \lambda_\infty r_\mathrm{func}} d\mu = \frac{\tau f_1 + g_1}{\tau + \lambda_\infty r} + \frac{\tau f_0 + g_0}{\tau + \lambda_\infty} = \frac{\tau f_1 + g_1}{\tau + \lambda_\infty r} - \frac{\tau f_1 + g_1}{\tau + \lambda_\infty} + \frac{\tau + 1}{\tau + \lambda_\infty}.
\]
As $\mathcal{X} = \{0,1\}$, we can choose $\varphi = \operatorname{id}$ and obtain that $m^{-1}N_{Y,1}^\sigma \overset{\mathbb{P}}{\rightarrow} \lambda_\infty r \, (\tau f_1 + g_1)/(\tau + \lambda_\infty r) = \gamma_1$, where in the last step we simply plugged in the expression for $\lambda_\infty$, which can be found explicitly by solving the equation above. We have thus established that $m^{-1}N_{Y,1}^{\sigma} \xrightarrow{\mathbb{P}} \gamma_1$, which implies the existence of a subsequence along which convergence holds almost surely; applying the dominated convergence theorem yields $\mathbb{E}[(m^{-1}N_{Y,1}^{\sigma})^k] \to \gamma_1^k$ for all $k \in \mathbb{N}$ along this subsequence, and by the uniqueness of limits, the convergence extends to the entire sequence. This completes the proof.
\end{proof}

This result shows that $m^{-1}N_{Y,1}^\sigma \overset{\mathbb{P}}{\rightarrow} \gamma_1$, which offers a more explicit interpretation of the result in Lemma~\ref{lemma:convergence_in_d} and emphasizes the intricate dependence of the limiting distribution of the permuted data on the initial parameters $(f_1$, $g_1$, $r, \tau)$, even in simple cases. Interestingly, under the null hypothesis, $\gamma_1 = g_1$, as expected. Furthermore, if $r = 1$, we find that $\gamma_1 = (\tau f_1 + g_1)/(\tau + 1)$, reflecting the fact that Algorithm~\ref{alg:star_algo} selects permutations uniformly at random when $r = 1$, consistent with \eqref{eq:stdPermutation_is_consistent}. Similarly, we can establish an analogous result for the convergence of $n^{-1}\sum_{i = 1}^{n}Z_{\sigma(i)}$ to $\nu_1$, where $\nu_1$ satisfies $\frac{\tau}{1+\tau} f_1 + \frac{1}{1+\tau} g_1 = \frac{\tau}{1+\tau} \nu_1 + \frac{1}{1+\tau} \gamma_1$, by leveraging the constraint that the total number of ones must remain conserved. \\

As for the analysis of power, Theorem~\ref{thm:general_consistency} already guarantees that the discrete density ratio permutation test is consistent when the test statistic is an integral probability metric that depends on $(N_{Y,0}^\sigma, \ldots, N_{Y,J}^\sigma)$. We now establish a further consistency result for an alternative choice of test statistic. In this setting, the null hypothesis \eqref{eq:testing_problem} is equivalent to 
\[
H_0: \frac{g_j}{g_0} = \frac{r_j f_j}{r_0 f_0} \quad \text{ for all } j \in [J],
\]
for fixed $r = (r_0, r_1, \ldots, r_J) \in R_+^{J+1}$. As before, we can assume $r_0 = 1$ without loss of generality. This motivates the introduction of 
\begin{equation}\label{eq:Tdiscrete_theory}
    T(Z_\sigma) = \frac{1}{nm} \sum_{j = 1}^J \left| r_j^{-1/2} N_{Y,j}^\sigma(\mathrm{tot}_0 - N_{Y,0}^\sigma) - r_j^{1/2}  N_{Y,0}^\sigma (\mathrm{tot}_j -N_{Y,j}^\sigma)\right|,
\end{equation}
which serves as an estimator of $D(f,g) \equiv D^r(f,g) := \sum_{j \in [J]}\left|r_j^{-1/2} g_j f_0 - r_j^{1/2} f_j g_0 \right|$, which is a population measure of discrepancy that characterizes the null. Similarly to $T_{\mathcal F, r}$ introduced in the main body, this quantity is invariant under taking the reciprocal transformations of $r$ and relabelling the sample. We can prove the following result.
\begin{prop}\label{prop:consistency_discrete}
Let $\mathcal{X}=\{0,\ldots, J\} =: \mathcal{J}$ with $J \geq 1$, $r = (r_0 = 1, r_1, \ldots, r_J) \in \mathbb R_+^{J+1}$, and define $H_0: g_j \propto r_j f_j \text{ for all } j \in \mathcal{J}$. Provided that $H > \lceil 1/\alpha - 1\rceil$ for fixed $\alpha \in (0,1)$, the discrete density ratio permutation test using \eqref{eq:Tdiscrete_theory} as its test statistic is consistent for $H_0$. 
\end{prop}
\begin{proof}
    Let $\mu$ be the counting measure, and equip the set $\mathcal{J}$ with the discrete topology. Then, arguing as in the proof of Corollary~\ref{lemma:mean_Nsigma}, we can use Lemma \ref{lemma:convergence_in_d} to show that $m^{-1} N_{Y,j}^\sigma = m^{-1} \sum_{i = n+1}^{n+m} \mathbbm{1}\{Z_{\sigma(i)} = j \}  \overset{\mathbb{P}}{\rightarrow} \lambda_\infty r_j \, (\tau f_j + g_j)/ (\tau + \lambda_\infty r_j)$, where $\lambda_\infty$ is a positive solution of $\sum_{k = 0}^J (\tau f_k + g_k)/(\tau + \lambda_\infty r_k)  = 1$. As a result, the generic $j$-th term of the test statistic \eqref{eq:Tdiscrete_theory} satisfies
\begin{align*}
    \frac{1}{nm}&\left\{r_j^{-1/2}N_{Y,j}^\sigma(\mathrm{tot}_0 - N_{Y,0}^\sigma) - r_j^{1/2}  N_{Y,0}^\sigma (\mathrm{tot}_j -N_{Y,j}^\sigma)\right\} \\
    & 
    \overset{\mathbb{P}}{\rightarrow}  r_j^{-1/2} \, \lambda_\infty r_j \frac{\tau f_j + g_j}{\tau + \lambda_\infty r_j} \frac{\tau f_0 + g_0}{\tau + \lambda_\infty r_0} - r_j^{1/2} \, \lambda_\infty r_0 \frac{\tau f_0 + g_0}{\tau + \lambda_\infty r_0} \frac{\tau f_j + g_j}{\tau + \lambda_\infty r_j} = 0,
\end{align*}
since $r_0 = 1$. This shows that $T(Z_\sigma) \overset{\mathbb{P}}{\rightarrow} 0$. 

As for the behaviour of our test statistic under for the unpermuted sample, we can use the weak law of large numbers to show that $T(Z) \overset{\mathbb{P}}{\rightarrow} D(f,g) > 0$ under the alternative. Since we assumed  $H > \lceil 1/\alpha - 1\rceil$, we have $\alpha(1+H)-1 > 0 $ and we can therefore apply Markov's inequality to see that
 \begin{align*}
    \mathbb{P}&\{p > \
    \alpha \}  =  \mathbb{P}\left\{1+\sum_{h=1}^H \mathbbm{1}\{T(Z_{\sigma^{(h)}}) \geq T(Z)\} > \alpha(1+H) \right\} \\
     & \leq \frac{\mathbb{E}\left[ \sum_{h=1}^H \mathbbm{1}\{T(Z_{\sigma^{(h)}}) \geq T(Z)\} \right]}{\alpha(1+H)-1} = \frac{H}{\alpha(1+H)-1}\mathbb{P}\{T(Z_{\sigma^{(1)}})\geq T(Z) \} \\
     & \leq \frac{H}{\alpha(1+H)-1} \left(\mathbb{P}\{T(Z_{\sigma^{(1)}})\geq \tfrac{1}{2}D(f,g) \} + \mathbb{P}\{T(Z)\leq \tfrac{1}{2} D(f,g) \}\right) \to 0,
 \end{align*}
 where the penultimate step follows from exchangeability. This implies $\mathbb{P}\{p \leq \
    \alpha \} \to 1$ and concludes the proof.
\end{proof}

Finally, coming back again to the case of binary data, the dependence of the minimax separation on $r$ can be analyzed more effectively compared to Theorem~ \ref{thm:LB_minimax}. In this regard, let $\mathcal{X}  = \{ 0,1\}$ and $r \geq 1$, and and consider the measure of discrepancy defined above, i.e.~$D(f,g) = \left|r^{-1/2} g_1 f_0 - r^{1/2} f_1 g_0 \right|$. For fixed $r \geq 1$ and $\rho > 0$, consider
\[
H_0: \frac{g_1}{g_0} = r \frac{ f_1}{ f_0} \quad \text{ vs. } \quad H_1^r(\rho): D(f,g) \geq \rho.
\]
Write $\Psi$ for the set of all tests, that is randomized functions of $(X_1, \ldots, X_n, Y_1, \ldots, Y_m)$, and $\Psi(\alpha)$ for the set of tests of size $\alpha$, with $\alpha \in (0,1)$. For $\beta \in (0,1 - \alpha)$, we may define the minimax separation as $\rho^*_r \equiv \rho_r^*(n, m, \alpha, \beta) := \inf\left\{\rho > 0:   \inf_{\varphi \in \Psi(\alpha)} \sup_{(f,g) \in H_1^r(\rho)} \mathbb{E}_P (1-\varphi) \leq \beta \right\},$
where $P = P_{f}^{\otimes n} \otimes P_{g}^{\otimes m}$. We now prove a lower bound on $\rho^*_r$.
\begin{prop}\label{prop:binary_LB}
Let $\rho^*_r$ the minimax separation defined above and suppose that $\alpha + \beta < 1/2$.  We have that $\rho^*_r(n, m, \alpha, \beta) \geq \sqrt{\frac{r}{(n \wedge m)(1+r)^2}\{1-2(\alpha + \beta)\}}$.
\end{prop}

 \begin{proof}
     For $ 0 < \rho^2 \leq \frac{r}{(1+r)^2}$, consider \[
     (f^{(0)}, g^{(0)}) = \left(\left(\frac{1}{2}, \frac{1}{2} \right), \left(\frac{1}{1+r}, \frac{r}{1+r} \right) \right) \quad \text{ and } \quad (f^{(1)}, g^{(1)}) = \left(\left(\frac{1-\gamma}{2}, \frac{1+\gamma}{2} \right), \left(\frac{1}{1+r}, \frac{r}{1+r} \right) \right),
     \]
     with $\gamma = \frac{1+r}{\sqrt{r}}\rho$. Observe that $D(f^{(0)}, g^{(0)}) = 0$ and $D(f^{(1)}, g^{(1)}) = \rho$. We will use the well-known fact that the squared Hellinger distance between two discrete probability distributions $p, q$ supported on $[J]$ is given by $\operatorname{H}^2(p,q) = \sum_{i \in [J]}(\sqrt{p_i}- \sqrt{q_i})^2$. Writing $P_{f^{(i)}}$ for the Bernoulli distribution with parameter $f^{(i)}$ for $i \in \{0,1\}$, we thus have 
\begin{align*}
    \operatorname{H^2}(P_{f^{(0)}}, P_{f^{(1)}}) & = \frac{1}{2} \left\{\left(1  - \sqrt{1 - \gamma} \right)^2 +  \left(1  - \sqrt{1 + \gamma} \right)^2 \right\} \leq \gamma^2 = \frac{(1+r)^2}{r}\rho^2,
\end{align*}
where the last inequality relies on $(1-\sqrt{1 \pm x})^2 \leq x^2$ , and further shows that $\operatorname{H^2}(P_{f^{(0)}}, P_{f^{(1)}}) \leq 1$ since $\rho^2 \leq \frac{r}{(1+r)^2}$. We can then bound the minimax risk using a standard Le-Cam two-point argument as
     \begin{align*}
         \alpha & + \sup_{(f,g) \in H_1^r(\rho)} \mathbb{E}_P (1-\varphi ) \geq \sup_{(f,g) \in H_0} \mathbb{E}_P \, \varphi + \sup_{(f,g) \in H_1^r(\rho)} \mathbb{E}_P (1-\varphi ) \\
         & \geq 1- \operatorname{TV}\left(P_{f^{(0)}}^{\otimes n} \otimes P_{g^{(0)}}^{\otimes m}, P_{f^{(1)}}^{\otimes n} \otimes P_{g^{(1)}}^{\otimes m} \right) \geq 1 - \left\{\operatorname{TV}\left(P_{f^{(0)}}^{\otimes n}, P_{f^{(1)}}^{\otimes n} \right) + \operatorname{TV}\left(P_{g^{(0)}}^{\otimes m}, P_{g^{(1)}}^{\otimes m} \right) \right\} \\
         & = 1 - \operatorname{TV}\left(P_{f^{(0)}}^{\otimes n}, P_{f^{(1)}}^{\otimes n} \right) \geq \frac{1}{2}\left(1 - \frac{1}{2}\operatorname{H^2}(P_{f^{(0)}}, P_{f^{(1)}})\right)^{2n} \geq \frac{1}{2}\left(1-n \operatorname{H^2}(P_{f^{(0)}}, P_{f^{(1)}}) \right) \\
         & \geq \frac{1}{2} \left(1  - \frac{n(1+r)^2}{r}\rho^2 \right),
     \end{align*}
     where in the fifth inequality we used the fact that $(1 - x)^n \geq 1 - n x$ for $n \in \mathbb{N}$ and $x \leq 1$. The last display is lower bounded by $\alpha+\beta$ if and only if $\rho^2 \leq \frac{r}{n(1+r)^2}\{1-2(\alpha + \beta)\}$. Note that $0 < \rho^2 \leq \frac{r}{(1+r)^2}$ is necessarily satisfied since $n \geq 1$ and $0 \leq \alpha + \beta < 1/2$. A similar construction, where $f^{(0)} = f^{(1)}$ and $g^{(1)}$ is a perturbation of $g^{(0)}$, gives an analogous lower bound with $m$ in place of $n$, and concludes the proof.
 \end{proof}
Proposition \ref{prop:binary_LB}  suggests that the testing problem is the hardest when $r=1$. As already mentioned in Subsection~\ref{sec:optimality}, this is in accordance with the goodness-of-fit testing problem, where the goal is to test the null hypothesis $f = f_0$ for a fixed density $f_0$, based on independent and identically distributed~samples $X_1, \ldots, X_n \sim f$. In this setting, the minimax separation rate depends on the choice of $f_0$, and it has been shown that the problem is hardest when $f_0$ is the uniform distribution \citep[e.g.][]{balakrishnan2019hypothesis}. In complete analogy, and according to the simulation results in Section \ref{sec:simulSynthetic}, Proposition \ref{prop:binary_LB} seems to indicate that $r = 1$ corresponds to the harder testing problem. In other words, more extreme shifts should be easier to detect. We validate this conjecture through simulations on synthetic data in Appendix~\ref{appendix:SimulDiscereteDRPT}.

Establishing the optimality of the rate $1/\sqrt{r (n \wedge m)}$ is more delicate. While the two-moment method used in the proof of Theorem \ref{thm:UB_minimax} suffices to derive an upper bound on $\rho^*_r$, it yields a loose dependence on $r$, even though it accurately captures the scaling with $n$ and $m$. Specifically, we can show that $\operatorname{Var}(T(Z)) \lesssim r/n$ and $\operatorname{Var}(T(Z_\sigma)) \lesssim r^2/n$. The first inequality follows from the independence between the $X$'s and the $Y$'s, while the second relies on the following proposition, which we include for completeness.

\begin{prop}\label{lemma:binary_rates}
    Assume $r \geq 1$ and $\mathcal{X} = \{ 0,1\}$. Define $N_{Y,1}^\sigma = \sum_{i = n+1}^{n+m}Z_{\sigma(i)}$, where $\sigma$ is sampled according to \eqref{eq:sampling_1}, and let $\gamma_1$ be as in \eqref{eq:gamma_1}. Then
    \[
    \mathbb{E}[(N_{Y,1}^\sigma - m \gamma_1)^2] \leq (1+r) n \wedge m + 2f_1(1-f_1)n + 2g_1(1-g_1)m.
    \]
\end{prop}
 \begin{proof}
     All the expectations are to be intended conditionally to $Z$. We will assume $n = \tau m$ throughout the proof to simplify the computations. Write $\mathcal{T} = \mathrm{tot_1}$, and define
     \[
     p_1^x = \frac{(n-\mathcal{T}+x)x}{(1+r)nm} \qquad\text{ and } \qquad p_2^x = \frac{r(\mathcal{T}-x)(m-x)}{(1+r)nm}.
     \]
     Let $\gamma(\mathcal{T})$ be the solution to $p_1^x = p_2^x$, i.e. 
     \[
     \gamma(\mathcal{T}) = \frac{1}{2(r-1)}\left\{(r-1)\mathcal{T}+n+rm - \sqrt{[(r-1)\mathcal{T}+n-rm]^2+4rmn} \right\}.
     \]
     We will show that $\mathbb{E}[(N_{Y,1}^\sigma - \gamma(\mathcal{T}))^2] \leq \frac{1+r}{2}n \wedge m$ and $\mathbb{E}[(m\gamma_1 - \gamma(\mathcal{T}))^2] \leq nf_1(1-f_1) +mg_1(1-g_1)$, which imply the desired result since $\mathbb{E}[(N_{Y,1}^\sigma - m\gamma_1)^2] \leq 2\mathbb{E}[\{N_{Y,1}^\sigma - \gamma(\mathcal{T})\}^2] + 2\mathbb{E}[\{m\gamma_1 - \gamma(\mathcal{T})\}^2]$. 
     Now, as for the latter, observe that $m\gamma_1 = \gamma(nf_1 + mg_1)$, and that $\gamma(\cdot)$ is $1-$Lipschitz. This is due to fact that 
     \[
     |\gamma^\prime(x)| \leq \frac{1}{2}\left(1+ \frac{|(r-1)x + n -rm|}{\sqrt{[(r-1)x + n -rm]^2 + 4rnm}} \right) \leq 1.
     \]
     Hence, 
\begin{align*}
    \mathbb{E}&[(m\gamma_1 - \gamma(\mathcal{T}))^2] = \mathbb{E}[\{\gamma(nf_1 + mg_1) - \gamma(\mathcal{T})\}^2] \leq  \mathbb{E}[\{(nf_1 + mg_1) - \mathcal{T}\}^2] \\
    & = \mathbb{E}\left[\left\{(nf_1 + mg_1) - \sum_{i=1}^n X_i - \sum_{i = 1}^m Y_i\right\}^2\right]  = \operatorname{Var}\left[ \sum_{i=1}^n X_i \right] + \operatorname{Var}\left[ \sum_{i=1}^m Y_i \right] = nf_1(1-f_1) +mg_1(1-g_1).  
\end{align*}
As for the other term, start by noticing that Algorithm \ref{alg:pairwise_sampler} works the same if at every time step $t \in \mathbb{N}$ we just choose a single couple $(i,j)$ with $i \in [n]$ and $j \in \{n+1, \ldots, n+m\}$ at random, and then switch $Z_{\sigma_t(i)}$ with $Z_{\sigma_t(j)}$ with probability equal to $r^{Z_{\sigma_t(i)}}/(r^{Z_{\sigma_t(i)}}+r^{Z_{\sigma_t(j)}})$. This is not efficient from a computational point of view, but simplifies the proof, since every time step $t$ corresponds at most to one switch. Now, let $K_t$ be the sum of the last $m$ observations after $t$ steps of this simplified algorithm, and observe that it has the same distribution of $N_{Y,1}^\sigma$ for every $t \in \mathbb{N}$ when the procedure is initialized at stationarity. Thus, for $\gamma \equiv \gamma(\mathcal{T})$ to ease notation, it follows that
\begin{align*}
    \mathbb{E}&[(K_{t+1} - \gamma)^2 | K_t] - (K_t - \gamma)^2 = p_1^{K_t} + p_2^{K_t} + 2\{(p_2^{K_t} - p_2^{\gamma}) - (p_1^{K_t} - p_1^{\gamma}) \}(K_t - \gamma) \\
    &  = p_1^{K_t} + p_2^{K_t} - \frac{2(K_t -\gamma)^2}{(1+r)nm} \sqrt{[(r-1)\mathcal{T}+n-rm]^2+4rmn} + \frac{2(r-1)}{(1+r)nm}(K_t - \gamma)^3 \\
    & \leq 1 - \frac{2(K_t -\gamma)^2}{(1+r)nm} \sqrt{[(r-1)\mathcal{T}+n-rm]^2+4rmn} + \frac{2(r-1)(m-\gamma)}{(1+r)nm}(K_t - \gamma)^2 \\
    & =  1 - \frac{2(K_t -\gamma)^2}{(1+r)nm} \left\{\sqrt{[(r-1)\mathcal{T}+n-rm]^2+4rmn} - (r-1)(m-\gamma)\right\} \\
    & = 1 - \frac{(K_t -\gamma)^2}{(1+r)nm} \left\{\sqrt{[(r-1)\mathcal{T}+n-rm]^2+4rmn} + (r-1)\mathcal{T} +n -(r-2)m \right\} \\
    & \leq 1 - \frac{(K_t -\gamma)^2}{(1+r)nm} \left\{|(r-1)\mathcal{T}+n-rm| + (r-1)\mathcal{T} +n -rm + 2m \right\} \\
    & \leq 1 - \frac{2(K_t -\gamma)^2}{(1+r)n} \leq 1 - \frac{2(K_t -\gamma)^2}{(1+r)n}.
\end{align*} 
Taking expectation with respect to $K_t$ under stationarity yields
\[
0 \leq 1- \frac{2}{(1+r)n}\mathbb{E}[\{K_t - \gamma(\mathcal{T})\}^2] = 1- \frac{2}{(1+r)n}\mathbb{E}[\{N_{Y,1}^\sigma - \gamma(\mathcal{T})\}^2],
\]
implying $\mathbb{E}[\{N_{Y,1}^\sigma - \gamma(\mathcal{T})\}^2] \leq (1+r) \, n/2.$
By symmetry, we can repeat the same computations for $N_{X,1}^\sigma := \mathcal{T} - N_{Y,1}^\sigma$ and $\nu(\mathcal{T}) := \mathcal{T} - \gamma(\mathcal{T})$ and get 
$\mathbb{E}[\{N_{Y,1}^\sigma - \gamma(\mathcal{T})\}^2] = \mathbb{E}[\{N_{X,1}^\sigma-\nu(\mathcal{T})\}^2] \leq (1+r) \, m/2$, which gives $\mathbb{E}[\{N_{Y,1}^\sigma - \gamma(\mathcal{T})\}^2] \leq \frac{1+r}{2} n \wedge m$ and concludes the proof.
 \end{proof}

\subsection{Simulation studies}\label{appendix:SimulDiscereteDRPT}
We introduce some variants of the discrete density ratio permutation test, each with $S = 50$ and $H=~99$. (E4) stands for the discrete density ratio permutation test with V-statistic~\eqref{eq:V-stat} and collision kernel $k(x,y)=\sum_{j=0}^J \mathbbm{1}\{x=j\}\mathbbm{1}\{y=j\}$. (E5) corresponds to the discrete density ratio permutation test with U-statistic~\eqref{eq:U_stat} and the same kernel. Finally,  (E6) identifies the discrete density ratio permutation test with test statistic~\eqref{eq:Tdiscrete_theory}. 

We start by replicating the setup used in the proof of Proposition \ref{prop:binary_LB}, selecting
\[
(f_r,g_r) =\left(\left(\frac{1-\gamma_r}{2}, \frac{1+\gamma_r}{2} \right), \left(\frac{1}{1+r}, \frac{r}{1+r} \right) \right) \quad \text{ with } \gamma_r = \frac{1+r}{\sqrt{r}}\eta.
\]
In Fig.~\ref{fig:combined4}(a), we assess the performance of the discrete (E6) as its test statistic over $3000$ repetitions with $n = m = 500$, $r \in \{0.1, 0.5, 1, 2, 10\}$ and plotting an estimate of the power function for varying $\eta \in \{0, \ldots,0.10\}$. The results empirically support the tightness of the lower bound of Proposition \ref{prop:binary_LB},  demonstrating that the power of the discrete density ratio permutation test increases when $r$ moves further away from $1$.

\begin{figure}[!ht]
    \centering
    \includegraphics[width=0.9\linewidth]{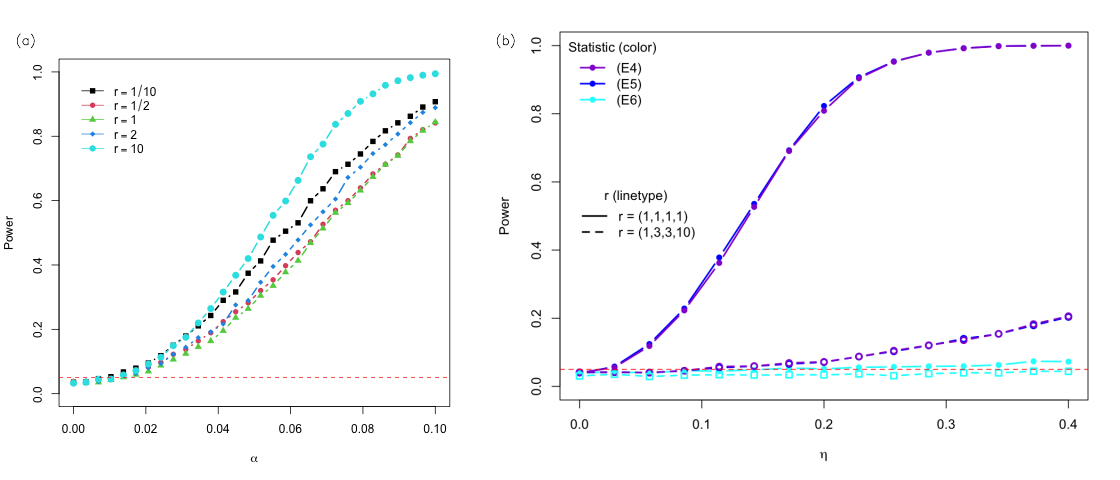}
    \caption{(a) Empirical power of (E6) in a synthetic binary data setting for varying $r \in \{0.1, 0.5, 1, 2, 10\}$. (b) Comparison of (E4), (E5), (E6) in two four-dimensional settings.}
    \label{fig:combined4}
\end{figure}

We now evaluate the performance of our reproducing kernel Hilbert space-based methodology. In this setting, when we use the kernel $k(x, y) := \sum_{j = 0}^J \mathbbm{1}\{x = j\}\,\mathbbm{1}\{y = j\}$, we can express the V-statistic~\eqref{eq:V-stat} and the U-statistic~\eqref{eq:U_stat} as
\begin{align*}
    V(x_1, \ldots, x_n, y_1, \ldots, y_m) 
    &= \sum_{j=0}^J \frac{1}{(n/m+\hat{\lambda}r_j)^2} 
    \biggl\{ \frac{\hat{\lambda}r_j}{n}\sum_{i=1}^n \mathbbm{1}_{\{x_i=j\}} 
              - \frac{1}{m} \sum_{i=1}^m \mathbbm{1}_{\{y_i=j\}} \biggr\}^2,
\end{align*}
and
\begin{align*}
    U(x_1, \ldots, x_n, y_1, \ldots, y_m) 
    &= V 
    - \sum_{j=0}^J \frac{\hat{\lambda}^2 r_j^2}{(n/m + \hat{\lambda}r_j)^2} 
      \biggl\{ \frac{1}{n^2}\sum_{i=1}^n \mathbbm{1}_{\{x_i=j\}} \biggr\}
    - \sum_{j=0}^J \frac{1}{(n/m + \hat{\lambda}r_j)^2} 
      \biggl\{ \frac{1}{m^2}\sum_{i = 1}^m \mathbbm{1}_{\{y_i=j\}} \biggr\},
\end{align*}
which shows that both $U$ and $V$ are functions of 
\((N_{Y,0}^\sigma, \ldots, N_{Y,J}^p)\), in analogy with~\eqref{eq:Tdiscrete_theory}.  In particular, for (E4) and (E5) it is sufficient to sample  \((N_{Y,0}^\sigma, \ldots, N_{Y,J}^p) \mid Z\) according to  \eqref{eq:cond_distr_NY} via the R function \texttt{rMFNCHypergeo},  without resorting to Algorithm~\ref{alg:pairwise_sampler}, which substantially  reduces the computational cost of Algorithm~\ref{alg:star_algo}. 

We apply this approach in two settings, described next. Before detailing the data-generating processes, note that our empirical findings here are consistent with Section~\ref{sec:simul}. The only practical difference is that the discrete density ratio permutation test is much faster to run, allowing substantially larger Monte Carlo sample sizes in each simulation. For the first setting, we set $J = 3$ and $n = m = 250$, where the $X_i$'s were independently drawn from a multinomial distribution over $\{0, 1, 2, 3\}$ with probabilities $p_X = (\frac{1}{8},\frac{1}{8},\frac{3}{8},\frac{3}{8})$, and the $Y_j$'s were sampled from a multinomial distribution with probabilities $p_Y = (\frac{1}{43},\frac{3}{43}, \frac{9 + 25\eta}{43}, \frac{30 - 25\eta}{43})$, with $\eta \in \{0, \ldots, 0.4\}$. When $\eta = 0$, the null hypothesis holds with $r = (1,3,3,10)$; increasing $\eta$ corresponds to larger deviations from the null. We repeated each experiment $5000$ times, comparing (E4), (E5) and (E6) in this setting. Fig.~\ref{fig:combined4}(b) presents the results in the dashed lines. Overall, the results show a notable drop in power when (E6) based on rejection sampling is used. Furthermore, the (E4) and (E5) are nearly indistinguishable, as expected, since the difference $V-U$ vanishes asymptotically as the sample size increases, while (E6) seems significantly less powerful than the reproducing kernel Hilbert space-based approach in this particular setting. For the second simulation setting, we consider a less extreme shift. Specifically, the shift vector is $r = (1, 1, 1, 1)$, so the problem reduces to the classical two-sample testing setting. The $X_i$'s are drawn independently and identically distributed~from the uniform distribution on $\{0, 1, 2, 3\}$, that is, $p_X' = \left(\frac{1}{4}, \frac{1}{4}, \frac{1}{4}, \frac{1}{4}\right)$, while the $Y_j$'s are sampled from a multinomial distribution with $p_Y^\prime = (\frac{1}{4},\frac{1}{4}, \frac{1+\gamma}{4},\frac{1-\gamma}{4})$, where $\gamma = \frac{25\eta}{43} ( \frac{1}{\sqrt{3}} + \frac{1}{\sqrt{10}})$. This specific choice of$\gamma$ ensures that $D(p_X, p_Y) = D(p_X^\prime, p_Y^\prime)$ for all values of~$\eta$, where the separation $D(f, g)$ is defined right after~\eqref{eq:Tdiscrete_theory}. The solid curve in Fig.~\ref{fig:combined4}(b) supports our earlier conjecture that near-uniform shifts are generally harder to detect.
\end{appendices}
\end{document}